\documentclass[reprint,twocolumn,superscriptaddress]{revtex4-2}
\usepackage[english]{babel}

\usepackage{amsmath,amssymb}
\usepackage{graphicx, color, graphpap}% Include figure files
\usepackage[caption=false,subrefformat=parens,labelformat=parens]{subfig}
\usepackage{overpic}
\usepackage{physics}
\usepackage{enumitem}
\usepackage{comment}
\usepackage{amsthm}
\usepackage{pstricks}
\usepackage{float}
\usepackage{multirow}
\usepackage[T1]{fontenc}
\usepackage{framed} % shaded
\usepackage{bm}
\usepackage{bbm}
\usepackage{mathtools}
\usepackage{xspace}
\usepackage{nicefrac}
\usepackage{enumitem}
\usepackage{thmtools}
\usepackage{thm-restate}
\usepackage{units}

\usepackage[dvipsnames]{xcolor}
\usepackage[colorlinks=true,citecolor=NavyBlue,linkcolor=NavyBlue,urlcolor=NavyBlue]{hyperref}
\usepackage[english,capitalise]{cleveref}

\definecolor{shadecolor}{gray}{0.9}

\newcommand{\bsym}[1]{\boldsymbol{#1}} 

\renewcommand{\Tr}[1]{\mathrm{Tr}\left[#1\right]}
\newcommand{\pTr}[2]{\mathrm{Tr}_{#1}\left[#2\right]}

\newcommand*\widebar[1]{%
  \text{%
    \vbox{%
      \hrule height 0.12ex 
      \kern 0.4ex 
      \hbox{%
        $\mkern-1.8mu#1\mkern0mu$% 
      }%
    }%
  }%
}

\newtheorem{theorem}{Theorem}

\newtheorem{lemma}[theorem]{Lemma}
\newtheorem{proposition}[theorem]{Proposition}
\newtheorem{definition}[theorem]{Definition}
\newtheorem{remark}[theorem]{Remark}
\theoremstyle{definition}
\newtheorem{prot}{Protocol}

\newcommand{\R}     {\mathbb{R}} %real numbers
\newcommand{\C}     {\mathbb{C}} %complex numbers
\newcommand{\N}     {\mathbb{N}} %natural numbers

\DeclareMathOperator\supp{supp}
\DeclareMathOperator{\Span}{span}

\newcommand{\Renyi}{R\'{e}nyi\xspace}

\newcommand{\renyiEnt}     {\mathbb{H}}

\newcommand{\renyiSandUp}  {\widetilde{H}^\uparrow}
\newcommand{\renyiSandDown}{\widetilde{H}^\downarrow}
\newcommand{\renyiPetzUp}  {\widebar{H}^\uparrow}
\newcommand{\renyiPetzDown}{\widebar{H}^\downarrow}

\newcommand{\frenyiSandUp}  {\widetilde{H}^{\uparrow,f}}
\newcommand{\frenyiSandDown}{\widetilde{H}^{\downarrow,f}}

\newcommand{\renyiDiv}     {\mathbb{D}}
\newcommand{\renyiSandDiv} {\widetilde{D}}
\newcommand{\renyiPetzDiv} {\widebar{D}}

\newcommand{\rel}[2]{\!\left(#1 \middle \| #2 \right)}
\newcommand{\con}[2]{\!\left(#1 \middle | #2 \right)}

\newcommand{\rotcurvearrowdown}{%
	\mathrel{%
		\mathchoice
		{\rotatebox[origin=c]{-90}{$\displaystyle\curvearrowright$}}%
		{\rotatebox[origin=c]{-90}{$\textstyle\curvearrowright$}}%
		{\rotatebox[origin=c]{-90}{$\scriptstyle\curvearrowright$}}%
		{\rotatebox[origin=c]{-90}{$\scriptscriptstyle\curvearrowright$}}%
	}%
}

\newcommand{\rotcurvearrowup}{%
	\mathrel{%
		\mathchoice
		{\rotatebox[origin=c]{-90}{$\displaystyle\curvearrowleft$}}%
		{\rotatebox[origin=c]{-90}{$\textstyle\curvearrowleft$}}%
		{\rotatebox[origin=c]{-90}{$\scriptstyle\curvearrowleft$}}%
		{\rotatebox[origin=c]{-90}{$\scriptscriptstyle\curvearrowleft$}}%
	}%
}

\newcommand{\frenyiMixedUp}{\widetilde{H}^{\rotcurvearrowdown,f}}
\newcommand{\frenyiMixedDown}{\widetilde{H}^{\rotcurvearrowup,f}}

\newcommand{\kapmixdown}{\kappa^{\rotcurvearrowup}}
\newcommand{\kapmix}{\kappa^{\rotcurvearrowdown}}
\newcommand{\kapup}{\kappa^\uparrow} %Upparrow normalization constant
\newcommand{\kapdown}{\kappa^\downarrow} %Downarrow normalization constant
\newcommand{\kapupbnd}{\underline{\kappa}^\uparrow}

\newcommand{\eps}{\varepsilon}
\newcommand{\leak}{\lambda_{\text{EC}}}
\newcommand{\eEV}{\eps_\mathrm{EV}}
\newcommand{\ePA}{\eps_\mathrm{PA}}
\newcommand{\eSec}{\eps_\mathrm{sec}}
\newcommand{\eSct}{\eps_\mathrm{sct}}
\newcommand{\eCor}{\eps_\mathrm{cor}}
\newcommand{\EATchannQKD}{\EATchann^{\text{QKD}}}
\newcommand{\EATchannQKDTag}{\widetilde{\EATchann}^{\text{QKD}}}
\newcommand{\OmegaAcc}{\Omega_{\mathrm{acc}}}
\newcommand{\OmegaLen}[1]{\Omega_{\mathrm{len} = #1}}
\newcommand{\OmegaEV}{\Omega_{\mathrm{EV}}}
\newcommand{\OmegaAT}{\Omega_{\mathrm{AT}}}
\newcommand{\Sacc}{S_{\mathrm{acc}}}

\newcommand{\Ctest}{\mathcal{C}^{\mathrm{test}}_{AB}}
\newcommand{\gen}{\mathtt{gen}}
\newcommand{\test}{\mathtt{test}}
\newcommand{\Ct}{\widehat{\mathcal{C}}_{\setminus \mathrm{gen}}}
\newcommand{\freq}{\operatorname{freq}}
\newcommand{\annA}{a} %Alice's announcement
\newcommand{\annB}{b} %Bob's announcement

\newcommand{\hUpQKD} {h_{\alpha}^{\uparrow,\mathrm{QKD}}}

\DeclareMathOperator{\GMap}{\mathcal{G}}

\DeclareMathOperator{\ZMap}{\mathcal{Z}}

\newcommand{\probst}{\mathbf{\Phi}}

\newcommand{\probstTag}{\widetilde{\probst}}

\newcommand{\EATchann}{\mathcal{M}} 
\newcommand{\EATchannProj}{\widebar{\mathcal{M}}} %projected EAT channel on finite dimensions
\newcommand{\EATchannProjInf}{\widetilde{\mathcal{M}}} %projected EAT channel on infinite dimensions
\newcommand{\CP}{\widehat{C}} %\hat{C} register
\newcommand{\alphCP}{\widehat{\mathcal{C}}}
\newcommand{\cP}{\widehat{c}} % outcome on \hat{C} register
\newcommand{\ProjIn}[2]{\mathcal{P}_{\mathrm{in},#1}^{(#2)}}
\newcommand{\ProjOut}[2]{\mathcal{P}_{\mathrm{out},#1}^{(#2)}}

\newcommand{\pf}{\operatorname{\mathtt{Pur}}} %Purifying function

\newcommand{\CPTP}{\operatorname{CPTP}} % completely positive trace preserving
\DeclareMathOperator{\id}{\mathord{\rm id}} %identity map
\newcommand{\defvar}{\coloneqq}
\newcommand{\ceil}[1]{\left\lceil #1 \right\rceil} %ceil
\newcommand{\dop}[1]{\operatorname{S}_{#1}} %Set of density operators
\newcommand{\pure}[1]{\ketbra{#1}{#1}} %density matrix of pure state
\newcommand{\ran}{\operatorname{ran}} %range function

\newcommand{\cB}{\mathcal{B}}

\newcommand{\cH}{\mathcal{H}}

\newcommand{\cK}{\mathcal{K}}

\newcommand{\cM}{\mathcal{M}}
\newcommand{\cN}{\mathcal{N}}

\newcommand{\cR}{\mathcal{R}}
\newcommand{\cS}{\mathcal{S}}
\newcommand{\cT}{\mathcal{T}}

\newcommand{\mbb}{\mathbb}
\newcommand{\mbf}{\mathbf}
\newcommand{\msf}{\mathsf}
\newcommand{\wt}{\widetilde}
\renewcommand{\ol}{\overline}
\newcommand{\ve}{\varepsilon}

\begin{document}
\title{Unconditional Security of Discrete-Modulated CV-QKD from Infinite-Dimensional MEAT}

%Add names and affiliations
\author{Lars Kamin}
\email{lars.kamin@outlook.com}
\affiliation{Institute for Quantum Computing and Department of Physics and Astronomy, University of Waterloo, Waterloo, Ontario N2L 3G1, Canada}

\author{Ian George}
\email{qit.george@gmail.com}
\affiliation{Centre for Quantum Technologies, National University of Singapore, Singapore 117543, Singapore}

\author{John Burniston}
\affiliation{Institute for Quantum Computing and Department of Physics and Astronomy, University of Waterloo, Waterloo, Ontario N2L 3G1, Canada}

\author{Florian Kanitschar}
\email{florian.kanitschar@outlook.com}
\affiliation{Atominstitut, Technische Universität Wien, Stadionallee 2, 1020 Vienna, Austria}
\affiliation{AIT  Austrian  Institute  of  Technology,  Center  for  Digital  Safety\&Security,  Giefinggasse  4,  1210  Vienna, Austria}
\affiliation{Vienna Center for Quantum Science and Technology (VCQ), Technische Universität Wien, 1020 Vienna, Austria}

\date{\today}
	
%--------------------------------------------------------
%                   Abstract
%--------------------------------------------------------
\begin{abstract}
Discrete-Modulated (DM) Continuous-Variable (CV) Quantum Key Distribution (QKD) is an experimentally attractive approach to quantum cryptography, offering high key rates over metropolitan-scale distances while relying on state-of-the-art telecom infrastructure. 
However, a fundamental gap has remained between this experimental promise and rigorous security: for more than two decades, DM CV-QKD has lacked a complete composable finite-size security proof against coherent attacks. Existing works either restrict the adversary to collective attacks, impose additional finite-dimensional assumptions, apply only to specific modulation formats, or fail to recover the known asymptotic rates. 
Here, we resolve this longstanding problem by establishing the first complete composable finite-size security proof for DM CV-QKD protocols against coherent attacks, incorporating imperfect detectors and both fixed- and variable-length protocol variants. Our proof introduces two central results of broader interest: an infinite-dimensional marginal-constrained entropy accumulation theorem (iMEAT) and a rigorous dimension-reduction technique that makes the infinite-dimensional security bounds numerically tractable. The resulting key rates recover the known asymptotic rates and outperform established finite-size bounds based on collective attacks, while yielding positive key rates beyond $\unit[70]{km}$ for experimentally relevant block sizes and parameters. Thus, our work closes a longstanding gap, places an experimentally attractive class of QKD protocols on a rigorous, composable security footing, and provides a pathway towards their practical deployment.
\end{abstract}

\maketitle

%********************************************************
%                    Main text
%********************************************************

%--------------------------------------------------------
%                   Introduction
%--------------------------------------------------------
\section{Introduction \label{sec:Intro}} 
Modern communication systems derive much of their efficiency from one simple idea: information is conveyed using carefully designed finite signal alphabets. Whether realised as quadrature-amplitude modulation in optical fibers or phase-shift keying often used in wireless links, these constellations represent the interface between digital information and the physical quantity used to transmit the signal. Quantum Key Distribution (QKD) \cite{Bennett_2014, Ekert_1991} faces the same design question. Rather than transmitting classical bits alone, legitimate users - Alice and Bob - distribute and measure quantum states, whose distinguishability fundamentally limits the information available to an eavesdropper (Eve). Consequently, the choice of signal alphabet is not merely an implementation detail, but directly influences both practicality and security. Continuous-Variable (CV) QKD protocols \cite{Ralph_1999, Usenko_2026} realise this alphabet through optical states measured by heterodyne detection, making them uniquely compatible with modern optical communication hardware. Historically, these states have been sampled from Gaussian distributions \cite{Cerf_2001, Grosshans_2002, Grosshans_2003} because the strong symmetry inherent to Gaussian Modulation (GM) provides an exceptionally powerful theoretical toolbox \cite{Wolf_2006, Garcia-Patron_2006, Navascues_2006} for security analyses \cite{Pirandola_2008, Leverrier_2015}. 
From a practical perspective, however, communication hardware is designed for finite constellations rather than continuous distributions. Discrete-Modulation (DM) therefore promises a considerably more practical realisation of CV-QKD, while preserving its compatibility with conventional telecom technology. Yet, replacing continuous modulation with a finite alphabet fundamentally alters the mathematical structure of the protocol; The Gaussian optimality arguments underpinning many security proofs no longer apply, requiring new theoretical tools to establish rigorous security guarantees. 

Establishing rigorous security guarantees for DM CV-QKD has progressed along several complementary directions. Early works demonstrated the feasibility of discrete modulation under restricted attack scenarios \cite{Heid_2006, Sych_2010} or for specific protocol configurations \cite{Zhao_2009, Bradler_2018, Matsuura_2021, Yamano_2024}. In parallel, approximation-based proof techniques have exploited the fact that sufficiently large constellations can closely reproduce the statistical properties of Gaussian modulation, enabling security arguments for arbitrary modulation schemes \cite{Denys_2021, Kaur_2021}.

A major breakthrough came with a conceptually different direction: numerical security proofs \cite{Coles_2016, Winick_2018} based on semidefinite programming (SDP). Rather than approximating Gaussian modulation, these methods directly optimise over the set of quantum states compatible with the observed measurement data, making them particularly well-suited for finite constellations. While the original SDP formulations relied on a photon-number cutoff assumption \cite{Ghorai_2019, Lin_2019}, continuity-bound-based dimension-reduction techniques \cite{Upadhyaya_2021, Lupo_2022} later removed this restriction and enabled composable finite-size security proofs against collective attacks \cite{Lupo_2022, Kanitschar_2023} and Gaussian attacks \cite{Staffieri_2026}. Extending these techniques to general attacks, however, has remained considerably more difficult. Existing approaches either reintroduce photon-number cutoffs \cite{Bauml_2024, Pascual-Garcia_2025, Navarro_2026} or fail to recover asymptotic key rates in the limit of large block sizes \cite{Primaatmaja_2024}. Most recently, Ref. \cite{Navarro_2026a} attempted to address this issue for fixed-length protocols using a conic reformulation of the optimisation problem, but the resulting security proof remains incomplete due to an assumption Eve's purification is finite-dimensional as well as technical issues with applying the finite-dimensional MEAT \cite{Arqand_2025} directly to an infinite-dimensional problem (for details, we refer to Appendix~\ref{apdx:Remarks_Paper}).

The challenge, therefore, remains to construct a composable security argument for general discrete modulation CV-QKD that rigorously:
\begin{enumerate}[label=(\roman*),itemsep=0pt]
    \item accounts for Eve's infinite-dimensional side-information and 
    \item avoids a photon-number cutoff assumption
\end{enumerate}
while
\begin{enumerate}[label=(\alph*),itemsep=0pt]
    \item remaining computationally tractable,
    \item recovering the tight asymptotic key rates \cite{Lin_2019}, and
    \item being applicable to both fixed- and variable-length protocols.
\end{enumerate}

In this work, we fully close this gap, fuelled by a combination of novel developments. Our starting point is the Rényi leftover-hashing lemma \cite{Dupuis_2022}, which bounds the trace distance between Alice's and Bob's (unknown) shared state after protocol execution and a state representing an ideal, perfectly random key that is completely unknown to Eve. This upper bound, equal to the protocol's security parameter $\epsilon$, is expressed in terms of the conditional Rényi entropy of the shared state between Alice and Bob. Our goal is to calculate this entropic quantity for a general discrete-modulated CV-QKD protocol via a constrained optimisation problem subject to experimental input. To this end, we first prove a new entropy accumulation theorem for reducing the calculation of the general $n$-round conditional Rényi entropy to a single-round optimization. What is new about this theorem is that, while it utilizes that a marginal of the $n$-round state is guaranteed, it allows the accrued side-information to be infinite-dimensional as needed for CV-QKD. For this reason, we identify the theorem as an infinite-dimensional, marginal-constrained entropy accumulation theorem (iMEAT). Because the iMEAT reduces to an optimisation problem where the optimisation variable is infinite-dimensional, we establish a dimension-reduction method for conditional Rényi entropies that bounds the value of the infinite-dimensional optimisation problem by the value of a finite-dimensional optimisation problem and a correction term. This allows us to bound the key rate of the infinite-dimensional process using a numerically implementable optimisation problem without introducing an unjustified photon-number cutoff assumption.

In contrast to earlier works \cite{Kanitschar_2023, Bauml_2024, Primaatmaja_2024} that required separate testing procedures to obtain bounds for the weight of the unknown quantum state, we develop a method that determines the weight in the same step as determining the key rate bound, removing slackness, and streamlining the protocol design while relying only on the same natural set of POVM measurements used to link the optimisation problem to the experimental observations. As a result, we obtain asymptotically tight lower bounds on the composable secure key rate for a general discrete-modulated CV-QKD protocol against coherent attacks without any mathematical assumptions limiting Eve's power.

This paper is structured as follows. In \cref{sec:Protocol}, we first introduce the general discrete-modulated continuous-variable QKD protocol we analyse in this work. Next, in Section \ref{sec:SecurityFramework}, we introduce and define the security framework used throughout. This is followed by our first major result in Section \ref{sec:InfDimMEAT}, the formulation and proof of an infinite-dimensional marginal constrained entropy accumulation theorem (Theorem \ref{thrm:MEAT_inf}). We then show how this result is used to prove security for both a fixed- and a variable-length variant of our objective protocol in Section \ref{sec:Security_proof}. While these already constitute valid security statements, the proven bounds cannot yet be evaluated. Therefore, in Section \ref{sec:DimReduction}, we follow with our second major result, the proof of a dimension reduction argument that allows us to rigorously reduce the infinite-dimensional security statements to finite-dimensional optimisation problems that can be solved subsequently. This concludes the security proof of our DM CV-QKD protocols. In Section \ref{sec:NumericalImplementation}, we explain details about the numerical implementation used to solve the underlying convex optimisation problems as well as the simulation model used to illustrate our findings. Plots showcasing the secure key rates achievable with our method can be found in Section \ref{sec:Results}. We conclude the paper in Section \ref{sec:Discussion}. Detailed derivations, proofs, and additional information can be found in the Appendices.

%--------------------------------------------------------
%                   General DM CV-QKD Protocol
%--------------------------------------------------------

\section{General DM CV-QKD Protocol}\label{sec:Protocol}
To establish a secret shared key, Alice and Bob follow a QKD protocol, a publicly known set of instructions. In what follows, we describe a general prepare-and-measure discrete modulated continuous-variable QKD protocol with $N_{\mathrm{St}}$ distinct signal states. Due to the source-replacement scheme \cite{Curty_2004, Ferenczi_2012}, an equivalent entanglement-based version of the protocol exists, and we can switch between those two equivalent descriptions. Note that Greek letters put in ket vectors represent coherent states.

\subsection{Protocol Description}\label{sec:ProtocolDescription}

\begin{prot}{(Generic DM-CV Protocol)}
\\
\textbf{Parameter List:}
\begin{description}[leftmargin = 2.5cm, style= sameline, align=right]
\item[$n \in \mathbb{N}$] Total number of protocol rounds.
\item[$\ell \in \mathbb{N}_0$] Length of the final key.
\item[$N_{\mathrm{St}}$] Number of signal states.
\item[$\ket{\alpha_k}$] Signal state with coherent state amplitude  $\alpha_k$.
\item[$\mathcal{S}$] Alphabet of the raw key registers $\{0, 1, ..., N_{\mathrm{St}}-1, \perp\}$.
\item[$N_{\mathrm{POVMs}}$] Number of coarse-grained POVM outcomes for Bob.
\item[$\{M_i^B\}_{i=1}^{N_{\mathrm{POVMs}}}$] POVM elements describing Bob's measurements in test rounds.
\item[$\{R_i^B\}_{i \in \mathcal{S}}$] Region operators, forming a POVM for Bob in generation rounds. 
\item[$\gamma$] Probability that a round is chosen to be a test round.
\item[$S_{\mathrm{acc}}$] Acceptance set of accepted frequencies, which is closed and convex.
\item[$\epsilon_{\mathrm{PA}}$] Security parameter related to the privacy amplification subroutine.
\item[$\epsilon_{\mathrm{EV}}$] Security parameter related to the error verification subroutine.
\item[$\OmegaAT$] Event of passing the acceptance test.
\item[$\OmegaEV$] Event of passing the error verification.
\item[$\OmegaAcc$] Event of the protocol not aborting ($\OmegaAcc = \OmegaAT \land \OmegaEV$).

\end{description}
\begin{itemize}
    \item[1] \textbf{\textit{State preparation---}} In each round, Alice independently chooses the round to be a \emph{test round} with probability $\gamma$, or a \emph{generation round} with probability $1-\gamma$. She randomly selects one out of $N_{\mathrm{St}}$ signal states according to some discrete probability distribution and prepares the corresponding coherent state $\ket{\alpha_k}$ for $\alpha_k \in \{\alpha_0, \alpha_1, ..., \alpha_{N_{\mathrm{St}}-1}\}$ with probability $p_k$, which is sent to Bob via the quantum channel. Alice stores the label $x_k$ for her choice of the signal state in her private classical register \(X_i\) with alphabet \(\mathcal{X}\), and computes a classical register \(C^A_i\) with alphabet \(\mathcal{C}^A\) for public announcement (including for instance the test/generation decision).

    \item[2] \textbf{\textit{Measurement---}} Bob receives a quantum signal and performs a heterodyne measurement, corresponding to a positive operator-valued measure (POVM) $\{F_{\beta}\}_{\beta \in \mathbbm{C}}$. The outcome of his measurement is a complex number $y_i$, which he keeps in his private classical register \(Y_i\) with alphabet \(\mathcal{Y}\). Furthermore, he computes a classical register \(C^B_i\) with alphabet \(\mathcal{C}^B\) for public announcement.

    \item[3] \textbf{\textit{Partitioning and announcements---}}
    For each round $i$, let $I_i$, with alphabet $\mathcal{I}$, denote the complete transcript of the public communication described in this step, including the announcements $C_i^A$ and $C_i^B$. This transcript is available to Eve, and Alice and Bob hold a local copy denoted by $\tilde{I}_i$ \footnote{We take $\tilde{I}_0$ and $I_0$ to be set to $\pure{\perp}$.}.

    After Bob has completed his measurement in round $i$, Alice and Bob announce $C_i^A$ and $C_i^B$. They may then perform further public classical communication, which can be two-way. Each message may depend only on the private data from the current round (\(X_i\) for Alice and \(Y_i\) for Bob), the public messages already exchanged in round \(i\), which includes \(C_i^A\) and \(C_i^B\), and fresh local randomness.
    
    Then, Alice and Bob compute $\CP_i=\phi(I_i)$ using a prescribed deterministic function $\phi$ which for generation rounds assigns $\CP_i=\mathrm{gen}$.
\end{itemize}

    The first three steps are repeated $n$ times.

\begin{itemize}
    \item[4] \textbf{\textit{Key map---}} Bob performs a reverse reconciliation key map on the key rounds based on his private registers $Y_i$ and the public announcements $I_i$ (and local randomness, if needed), to produce a classical register $S_i$. Therefore, he discretises his measurement outcomes to elements of the set $\mathcal{S} = \{0, 1,..., N_{\mathrm{St}}-1, \perp\}$, where the symbol $\perp$ indicates that the round has been discarded or a test round. The former allows Bob to include postselection in his protocol.

    \item[5] \textbf{\textit{Statistical testing or variable-length decision---}} Alice and Bob compute the frequency distribution $\mathbf{F}^{\mathrm{obs}}$ of the observed string of values $\cP^n_1$ on the registers \(\CP_1^n\)
    \begin{equation}\label{eq:freqdefn}
        F^{\mathrm{obs}}(\cP) \defvar \frac{\text{number of occurrences of $\cP$ in $\cP_1^n$}}{n}.
    \end{equation}
    Depending on which protocol variant is implemented, Alice and Bob proceed as follows:
    \begin{itemize} 
        \item[a)] For the fixed-length variant, if $\mathbf{F}^{\mathrm{obs}} \in S_{\mathrm{acc}}$, the predefined acceptance set, they proceed, otherwise they abort. Passing this stage is denoted by the event $\OmegaAT$.

        \item[b)] For the variable-length variant, Alice and Bob determine the key length $\ell_{\mathrm{var}}$ by calculating a function based on the observations $\cP_1^n$ in registers $\CP_1^n$, see \cref{thrm:variable_length_security}.
    \end{itemize}

    \item[6] \textbf{\textit{Error correction and error verification---}} In the next protocol stage, Alice and Bob aim to reconcile their raw keys.
    \begin{itemize}
        \item[a)] For the fixed-length protocol, Alice and Bob publicly exchange information, disclosing at most $\lambda_{\mathrm{EC}}$ bits of information. 
        \item[b)] In the variable-length case, they disclose $\lambda_{\mathrm{EC}}(\hat{c}_1^n)$, again based on the observations stored in registers $\CP_1^n$. 
    \end{itemize}
    Together with her private data, Alice computes a guess $\bar{S}_1^n$ on Bob's raw key string $S_1^n$. 
    
    Next, they exchange a 2-universal hash of length $\left\lceil\log(\nicefrac{1}{\epsilon_{\mathrm{EV}}}) \right\rceil$. They accept if the hashes agree and abort the protocol otherwise. Let register \(L\) contain the full data communicated in this step and we denote the event of successfully passing error verification by $\OmegaEV$.

    \item[7] \textbf{\textit{Privacy amplification---}} Alice and Bob randomly choose a 2-universal hash function from some family and announce their choice publicly.
    \begin{itemize}
        \item[a)] In the fixed-length case, the hash function maps $\mathcal{S}^n$ to a final key of fixed length $\ell$.
        \item[b)] In the variable-length case, the hash function maps $\mathcal{S}^n$ to a final key of length $\ell_{\mathrm{var}} \in \N_0$. We denote the event of the protocol outputting a key length $m\in \N_0$ with $\OmegaLen{m}$.
    \end{itemize}
    
    Finally, Alice and Bob apply the two-universal hash function to their raw keys $S_1^n$ and $\bar{S}_1^n$ producing their final keys $K_B$ and $K_A$, respectively.
\end{itemize}
\end{prot}

\subsection{Announcements and Conditioning on Test and Generation rounds \label{sec:FormalisationProtocol}}
In order to rigorously represent all protocol steps in the security proof, we require some notation. We split Alice's and Bob's POVMs by their announcements. Given the distinction between test and generation rounds, which are decided by Alice, we additionally split Alice's set of announcements and the corresponding measurement outcomes into test and generation subsets, i.e.
\begin{align}
	\mathcal{C}^A &= \mathcal{C}^{A,\test} \cup \mathcal{C}^{A,\gen},\\
	\mathcal{X} &= \mathcal{X}^{\test} \cup \mathcal{X}^{\gen}.
\end{align}
Then, as shown in \cite{Kamin_2025a}, we find Alice's sets of POVM elements can be written as
\begin{align}
	\{M^{A|\test}_{\annA,x} \}_{\annA \in \mathcal{C}^{A,\test}, x \in \mathcal{X}^{\test}_{\annA} } &= \{M^{A| \test}_k\}_k ,\\
	\{M^{A|\gen}_{\annA,x} \}_{\annA \in \mathcal{C}^{A,\gen}, x \in \mathcal{X}^{\gen}_{\annA} } &= \{M^{A| \gen}_k\}_k ,
\end{align}
where
\begin{align}
	\mathcal{X^{\test}} = \bigcup_{\annA \in \mathcal{C}^{A,\test}} \mathcal{X}^{\test}_{\annA}, \quad \mathcal{X^{\gen}} = \bigcup_{\annA \in \mathcal{C}^{A,\gen}} \mathcal{X}^{\gen}_{\annA},
\end{align}
and the conditional POVM elements were constructed using the partition operators
\begin{align}\label{eq:test and gen partition ops}
	\Pi^{\test} = \sum_{\substack{\annA \in \mathcal{C}^{A,\test}, \\ x \in \mathcal{X}^{\test}_{\annA}}} M_{\annA,x}^A, \quad
	\Pi^{\gen} = \sum_{\substack{\annA \in \mathcal{C}^{A,\gen}, \\ x \in \mathcal{X}^{\gen}_{\annA}}} M_{\annA,x}^A.
\end{align}
See \cite[Appendix.~H]{Kamin_2025a} for more details on the construction of the conditioned POVM elements for the case where Alice's POVM elements are not rank one projectors.

Equivalently, we also partition Bob's announcements as
\begin{align}
	\mathcal{C}^B &= \mathcal{C}^{B,\test} \cup \mathcal{C}^{B,\gen},\\
	\mathcal{Y} &= \mathcal{Y}^{\test} \cup \mathcal{Y}^{\gen},
\end{align}
where
\begin{align}
	\mathcal{Y^{\test}} = \bigcup_{\annB \in \mathcal{C}^{B,\test}} \mathcal{Y}^{\test}_{\annB}, \quad \mathcal{Y^{\gen}} = \bigcup_{\annB \in \mathcal{C}^{B,\gen}} \mathcal{Y}^{\gen}_{\annB}.
\end{align}
Additionally, we implement a coarse-graining on Bob's full POVM $\{F_\beta\}_{\beta \in \C}$, which is different in test and generation rounds. This coarse-graining maps the underlying POVM $\{F_\beta\}_{\beta \in \C}$ to the POVM elements $\{M^B_i\}_{i=1}^{N_{\mathrm{POVMs}}}$ in test rounds (see Appendix \ref{apdx:POVMexpressions}) and to the region operators $\{R_i^B\}_{i\in \mathcal
S}$ (see Section \ref{sec:RegionOpDef}).

Then, we also partition these POVM elements according to Bob's announcements such that
\begin{align}
	\{M^B_{\annB,y} \}_{\annB \in \mathcal{C}^{B,\test}, y \in \mathcal{Y}_{\annB}^{\test}} &= \{M^B_i\}_{i=1\dots N_{\mathrm{POVMs}}}, \\
    \{R^B_{\annB,y}\}_{\annB \in \mathcal{C}^{B,\gen}, y \in \mathcal{Y}_{\annB}^{\gen}} &=
    \{R_i^B\}_{i=1}^{N_{\mathrm{St}}+1}.
\end{align}
We highlight that picking such different POVM elements in test and generation rounds is only possible because they originate from the \emph{same} underlying POVM $\{F_\beta\}_{\beta \in \C}$.

Finally, for brevity, let us define the set
\begin{equation}
	\begin{split}
		\Ctest \defvar \{(\annA,x,\annB,y) | 
		\annA &\in \mathcal{C}^{A,\test}, x \in \mathcal{X}^{\test}_{\annA} , \\
        &\annB \in \mathcal{C}^{B,\test}, y \in \mathcal{Y}^{\test}_{\annB} \},
	\end{split}
\end{equation}
as the set containing all announcements during test rounds. In addition, we assume that the alphabet \(\mathcal{I}\) of registers \(I_i\) contains \(\Ctest\).

%--------------------------------------------------------
%                   Security Framework
%--------------------------------------------------------
\section{Security Framework}\label{sec:SecurityFramework}

\subsection{Security Definition}
In this work, we use the variable-length security definitions from Ref.~\cite[Sec.~VI.A]{Portmann_2022}.

\begin{definition}[Variable-Length $\varepsilon$-security]\label{Def:Variable length eps security}
	Let \(\OmegaLen{m}\) be the event of a QKD protocol generating a final key of length \(m \in \N_0\). A variable-length QKD protocol is \emph{\(\eSec\)-secure} if for any input state \(\sigma_{A^n B^n}\) the resulting output state \(\sigma_{K_AK_BE}\) satisfies
	\begin{equation} \label{eq:security}
        \begin{split}
             \frac{1}{2} \sum_{m=0}^{\infty} \Pr[\OmegaLen{m}] \Bigg\lVert &\sigma_{K_AK_BE| \OmegaLen{m}} \\ & - \mathbb{K}_{K_AK_B}^m \otimes \sigma_{E| \OmegaLen{m}} \Bigg\rVert_1 \leq \eSec,
         \end{split}
	\end{equation}
    where \(\mathbb{K}_{K_AK_B}^m\) denotes the state of a perfect uniform shared key of length $m$:
    \begin{equation}
        \mathbb{K}_{K_AK_B}^m = \sum_{k \in \{0,1\}^m } \frac{\dyad{kk}_{K_AK_B}}{2^m}.
    \end{equation}
    Furthermore, a variable-length QKD protocol is \emph{\(\eSct\)-secret} if
	\begin{equation} \label{eq:secrecy}
        \begin{split}
             \frac{1}{2} \sum_{m=0}^{\infty} \Pr[\OmegaLen{m}] \Bigg\lVert &\sigma_{K_BE| \OmegaLen{m}} \\ & - \mathbb{U}_{K_B}^m \otimes \sigma_{E| \OmegaLen{m}} \Bigg\rVert_1 \leq \eSct,
         \end{split}
	\end{equation}
    where \(\mathbb{U}_{K_B}^m\) is a fully mixed state of dimension \(2^m\), and \emph{\(\eCor\)-correct} if
	\begin{equation}
		\Pr[K_A \neq K_B \wedge \Omega_{\text{acc}}] \leq \eCor.
	\end{equation}
    A variable-length QKD protocol that is \(\eCor\)-correct and \(\eSct\)-secret, is \(\eSec = \eSct + \eCor \) secure \cite{Tupkary_2024}. 
\end{definition}

\begin{remark}
    For P\&M protocols, one can restrict the input states to states that satisfy the marginal constraint \(\pTr{B^n}{\sigma_{A^nB^n}} = \tau_A^{\otimes n}\) for some state \(\tau_A\) defined by Alice's choices. This is possible because we assume that Eve cannot access Alice's lab. In the present work, this marginal constraint is always imposed.
\end{remark}

\begin{remark}\label{rem:Fixed-Length security}
    Fixed-length security follows directly from variable-length security, by selecting only one key length in the sum in \cref{eq:security} and identifying this with the fixed key length \(\ell\) of the fixed-length protocol; all other events are combined to one single abort event. Furthermore, secrecy and correctness imply security in the same way.
\end{remark}

\subsection{Source Replacement Scheme}\label{sec:sourceReplacement}
Under the source replacement scheme \cite{Curty_2004, Ferenczi_2012}, Alice's state preparation process in each round can be viewed as having her first prepare a \emph{pure} state, 
\begin{equation}
	\ket{\psi}_{AA'} = \sum_{k=1}^{N_{\mathrm{St}}}  \sqrt{p(k)} \ket{k}_A \ket{\alpha_k}_{A'},
\end{equation}
where $p(k)$ is the total probability of Alice sending state $\ket{\alpha_k}_{A'}$ to Bob, and then performing a measurement on $A$ described by the following POVM elements:
\begin{equation}
	M_k^A = \ketbra{k}{k},
\end{equation}
for all $k=1 \dots d_A$. Furthermore, let us write $\sigma_A := \pTr{A'}{\ketbra{\psi}{\psi}}$ to denote the marginal of the single-round state Alice prepares in the source-replaced picture.

\subsection{QKD Channel}
For completeness, we define the channel $\EATchannQKD$ along the lines of \cite{Kamin_2025a}, but for reverse reconciliation.
\begin{definition}[QKD Channel (adapted from \cite{Kamin_2025a}]\label{def:QKD channel}
    Let \(\hat{\GMap}: AB \rightarrow ABXY\hat{S}_Q\tilde{I}\) be the CPTNI map as described in \cite[Appendix A]{Lin_2019} of a QKD protocol defined with Alice's POVM elements conditioned on a generation round, i.e. \(\{M_{\annA,x}^{A|\gen}\}\). We extend \(\hat{\GMap}\) to a CPTP map \(\GMap\) through the following construction. Let \(\hat{\mathcal{S}}\) be the alphabet of \(\hat{S}_Q\) and \(\mathcal{S}= \hat{\mathcal{S}}\cup\{\perp\}\) be the alphabet of the extended register \(S_Q\). We first define the map \(\hat{\GMap}\) for all \(\rho \in \dop{=}(AB)\) as
    \begin{align}
        &\hat{\GMap}: AB \rightarrow ABXYS_Q\tilde{I}, \\
        &\hat{\GMap}(\rho) \defvar \sum_{(\annA, \annB) \in \mathcal{K} } K_{\annA,\annB} \rho K_{\annA,\annB}^{\dagger},
    \end{align}
    where $\mathcal{K} = \{\annA \in \mathcal{C}^{A|\gen}, \annB \in \mathcal{C}^B | \exists x \textrm{ s.t. } g(\annA,\annB) \neq \perp \}$ is the set of kept announcements and the Kraus operators $K_{\annA, \annB}$ are given by
    \begin{equation}\label{eq:discard_Kraus_op}
        K_{\annA,\annB} \defvar \begin{aligned}[t]
        \sum_{\substack{x \in  \mathcal{X}^{\gen}_{\annA} \\ y \in \mathcal{Y}^{\gen}_{\annB} }} \sqrt{M_{\annA,x}^{A|\gen}} \otimes \sqrt{R_{\annB,y}^{B}} \otimes \ket{x}_X \otimes \ket{y}_Y \\ \otimes \ket{g(\annA,\annB,y)}_{S_Q} \ket{\annA,\annB}_{\tilde{I}},
        \end{aligned}
    \end{equation}
    where $g$ is the key map for reverse reconciliation. We extend this CPNTI map to $\GMap: AB \rightarrow ABXYS_Q\tilde{I}$ by adding the Kraus operator
    \begin{equation}
        K_{\perp} \defvar \sqrt{\left(I - \sum_{\annA, \annB \in \mathcal{K}} K_{\annA, \annB}^{\dagger} K_{\annA, \annB} \right)} \otimes \ket{\perp}_{XY S_Q\tilde{I}}.
    \end{equation}
    Additionally, let \(V_{\GMap}:AB \rightarrow ABXYS_Q\tilde{I}I\) be the Stinespring dilation of \(\GMap\)\footnote{Due to the properties of \(\hat{\GMap}\) its Kraus operators are indexed by the announcements contained in register \(\tilde{I}\). Therefore, the environment in a Stinespring dilation of \(\GMap\) contains a copy of \(\tilde{I}\) which is exactly Eve's copy of the announcement register \(I\).}, which has the form
    \begin{equation}
        V_{\GMap} = \sum_{i=(\annA, \annB),\perp} K_{\annA, \annB} \otimes \ket{i}_{I}.
    \end{equation}
    Furthermore, we define the following isometries
    \begin{align}
        V_{\test} &\defvar \sqrt{\Pi^{\test}} \otimes \ket{t}_F + \sqrt{\Pi^{\gen}} \otimes \ket{g}_F, \\
        V_{\ZMap} &\defvar \sum_{j \in \mathcal{S}} \ket{j}_{S} \otimes Z_j, \\
        V_{\mathrm{meas}} &\defvar \sum_{\substack{(\annA, x, \annB, y) \in \Ctest}} \sqrt{M_{\annA,x}^{A|\test}} \otimes \sqrt{M_{\annB,y}^{B}}  \\
        &\qquad \otimes \ket{\annA,x,\annB,y}_{I} \otimes \ket{\annA,x,\annB,y}_{\tilde{I}} \otimes \ket{\perp}_{S_QXY}, \\
        V_{\phi} &\defvar \sum_{j \in \mathcal{I}} \ketbra{j}{j}_I  \otimes \ket{\phi(j)}_{\CP}.
    \end{align}
    where the orthonormal projectors \(Z_j\) satisfy \(\sum_{j \in \mathcal{S}} Z_j = \mathbb{I}_{S_Q} \). Finally, let us define the concatenation of all isometries as
    \begin{equation}
        W \defvar \begin{aligned}[t]
            V_{\ZMap} V_{\phi} \big( &V_{\mathrm{meas}}\otimes\ketbra{t}{t}_F \\
            &\quad + V_{\GMap}\otimes\ketbra{g}{g}_F \big)V_{\test}.
        \end{aligned}
    \end{equation}
    Then, we define the channel \(\EATchannQKD\) of a QKD protocol for any \(\rho_{AB}\) as
    \begin{equation}
        \EATchannQKD[\rho_{AB}] \defvar \pTr{\mathrm{all} \setminus SI \CP}{W \rho_{AB} W^{\dagger}},
    \end{equation}
    where with \(\pTr{\mathrm{all} \setminus SI \CP}{\dots}\) we indicate tracing out all systems apart from \(SI \CP\). Finally, we also define the map creating the statistics in test rounds as
    \begin{equation}\label{eq:stat_generating_map}
        \probst[\sigma] \defvar \sum_{\cP \in \Ct} \sum_{\substack{(\annA,x,\annB,y) \\ \in \phi^{-1} (\cP)}} \Tr{\left(M_{\annA,x}^{A|\test} \otimes M_{\annB,y}^{B}\right) \sigma} \hat{e}_{\cP},
    \end{equation}
    where \(\hat{e}_{\cP}\) is a unit vector with a one at the \(\cP\)-th position and \(\phi^{-1}\) stands for the preimage.
\end{definition}

\begin{remark}
    We again highlight that the choice of Bob's POVM elements to differ in test and generation rounds is only possible because they are derived from one underlying POVM that is coarse grained in different ways in the classical post processing.
\end{remark}

\begin{remark}
    In principle, one can define $\EATchannQKD$ even if the POVM elements and the processing differ from round to round. The construction will still remain the same and for notational simplicity we only presented this case here. However, for full generality in the remainder, we will state our theorems with channels $\EATchannQKD_j$ and $\probst_j$ for $j=1,\dots,n$.
\end{remark}

Since we will make extensive use of this notation, we also define the states shared between Alice and Bob conditioned on a test and generation round as
\begin{align}
    \rho_{AB|\test} \defvar \frac{\sqrt{\Pi^{\test}} \rho_{AB} \sqrt{\Pi^{\test}} }{\gamma}, \\  
    \rho_{AB|\gen} \defvar \frac{\sqrt{\Pi^{\gen}} \rho_{AB} \sqrt{\Pi^{\gen}} }{1-\gamma}.
\end{align}

Finally, we define a set of channels $\{\mathcal{N}_j\}_{j=1,\dots,n}$ combining the QKD channel from \cref{def:QKD channel}, representing the QKD protocol from \cref{sec:ProtocolDescription} with Eve's attack. Such a channel is convenient to represent the QKD protocol and Eve's attack in the MEAT framework of Ref~\cite{Arqand_2025}. A schematic representation of this construction is shown in \cref{fig:Concatenated_channels}. 

Crucially, as described in \cite{Tupkary2026}, this representation allows one to make the announcements right after Bob's measurements and include Eve possibly adapting her attack based on this announcement.

\begin{definition}[MEAT Channels including Eve's Attack]\label{def:MEAT_channel_with_Eve}
    Let $\psi_{A_{1}^{n}(A^{\prime})_{1}^{n}}$ be the pure state sent by Alice, given through the source-replacement scheme and satisfying the marginal constraint
    \begin{equation}
        \pTr{(A')_{1}^{n}}{\psi_{A_{1}^{n}(A^{\prime})_{1}^{n}}} = \bigotimes_{j=1}^{n}\sigma_{A_j}^{(j)}.
    \end{equation}
    Furthermore, let $\mathcal{E}_0:(A')_1^n \to B_1^n E_0'$ be Eve's initial attack channel and let $\omega_{A_1^nB_1^n E_0}$ be a purification of $\mathcal{E}_0\left(\psi_{A_{1}^{n}(A^{\prime})_{1}^{n}}\right)$, such that Eve initially holds the purification $E_0$. We extend the state $\omega$ trivially on the announcement register $I_0$ as
    \begin{equation}
        \omega_{A_1^nB_1^n E_0 I_0} = \omega_{A_1^nB_1^n E_0} \otimes \pure{\perp}_{I_0}.
    \end{equation}
    Additionally, we define Eve's attack channels $\mathcal{E}_j$ as
    \begin{equation}
        \mathcal{E}_{j}: E_{j-1} B_j^n I_0^{j-1} \longrightarrow E_j B_{j}^n I_0^{j-1},
    \end{equation}
    for all $j=1, \dots, n$. Eve's attack channels may depend on the classical values in \(I_0^{j-1}\), however they preserve those values, that is
    \begin{equation}
        \mathcal{E}_j \left(\tau\otimes \pure{i_0^{j-1}}\right)
        = \mathcal E_{j\mid i_0^{j-1}}(\tau) \otimes \pure{i_0^{j-1}},
    \end{equation}
    for all $\tau\in \dop{=}(B_j^nE_{j-1})$ and for every classical history \(i_0^{j-1}\). Finally, we define the channels onto which we apply the iMEAT as
    \begin{equation}
        \mathcal{N}_j \defvar (\EATchannQKD_j \otimes \id_{B_{j+1}^n E_{j} I_0^{j-1}}) \circ (\id_{A_j} \otimes \mathcal{E}_j),
    \end{equation}
    which acts as
    \begin{align}
        \mathcal{N}_j: A_j B_j^n E_{j-1} I_0^{j-1}  &\xrightarrow{\mathcal{E}_j} A_j B_j^n E_j I_0^{j-1} \\
        &\xrightarrow{\EATchannQKD_j} S_j \hat{C}_j B_{j+1}^n E_j I_0^j.
    \end{align}
    Then, we can write the state after key mapping and before statistical testing, error correction, and privacy amplification, generated from the QKD protocol as described in \cref{sec:ProtocolDescription}, as
    \begin{equation}
        \rho_{S_1^n \hat{C}_1^n E_n I_1^n} = \mathcal{N}_n \circ \dots \circ \mathcal{N}_1 [\omega_{A_1^nB_1^n E_0 I_0}].
    \end{equation}
\end{definition}

\begin{figure*}
    \centering
    \includegraphics[width=\linewidth]{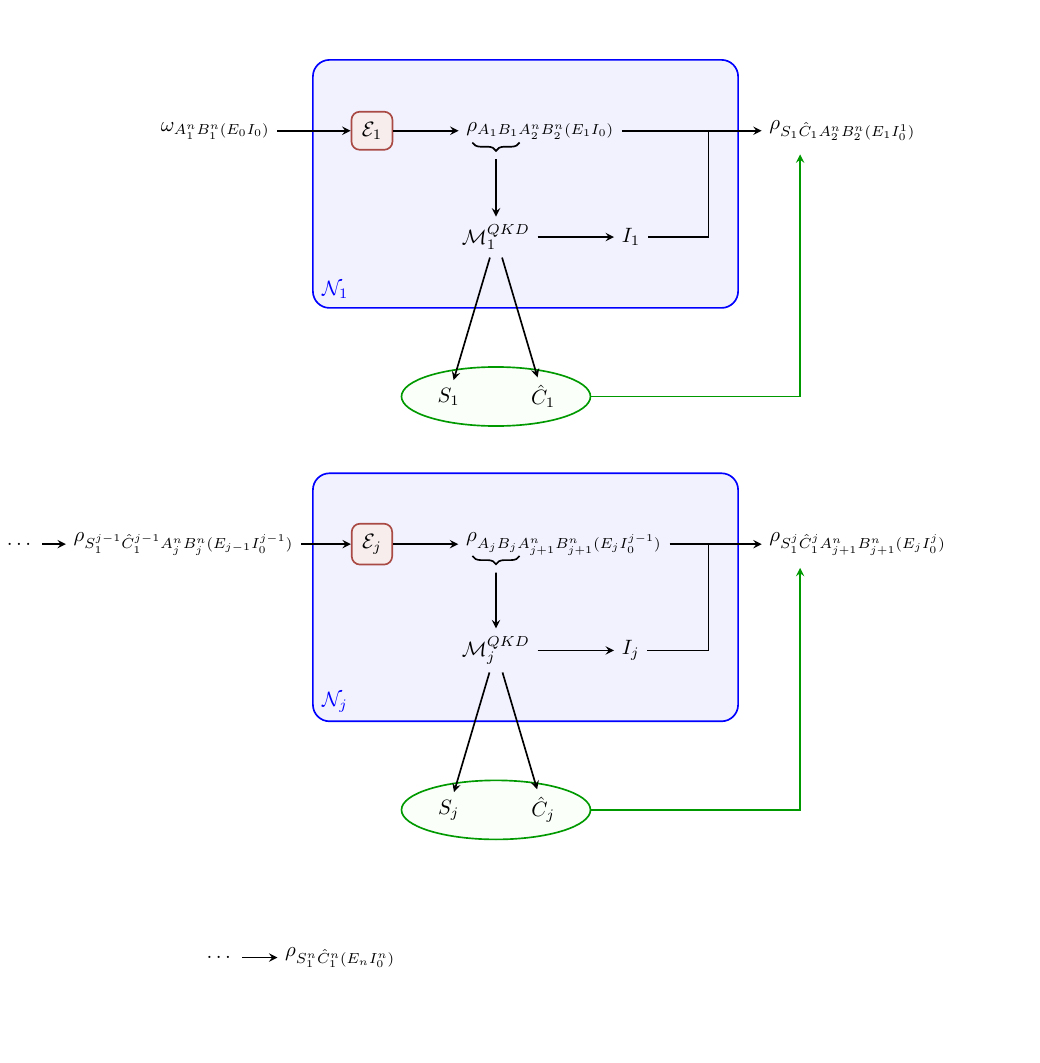}
    \caption{Schematic representation of the concatenated channel structure underlying the security proofs for fixed- and variable-length of \cref{thrm:Fixed_length_security,thrm:variable_length_security}.}
    \label{fig:Concatenated_channels}
\end{figure*}

\subsection{Leftover-hashing Lemma}
The leftover-hashing lemma (LHL) \cite{Renner_2006} relates the trace distance between the state produced by the protocol and an ideal state - a uniformly random key that is decoupled from Eve (see Eq.~\eqref{eq:security}) - to an entropy of the underlying quantum state and the length of the final key. In established formulations of the LHL, this entropic quantity is the (smoothed) min-entropy. However, expressing the final security statement in terms of the min-entropy requires additional conversion steps later in the proof, resulting in a loss of tightness. The LHL derived by Dupuis \cite{Dupuis_2022} makes the \Renyi entropy the quantity of interest that relates the extractable key length to the underlying quantum state.

For a classical-quantum state $\rho_{AE}$, a function $u:~ \mathcal{A} \rightarrow \mathcal{Z}$ drawn uniformly random from a set of two-universal hash functions with $\mathcal{Z} = \{0,1\}^\ell$, and $\alpha \in (1,2)$, we find
\begin{equation}
    \begin{aligned}
        &\frac{1}{2} \mathbb{E}\left|\left|\rho_{ZE|g} - \frac{1}{|\mathcal{Z}|} \mathbb{I}_Z \otimes \rho_E \right|\right|_1 \leq 2^{\frac{2(1-\alpha)}{\alpha}} 2^{\frac{1-\alpha}{\alpha} \left(\tilde{H}^{\uparrow}_{\alpha}(A|E)_{\rho} - \ell\right)}.
    \end{aligned}
\end{equation}

Applied to QKD, the abstract register $A$ in the leftover-hashing lemma is identified with Alice's $n$-round secret register $S_1^n$, while $Z$ is the register containing the final key obtained by applying the random hash to $S_1^n$. Thus, after sifting and key map, the considered state is of the form $\rho_{S_1^nE}$ and privacy amplification produces $K_A = u(S_1^n)$. Consequently, the trace-distance security of the final key is directly bounded in terms of the conditional \Renyi entropy $\renyiSandUp_{\alpha}(S_1^n |E)_{\rho}$, where $E$ now contains all of Eve's information including the announcements $I_0^n$.

This makes the role of the $n$-round \Renyi entropy clear: once the protocol has reduced to the state $\rho_{S_1^nE}$, the privacy amplification part of the security proof now requires us to obtain a lower bound on $\renyiSandUp_{\alpha}(S_1^n|E)_{\rho}$. The length of the final key determines how much entropy is extracted by the hash. Here, as we will see, we will not have to convert to other entropic quantities in the remaining part of the proof, which removes several sources of looseness in the security proof. 
The remaining task therefore is to bound this \Renyi entropy of the raw secret bit string generated by the protocol, considering all of Eve's classical and quantum side-information.

However, the continuous nature of the signals exchanged in the present DM CV-QKD protocol poses an additional challenge, because this makes Eve's quantum system genuinely infinite-dimensional. While techniques for bounding such \Renyi entropies have been developed recently \cite{Arqand_2025, Fawzi_2026}, they require Eve's system to be finite-dimensional. This motivated our first major contribution: the proof of a fully infinite-dimensional marginal- constrained entropy accumulation theorem (iMEAT).

%--------------------------------------------------------
%          Infinite-Dimensional MEAT (Main text)
%--------------------------------------------------------

\section{Infinite-Dimensional Marginal-Constrained Entropy Accumulation Theorem}\label{sec:InfDimMEAT}

As discussed above, the \Renyi LHL \cite{Dupuis_2022} reduces the problem of a QKD security proof to bounding the $n$-round conditional \Renyi entropy in the presence of infinite-dimensional side information. We obtain such a bound by generalizing the marginal-constrained entropy accumulation theorem of Ref.~\cite{Arqand_2025} to infinite dimensions. The proof of this generalization is given in \cref{App:Inf_MEAT_Proof}.

Before stating the resulting theorems, we briefly discuss the scope of the generalization and the dimensional restrictions that remain necessary.

The marginal-constrained entropy accumulation theorem of Ref.~\cite{Arqand_2025} yields a lower bound on the $n$-round \Renyi entropy of the form
\begin{equation}
\renyiSandUp_\alpha(S_1^n \mid \CP_1^n E_n)_{\rho_{|\Omega}}
\geq n h^\uparrow_{\alpha} - \frac{\alpha}{\alpha-1} \log\frac{1}{p_\Omega},
\end{equation}
where $h^\uparrow_{\alpha}$ is a constant that will be defined below. The state $\rho$ is generated from a sequence of channels as
\begin{equation}
    \rho_{S_1^n\CP_1^nE_n} = \EATchann_n \circ \dots \circ \EATchann_1 \left(\omega_{A_0^{n-1}E_0}\right),
\end{equation}
where $\omega_{A_0^{n-1}E_0}$ satisfies a marginal constraint for the register $A_0^{n-1}$. In Ref.~\cite{Arqand_2025}, the bound on the $n$-round \Renyi entropy was established only under finite-dimensionality assumptions.

We likewise extend the following bound for $f$-weighted entropies, which is needed to prove security for variable-length protocols:
\begin{equation}
H_\alpha^{\uparrow,f_{\mathrm{full}}}(S_1^n \mid \CP_1^n E_n)_\rho \geq \sum_j\min_{\cP_1^{j-1}} \kappa_{\cP_1^{j-1}}.
\end{equation}
The constants $\kappa_{\cP_1^{j-1}}$ will also be defined below; we omit their definition from this informal overview.

The proofs of both finite-dimensional results rely on a de-Finetti reduction, specifically on a version that preserves a fixed marginal \cite[Corollary~3.2]{Fawzi_2015}. This reduction is applied to the secret register $S$. Consequently, any attempt to establish an analogous bound on the \Renyi entropy for an infinite-dimensional secret register by following the same proof strategy would fail at this step; the finite-dimensional de-Finetti reduction used in the proof is no longer applicable.

Moreover, when $S$ is infinite-dimensional, the \Renyi entropy and its $f$-weighted counterpart are, in general, no longer uniformly continuous and need not be finite-valued. Both properties are essential for a separate proof step relying on a strong-duality argument. Thus, in the absence of an alternative proof technique, the present argument can only establish a generalization in which the secret register remains finite-dimensional.

The announcement register $\CP$ must also remain finite-dimensional. In our context, it is unclear how to extend the relative entropy between the probability distribution induced by the state on $\CP$ and a reference distribution while preserving the properties required by the proof. Such an extension must account for the possibility that the relative entropy takes the value $+\infty$ when the two distributions have incompatible supports. This occurs, for example, when the induced distribution assigns positive probability to a set to which the reference distribution assigns probability zero. The problematic set may therefore be negligible for the reference distribution but not for the induced distribution. The current formulation of the marginal-constrained entropy accumulation theorem does not cover such cases, and we therefore retain the assumption that $\CP$ is finite-dimensional.

These obstructions are specific to $S$ and $\CP$ and do not arise for $A_j$ and $E_j$. Our generalization therefore allows $A_j$ and $E_j$ to be infinite-dimensional, while $S$ and $\CP$ remain finite-dimensional. We state the two resulting theorems below and defer their proofs to \cref{App:Inf_MEAT_Proof}.

The first theorem is the direct generalization of the marginal constrained entropy accumulation theorem for infinite dimensions, which can be used directly for finite-size security proofs of fixed-length protocols as it gives a lower bound on the $n$-round Renyi entropy.

\begin{restatable}[Marginal-constrained entropy accumulation theorem in Infinite Dimensions]{theorem}{MEATInf}\label{thrm:MEAT_inf}
% \begin{theorem}[Marginal-constrained entropy accumulation theorem in Infinite Dimensions]\label{thrm:MEAT_inf}
    Let $A_j$, $j=0,\dots, n-1$, and $E_j$, $j=0,\dots, n$ be registers with possibly \emph{infinite dimensional separable} Hilbert spaces. In contrast, assume that $S_j$ and $\CP_j$ for $j=1,\dots, n$ are registers with underlying \emph{finite}-dimensional Hilbert spaces.
    
    For each $j\in\{1,2,\dots,n\}$, take a state $\sigma^{(j-1)}\in\dop{=}(A_{j-1})$ and a channel $\mathcal{M}_j\in\CPTP(A_{j-1}E_{j-1},S_j\CP_jE_j)$, where each $\CP_j$ is classical and all the $\CP_j$ are isomorphic to a single register $\CP$ with alphabet $\alphCP$.
    
    Let $\rho$ be a state of the form $\rho_{S_1^n\CP_1^nE_n}=\mathcal{M}_n\circ\cdots\circ\mathcal{M}_1[\omega_{A_0^{n-1}E_0}]$ for some $\omega\in\dop{=}(A_0^{n-1}E_0)$, such that $\omega_{A_0^{n-1}}=\sigma_{A_0}^{(0)}\otimes\cdots\otimes\sigma_{A_{n-1}}^{(n-1)}$. Furthermore, suppose that $\rho = p_\Omega \rho_{|\Omega} + (1-p_\Omega) \rho_{|\overline{\Omega}}$ for some $p_\Omega \in (0,1]$, normalized states $\rho_{|\Omega},\rho_{|\overline{\Omega}}$, and an event $\Omega$ defined only on the classical register $\CP_1^n$.
    
    Let $S_\Omega$ be a closed and convex set of probability distributions on the alphabet $\alphCP$, such that for all $\cP_1^n$ with nonzero probability in $\rho_{|\Omega}$, the frequency distribution $\freq_{\cP_1^n}$ lies in $S_\Omega$. 
    Then, for any $\alpha\in(1,\infty)$, we have:
    \begin{align}\label{eq:MEAT_inf}
        \renyiSandUp_\alpha(S_1^n | \CP_1^n E_n)_{\rho_{|\Omega}} \geq  n  h^\uparrow_{\alpha}
        - \frac{\alpha}{\alpha-1} \log\frac{1}{p_\Omega}, 
        \end{align}
        where (recalling $\bsym{\nu}_{\CP}$ denotes the distribution on $\CP$ induced by any state $\nu_{\CP}$)
    \begin{equation}
        h^\uparrow_{\alpha} = 
        \begin{aligned}[t]
        \inf_{\mbf{q} \in S_\Omega} &\inf_{\nu\in\Sigma_{S\CP E\widetilde{E}}} \Bigg( \frac{\alpha}{{\alpha}-1}D\left(\mbf{q} \middle\Vert \bsym{\nu}_{\CP}\right) \\
        &+\sum_{\cP\in\supp(\bsym{\nu}_{\CP})}q(\cP)\renyiSandUp_{\alpha}(S|E\widetilde{E})_{\nu_{|\cP}}  \Bigg),
        \end{aligned}
    \end{equation}
    with $\Sigma_{S\CP E \widetilde{E}}$ being the marginal-constrained convex range (see \cite[Definition~4.2]{Arqand_2025}) of the channels $\EATchann_j \otimes \id_{\widetilde{E}}$ and states $\sigma^{(j-1)}$, where $\widetilde{E}$ is a register of large enough dimension to serve as a purifying register for any of the $A_{j-1}E_{j-1}$ registers. Furthermore, $S$ and $E$ are common registers into which all $S_j$ and $E_j$, respectively, can be isometrically embedded.
\end{restatable}
% \end{theorem}

Next, we state the generalization to infinite dimensions of the entropy accumulation theorem for $f$-weighted entropies. For the definition of $f$-weighted entropies see \cref{apdx:MathDef}. As we will see later in \cref{sec:Variable_length_security} this theorem yields finite-size security for variable-length protocols.

\begin{restatable}[Infinite-dimensional $f$-weighted marginal-constrained entropy accumulation theorem]{theorem}{fWeightedMEATInf}\label{thrm:fweighted_meat_simp}
    Let $A_j$, $j=0,\dots, n-1$, and $E_j$, $j=0,\dots, n$ be registers with possibly \emph{infinite dimensional separable} Hilbert spaces. In contrast, assume that $S_j$ and $\CP_j$ for $j=1,\dots n$ are \emph{finite} registers.
    
	Furthermore, for each $j\in\{1,2,\cdots,n\}$, take a state $\sigma^{(j-1)}\in\dop{=}(A_{j-1})$, and a channel $\mathcal{M}_j\in\CPTP(A_{j-1}E_{j-1},S_j\CP_jE_j)$, such that $\CP_j$ are classical. Let $\rho$ be a state of the form $\rho_{S_1^n\CP_1^nE_n}=\EATchann_n\circ\cdots\circ\EATchann_1[\omega_{A_0^{n-1}E_0}]$ for some $\omega\in\dop{=}(A_0^{n-1}E_0)$, such that $\omega_{A_0^{n-1}}=\sigma_{A_0}^{(0)}\otimes\cdots\otimes\sigma_{A_{n-1}}^{(n-1)}$. For each $j$, suppose that for every value $\cP_1^{j-1}$, we have a tradeoff function $f_{|\cP_1^{j-1}}$ on registers $\CP_j$ and define the following tradeoff function on $\CP_1^n$:
	\begin{align}
		\label{eq:full qes_simp2}
		f_\mathrm{full}(\cP_1^n) \defvar \sum_{j=1}^n f_{|\cP_1^{j-1}}(\cP_j).
	\end{align}
	Then for any $\alpha\in (1,\infty)$ we have
	\begin{align}\label{eq:qes eat_simp}
		\begin{aligned}
			&\widetilde{H}_\alpha^{\uparrow,f_\mathrm{full}}(S_1^n|\CP_1^nE_n)_\rho \geq \sum_{j=1}^n\min_{\cP_1^{j-1}} \kappa_{\cP_1^{j-1}} \\
			\text{where} \\
            &\kappa_{\cP_1^{j-1}} \defvar \inf_{\nu\in\Sigma_j} H^{\uparrow, f_{|\cP_1^{j-1}}}_{\alpha}(S_j| \CP_j E_j \widetilde{E})_{\nu},
		\end{aligned}
	\end{align}
	defining 
    \begin{equation}
        \begin{aligned}[t]
        \Sigma_j \defvar \Big\{
            &\left(\EATchann_j \otimes \id_{\widetilde{E}}\right)\left[\omega_{A_{j-1} E_{j-1}\widetilde{E}}\right] 
        \mid \\
        &\omega \in \dop{=}(A_{j-1}E_{j-1}\widetilde{E}), \; \omega_{A_{j-1}}=\sigma_{A_{j-1}}^{(j-1)} \Big\}.
        \end{aligned}
    \end{equation} 
    with $\widetilde{E}$ being a register of large enough dimension to serve as a purifying register for any of the $A_{j-1}E_{j-1}$ registers. We can assume $\widetilde{E}$ to be separable, because $n$ is finite and all Hilbert spaces under consideration are separable.
	
	Consequently, if we instead define the following ``normalized'' tradeoff function on $\CP_1^n$:
    \begin{equation}\label{eq:fullQESnorm_simp}
        \begin{aligned}
		&\hat{f}_\mathrm{full}( \cP_1^n) \defvar \sum_{j=1}^n \hat{f}_{|\cP_1^{j-1}}(\cP_j)
	\end{aligned}
    \end{equation}
	where
    \begin{equation}
    \hat{f}_{|\cP_1^{j-1}}(\cP_j) \defvar f_{|\cP_1^{j-1}}(\cP_j) + \kappa_{\cP_1^{j-1}},
    \end{equation}
	then
	\begin{align}\label{eq:chainQESnorm_simp}
		\widetilde{H}^{\uparrow,\hat{f}_\mathrm{full}}_\alpha(S_1^n | \CP_1^n E_n)_\rho \geq 0.
	\end{align}
\end{restatable}

At a high level, the proofs of both theorems proceed in the following way. We first introduce projection channels that project the infinite-dimensional spaces onto finite-dimensional truncation spaces, allowing us to apply the finite-dimensional version of the marginal-constrained $f$-weighted entropy accumulation theorem from above, alongside parts of the proof for the finite-dimensional MEAT. 

For the marginal-constrained $f$-weighted entropy accumulation theorem, the key step is to establish the uniform continuity of the $f$-weighted entropy based on the uniform continuity of the conditional \Renyi entropy. We then take the limit to infinity along the truncations introduced by the projection maps, which suffices to prove the infinite-dimensional marginal-constrained $f$-weighted entropy accumulation theorem.

However, this approach does not extend to the standard marginal-constrained entropy accumulation theorem because the objective in $h^\uparrow_{\alpha}$ is not uniformly continuous, and hence we cannot swap the limit of projections with the infimum. Moreover, even worse, in infinite dimensions, the set of normalized states is not compact, which was a core requirement in the proof of \cite[Theorem~4.1]{Arqand_2025} for applying strong duality. Changing the topology does not help either: the set of normalized states is weak* compact by the Banach--Alaoglu theorem; however, the $f$-weighted entropy is not weak* lower semi-continuous, and hence minimax theorems (e.g., Sion's) do not apply.

Given these obstacles, we instead use the Fenchel--Moreau theorem \cite[Theorem~13.37]{Bauschke_2017} to establish the required strong duality. The advantage here is that it does not require the compactness of the underlying sets. It does, however, require introducing a perturbation function and proving its lower semicontinuity. Equipped with this strong duality result, which generalizes to infinite dimensions, we prove our infinite-dimensional marginal-constrained entropy accumulation theorem (iMEAT). For detailed proofs, we refer to \cref{App:Inf_MEAT_Proof}.

%--------------------------------------------------------
%                   Security Proof
%--------------------------------------------------------

\section{Security Proof}\label{sec:Security_proof}
Equipped with an infinite-dimensional marginal constrained entropy accumulation theorem, we are now ready to provide a full security proof for discrete-modulated continuous variable QKD protocols. After two short sections with preparations, we will first prove security for the fixed-length protocol, followed by security for the variable-length protocol.

\subsection{Fixed Length}\label{sec:Fixed_length_security}
We first show the fixed-length security of the protocol described in \cref{sec:ProtocolDescription}. Thus, the goal of this section is to prove the following theorem.

\begin{theorem}[Fixed-Length Security for Infinite Dimensional Side Information]\label{thrm:Fixed_length_security}
Let $\{\mathcal{N}_j\}_{j=1}^n$ be the channels defined in \cref{def:MEAT_channel_with_Eve} that implement both the QKD channels $\EATchannQKD_j$ and Eve's corresponding attack channels $\mathcal{E}_j$ for each $j=1,\dots,n$ and let $\mathcal{E}_0$ implement Eve's initial attack. Let $\alpha \in (1,2)$, and $\psi_{A_1^n(A')_1^n}$ be the states sent by Alice according to the source replacement scheme. Furthermore, let $\ePA, \eEV \in (0,1]$ and let the state $\rho_{S_1^n \hat{C}_1^n E_n I_1^n}$ be generated from channels $\{\mathcal{N}_j\}$ according to:
\begin{enumerate}[label=\roman*)]
    \item $\pTr{(A')_1^n}{\psi_{A_1^n(A')_1^n}} = \bigotimes_{j=1}^n \sigma_{A_j}^{(j)}$.
    \item $\omega_{A_1^nB_1^n E_0}$ is a purification of $\mathcal{E}_0\left(\psi_{A_1^n(A')_1^n}\right)$.
    \item $\omega_{A_1^nB_1^n E_0 I_0} = \omega_{A_1^nB_1^n E_0} \otimes \pure{\perp}_{I_0}$,
    \item $\rho_{S_1^n \hat{C}_1^n E_n I_1^n } = \mathcal{N}_n \circ \dots \circ \mathcal{N}_1 [\omega_{A_1^nB_1^n E_0 I_0}].
    $
\end{enumerate}
Most importantly, all registers besides $S_1^n$ and $\CP_1^n$ can be infinite dimensional.

Then, the fixed-length protocol from \cref{sec:ProtocolDescription} is $\ePA$-secret and $\eEV$-correct, hence $(\ePA+\eEV)$-secure, when the length $\ell$ of the final key satisfies:
\begin{equation}
    \ell \le n \hUpQKD - \leak - \ceil{\log \frac{1}{\eEV}} - \frac{\alpha}{\alpha-1} \log \frac{1}{\ePA} + 1,
\end{equation}
where $\leak$ is the length of the error correction string. Moreover, the single-round quantity $\hUpQKD$ is given by:
\begin{equation}
    \hUpQKD = \begin{aligned}[t]
        \inf_{\mbf{q} \in \Sacc} \inf_{\nu \in \Sigma} \Bigg(&\frac{\alpha}{\alpha-1} D\left(\mbf{q} || \bsym{\nu}_{\CP}\right)\\
        &+ \sum_{\cP} q(\cP) \renyiSandUp_{\alpha}(S|E I)_{\nu_{|\cP}} \Bigg),
    \end{aligned}
\end{equation}
where $\Sacc$ is the closed and convex acceptance set and $\Sigma$ is the marginal convex range of the channels $(\EATchannQKD_j \otimes \id_E)$ and defined by:
\begin{align}
    \Sigma_j &\defvar \begin{aligned}[t] \Big\{ & \left(\EATchannQKD_j \otimes \id_E\right)\left[\tau_{A_j B_j E}\right] | \\ &\tau \in \dop{=}(A_j B_j E) \text{ pure}, \tau_{A_j} = \sigma_{A_j}^{(j)} \Big\},
    \end{aligned} \\
    \Sigma &\defvar \mathrm{conv}\left( \bigcup_{j=1}^n \Sigma_j \right).
\end{align}
\end{theorem}

\begin{proof}
We prove correctness and secrecy separately, as together they directly imply security. The $\epsilon_{\mathrm{EV}}$ correctness of the objective DM CV-QKD protocol follows from 2-universal hashing,
\begin{equation}
\begin{aligned}
    &\mathrm{Pr}\left[ K_A \neq K_B \land \OmegaAcc \right] \leq \mathrm{Pr}\left[ K_A \neq K_B \land \OmegaEV\right]\\
    &\leq \mathrm{Pr}\left[ S_1^n \neq \bar{S}_1^n \land \OmegaEV \right] \leq \mathrm{Pr}\left[ \OmegaEV | S_1^n \neq \bar{S}_1^n \right]\\
    & \leq 2^{- \left\lceil \log\left( \frac{1}{\epsilon_{\mathrm{EV}}} \right) \right\rceil}\\
    &\leq \epsilon_{\mathrm{EV}},
\end{aligned}
\end{equation}
where $\OmegaEV$ is the event of successfully passing the error-verification step and $\OmegaAcc$ is the event of the protocol not aborting.

The remaining task is to prove secrecy, for which will use the iMEAT we stated in the previous section. We denote by \(\rho\) the state extended by the public error-correction and verification transcript \(L\), before conditioning on acceptance, and by \(\rho_{|\OmegaAcc}\) the state after conditioning on acceptance. Starting from the security definition of \cref{eq:security}, restricting the sum to exactly one key length, and applying our infinite-dimensional LHL, \cref{prop:inf-LHL}, we obtain
\begin{equation}
    \begin{aligned}
 &\frac{1}{2}\mathrm{Pr}\left[\OmegaAcc\right]      \Bigg\lVert \sigma_{K_BI_1^nELG| \OmegaAcc} - \mathbbm{U}_{K_B}^\ell \otimes \sigma_{I_1^nELG| \OmegaAcc} \Bigg\rVert_1  \\
 & \leq \mathrm{Pr}\left[\OmegaAcc\right] 2^{\frac{(1-\alpha)}{\alpha}} 2^{\frac{(1-\alpha)}{\alpha} \left( \renyiSandUp_{\alpha}(S_1^n |I_1^n EL)_{\rho_{|\mathrm{acc}}} - \ell \right)}\\
 & = \mathrm{Pr}\left[ \OmegaEV | \OmegaAT \right] \mathrm{Pr}\left[ \OmegaAT \right] 2^{\frac{(1-\alpha)}{\alpha} \left(\renyiSandUp_{\alpha}(S_1^n |I_1^n EL)_{\rho_{|\mathrm{acc}}} - \ell +1 \right)},
    \end{aligned}
\end{equation}
where we used that 2-universal hashing is $\sqrt{2}$-randomizing, see \cite[Lemma~3.5.5]{Kamin_2026}. Recall that $\mathbbm{U}^\ell$ denotes a fully mixed state of dimension $2^\ell$ and the register $G$ stores the choice of the hash function and $L$ contains the information exchanged during error correction. Next, we simplify the entropic term: 
\begin{equation}
    \begin{aligned}
        &\renyiSandUp_{\alpha}(S_1^n|I_1^n EL)_{\rho_{|\OmegaAcc}} = \renyiSandUp_{\alpha}(S_1^n|I_1^n EL)_{\rho_{|\Omega_{\mathrm{EV} \land \mathrm{AT}}}}\\
        & \geq \renyiSandUp_{\alpha}(S_1^n|I_1^nEL)_{\rho_{|\OmegaAT}} - \frac{\alpha}{(\alpha-1)} \log\left( \frac{1}{\mathrm{Pr}\left[ \OmegaEV | \OmegaAT \right]} \right)\\
        & \geq \renyiSandUp_{\alpha}(S_1^n|I_1^nE)_{\rho_{|\OmegaAT}}- \lambda_{\mathrm{EC}} - \left\lceil \log\left( \frac{1}{\epsilon_{\mathrm{EV}}} \right) \right\rceil\\
        &~~~- \frac{\alpha}{(\alpha-1)} \log\left( \frac{1}{\mathrm{Pr}\left[ \OmegaEV | \OmegaAT \right]} \right) \ , 
    \end{aligned}
\end{equation}
where the first inequality is \cite[Lemma B.5]{Dupuis_2020}, whose proof extends to the conditioning system being infinite-dimensional and the second is the chain rule for classical side-information, which we prove for quantum states with separable Hilbert spaces in \cref{prop:cl-chain-rule}.

We now wish to bound $\renyiSandUp_{\alpha}(S_1^n|I_1^nE)_{\rho_{|\OmegaAT}}$. As we consider the source-replacement scheme (see Section \ref{sec:sourceReplacement}), we are guaranteed $\pTr{(A')_1^n}{\psi_{A_1^n(A')_1^n}} = \bigotimes_{j=1}^n \sigma_{A_j}^{(j)}$, i.e.~Alice's marginal state is constrained. Thus, we search for a lower bound of the Rényi entropy from above for all states with a given marginal. 

The iMEAT, \cref{thrm:MEAT_inf} which we stated in Section~\ref{sec:InfDimMEAT}, provides such a bound. Conditioned on passing the acceptance test, we obtain
\begin{equation}
     \renyiSandUp_{\alpha}(S_1^n | I_1^nE)_{\rho_{|\OmegaAT}} \geq n h_{\alpha}^{\uparrow} - \frac{\alpha}{(\alpha-1)} \log\left( \frac{1}{\mathrm{Pr}\left[ \Omega_{\mathrm{AT }}\right]} \right).
\end{equation}

Here, $h_{\alpha}^{\uparrow}$ is calculated as
\begin{equation}\label{eq:fixedLengthBound}
    h_{\alpha}^{\uparrow} =\begin{aligned}[t] \inf_{\mbf{q} \in \Sacc} \inf_{\nu \in \Sigma'}  \Bigg( &\frac{\alpha}{\alpha-1} D\left(\mbf{q} || \bsym{\nu}_{\CP}\right) \\ 
    &+ \sum_{\cP} q(\cP) \renyiSandUp_{\alpha}(S|\bar{E} I \tilde{E})_{\nu_{|\cP}} \Bigg),
    \end{aligned}
\end{equation}
where $\bar{E}$ contains the embedded registers $\bar{E}_j \defvar B_{j+1}^n E_{j} I_0^{j-1}$ and $\Sigma'$ is the marginal-constrained convex range, see \cite[Definition]{Arqand_2025} of the channels $\mathcal{N}_j$ defined by
\begin{align}
    \Sigma' &\defvar \mathrm{conv}\left( \bigcup_{j=1}^n \Sigma_j'\right), \\
    \Sigma'_j &\defvar \begin{aligned}[t]
         \Big\{ &\left(\mathcal{N}_j \otimes \id_{\tilde{E}}\right)\left[\omega_{A_j B_j^n E_{j-1} I_0^{j-1} \tilde{E}}\right] | \\
         &\omega \in \dop{=}(A_j B_j^n E_{j-1} I_0^{j-1} \tilde{E}), \; \omega_{A_j} = \sigma_{A_j}^{(j)} \Big\}.
        \end{aligned}
\end{align}

Then, the only remaining part to prove is the claim about the marginal convex range. We will show this in three steps. First, we relate the marginal-constrained convex range $\Sigma'$ built from the channels $\mathcal{N}_j$ to the marginal-constrained convex range $\Sigma''$ built from the channels $\EATchannQKD_j$, effectively allowing an arbitrary input for $\EATchannQKD_j$ instead of the one generated by $\mathcal{E}_j$. Secondly, we define another set $\breve{\Sigma}$ containing extensions of the states in $\Sigma''$ and is related to $\Sigma$ through an isometry. Finally, we relate these sets by showing that $\hUpQKD$ evaluated over $\Sigma$ is a lower bound due to data processing or isometric invariance of the \Renyi entropy.

\begin{widetext}
We begin with the definition of $\Sigma''$. For notational simplicity, let us define $\bar{E}_{j} \defvar B_{j+1}^n E_{j} I_0^{j-1}$. It then follows that
\begin{align}
    \Sigma'_j &= \left\{ (\EATchannQKD_j \otimes \id_{B_{j+1}^n E_j I_0^{j-1} \tilde{E}}) \circ (\mathcal{E}_j \otimes \id_{A_j\tilde{E}}) [\omega_{A_j B_j^n E_{j-1} I_0^{j-1} \tilde{E}}] \mid \omega \in \dop{=}(A_j B_j^n E_{j-1} I_0^{j-1} \tilde{E}), \omega_{A_j} = \sigma_{A_j}^{(j)} \right\} \\
    &= \left\{ (\EATchannQKD_j \otimes \id_{\bar{E}_j \tilde{E}})[\tau_{A_j B_j \bar{E}_j \tilde{E}}] \mid \tau = (\mathcal{E}_j \otimes \id_{A_j\tilde{E}})[\omega_{A_j B_j^n E_{j-1} I_0^{j-1} \tilde{E}}], \omega \in \dop{=}(A_j B_j^n E_{j-1} I_0^{j-1} \tilde{E}), \omega_{A_j} = \sigma_{A_j}^{(j)} \right\} \\
    &\subseteq \left\{ (\EATchannQKD_j \otimes \id_{\bar{E}_j \tilde{E}})[\tau_{A_j B_j \bar{E}_j \tilde{E}}] \mid \tau \in \dop{=}(A_j B_j \bar{E}_j \tilde{E}), \tau_{A_j} = \sigma_{A_j}^{(j)} \right\}\\
    &= \Sigma''_j
\end{align}
\end{widetext}
Thus, the marginal convex ranges contain each other, meaning that
\begin{equation}
    \Sigma' = \mathrm{conv}\left( \bigcup_{j=1}^n \Sigma'_j \right) \subseteq \mathrm{conv}\left( \bigcup_{j=1}^n \Sigma''_j \right) \eqqcolon \Sigma'',
\end{equation}
which completes the first step.

For step two, let $\breve{E}$ be a separable Hilbert space acting as a purifying register for $A_j B_j \bar{E}_j \tilde{E}$ such that $\tau_{A_j B_j \bar{E}_j \tilde{E} \breve{E}}$ is pure for $\tau \in \dop{=}\left(A_j B_j \bar{E}_j \tilde{E} \right)$ and all $j$. We can assume $\breve{E}$ to be separable because all $A_j$, $B_j$, $E_j$ and $\tilde{E}$ are separable. We define
\begin{equation}
    \breve{\Sigma}_j \defvar \begin{aligned}[t]
         \Big\{ &\left(\EATchannQKD_j \otimes \id_{\bar{E}_j \tilde{E} \breve{E}} \right)\left[\tau_{A_j B_j \bar{E}_{j} \tilde{E} \breve{E} } \right] | \\
         &\tau \in \dop{=}(A_j B_j \bar{E}_{j} \tilde{E} \breve{E}) \text{ pure}, \; \tau_{A_j} = \sigma_{A_j}^{(j)} \Big\},
        \end{aligned}
\end{equation}
where we note that because the input states that define $\breve{\Sigma}_{j}$ are purifications of input states that define $\Sigma_{j}''$. Hence, for a $j \in [n]$ and a state $\sigma \in \Sigma_{j}''$, there exists an extension of this state in $\breve{\Sigma}_{j}$.

Furthermore, we define the corresponding marginal-constrained convex range $\breve{\Sigma}$, again see \cite[Definition~4.2]{Arqand_2025}, by
\begin{align}
    \breve{\Sigma} \coloneq \mathrm{conv}\left( \bigcup_{j=1}^n \breve{\Sigma}_{j} \right) \ , 
\end{align}
which by construction embeds the registers $\bar{E}_{j} \tilde{E} \breve{E}$ in $\bar{E} \tilde{E} \breve{E}$. The register $\bar{E}$ in which we embed all registers $\bar{E}_j$ exists since all $\bar{E}_j$ are separable and thus have a countable basis.

Similarly, let $\hat{E}$ be a separable Hilbert space that can act as a purifying register for $A_j B_j$ such that $\tau_{A_j B_j \hat{E}}$ is pure for all $j$, and define $\Sigma_j$ as in the theorem statement as:
\begin{equation}
    \Sigma_j \defvar \begin{aligned}[t] 
        \Big\{ & \left(\EATchannQKD_j \otimes \id_{\hat{E}} \right)\left[\tau_{A_j B_j \hat{E}}\right] | \\ &\tau \in \dop{=}(A_j B_j \hat{E}) \text{ pure}, \tau_{A_j} = \sigma_{A_j}^{(j)} \Big\}.
    \end{aligned}
\end{equation}

To conclude step two, we construct a global isometry $U$ mapping $\bar{E} \tilde{E} \breve{E}$ to $\hat{E}$. Since both $\bar{E} \tilde{E} \breve{E}$ and $\hat{E}$ are separable Hilbert space, there exists a unitary $U: \bar{E} \tilde{E} \breve{E} \to \hat{E} $ which can be defined by the following construction. Let $\{\ket{e_k}\}_{k\in \N}$ and $\{\ket{f_j}\}_{j\in \N}$ be bases of $\bar{E} \tilde{E} \breve{E}$ and $\hat{E}$ respectively. Then, $U$ can be uniquely defined by its action on the basis
\begin{equation}
    U \ket{e_k} = \ket{f_k},
\end{equation}
for all $k \in \N$. Thus, there exists a unitary $U$ such that 
\begin{equation}
    U \breve{\Sigma} U^{\dagger} \subseteq  \Sigma.
\end{equation}

As the last step of this proof, we show that $\hUpQKD$ is a bound on $h_{\alpha}^{\uparrow}$ originating form the channels $\mathcal{N}_j$. We find 
{\allowdisplaybreaks
\begin{align}
    h_{\alpha}^{\uparrow}
    &\geq \begin{aligned}[t]
        \inf_{\mbf{q} \in \Sacc} \inf_{\nu \in \Sigma''} &\Bigg( \frac{\alpha}{\alpha-1} D\left(\mbf{q} \Vert \bsym{\nu}_{\CP}\right) \\ 
        &+ \sum_{\cP} q(\cP) \renyiSandUp_{\alpha}(S|\bar{E} \tilde{E} I )_{\nu_{|\cP}} \Bigg)
    \end{aligned}  \\
    &\geq \begin{aligned}[t]
        \inf_{\mbf{q} \in \Sacc} \inf_{\nu \in \breve{\Sigma}} \Bigg( &\frac{\alpha}{\alpha-1} D\left(\mbf{q} \Vert \bsym{\nu}_{\CP}\right) \\ 
        &+ \sum_{\cP} q(\cP) \renyiSandUp_{\alpha}(S|\bar{E} \tilde{E} I )_{\nu_{|\cP}} \Bigg)
    \end{aligned}  \\
    &\geq \begin{aligned}[t]
        \inf_{\mbf{q} \in \Sacc} \inf_{\nu \in \breve{\Sigma}} \Bigg( &\frac{\alpha}{\alpha-1} D\left(\mbf{q} \Vert \bsym{\nu}_{\CP}\right) \\ 
        &+ \sum_{\cP} q(\cP) \renyiSandUp_{\alpha}(S|\bar{E} \tilde{E} \breve{E} I )_{\nu_{|\cP}} \Bigg)
    \end{aligned}  \\
    &\geq \begin{aligned}[t] 
        \inf_{\mbf{q} \in \Sacc} \inf_{\nu \in \Sigma} &\Bigg( \frac{\alpha}{\alpha-1} D\left(\mbf{q} \Vert \bsym{\nu}_{\CP}\right) \\ 
        &+ \sum_{\cP} q(\cP) \renyiSandUp_{\alpha}(S|\hat{E} I )_{\nu_{|\cP}} \Bigg),
    \end{aligned} \\
    &= \hUpQKD,
\end{align}
}
where the first inequality follows because $\Sigma' \subseteq \Sigma''$, and the second inequality holds since $\breve{\Sigma}$ contains extensions of every state contained in $\Sigma''$. The third inequality follows by strong subadditivity of \Renyi entropies \cref{prop:Inf-strong-subadd} (or see e.g. \cite[Corollary~5.4]{Tomamichel_2016} for finite dimensions). The fourth inequality is due to $U \breve{\Sigma} U^{\dagger} \subseteq  \Sigma$ and the data processing inequality, see \cref{app:inf-dim-DPI}. Finally, the equality to $\hUpQKD$ follows by renaming $\hat{E}$ to $E$ and the theorem statement follows.
\end{proof}

\subsection{Variable Length}\label{sec:Variable_length_security}
Next, we proceed with the variable-length security proof. Our goal is to prove the following theorem.

\begin{theorem}[Variable-Length Security for Infinite Dimensional Side Information]\label{thrm:variable_length_security}
Let $\{\mathcal{N}_j\}_{j=1}^n$ be the channels defined in \cref{def:MEAT_channel_with_Eve} that implement both the QKD channels $\EATchannQKD_j$ and Eve's corresponding attack channels $\mathcal{E}_j$ for each $j=1,\dots,n$ and let $\mathcal{E}_0$ implement Eve's initial attack. Let $\alpha \in (1,2)$, and $\psi_{A_1^n(A')_1^n}$ be the states sent by Alice according to the source replacement scheme. Furthermore, let $\ePA, \eEV \in (0,1]$ and let the state $\rho_{S_1^n \hat{C}_1^n E_n I_1^n}$ be generated from channels $\{\mathcal{N}_j\}$ according to:
\begin{enumerate}[label=\roman*)]
    \item $\pTr{(A')_1^n}{\psi_{A_1^n(A')_1^n}} = \bigotimes_{j=1}^n \sigma_{A_j}^{(j)}$.
    \item $\omega_{A_1^nB_1^n E_0}$ is a purification of $\mathcal{E}_0\left(\psi_{A_1^n(A')_1^n}\right)$.
    \item $\omega_{A_1^nB_1^n E_0 I_0} = \omega_{A_1^nB_1^n E_0} \otimes \pure{\perp}_{I_0}$,
    \item $\rho_{S_1^n \hat{C}_1^n E_n I_1^n } = \mathcal{N}_n \circ \dots \circ \mathcal{N}_1 [\omega_{A_1^nB_1^n E_0 I_0}].
    $
\end{enumerate}
Most importantly, all registers besides $S_1^n$ and $\CP_1^n$ can be infinite dimensional. 
Finally, for each \(j\) and each history \(\hat c_1^{j-1}\), fix a tradeoff function \(f_{|\hat c_1^{j-1}}:\widehat{\mathcal C}_j\to\mathbb R\).

Then, the variable-length protocol from \cref{sec:ProtocolDescription} is $\ePA$-secret and $\eEV$-correct, hence $(\ePA+\eEV)$-secure, when the length $\ell_{\mathrm{var}}$ of the final key satisfies:
\begin{equation}\label{eq:l_var_ffull}
    \ell_{\mathrm{var}} \le \begin{aligned}[t]
        &\hat{f}_{\text{full}}(\cP_1^n) - \leak(\cP_1^n) - \ceil{\log \frac{1}{\eEV}} \\
        &- \frac{\alpha}{\alpha-1} \log \frac{1}{\ePA} + 1
    \end{aligned} 
\end{equation}
where
\begin{align}
    \hat{f}_{\text{full}}(\cP_1^n) &\defvar \sum_{j=1}^n \left( f_{|\cP_1^{j-1}}(\cP_j) + \kapup_j \right), \\
    \kapup_j &\defvar \inf_{\nu \in \Sigma_j} \widetilde{H}_{\alpha}^{\uparrow,f_{|\cP_1^{j-1}}}(S_j | \CP_j I_j E)_{\nu}
\end{align}
and the sets $\Sigma_j$ is defined as
\begin{equation}
    \Sigma_j \defvar
    \begin{aligned}[t]
        \Big\{ &\left(\EATchannQKD_j \otimes \id_E\right)\left[\tau_{A_j B_j E}\right] |\\ &\tau \in \dop{=}(A_j B_j E) \text{ pure}, \tau_{A_j} = \sigma_{A_j}^{(j)} \Big\}.
    \end{aligned}
\end{equation}
Finally, $\leak(\cP_1^n)$ is the length of the error correction data used given the announcements $\cP_1^n$.

In particular, the same security statement holds if we define $\hat{f}_{\text{prot}}(\cP_1^n) \defvar \sum_{j=1}^n \left( f_{|\cP_1^{j-1}}(\cP_j) + \kapupbnd_j \right)$ where $\kapupbnd_j$ is any lower bound on $\kapup_j$ and choose $\ell_{\mathrm{var}}$ as
\begin{equation}\label{eq:l_var_fprot}
    \ell_{\mathrm{var}} \le \begin{aligned}[t]
        &\hat{f}_{\text{prot}}(\cP_1^n) - \leak(\cP_1^n) - \ceil{\log \frac{1}{\eEV}} \\
        &- \frac{\alpha}{\alpha-1} \log \frac{1}{\ePA} + 1.
    \end{aligned} 
\end{equation}
\end{theorem}

\begin{proof}
We first show how \cref{eq:l_var_ffull} follows from \cref{eq:l_var_fprot}. Let $\cP_1^n$ be arbitrary, then:
\begin{equation}
    \begin{aligned}
    \hat{f}_{\text{prot}}(\cP_1^n) &= \sum_{j=1}^n \left( f_{|\cP_1^{j-1}}(\cP_j) + \kapupbnd_j \right) \\
    &\leq \sum_{j=1}^n \left( f_{|\cP_1^{j-1}}(\cP_j) + \kapup_j \right) = \hat{f}_{\text{full}}(\cP_1^n) 
\end{aligned}
\end{equation}
Thus, $\ell_{\mathrm{var}}(\hat{f}_{\text{prot}}) \leq \ell_{\mathrm{var}}(\hat{f}_{\text{full}})$ and the key length defined in terms of $\hat{f}_{\text{prot}}$ is secure since any variable key with length less than or equal to $\ell_{\mathrm{var}}$ as in \cref{eq:l_var_ffull} is secure.

Now, we prove the remainder of the theorem with regards to the security claim with respect to $\hat{f}_{\text{full}}$. This claim follows essentially from \cite[Thm.~12]{Kamin_2025a} with our infinite-dimensional LHL, \cref{prop:inf-LHL}, as long as $\widetilde{H}_{\alpha}^{\uparrow,\hat{f}_{\text{full}}}(S_1^n | \CP_1^n I_1^n E_n)_{\rho} \geq 0$ is maintained.

By the $f$-weighted marginal-constrained entropy accumulation theorem, \cref{thrm:fweighted_meat_simp}, it holds $\widetilde{H}_{\alpha}^{\uparrow,\hat{f}_{\text{full}}}(S_1^n | \CP_1^n I_1^n E_n)_{\rho} \geq 0$ if 
\begin{equation}
    \kapup_j \leq  \inf_{\nu \in \Sigma_j'} \widetilde{H}_{\alpha}^{\uparrow,f_{|\cP_1^{j-1}}}(S_j | \CP_j I_j \bar{E}_j \tilde{E})_{\nu},
\end{equation} 
where
\begin{equation}
    \Sigma'_j \defvar \begin{aligned}[t]
         \Big\{ &\left(\mathcal{N}_j \otimes \id_{\tilde{E}}\right)\left[\omega_{A_j B_j^n E_{j-1} I_0^{j-1} \tilde{E}}\right] | \\
         &\omega \in \dop{=}(A_j B_j^n E_{j-1} I_0^{j-1} \tilde{E}), \; \omega_{A_j} = \sigma_{A_j}^{(j)} \Big\},
        \end{aligned}
\end{equation}
and $\bar{E}_j \defvar B_{j+1}^n E_{j} I_0^{j-1}$.
Therefore, as in the proof of \cref{thrm:Fixed_length_security}, the only claim left to show is the change in marginal-constrained ranges. Then, following the same steps as in \cref{thrm:Fixed_length_security}, now applied to the simpler sets $\Sigma_j$ instead of $\Sigma$, the theorem follows.
\end{proof}

\subsection{Interim Résumé}
We have successfully proven security for both the fixed- and the variable-length variant of the abstract DM CV-QKD protocol described in \cref{sec:ProtocolDescription}. For the fixed-length security statement, it remains to calculate  $\hUpQKD$ as given in Eq. (\ref{eq:fixedLengthBound}). Most notably, this requires optimizing over $\renyiSandUp_{\alpha}(S|IE)$, involving Eve's infinite-dimensional register as side-information. % and the optimization variable $\mbf{q}$, which, for heterodyne measurements, is a continuous quantity. 
For the variable-length security statement, similar problems arise. In addition, it remains to find a suitable tradeoff function $f$.

While our method eventually overcomes all the aforementioned obstacles, we begin by developing a dimension-reduction argument for Rényi entropies, similar to that in Ref. \cite{Upadhyaya_2021} for von Neumann entropies, to help us rigorously address infinite-dimensional registers in the numerical optimisation problem.

%--------------------------------------------------------
%                   Dimension Reduction
%--------------------------------------------------------

\section{Dimension Reduction}\label{sec:DimReduction}
\newcommand{\Pibar}{\widebar{\Pi}}

In this section, we derive our dimension-reduction argument, which enables us to rigorously account for infinite-dimensional quantum registers. Although the dimension reduction argument and its proof represents a substantial contribution of this paper, to improve readability we will only state the main dimension-reduction theorems in the main text and defer full proofs to \cref{App:Dim_Red_Proof}.

We begin with introductory remarks and definitions required for the full theorem statement, followed by the theorem statement for the variable-length and the fixed-length problem.

We first consider the fixed-length setting, for which the general formulation simplifies under a natural condition on the marginal convex range. Specifically, we assume that the marginal convex range, see \cite[Def. 4.2]{Arqand_2025}, collapses to a single set. This is the case if all channels $\EATchannQKD_j$ and the states $\sigma_{A_{j}}^{(j)}$ are isomorphic to a single channel and state. If this is not the case, then all feasible states are a probabilistic mixture of states produced by each channel $\EATchannQKD_j$ and the standard conversions to convex optimization problems from Refs. \cite{Kamin_2025, Chung_2025, Navarro_2026} do not apply directly. This simplification is not required for variable-length key rates and their optimization in $\kapup_j$.

We begin by defining the tagged QKD channel and statistics generating map, which will allow us to describe the action of the infinite dimensional channel and map up to certain corrections.

\begin{definition}[Tagged Channels and Statistic-Generating Maps]\label{def:Tagged_channels_maps}
    Let \(\EATchannQKD_j\) be the QKD channel defined in \cref{def:QKD channel}, constructed from Alice's and Bob's POVMs \(\{M_k^A\}_{k=1}^{N_{\mathrm{St}}}\) and \(\{R_k^B\}_{k \in \mathcal{S}}\), \(\{M_k^B\}_{k=1}^{N_{\mathrm{POVMs}}}\), respectively.

    Let $\Pi$ be a projection on Bob's system $B$, and define
    \begin{equation}
        B_{N_c} \defvar \supp(\Pi), \qquad \Pibar \defvar I_B-\Pi,
    \end{equation}
    so that $\Pi$ projects onto $B_{N_c}$.
    
    For each of Bob's POVM elements, define the corresponding tagged POVM element on \(\widetilde{B} \defvar B_{N_c}\oplus\Span\{\ket{\star}\}\) by
    \begin{equation}\label{eq:tagged-POVM-elements}
        \widetilde{M}_k^B \defvar \left.\Pi M_k^B\Pi\right|_{B_{N_c}}\oplus 0_{\star},
    \end{equation}
    where $\left.\Pi M_k^B\Pi\right|_{B_{N_c}}$ indicates the restriction of the projected POVM element to the space $B_{N_c}$. In addition, define the POVM element
    \begin{equation}
        \widetilde{M}_\star^B \defvar 0_{B_{N_c}}\oplus\pure{\star},
    \end{equation}
    which completes the POVM elements $\{\widetilde{M}_k^B \}_{k=1}^{N_{\mathrm{POVMs}}}$ to a POVM on \( \widetilde{B} = B_{N_c}\oplus\Span\{\ket{\star}\}\), since
    \begin{equation}
            \begin{aligned}
            \sum_{k=1}^{d_B} \widetilde{M}_k^B + \widetilde{M}_{\star}^B
            &= \left.\Pi\left(\sum_{k=1}^{d_B}M_k^B\right)\Pi\right|_{B_{N_c}} \oplus \pure{\star} \\
            &= I_{B_{N_c}} \oplus \pure{\star} \\
            &= I_{B_{N_c} \oplus \Span\{\ket{\star}\} }.
        \end{aligned}
    \end{equation}
    Similarly, we construct the tagged region operators $\widetilde{R}_k^B$ for each $k \in \mathcal{S}$. Next, we define the tagged map \(\widetilde{\GMap}\). For each $(\annA,\annB) \in \mathcal{K}$, let
    \begin{equation}
        \widetilde{K}_{\annA,\annB} \defvar \begin{aligned}[t] \sum_{\substack{x \in \mathcal{X}_{\annA}^{\gen} \\  y \in \mathcal{Y}_\annB}} \sqrt{M_{\annA,x}^{A|\gen}} &\otimes \sqrt{\widetilde{R}_{\annB,y}^B} \otimes \ket{x}_X \otimes \ket{y}_Y  \\ 
        &\otimes \ket{g(\annA,\annB,y)}_{S_Q} \ket{\annA,\annB}_{\widetilde{I}}.
        \end{aligned}
    \end{equation}
    The corresponding discard Kraus operator, which completes $\widetilde{\GMap}$ to a CPTP map, is defined by
    \begin{equation}
        \widetilde K_\perp \defvar \sqrt{I-\sum_{(\annA,\annB)\in\mathcal{K}} \widetilde{K}_{\annA,\annB}^\dagger \widetilde{K}_{\annA,\annB}} \otimes \ket{\perp}_{XY S_Q\widetilde{I}},
    \end{equation}
    where the identity acts on \(A\bigl(B_{N_c}\oplus\Span\{\ket{\star}\}\bigr)\). Moreover, since the kept Kraus operators \(\widetilde K_{\annA,\annB}\) vanish on the \(\star\)-subspace, the completion Kraus operator maps this subspace to the discard output with probability one.

    We define the tagged QKD channel \(\EATchannQKDTag_j\) as in \cref{def:QKD channel}, replacing \(\GMap\) with \(\widetilde{\GMap}\) and replacing Bob's POVM by
    \begin{equation}
        \{\widetilde M_k^B\}_{k=1}^{N_{\mathrm{POVMs}}}\cup\{\widetilde M_\star^B\}.
    \end{equation}
    For completeness, the test-round isometry is extended using this tagged POVM, with the additional outcome ``$\star$'' mapped to the discard announcement ``$\perp$''.

    Finally, we define the tagged statistic-generating map by retaining only the original, non-discard test outcomes. For every \(\sigma\in\dop{\leq}(AB_{N_c})\), let
    \begin{equation}\label{eq:tagged_stat_map}
        \probstTag_j[\sigma] \defvar \sum_{\cP\in\Ct}\sum_{\substack{(\annA,x,\annB,y) \\ \in \phi^{-1}(\cP)}}\Tr{\left(M_{\annA,x}^{A|\test}\otimes\Pi M_{\annB,y}^B\Pi\right)\sigma}\hat e_{\cP}.
    \end{equation}
    In particular, the tag outcome associated with \(\widetilde M_\star^B\) is not included in \(\probstTag_j\).
\end{definition}
\begin{remark}[Identification of projected operators]
    Throughout this section, we identify every operator on \(A\otimes B\) supported on \(A\otimes B_{N_c}\), where \(B_{N_c}\defvar\supp(\Pi)\), with its restriction to \(A\otimes B_{N_c}\). We also use the same symbol for its canonical extension to \(A\otimes\bigl(B_{N_c}\oplus \Span\{\ket{\star}\}\bigr)\). In particular, for such an operator \(X\), we  write \(X_{AB_{N_c}}\) and identify \(\pTr{B_{N_c}}{X}=\pTr{B}{X}\).
\end{remark}
\begin{remark}[Action of the tagged QKD channel]\label{rem:tagged-channel-action}
    By construction, \(\EATchannQKDTag_j\) removes coherences between the
    \(B_{N_c}\)- and \(\star\)-subspaces and maps the \(\star\)-component
    deterministically to the discard output \(\perp\). Moreover, on inputs
    supported on \(AB_{N_c}\), the tagged and original QKD channels have
    the same output after the measured systems are traced out. Indeed, for
    every \(X_{AB_{N_c}R}\) and every Bob POVM element \(M_k^B\),
    \begin{align*}
        &\pTr{B}{\left(I_A\otimes\sqrt{M_k^B}\otimes I_R\right)
        X\left(I_A\otimes\sqrt{M_k^B}\otimes I_R\right)} \notag\\
        &= \pTr{B}{X\left(I_A\otimes M_k^B\otimes I_R\right)}\\
        &= \pTr{B_{N_c}}{X\left(I_A\otimes\widetilde M_k^B\otimes I_R\right)} \\
        &=\pTr{B_{N_c}}{\left(I_A\otimes\sqrt{\widetilde M_k^B}\otimes I_R\right)
        X\left(I_A\otimes\sqrt{\widetilde M_k^B}\otimes I_R\right)}.
    \end{align*}
    The same identity remains valid after Alice's measurement and the subsequent classical postprocessing.
\end{remark}

\subsection{Variable-Length}\label{sec:DM_variable_length}
This subsection relies on several crucial special requirements that are satisfied in the QKD protocol we described in \cref{sec:ProtocolDescription}. First, we require that Alice decides if each round is a test or key generation round in a manner that solely depends on her choice of signal state for that round. Second, we assume Alice's method of deciding between test and generation round results in a \emph{fixed} probability of testing, $\gamma \in (0,1)$. If both requirements are satisfied, one can apply the results in this section.

\begin{lemma}[Normalized gentle measurement lemma]\label{lem:Norm_gentle_measurement_lemma}
    Let \(\sigma \in \dop{=}(A)\) be a normalized state and \(Q\) an orthogonal projection such that \(p\defvar\Tr{Q\sigma}>0\). Then
    \begin{equation}
        \frac{1}{2}\left\|\sigma-\frac{Q\sigma Q}{p}\right\|_1\leq\sqrt{1-p}.
    \end{equation}
That is, the trace distance between the initial state and the re-normalized post-projective-measurement state is bounded above as a function of the probability of the measurement outcome.
\end{lemma}
\begin{proof}
    %This uses square root fidelity!
    Let \(\ket{\psi}_{AR} \in \dop{=}(AR)\) be a purification of \(\sigma\). Then
    \begin{equation}
        \ket{\psi_Q}_{AR}\defvar\frac{(Q\otimes I_R)\ket{\psi}_{AR}}{\sqrt{p}}
    \end{equation}
    is a purification of \(Q\sigma Q/p\), and
    \begin{equation}
        \abs{\braket{\psi}{\psi_Q}}=\frac{\bra{\psi}(Q\otimes I_R)\ket{\psi}}{\sqrt{p}}=\sqrt{p}.
    \end{equation}
    Hence, Uhlmann's theorem (see \cite{Hou_2012} for an elementary proof of Uhlmann's theorem in infinite dimensions) gives
    \begin{equation}
        F\left(\sigma,\frac{Q\sigma Q}{p}\right)\geq\sqrt{p}.
    \end{equation}
    The claim now follows from the Fuchs--van de Graaf inequality:
    \begin{equation}
        \frac{1}{2}\left\|\sigma-\frac{Q\sigma Q}{p}\right\|_1\leq\sqrt{1-F\left(\sigma,\frac{Q\sigma Q}{p}\right)^2}\leq\sqrt{1-p}.
    \end{equation}
\end{proof}

Making use of this lemma we prove the first step towards the dimension reduction statement, namely a bound on the entropy of the infinite dimensional state in terms of the entropy of the tagged state and channel.

\begin{lemma}[Entropy bound for the tagged QKD channel]\label{lem:tagged_QKD_entropy_bound}
    Fix \(1\leq j \leq n\) and \(\alpha\in(1,2)\). Let \(\hat{\rho}_{AB|\gen}\in\dop{=}(AB)\), let \(\hat{\rho}_{ABE|\gen}\) be a purification of \(\hat{\rho}_{AB|\gen}\), and define
    \begin{align}
        \rho_{AB_{N_c}|\gen}&\defvar\left(I_A\otimes\Pi\right)\hat{\rho}_{AB|\gen}\left(I_A\otimes\Pi\right),\\
        W_{|\gen}& \defvar  1-\Tr{\rho_{AB_{N_c}|\gen}},\\
        \widetilde{\rho}_{A\widetilde{B}|\gen} &\defvar \begin{aligned}[t]
        &\rho_{AB_{N_c}|\gen} + \Bigl(\pTr{B}{\hat{\rho}_{AB|\gen}} \\
        &-\pTr{B_{N_c}}{\rho_{AB_{N_c}|\gen}} \Bigr) \otimes \pure{\star}_{\wt{B}}.
        \end{aligned}
    \end{align}
    Let \(\widetilde{\rho}_{A\widetilde{B}E\widetilde E|\gen}\) be a purification of \(\widetilde{\rho}_{A\widetilde{B}|\gen}\), and define
    \begin{align}
        \omega_{|\gen}&\defvar\left(\EATchannQKD_j\otimes\id_E\right)\left[\hat{\rho}_{ABE|\gen}\right],\\
        \widetilde{\omega}_{|\gen}&\defvar\left(\EATchannQKDTag_j\otimes\id_{E\widetilde E}\right)\left[\widetilde{\rho}_{A\widetilde{B}E\widetilde{E}|\gen}\right].
    \end{align}
    Then
    \begin{equation}
        \renyiSandUp_{\alpha}(S|IE)_{\omega_{|\gen}} \geq \renyiSandUp_{\alpha}(S|IE\widetilde E)_{\widetilde{\omega}_{|\gen}} - \Delta_{\alpha}\!\left(\sqrt{W_{|\gen}}\right).
    \end{equation}
\end{lemma}
\begin{proof}
    See \cref{App:Dim_Red_Proof}.
\end{proof}

With the previous lemma at hand, we can state the final dimension reduction proof.

\begin{theorem}\label{thrm:Dim_red_kappa_special}
    Fix \(1 \leq j \leq n \), \(\alpha\in(1,2)\) and assume the same conditions as in \cref{thrm:variable_length_security}. Throughout this theorem, \(\hat{\rho}\) denotes a normalized state on the possibly infinite-dimensional system \(AB\), whereas \(\rho\) denotes a projected, subnormalized state on \(AB_{N_c}\).

    Let \(\Pi\) be a finite-rank projection on Bob's system \(B\), with support \(B_{N_c}\), satisfying
    \begin{equation}
        \left[\Pi,M_k^B\right]=0
    \end{equation}
    for every \(k\) in Bob's test round POVMs. Write \(\Pibar=I_B-\Pi\) for the projection onto the complementary subspace. Suppose that Bob's test-round POVM contains an element \(V_1\) such that
    \begin{equation}
        \Pibar\leq rV_1
    \end{equation}
    for some \(r\geq1\). For each of Alice's outcomes \((\annA,x)\), let \(\nu_{\annA,x,1}\) denote the component of \(\bsym{\nu}_{|\test}\) corresponding to the joint probability, conditioned on a test round, that Alice obtains \((\annA,x)\) and Bob obtains the outcome associated with \(V_1\), that is
    \begin{equation}
        \nu_{\annA,x,1} = \Tr{\hat{\rho}_{AB|\test}\left(M_{\annA,x}^{A|\test}\otimes V_1\right)},
    \end{equation}
    where \(\hat{\rho}_{AB|\test}\in\dop{=}(AB)\) is the corresponding normalized test-round state. Finally, the tagged QKD channel \(\EATchannQKDTag_j\) and the tagged statistic-generating map \(\probstTag_j\) are defined as in \cref{def:Tagged_channels_maps}.

    Then, the normalization constant \(\kapup_j\) from \cref{thrm:variable_length_security} obeys the lower bound
    \begin{widetext}
    \begin{align}
        \begin{aligned}
            \kapup_j\geq & \inf_{\substack{\rho\in\dop{\leq}(AB_{N_c}),\ \bsym{\nu}_{|\test}\in\mathbb{P}(\Ct),\\ W_{\test}^{\annA,x},W_{\gen}\in[0,1]}}\begin{aligned}[t]
                \frac{\alpha}{1-\alpha}\log\Bigg(&\gamma\sum_{\cP\in\Ct}\bsym{\nu}_{|\test}(\cP)2^{\frac{\alpha-1}{\alpha}f(\cP)}\\
                &+(1-\gamma)2^{-\frac{\alpha-1}{\alpha}\left(\renyiSandUp_{\alpha}(S|IE\widetilde{E})_{\widetilde{\omega}_{|\gen}}-\Delta_{\alpha}\left(\sqrt{W_{\gen}}\right)-f(\gen)\right)}\Bigg)
            \end{aligned}\\
            &\quad\textrm{s.t. }\begin{aligned}[t]
                &\pTr{B_{N_c}}{\rho}\leq\sigma_A^{(j)},\\
                &\widetilde{\rho} = \rho + \left( \sigma_A^{(j)} - \pTr{B_{N_c}}{\rho} \right) \otimes \pure{\star}, \\
                &W_{\test}=\sum_{\annA,x}p(\annA,x|\test)W_{\test}^{\annA,x},\\
                &W=(1-\gamma)W_{\gen}+\gamma W_{\test},\; \Tr{\rho}=1-W,\\
                &\Tr{\rho_{|\gen}}=1-W_{\gen},\; \Tr{\rho_{|\test}}=1-W_{\test},\\
                &\Tr{\rho_{|(\test\wedge\annA,x)}}=1-W_{\test}^{\annA,x}, \quad
                \forall \;\annA \in \mathcal{C}^{A,\test}, \; x \in \mathcal{X}_{\annA}^{\test} \text{ s.t. } p(\annA,x|\test)>0\\
                &\widetilde{\omega}_{|\gen}=\EATchannQKDTag_j\left[\widetilde{\rho}_{|\gen}\right],\\
                &\probstTag_j[\rho_{|\test}]\leq\bsym{\nu}_{|\test}\leq\probstTag_j[\rho_{|\test}]+W_{\test}\bsym{\delta},\\
                &p(\annA,x|\test)W_{\test}^{\annA,x}\leq r\nu_{\annA,x,1} \; \forall \;\annA \in \mathcal{C}^{A,\test}, \; x \in \mathcal{X}_{\annA}^{\test}.
            \end{aligned}
        \end{aligned}
    \end{align}
    \end{widetext}
    Here, \(\widetilde{\rho}_{|\test}\) and \(\widetilde{\rho}_{|\gen}\) are obtained by conditioning \(\widetilde{\rho}\) on Alice's test and generation event. Their respective projected components are denoted by
    \(\rho_{AB_{N_c}|\test},\rho_{AB_{N_c}|\gen}
    \in\dop{\leq}(AB_{N_c})\). Moreover, \(E\widetilde{E}\) purifies \(\widetilde{\rho}_{|\gen}\), and the tagged QKD channel $\EATchannQKDTag_j$ is understood to act trivially on these registers. The inequality \(\pTr{B_{N_c}}{\rho}\leq\sigma_A^{(j)}\) means that \(\sigma_A^{(j)}-\pTr{B_{N_c}}{\rho}\) is positive semidefinite, while the vector inequalities are understood componentwise. Finally, the entries of \(\bsym{\delta}\) are defined as
    \begin{equation}
        \delta_{\cP}\defvar\sum_{\substack{(\annA,x,\annB,y)\in\phi^{-1}(\cP)}}\left\|M_{\annA,x}^{A|\test}\otimes \Pibar M_{\annB,y}^B\Pibar\right\|_{\infty}, 
    \end{equation}
    and for \(\epsilon\in[0,1]\) the correction term $\Delta_{\alpha}$ is given by 
    \begin{widetext}
        \begin{equation}
        \Delta_{\alpha}(\epsilon)\defvar\min\left\{
        \begin{aligned} &\log(1+\epsilon)+\frac{1}{\alpha-1}\log\!\left(1+\epsilon d_S^{\alpha-1}-\frac{\epsilon^\alpha}{(1+\epsilon)^{\alpha-1}}\right),\\ &\frac{\alpha}{\alpha-1}\log\!\left(1+\epsilon d_S^{\frac{\alpha-1}{\alpha}}\right),\\ &\log(1+\epsilon)+\frac{\alpha}{\alpha-1}\log\!\left(1+\epsilon d_S^{\frac{\alpha-1}{\alpha}}-\frac{\epsilon^{2-\frac{1}{\alpha}}}{(1+\epsilon)^{\frac{\alpha-1}{\alpha}}}\right)
        \end{aligned}
        \right\},
    \end{equation}

    \end{widetext}
        where \(d_S=\dim(S)\).
\end{theorem}
\begin{remark}
    We note that the marginal constraint $\mathrm{Tr}_{B_{N_c}}\left[\rho\right]\leq\sigma_A^{(j)}$ was suggested by Ref. \cite{Reichert_2026} already prior to this work as an improvement of the security proof framework of Ref. \cite{Kanitschar_2023}. Here, the constraint follows from different principles, directly from our dimension reduction argument.
\end{remark}

\begin{remark}
    We will call the constraints
    \begin{equation}
        p(\annA,x|\test)W_{\test}^{\annA,x}\leq r\nu_{\annA,x,1},
    \end{equation}
    the dimension reduction constraints and they play the same role as the squashing constraints found in \cite[Theorem 17]{Kamin_2025a}. Any differences stem from squashing maps satisfying $\probst[\rho_{|\test}] = \bsym{\nu}_{\CP}$.
\end{remark}
\begin{proof}
    See \cref{App:Dim_Red_Proof}.
\end{proof}

\begin{remark}
    Based on the formulation of \cref{thrm:Dim_red_kappa_special}, one can apply any of the reformulations of the \Renyi entropy in terms of convex functions as found in \cite{Kamin_2025a} using the duality argument in \cref{prop:Petz-Sandwiched-duality} in place of the finite-dimensional one used in Ref.~\cite{Kamin_2025a}. We believe that one can equally extend the duality argument required for the methods in Ref.~\cite{Navarro_2026} to infinite dimensions.
    
    Similarly, one can use the simplifications to bound the mixed $f$-weighted normalization constants $\kapmix_j$ as introduced in Ref.~\cite{Kamin_2026}.
\end{remark}

\subsection{Fixed-Length}
As before, we only consider the case where all QKD channels $\EATchannQKD_j$ are isomorphic to a single one. We therefore suppress the round index and write \(\sigma_A\), \(\EATchannQKDTag\), and \(\probstTag\) for the corresponding common objects.

%\subsubsection{Special Case: Single $\gen$ Announcement}
\begin{theorem}\label{thrm:Dim_red_halpha_special}
    Fix \(\alpha\in(1,2)\) and assume the same conditions as in \cref{thrm:Fixed_length_security}. Let \(\Pi\) be a finite-rank projection with support \(B_{N_c}\) satisfying \([\Pi,M_k^B]=0\) for every element \(M_k^B\) of Bob's test-round POVM, and define \(\Pibar \defvar I_B-\Pi\). Suppose that Bob's test POVM contains \(V_1\) that satisfies 
    \begin{equation}
        \Pibar \leq r V_1,
    \end{equation}
    where $r\geq1 $ is a constant. For each of Alice's outcomes \((\annA,x)\), let \(\nu_{\annA,x,1|\test}\) denote the component of \(\bsym{\nu}_{|\test}\) corresponding to the joint probability, conditioned on a test round, that Alice obtains \((\annA,x)\) and Bob obtains the outcome associated with \(V_1\), that is
    \begin{equation}
        \nu_{\annA,x,1|\test} = \Tr{\hat{\rho}_{AB|\test}\left(M_{\annA,x}^{A|\test}\otimes V_1\right)},
    \end{equation}
    where \(\hat{\rho}_{AB|\test}\in\dop{=}(AB)\) is the corresponding normalized test-round state. Finally, the tagged QKD channel \(\EATchannQKDTag\) and the tagged statistic-generating map \(\probstTag\) are defined as in \cref{def:Tagged_channels_maps}.
    
    Then, the single-round quantity $\hUpQKD$ from \cref{thrm:Fixed_length_security} obeys the lower bound
\begin{widetext}
\begin{equation}
    \begin{aligned}
        \hUpQKD \geq&
        \inf_{\substack{\mbf q\in\Sacc,\;\bsym{\nu}_{\CP} \in\mathbb P(\alphCP),\;\rho\in\dop{\leq}(AB_{N_c}),\\ W_{\test}^{\annA,x},W_{\gen}\in[0,1]}}
        \Bigg(
            \frac{\alpha}{\alpha-1}D\!\left(\mbf q\middle\Vert\bsym{\nu}_{\CP}\right)
            +q(\gen)\renyiSandUp_{\alpha}(S|IE\widetilde E)_{\widetilde{\omega}_{|\gen}}
            -q(\gen)\Delta_{\alpha}\!\left(\sqrt{W_{\gen}}\right)
        \Bigg)\\
        & \quad\textrm{s.t.}\quad
        \begin{aligned}[t]
            &\pTr{B_{N_c}}{\rho}\leq\sigma_A,\\
            &\widetilde{\rho}=\rho+\left(\sigma_A - \pTr{B_{N_c}}{\rho}\right)\otimes\pure{\star},\\
            &W_{\test}=\sum_{\annA,x}p(\annA,x|\test)W_{\test}^{\annA,x},\\
            &W=(1-\gamma)W_{\gen}+\gamma W_{\test},\qquad \Tr{\rho}=1-W,\\
            &\Tr{\rho_{|\gen}}=1-W_{\gen},\qquad \Tr{\rho_{|\test}}=1-W_{\test}, \\
                &\Tr{\rho_{|(\test\wedge\annA,x)}}=1-W_{\test}^{\annA,x} \;\forall \;\annA \in \mathcal{C}^{A,\test}, \; x \in \mathcal{X}_{\annA}^{\test} \text{ s.t. } p(\annA,x|\test)>0,\\
                &\widetilde{\omega}_{|\gen}=\EATchannQKDTag\left[\widetilde{\rho}_{|\gen}\right],\\
                &\bsym{\nu}_{\CP}=\left(\gamma\bsym{\nu}_{|\test},1-\gamma\right)^T,\\
                &\probstTag[\rho_{|\test}]\leq\bsym{\nu}_{|\test}\leq\probstTag[\rho_{|\test}]+W_{\test}\bsym{\delta},\\
                &p(\annA,x|\test)W_{\test}^{\annA,x}\leq r\nu_{\annA,x,1|\test}\qquad \forall \;\annA \in \mathcal{C}^{A,\test}, \; x \in \mathcal{X}_{\annA}^{\test},
            \end{aligned}
    \end{aligned}
\end{equation}
\end{widetext}
    where \(\Delta_{\alpha}\) and \(\bsym{\delta}\) are defined as in
    \cref{thrm:Dim_red_kappa_special}. Additionally, \(\widetilde{\rho}_{|\test}\) and \(\widetilde{\rho}_{|\gen}\) are obtained by conditioning \(\widetilde{\rho}\) on Alice's test and generation event. Their respective projected components are denoted by
    \(\rho_{AB_{N_c}|\test},\rho_{AB_{N_c}|\gen}
    \in\dop{\leq}(AB_{N_c})\). Moreover, \(E\widetilde{E}\) purifies \(\widetilde{\rho}_{|\gen}\), and the tagged QKD channel $\EATchannQKDTag$ is understood to act trivially on these registers. The inequality \(\pTr{B_{N_c}}{\rho}\leq\sigma_A\) means that \(\sigma_A - \pTr{B_{N_c}}{\rho}\) is positive semidefinite, while the vector inequalities are understood componentwise. 
\end{theorem}
\begin{proof}
    The proof follows analogously to \cref{thrm:Dim_red_kappa_special} by using \cref{lem:tagged_QKD_entropy_bound}.
\end{proof}

\begin{remark}
    As before, based on the formulation of \cref{thrm:Dim_red_halpha_special}, one can apply any of the reformulations of the \Renyi entropy in terms of convex functions as found in \cite{Kamin_2025a} using the duality argument in \cref{prop:Petz-Sandwiched-duality} in place of the finite-dimensional on used in Ref.~\cite{Kamin_2025a}. We believe that one can equally extend the duality argument required for the methods in Ref.~\cite{Navarro_2026} to infinite dimensions.
\end{remark}

%--------------------------------------------------------
%                   Numerical Implementation
%--------------------------------------------------------
\section{Numerical Implementation}\label{sec:NumericalImplementation}
We will use the framework developed in Ref.~\cite{Kamin_2025a} to evaluate our key rates. The duality argument required for the proof of \cite[Theorem~8]{Kamin_2025a} is that for any pure state $\ket{\rho}_{ABC}$ where $\abs{A} < +\infty$ it holds
\begin{equation}
    \renyiPetzUp_{\alpha}(A \vert B)_{\rho} + \renyiSandDown_{\beta}(A \vert C)_{\rho} = 0,
\end{equation}
for $\alpha \beta = 1$ and $\alpha, \beta \in [0,\infty]$. Most importantly, the simplifications in Ref.~\cite{Kamin_2025a} based on this duality relation carry forward to infinite dimensions as we show in \cref{prop:Petz-Sandwiched-duality}. The remainder of the arguments in Ref.~\cite{Kamin_2025a} carry over without any issues.

Finally, we note that for the variable length key rates, we will use the mixed normalization constant $\kapmix_j$ instead of $\kapup_j$ because it is the optimal normalization constant, see \cite[Corollary 10.2.4.]{Kamin_2026}, for a choice of the tradeoff function based on the methods in Ref.~\cite{Kamin_2025a} which lower bound $\renyiSandUp_{\alpha}$ with $\renyiSandDown_{\alpha}$.

\subsection{Basis Choice}
Quantum optical considerations \cite{Upadhyaya_2021} show that, for a fixed photon number cutoff, i.e., for a fixed system dimension, the out-of-subspace weight $W$ is minimised using a displaced photon-number basis. Consequently, we use a displaced photon number basis $\ket{n, \beta_i} \defvar \hat{D}(\beta_i)\ket{n}$ for numerical simulations and key rate calculations, where $\hat{D}(\beta) \defvar \exp\left( \beta \hat{a}^{\dagger} - \beta^* \hat{a} \right)$ is the displacement operator and $\ket{n}$ for $n\in\N_0$ denotes the $n$-th Fock state. For each symbol $i$, we estimate the displacement $\beta_i$ from the raw measurement data collected during the test rounds in which symbol $i$ was sent. Importantly, this estimate does not need to coincide exactly with the actual displacement of the received quantum states. If the estimate is imperfect, the resulting basis is simply suboptimal and may lead to a slightly larger out-of-subspace weight $W$; this does not compromise security. Furthermore, since expectation values $\Tr{\rho X}$ of arbitrary observables $X$ are invariant under the change of basis, the dimension reduction argument remains valid.

For detailed derivations of relevant quantities and the representation of the displaced partial trace and its adjoint in this special basis, we refer to Ref.~\cite{Upadhyaya_Thesis_2021}. In particular, the implementations of the partial trace and its adjoint must account for the basis change transformation.

We want to emphasise that the displaced photon-number basis represents solely a physically motivated implementation choice; the security argument is general and does not depend on the particular choice of basis. 

\subsection{Measurements}
In the present DM CV-QKD protocol, Bob performs heterodyne detection, i.e. his measurement outcome is a complex number $\Upsilon \in\mathbbm{C}$ which, in the case of ideal detectors, can be related to a scaled projection of the incident quantum state onto the coherent state $\ket{\Upsilon}$, which leads to the POVM $\{E_{\Upsilon}\}_{\Upsilon \in \mathbbm{C}} = \{\frac{1}{\pi} \ketbra{\Upsilon}{\Upsilon}\}_{\Upsilon \in \mathbbm{C}}$. In the case of non-ideal, trusted detectors with efficiency $\eta_d$ and electronic noise $\nu_{\mathrm{el}}$, the heterodyne measurement can be described \cite{Lin_2020} as a scaled projection on displaced thermal states with mean photon number $\bar{n} = \frac{1-\eta_d+\nu_{\mathrm{el}}}{\eta_d}$, leading to POVM elements $\{G_{\zeta}\}_{\zeta \in \mathbbm{C}} = \left\{ \frac{1}{\eta_d \pi} \hat{D}\left( \frac{\zeta}{\sqrt{\eta_d}}\right) \rho_{\mathrm{th}}(\bar{n}) \hat{D}^{\dagger}\left( \frac{\zeta}{\sqrt{\eta_d}}\right) \right\}_{\zeta \in \mathbbm{C}}$.

Independent of the chosen detector model, Bob's heterodyne measurement outcomes are complex numbers corresponding to points in phase space, posing two challenges. First, the ultimate goal of a QKD protocol is to generate a shared digital secret key, which, by definition, consists of discrete bits. Therefore, although the measurement outcomes are continuous, they must ultimately be discretized for classical postprocessing. Second, the employed security-proof framework cannot directly handle registers that hold continuous quantities. Both challenges can be addressed by coarse-graining Bob's measurement results. 

\begin{figure}
\centering
\includegraphics[width=0.95\columnwidth]{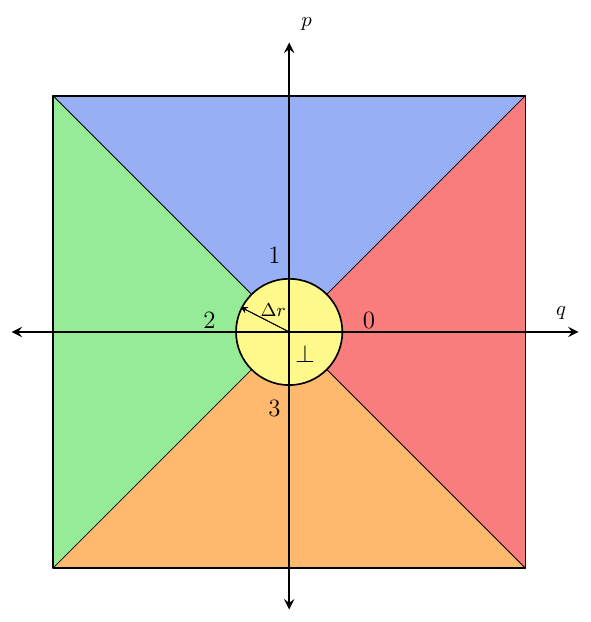}
\caption{QPSK key map with radial postselection. Measurement outcomes in the shaded areas are assigned the corresponding logical value. Outcomes in the central postselection region are discarded (assigned the value $\perp$). \label{fig:KeyMap}}
\end{figure}

\subsubsection{Key rounds}\label{sec:RegionOpDef}
For key rounds, Bob performs a reverse-reconciliation key map which assigns his measurement outcomes to bit values or assigns the discard symbol $\perp$. Therefore, he divides the phase space into $N_{\mathrm{St}}+1$ regions $\mathcal{R}_z$, $z \in \{0,1,..., N_{\mathrm{St}}-1, \perp \}$ and associates measurement outcomes lying in those regions with the corresponding key symbol $z$. Formally, this corresponds to coarse-graining the POVM into \textit{region operators}
\begin{equation}
    R_z^B = \int_{\mathcal{R}_z} F_{\zeta}~d^2\zeta,
\end{equation}
where $F_{\zeta}$ stands for the fine-grained POVM of the heterodyne measurement, which could be equal to $E_{\zeta}$, $G_{\zeta}$, or any other POVM corresponding to a model of heterodyne detection. While, in principle, the choice of the regions $\mathcal{R}_z$ is arbitrary, the chosen constellation, along with quantum-optical considerations \cite{Kanitschar_2022}, usually already hints at their optimal shape, while their optimal size can usually be obtained by optimising only a few parameters. In the present work, for illustration purposes, we chose QPSK modulation, i.e. $N_{\mathrm{St}} = 4$, and a simplified postselection scheme that only removes measurement outcomes within a radius of $\Delta_r$ around the origin. Thus, we obtain the region operators
\begin{equation}
    R_z^B := \frac{1}{\pi} \int_{\Delta_r}^{\infty} \int_{\frac{2z-1}{N_{\text{St}}} \pi}^{\frac{2z+1}{N_{\text{St}}} \pi} r |r e^{i \phi}\rangle \langle r e^{i \phi}|~d\phi ~dr.
\end{equation}
for $z=0\dots N_{\mathrm{St}}-1$ and $\R^B_{\perp} = \mathbbm{1} - \sum_{i=0}^{3} R_i^B$.

\subsubsection{Test rounds}
Earlier works \cite{Lin_2019, Upadhyaya_2021, Kanitschar_2022, Kanitschar_2023} coarse-grained key round measurements into subsets of the phase space, while test rounds were coarse-grained into observables such as $\hat{q}, \hat{p}, \hat{n}$ or their displaced counterparts. Those observables, however, are unbounded, which introduces additional challenges into the security proof and are not directly compatible with the framework used in this paper. Working directly with the heterodyne measurements would lead to an uncountably infinite number of observations $\nu(\cP)$, which cannot be handled in a numerical security argument. Thus, in this work, for test rounds, we follow a different avenue.

For test rounds, our security proof framework cannot directly handle continuous, fine-grained statistics; hence, we need to coarse-grain the test data. However, in contrast to the key rounds, the particular choice for test rounds is less obvious. Previous works \cite{Lin_2019, Upadhyaya_2021, Kanitschar_2022, Kanitschar_2023} have considered observables such as $\hat{q}, \hat{p}, \hat{n}$ or their displaced counterparts, which leads to problems of boundedness in the finite-size regime, coarse-grained the measurement outcomes into small squares \cite{Lupo_2022}, or coarse-grained the outcomes into sectors and ring-sectors centered at the origin \cite{Bauml_2024, Pascual-Garcia_2025, Primaatmaja_2024}.

In this work, we propose using POVMs corresponding to displaced rings in phase space, seamlessly integrating with the displaced Fock basis approach, motivated by high weights within a cutoff space with a fixed maximal photon number. Let us denote the number of rings in each displaced Fock basis by $N_R$ and the associated radii $0 = r_0 \leq r_1 \leq ... \leq r_{N_R}$ such that $B_0$ denotes the circle with outer radius $r_1$, $B_1$ denotes the ring with inner radius $r_1$ and outer radius $r_2$, up to $B_{N_R}$ denoting the area ranging from radius $r_{N_R}$ to infinity. Furthermore, let us denote the displacement associated with signal state $i$ by $\beta_i$. 

As we derive in Appendix \ref{apdx:POVMexpressions}, for the ideal detector scenario, this approach leads to Bob's test round POVM elements
\begin{equation}
    \begin{aligned}
    &M_{\beta_i, k}^{\textrm{ideal}} \\
    &= \ket{i}\!\!\bra{i}_A \otimes \sum_{n=0}^{\infty} \ket{n,\beta_i}\!\!\bra{n, \beta_i}_B  \frac{\Gamma(n+1, r_{k}^2) - \Gamma(n+1, r_{k+1}^2)}{n!},
\end{aligned}
\end{equation}
while for the non-ideal trusted detector scenario, we obtain
\begin{equation}
\begin{aligned}
  &M_{\beta_i, k}^{\textrm{non-ideal}}\\
  &= \ket{i}\!\!\bra{i}_A \otimes \sum_{m=0}^{\infty} \ket{m, \beta_i}\!\!\bra{m, \beta_i}_B C_{m,m} \\
  & ~~~\times\sum_{j=0}^{m} \begin{pmatrix}
    m \\ m-j
\end{pmatrix} \frac{\left(\Gamma(j+1,a r_{k}^2) - \Gamma(j+1, a r_{k+1}^2)\right)}{a^{j+1}b^jj!},    
\end{aligned}
\end{equation}
where $C_{m,n}$ depends on the detection efficiency and the electronic noise. 

\subsection{Simulation model}
To illustrate how the analyzed DM CV-QKD protocol behaves under our security argument, we simulate the quantum channel as a realistic physical channel in the absence of Eve. In the context of optical fiber communication, such a channel can be modeled by a phase-invariant Gaussian channel with transmittance $\eta$ and excess-noise $\xi$. However, we note that this is not a restriction of our method but solely chosen for illustration purposes. Our security proof does not rely on any channel assumption.

In the analyzed protocol, Alice prepares a coherent state $\ket{\alpha_i}$ and sends it to Bob through the quantum channel. In our model, the channel transforms the coherent state into a displaced thermal state centered at $\beta_i := \sqrt{\eta}\alpha_i$ with mean photon number $\bar{n} = \frac{\eta \xi}{2}$.

This allows us to simulate the observations for the different POVM elements. For the ideal detector scenario, we obtain
  \begin{align}
    m_{\beta_i, k}^{\textrm{ideal}} &= \sum_{n=0}^{\infty}  \frac{\Gamma(n+1, r_{k}^2) - \Gamma(n+1, r_{k+1}^2)}{n!} \frac{\bar{n}^n}{(1+\bar{n})^{n+1}},
\end{align}  
where, again, we replace the upper bound of the last sum by our finite cutoff $N_c$.

Similarly, we obtain for the non-ideal detector scenario,
\begin{equation}
    \begin{aligned}
    m_{\beta_i, k}^{\textrm{non-ideal}} &= \sum_{ m=0}^{\infty} \begin{pmatrix}
    m \\ m-j
\end{pmatrix}  \frac{\bar{n}^m C_{m,m}}{(1+\bar{n})^{m+1}}   \\
& ~~~\times \sum_{j=0}^{m} \frac{ \left(\Gamma(j+1,a r_{k}^2) - \Gamma(j+1, a r_{k+1}^2)\right)}{a^{j+1}b^jj!}
\end{aligned}  
\end{equation}

For detailed derivations, we refer interested readers to Appendix \ref{apdx:POVMexpressions}.

%--------------------------------------------------------
%                   Results
%--------------------------------------------------------

\section{Results}\label{sec:Results}
We illustrate our general security proof framework for the special case of a quadrature phase-shift keying (QPSK) protocol. Alice prepares one of $N_{\mathrm{St}} = 4$ coherent states with fixed amplitude, $\{\alpha_{\mathrm{coh}}, 1i \alpha_{\mathrm{coh}}, -\alpha_{\mathrm{coh}}, -1i \alpha_{\mathrm{coh}}\}$, $\alpha_{\mathrm{coh}}\in \mathbbm{R}$ with equal probability, while Bob performs heterodyne detection. Although our security proof framework applies independently of which party performs error-correction, we consider reverse reconciliation throughout, as it is known to outperform direct reconciliation for CV-QKD for losses exceeding $3$dB. Bob subsequently performs the key mapping by partitioning phase-space into four regions corresponding to the key symbols $\{0,1,2,3\}$, and a fifth region corresponding to the symbol $\perp$, which denotes discarded symbols resulting from radial postselection.

\subsection{Implementation Details}
We simulate Bob's observations using a noisy and lossy Gaussian channel with transmittance $\eta$ and excess noise $\xi$. We model the excess noise as preparation noise introduced on Alice's side of the channel. When reporting results as a function of transmission distance, we relate channel transmittance to the distance $L$ according to $\eta = 10^{-0.02 L}$, corresponding to an attenuation of  $-0.2\mathrm{dB/km}$, representative of standard optical fibre at telecom wavelengths. In our simulation, we fix the displacements and therefore the chosen displaced bases according to the expected channel behaviour ($\sqrt{\eta} \alpha_i$).

The numerical implementation involves several parameters that could, in principle, be optimized further to improve the key rates. These include the number of coarse-grained POVM rings $N_{\mathrm{POVMs}}$, the outer radius of the final POVM ring $r_{\mathrm{POVM, max}}$, and the testing ratio $\gamma$ which specifies the fraction of signals disclosed for statistical testing. To keep the numerical analysis tractable and the presentation consistent across all simulations, we fix these parameters to $N_{\mathrm{POVMs}} = 30$, $r_{\mathrm{POVM, max}} = 4.5$, and $\gamma = 10\%$. 

We truncate Bob's Hilbert space at $N_c = 25$, thereby describing his system in a $N_c+1 = 26$-dimensional Hilbert space. This value provides a compromise between the truncation-induced weight correction, numerical stability, and computational cost. As discussed below, the treatment of the weight correction is an important distinction between our security analysis and previous approaches.

For error correction, we assume an efficiency of $\beta_{\mathrm{EC}} = 95\%$, which is a commonly adopted value in the QKD literature and allows for a direct comparison with previous results. We note, however, that maintaining a constant error-correction efficiency of $95\%$ over the wide range of signal-to-noise ratios considered here may not be achievable in practice \cite{Leverrier_2023}. 

The total security parameter is fixed to $\epsilon_{\mathrm{sec}} = 10^{-80}$, split equally between correctness and secrecy. We deliberately choose this unusually stringent value to demonstrate a particular advantage of our security framework: reducing the security parameter to extremely small values incurs essentially no additional key rate penalty. In the parameter regimes considered, we find numerically that key rates obtained for $\epsilon_{\mathrm{sec}} = 10^{-80}$ in plots are virtually indistinguishable from $\epsilon_{\mathrm{sec}} = 10^{-20}$. Consequently, our $\epsilon_{\mathrm{sec}} = 10^{-80}$ key rate curves can be directly compared to those obtained with results using more conventional security parameters, e.g., $\epsilon_{\mathrm{sec}} = 10^{-20}$ 

\begin{remark}
The negligible dependence of our key rates on the security parameter is a consequence of specific features of our security proof framework. First, our proof avoids smoothing, and therefore does not introduce the additional $\epsilon_{\mathrm{smooth}}$-dependent correction term that becomes increasingly costly as the security parameter is reduced. Second, we do not introduce a separate security parameter $\epsilon_{\mathrm{AT}}$ to account for failure of the acceptance testing procedure via concentration inequalities. Instead, the acceptance test is incorporated directly into the security analysis, and no separate portion of the overall security budget needs to be allocated. 

This is particularly advantageous when demanding very small security parameters. In conventional finite-size analyses, obtaining very small acceptance-test failure probabilities leads to increasingly conservative statistical bounds, resulting in substantial penalties in key rate. In our framework, in contrast, reducing the overall security parameter over several orders of magnitude has a negligible effect on the key rate. This demonstrates that our framework can provide exceptionally strong composable security without a practically relevant reduction in key rate, a feature that is particularly relevant in the context of QKD standardisations and stringent security requirements.
\end{remark}

This demonstrates that the security guarantees of our framework can be strengthened to values relevant for QKD standardisation and practical security requirements without incurring a significant key rate penalty.

\subsection{Simulation Results}
In this section, we numerically evaluate the performance implied by our security proof for a QPSK protocol. We first present unique-acceptance fixed-length key rates, which are the standard way of reporting key rates in the literature and therefore allow for an efficient comparison with previous works. However, such key rates are impractical in the sense that they require the experimental observations to match the expected channel behaviour exactly. Ref. \cite{Kanitschar_2022} addressed this issue and suggested non-unique acceptance key rates, which enlarge the set of accepted observations and therefore increase the probability of a successful protocol execution. However, this comes at the cost of reduced one-shot key rates.

Here, we overcome this trade-off by proving full variable-length security, which allows Alice and Bob to adapt the length of the produced key to the quality of their observations. By \cite[Section~VI.A]{Kamin_2025a}, the fixed- and variable-length key rates obtained from the security proofs in Sections~\ref{sec:Security_proof} and~\ref{sec:DimReduction} coincide for the expected observations, up to small deviations due to numerical imprecision.

Our numerical analysis indicates that the variable-length problem is slightly less stable than its fixed-length counterpart, especially in the region where the key rate starts dropping quickly. Nevertheless, for most parameter regimes, the resulting differences are small on the scale of the plots presented below. For clarity, and to allow for comparison with previous works, we therefore plot the conventional unique-acceptance fixed-length key rates throughout this section and illustrate the advantages and practical implications of our variable-length security argument separately.

We first investigate the dependence of the key rate on the number of signals sent or equivalently protocol rounds. We fix the transmission distance to $L=\unit[10]{km}$, corresponding to a loss of $2$~dB, the coherent state amplitude to $\alpha_{\mathrm{coh}} = 0.90$ and the radial postselection parameter $\Delta_r = 0.45$. These parameter values are close to optimal for the considered setting. In Figure~\ref{fig:KRvsN}, we compare our finite-size key rates against coherent attacks and security parameter $\epsilon_{\mathrm{sec}} = 10^{-80}$ with the finite-size key rates against collective attacks obtained from the method in Ref. \cite{Kanitschar_2023}, using a security parameter of $\epsilon_{\mathrm{sec}} = 10^{-10}$, as well as the asymptotic key rates obtained from the method of Ref. \cite{Upadhyaya_2021}. Since $10\%$ of all protocol rounds are used for statistical testing in our protocol, we plot $90\%$ of the asymptotic key rate, which assumes that no rounds are sacrificed for testing, to ensure a fair comparison.

\begin{figure}
\centering
\includegraphics[width=0.95\columnwidth]{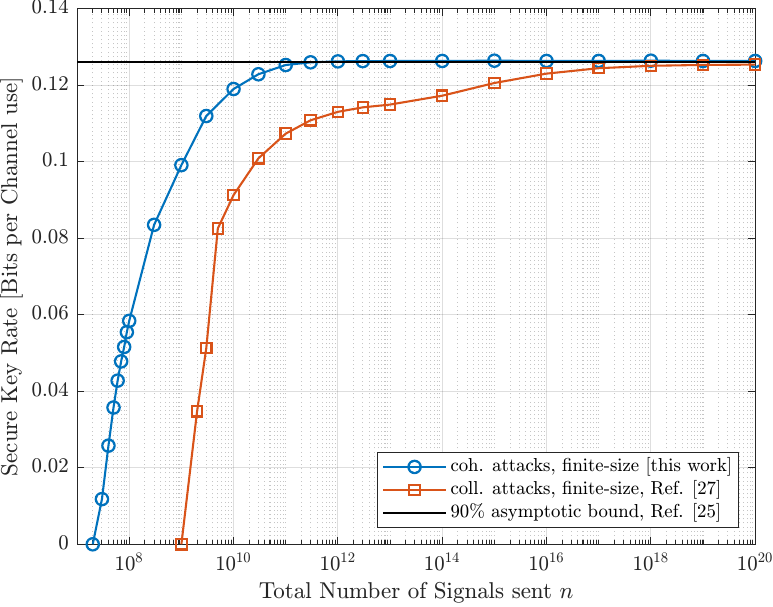}
\caption{Secure key rates over the total number of signals sent ($n$). We fix the transmission distance to $L=\unit[10]{km}$, corresponding to $2$dB of loss, the coherent state amplitude $\alpha_{\mathrm{coh}} = 0.90$, and the radial postselection parameter $\Delta_r = 0.45$ and compare our coherent attack finite-size key rates with the asymptotic rates of Ref. \cite{Upadhyaya_2021} and the collective attack finite-size key rates of Ref. \cite{Kanitschar_2023}. \label{fig:KRvsN}}
\end{figure}

Our coherent attack finite-size key rates converge rapidly towards the asymptotic rates, reaching almost asymptotic behaviour already for block sizes of $3\times 10^{11}$. While such block sizes are at the upper end of what is currently achievable with state-of-the-art technology, our results remain substantially advantageous at much smaller block sizes. In particular, the coherent attack rates obtained in this work significantly outperform the collective attack finite-size rates of Ref. \cite{Kanitschar_2023} throughout the parameter regime considered, with the largest relative improvement occurring in the practically relevant regime of small block sizes. 
This improvement is particularly notable because it is achieved under a stronger security requirement, and with a security parameter of $10^{-80}$ compared to $10^{-10}$ in Ref. \cite{Kanitschar_2023}. The improvement results from a fundamentally different security proof approach leading to several advances, which eliminate sources of looseness present in previous analyses and reduce the resulting finite-size corrections. Additionally, our weight estimation approach avoids finding separate worst-cases for the weight estimation and the worst-case state for the entropy minimisation by solving one unified optimisation problem. Most importantly from a practical perspective, our analysis yields a non-zero key rate for block sizes as small as $3\times 10^{7}$. This extends the regime of non-zero finite-size key rates to block sizes that are readily accessible with state-of-the-art continuous-variable systems. 

We next investigate the dependence of our coherent attack finite-size key rates on the channel loss (transmission distance). As before, we include the asymptotic key rate bound of Ref. \cite{Upadhyaya_2021} as a reference. The asymptotic key rates and their finite-size counterparts are, however, not directly equivalent: as mentioned before, the asymptotic analysis does not require rounds to be sacrificed for statistical testing, whereas finite-size security requires a fraction of the transmitted rounds to be disclosed in order to obtain estimators for Bob's observations. The resulting testing overhead accounts for the small gap between the asymptotic curve and the finite-size curve with the largest block size.

We consider three different block sizes, $n \in \{10^8, 10^{10}, 10^{12} \}$, representing relatively small, intermediate, and large block sizes accessible with state-of-the-art CV-QKD systems. For each data point, we independently optimise the \Renyi parameter $\alpha$, the coherent state amplitude $\alpha_{\mathrm{coh}}$, and the radial postselection parameter $\Delta_r$ using a coarse-grained parameter search. The resulting key rate curves are shown in Figure \ref{fig:KRvsLoss}. Positive key rates are obtained up to transmission distances of $\unit[18]{km}$, $\unit[54]{km}$, and $\unit[78]{km}$ for $n= 10^8$, $n=10^{10}$, and $n=10^{12}$ respectively. Thus, even for comparatively small block sizes, our security analysis permits key generation over practically relevant transmission distances, while increasing the block size substantially extends the achievable range.

\begin{figure}
\centering
\includegraphics[width=0.95\columnwidth]{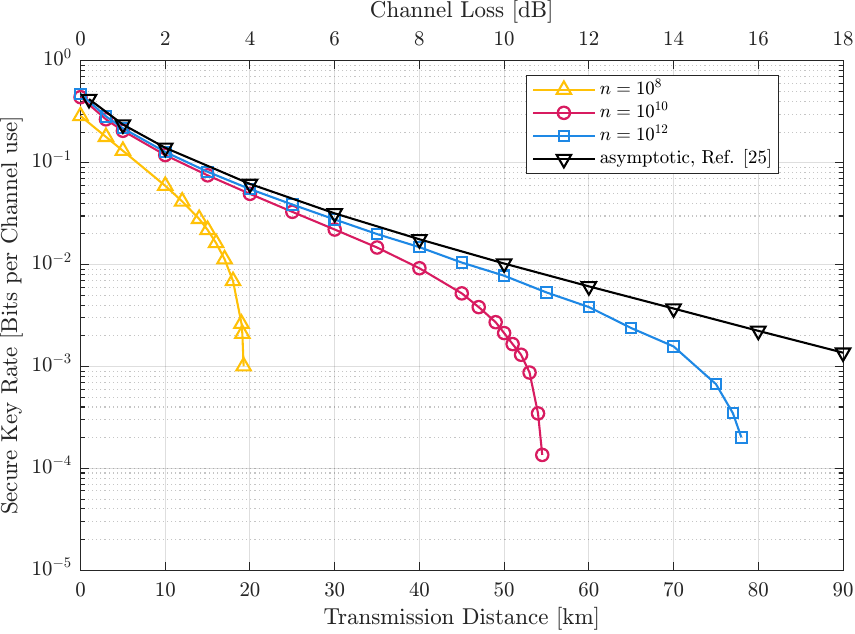}
\caption{Secure key rates over transmission distance $L$ in $\unit{km}$ (channel loss in $\unit{dB}$) compared to the asymptotic key rates from Ref. \cite{Upadhyaya_2021}. Here, we assumed optimal detectors. Each key rate point was optimised over the \Renyi parameter $\alpha$, the coherent state amplitude $\alpha_{\mathrm{coh}}$, and the radial postselection parameter $\Delta_r$. \label{fig:KRvsLoss}}
\end{figure}

The results so far assume ideal detectors. We now extend our analysis to trusted, non-ideal detectors following the modelling approach of Ref. \cite{Lin_2020}. As a representative example, we consider a detection efficiency of $\eta_{\mathrm{det}} = 0.72$ and an electronic noise of $\nu_{\mathrm{el}} = 0.04$ (nu). We note that although the excess noise of $\xi = 0.01$ (nu) remains the same as for previous plots, the total noise is increased by the additional electronic noise. Furthermore, the non-unit detection efficiency introduces additional loss, such that the total system loss is larger than the (untrusted) channel loss shown on the horizontal axis. We continue to parametrise the latter by the corresponding transmission distance. 

For each data point, we again optimise the \Renyi parameter $\alpha$, the coherent state amplitude $\alpha_{\mathrm{coh}}$, and the radial postselection parameter $\Delta_r$  using the same coarse-grained search as in the previous simulations. Figure \ref{fig:KRvsLossNonideal} shows the resulting key rates for block sizes $n \in \{10^8, 10^{10}, 10^{12} \}$. We obtain non-zero key rates up to transmission distances of $\unit[14]{km}$, $\unit[47]{km}$, and $\unit[72]{km}$ for $n=10^8$, $n=10^{10}$, and $n=10^{12}$, respectively. 
These results demonstrate that composable security against coherent attacks remains achievable in the presence of realistic detector imperfections and over practically relevant transmission distances. Although detector imperfections and electronic noise reduce the achievable distances compared with the ideal detector case, the protocol retains practically relevant finite-size key rates for all three considered block sizes.

\begin{figure}
\centering
\includegraphics[width=0.95\columnwidth]{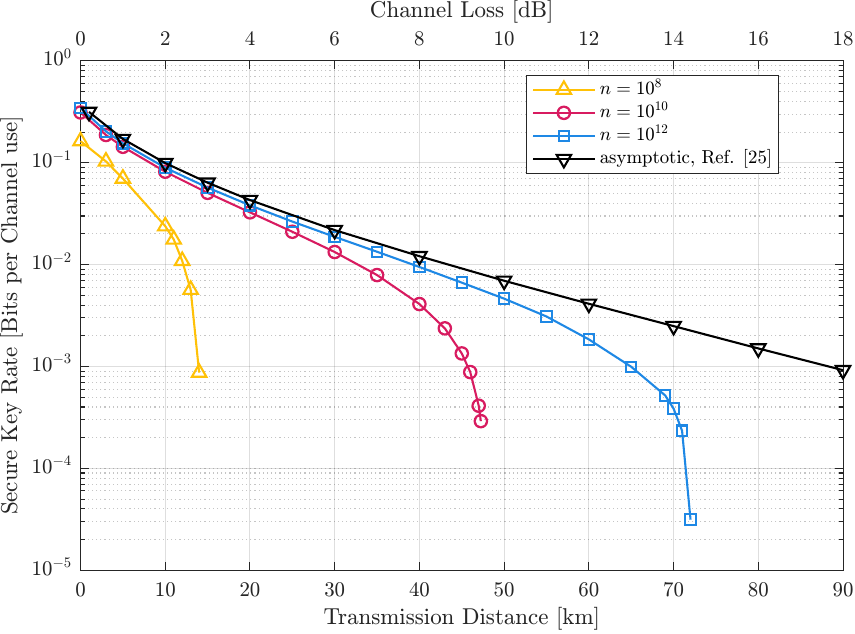}
\caption{Secure key rates over transmission distance $L$ in $\unit{km}$ (channel loss in dB) for trusted, non-ideal parameters with $\eta_{\mathrm{det}} = 0.72$ and $\nu_{\mathrm{el}} = 0.04$ compared to the asymptotic key rates from Ref. \cite{Upadhyaya_2021}. Each key rate point was optimised over the \Renyi parameter $\alpha$, the coherent state amplitude $\alpha_{\mathrm{coh}}$, and the radial postselection parameter $\Delta_r$.  \label{fig:KRvsLossNonideal}}
\end{figure}

\subsubsection{Variable-length security}
We finally illustrate the practical advantage of the full variable-length security statement proved in Section \ref{sec:Variable_length_security}. As discussed above, fixed-length security is the de facto standard in theoretical DM CV-QKD literature, but its silent dependence on a successful protocol execution is poorly aligned with experimental fluctuations. In a fixed-length analysis, the communicating parties must first specify the acceptance set $\Sacc$ before protocol execution. The key length is then determined by the worst-case state compatible with outcomes in this set. Consequently, the protocol produces a fixed key length $\ell_{\textrm{fixed}}$ whenever the observations fall within the acceptance set, and a key of length $0$ otherwise.

Unique acceptance refers to the limiting case in which the acceptance region contains only the accepted statistics, which makes successful protocol executions virtually impossible. Ref. \cite{Kanitschar_2023} addressed this issue by introducing non-unique acceptance, in which the acceptance region is enlarged to include a range of possible observations. This increases the probability of accepting a protocol run, but comes at a cost: enlarging the acceptance region also enlarges the set of possible states that must be considered in the security analysis, lowering the key rate obtained upon acceptance. Thus, upon acceptance $\ell_{\textrm{fixed}}^{\mathrm{non-u.a.}} < \ell_{\textrm{fixed}}^{\mathrm{u.a.}}$. 

Our variable-length security statement takes a different approach (see Ref. \cite{Tupkary_2024} for an introduction to variable-length security). Rather than assigning a single key length to all observations within a set, the key length is determined directly from the observed statistics. For observations coinciding with the target statistics, the length of the resulting variable-length key agrees, up to numerical deviations, with the corresponding unique acceptance fixed-length key length. Away from the target statistics, however, the variable-length analysis adapts the key length to the quality of the observed data, thereby avoiding the binary acceptance decision of fixed-length security.

To illustrate this difference, we consider a trusted non-ideal detector with efficiency $\eta_{\mathrm{det}} = 0.72$ and electronic noise of $\nu_{\mathrm{el}} = 4\times 10^{-2}$, and use a block size of $n=10^{10}$. We assume a Gaussian channel which is characterised by its loss $\eta$ and its excess noise $\xi$. This assumption is made solely to construct a simple illustrative model of fluctuating channel conditions; our security proof itself does not rely on any assumption on the channel. We model the evolution of the channel parameters as random walk, initialised at the target values $\mu_{\mathrm{loss}} = 3 \mathrm{dB}$, and $\mu_{\mathrm{noise}} = 10^{-2}$, corresponding to a $\unit[15]{km}$ fibre channel. After each protocol execution, loss and noise are independently updated according to $\Delta_{\mathrm{loss}} \sim \mathcal{N}(0, \sigma_{\mathrm{loss}}^2)$ with $\sigma_{\mathrm{loss}}^2 = 2.5\times10^{-4}$ and $\Delta_{\mathrm{noise}} \sim \mathcal{N}(0, \sigma_{\mathrm{noise}}^2)$ with $\sigma_{\mathrm{noise}}^2 = 2.5\times10^{-6}$, with each update applied relative to the parameters of the previous execution. The resulting simulated evolution of, in total, $100$ channel parameter sets is shown in Figure \ref{fig:RandomWalk} in the loss-noise parameter space. Under the Gaussian channel assumption, it is sufficient to specify this region in terms of $\eta$ and $\xi$, rather than separately specifying acceptance regions for every coarse-grained POVM outcome. 

The fixed-length key is determined by the worst-case state explaining observations within $\mathcal{S}_{\textrm{acc}}$, which in our example corresponds to the largest loss and excess noise allowed by the acceptance region, which we set to $[\unit[2.9]{dB}, \unit[3.1]{dB}] \times [7.5\times 10^{-3}, 12.5 \times 10^{-3}]$ in the $(\eta,\xi)$ space. We obtain $R_{\textrm{fixed}}^{\textrm{non-u.a.}} = 2.95 \times 10^{-2}$.

For the variable-length analysis, we take the target statistics to be the centre of the acceptance region, $\mu_{\textrm{loss}} = \unit[3.0]{dB}$, and $\mu_{\textrm{noise}} = 10^{-2}$. The variable-length security bound of \cref{eq:l_var_fprot} then determines the key length from the observed data. Specifically, the optimisation provides the coefficients $\mbf{f}$ of a linear function (the tradeoff function) and an affine offset $\kappa$, which, together with Bob's observed coarse-grained statistics $\mbf{v}_{\textrm{obs}}$, the error-correction leakage and second-order correction terms, yield a lower bound on the length of the variable-length key,
\begin{equation}
    \ell_{\textrm{var}} = \kappa + \mbf{f}\cdot \mbf{v}_{\textrm{obs}} - \lambda_{\mathrm{EC}}(\mbf{v}_{\textrm{obs}}) - \textrm{second-order~corrections}.
\end{equation}

Here, we emphasise the high degree of practicality offered by our variable-length security proof approach. All information required by Alice and Bob to evaluate the key rate is contained in $\mbf{f}$ and $\kappa$. The only computation required to determine the length of the key is an inner product between two vectors whose dimension is given by the number of observables, followed by the addition of a constant. Consequently, both the memory requirements and the computational overhead related to the key length calculation are minimal, effectively addressing a recurring challenge in the practical implementation of QKD.

We simulate $100$ protocol executions, with the first execution initialised at the target statistics and each subsequent run generated by one step of the random walk described above. Figure \ref{fig:Fixed_vs_variable_length} shows the resulting fixed- and variable-length key rates as a function of the execution index. The fixed-length protocol shows the expected binary behavior: it produces the key rate $R_{\textrm{fixed}}^{\textrm{non-u.a.}}$ whenever the observed statistics lie within $\mathcal{S}_{\textrm{acc}}$ and aborts otherwise. In contrast, the variable-length key rate changes with the channel parameters, reflecting the quality of the data obtained in every execution.

\begin{figure}
\subfloat[Fixed-length vs. variable-length key rates. \label{fig:Fixed_vs_variable_length}]{
    \includegraphics[width=0.48\textwidth]{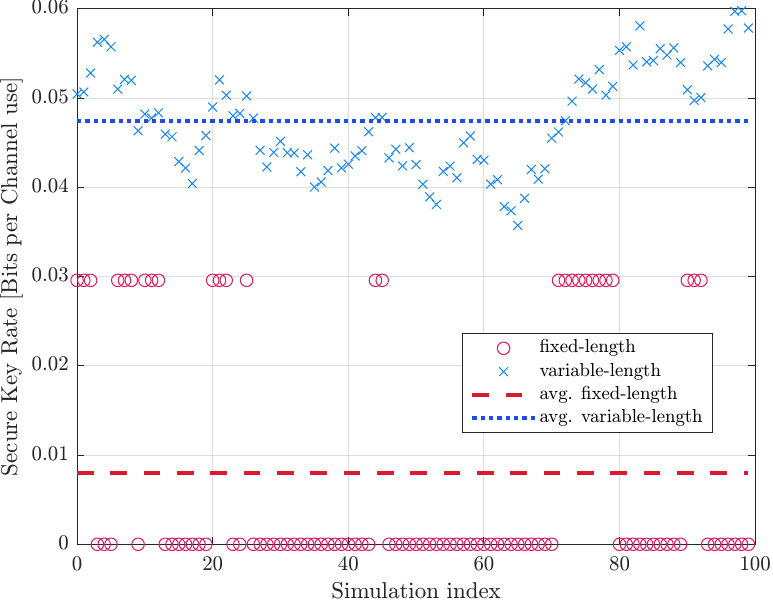}}\\
\subfloat[Evolution of simulated channel parameters. \label{fig:RandomWalk}]{
    \includegraphics[width=0.48\textwidth]{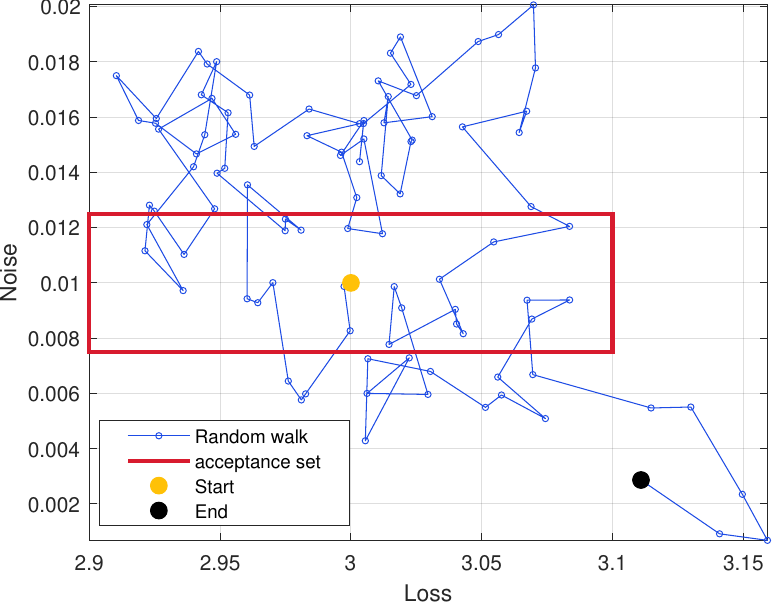}}
\caption{Illustration of fixed-length and variable-length key rates for $100$ simulations of a random-walk Gaussian channel. (a) Comparison of fixed-length (red circles) and variable-length (blue crosses) key rates for $100$ protocol runs. Horizontal lines represent the respective average key rate obtained over $100$ runs. (b) Illustration of the evolution of the channel parameters from run 1 (yellow dot) to run 100 (black dot). The red box marks the acceptance set.}
\end{figure}

Two features are particularly noteworthy. 
\begin{itemize}
\item[(i)] The variable-length protocol continues to produce a non-zero key throughout the simulated channel fluctuations, whereas the fixed-length protocol aborts frequently. Whenever the random walk leaves this region, the fixed-length protocol produces zero key, even though the observed statistics may still be compatible with secure key generation. In contrast, the variable-length protocol continuously adapts the key length to the observed statistics and therefore continues to generate a nonzero key throughout the simulated trajectory.

\item[(ii)] Even for the initial execution, for which the parameters coincide with the target statistics, the variable-length rate exceeds the non-unique acceptance fixed-length rate. This is because the latter is determined by the worst-case point in the entire acceptance region, whereas the variable-length analysis exploits the recorded statistics. In particular, again due to \cite[Section~VI.A]{Kamin_2025a}, at the target statistics the variable-length rate essentially recovers the unique acceptance rate
 \begin{equation*}
    R_{\textrm{var}}(\textrm{target}) \simeq R_{\textrm{fixed}}^{\textrm{u.a.}}(\textrm{target}) > R_{\textrm{fixed}}^{\textrm{non-u.a.}}(\mathcal{S}_{\textrm{acc}})
\end{equation*}
where the last inequality reflects the penalty associated with enlarging the acceptance region. 
\end{itemize}

Variable-length security avoids both sources of conservativism. For our simulation, this leads to an average variable-length key rate of $47.4 \times 10^{-3}$, corresponding to an improvement by a factor of $6$ compared to the average fixed-length key rate of $7.98 \times 10^{-3}$. The exact improvement depends on the model for the channel fluctuations, its parameters, the particular choice of the acceptance region, and the particular random walk that manifested. In particular the size of the acceptance region determines the trade-off between the probability of accepting a protocol run and the key rate obtained upon acceptance, as discussed in detail in Ref. \cite{Kanitschar_2023,Kamin2026Phys.Rev.Res.}. The present simulation is therefore intended as an illustration rather than a quantitative prediction for a particular experimental implementation.

%--------------------------------------------------------
%                  Discussion
%--------------------------------------------------------

\section{Discussion}\label{sec:Discussion}
A full composable security proof for general discrete-modulated continuous-variable QKD protocols against coherent attacks has remained an open problem for more than two decades. Previous results established security for restricted attack models, specific modulation formats, or under additional assumptions. In this work, we close this gap by providing a general composable finite-size security proof for DM CV-QKD against coherent attacks, covering both fixed- and variable-length variants of the protocol. In particular, our security argument does not rely on restricting the adversary or on assuming finite-dimensional descriptions of the underlying quantum systems.

The proof rests on two fundamental advances that are of interest beyond the specific QKD application here. First, at the information-theoretic level, we introduce and prove an infinite-dimensional marginal-constrained entropy accumulation theorem (iMEAT), which provides a lower bound on the $n$-round \Renyi entropy in the presence of \emph{infinite-dimensional} side-information. We further extend this result to $f$-weighted entropies, which is essential for variable-length security.
Second, we develop a dimension-reduction argument that rigorously relates the original infinite-dimensional optimisation problems to finite-dimensional problems that can be solved numerically. A particular advantage of our approach is that the weight of the obtained finite-dimensional subspace is estimated directly within the key rate optimisation. This differs fundamentally from previous approaches \cite{Kanitschar_2023} and those recently proposed in Ref.~\cite{Navarro_2026a}, which determine the weight either through separate optimization problems or through additional statistical testing. The resulting formulation provides a unified way of accounting for the contribution of the truncated subspace without introducing a separate weight estimation step. This removes the looseness arising from separately taking the worst-case weight bound and the worst-case state for the entropy minimisation.

A further conceptual ingredient of our approach is handling the experimental input directly in terms of (coarse-grained) POVMs. Previous numerical security analyses of DM CV-QKD \cite{Ghorai_2019, Lin_2019, Upadhyaya_2021, Kanitschar_2023} have typically introduced a set of observables, such as quadrature or photon-number operators and their higher moments, whose expectation values are then used to constrain the adversarial state. While this formulation is mathematically valid, it introduces additional complications in the infinite-dimensional setting, particularly when the considered observables are unbounded. In our approach, we instead formulated the experimental input in terms of probabilities. This provides a more direct connection between the experimentally obtained data and the heterodyne measurement underlying the protocol and avoids the need to introduce unbounded observables as intermediate quantities. In addition, this change is not merely conceptual. It plays a key role in overcoming two technical obstacles: it allows us to formulate the finite-size problem in terms of bounded measurement outcomes and allows us to incorporate experimental input as probabilities into the optimization problem that follows from our iMEAT and dimension reduction theorems.

Our numerical results demonstrate that these developments lead to substantial practical improvements. For the QPSK protocol considered in this work, we obtain non-negative key rates for block sizes below $10^8$ for low to medium transmission distances and reach transmission distances beyond $\unit[70]{km}$ for block sizes of $10^{12}$ and experimentally viable detector parameters. Moreover, our coherent-attack finite-size key rates substantially outperform previously reported finite-size rates based on weaker security guarantees. Importantly, these results are obtained with the stringent security parameter $\epsilon_{\mathrm{sec}} = 10^{-80}$. Reducing the security parameter to this level causes only a minimal additional key rate penalty, illustrating an important practical advantage of our framework that addresses a recurring demand in practical system developments and ongoing QKD certification and standardisation efforts.

Our variable-length security result provides a further road towards making composable coherent attack security practical. Rather than fixing the key length based on a predefined acceptance set, the variable-length protocol can adapt the key length based on the quality of the observed statistics. This flexibility avoids the impracticality of requiring observations to coincide or closely match the target statistics. More importantly, it avoids two sources of unnecessary key rate loss inherent to fixed-length approaches. First, observations that fall outside the predefined acceptance region need not be discarded solely because they differ from the target statistics. Second, among accepted observations, one need not assign every data set the key length corresponding to the least favourable state compatible with the entire acceptance region. Instead, the variable key length can be tailored to the actual observed statistics, allowing favourable data sets to yield longer keys.

Crucially, this improved statistical efficiency does not come at the expense of practical implementability. The information required to determine the key length is contained in a vector $\mbf{f}$ and a constant $\kappa$. The key length calculation requires only an inner product of a vector containing the experimental observations with $\mbf{f}$, followed by the addition of the constant $\kappa$. Thus, the additional processing required to determine the key length is negligible, while the memory requirements remain extremely low. Our method therefore combines the adaptivity of variable-length security with a computationally lightweight practical implementation, addressing two important obstacles in the practical deployment of DM CV-QKD systems simultaneously. Our simulations demonstrate that our variable-length security argument yields substantial improvements in the expected key rate, without sacrificing practical implementability.

Beyond the theoretical advances, our results remove a fundamental obstacle towards the practical deployment of DM CV-QKD. Experimental platforms for CV-QKD have made rapid progress \cite{Pietri_2024, Aldama_2025, Hajomer_2025, Jaksch_2026, Ng_2026, Zhan_2026, Wu_2026}, demonstrating increasing experimental maturity. However, a full composable finite-size proof against coherent attacks has so far been lacking, leaving a gap between the capabilities of experimental implementations and the security guarantees required for practical cryptographic applications. By closing this gap, our work establishes a security foundation for DM CV-QKD that can be directly connected to these experimental developments. Together with the practically relevant key rates and transmission distances demonstrated in our simulations, this makes an experimental demonstration of composable secure key against coherent attacks a realistic next step.

Several further steps remain toward fully practical systems. 
While our security proof applies to arbitrary discrete modulation schemes, our numerical analysis focused on QPSK, both because of its practical relevance and to keep the computational demands low. Extending the numerical analysis to higher modulation formats, which are known to yield key rates closer to the Gaussian modulation limit, is a natural direction for future work. In parallel, the development and adaptations of more efficient numerical formulations \cite{Pascual-Garcia_2025, Navarro_2026}, as well as alternative optimisation techniques with improved scaling \cite{Temesi_2026}, could allow the exploration of a broader range of modulation formats and parameter regimes.

Incorporating additional device imperfections and potential side channels into the security proof would further narrow the gap between theoretical description and implementation security. Finally, extending the framework to multi-user and networked settings \cite{Bian_2023, Hajomer_2024, Kanitschar_2026} would open continuous-variable QKD to more general quantum communication architectures.

\begin{acknowledgements}
F.K. acknowledges support from the European Research Council Consolidator Grant, Grant Agreement Number 101043705 (Cocoquest), from the Digital Europe Program, Grant Agreement Number 101305042 (TrustQ), and the Dieberger-Skalicky Foundation. This project is supported by the Ministry of Education, Singapore, through grant T2EP20124-0005. This project is supported by the National Research Foundation, Singapore under the NRF Postdoctoral award.
L.K. and J.B. conducted their research at the Institute for Quantum Computing, at the University of Waterloo, which is supported by Innovation, Science, and Economic Development Canada and the NSERC Alliance. NSERC provided support under the Discovery Grants Program, Grant No. 341495, Alliance Grant QUINT, and Alliance Grant ReFQ.

The authors thank Mil\'{a}n Mosonyi for referring them to \cite{Mosonyi_2023}. The authors additionally thank Norbert Lütkenhaus for discussions during early stages of this project, Marcus Huber for helpful general advice, and Amir Arqand for insight on generalizations of the marginal-constrained entropy accumulation theorem.

The authors used AI tools for purposes of spellchecking, grammatical corrections, and stylistic polishing of the final manuscript. Beyond this, the proof of Proposition \ref{prop:non-var-form-of-petz-up} is a simplified version of a proof generated by ChatGPT 5.6. The proof of Proposition \ref{prop:cl-chain-rule} uses \cite[Theorem 2]{Rotfeld_1969}, which the authors became aware of because ChatGPT 5.6 had referenced it in an earlier attempt at \cref{prop:non-var-form-of-petz-up}. Upon review, ChatGPT-6 Astra found a technical error in the choice of function to interpolate to prove the leftover hashing lemma, which we had inherited from \cite{Dupuis_2022}. It requires a mild alteration that it suggested--- see Section \ref{app:inf-dim-LHL-against-inf-side-info} for more details. Additionally, the authors discussed the findings in \cref{App:Inf_MEAT_Proof} with ChatGPT 5.6 which pointed out gaps and inconsistencies in the initial drafts.

The authors independently verified the AI-assisted mathematical arguments and references incorporated into the manuscript and take full responsibility for its contents.
\end{acknowledgements}

\bibliography{new_references.bib}

\onecolumngrid
\appendix
 
%********************************************************
%                    Appendix
%********************************************************

%--------------------------------------------------------
%                  Mathematical Definitions
%--------------------------------------------------------
\section{Mathematical Definitions}\label{apdx:MathDef}
In this appendix, we provide definitions of central mathematical objects used in this work.
\subsection{\texorpdfstring{\Renyi}{Renyi} entropies}
Our analysis relies on various \Renyi divergences and entropies defined in Ref.~\cite{Tomamichel_2016}. However, we at times require the definitions generalized to density matrices that act on separable Hilbert spaces. This is so that we model infinite-dimensional quantum systems. Conveniently, the technical details of said generalization of \Renyi divergences has been worked out in \cite{Mosonyi_2023}. For this reason, here we present the finite-dimensional definitions for clarity and refer the reader to Appendix \ref{app:inf-dim-stuff} for omitted technical details.

\begin{definition}\label{def:Renyi Divergence}
    Let $\alpha \in (0,1) \cup (1,\infty)$, and $\rho,\sigma \in \dop{\leq}(A)$ with \(\Tr{\rho} \neq 0\), the \emph{minimal quantum \Renyi divergence} (or sandwiched quantum \Renyi divergence) is defined as
    \begin{equation}
        \renyiSandDiv_\alpha(\rho||\sigma) \defvar \frac{1}{\alpha-1} \log \frac{\Tr{ \left( \sigma^{\frac{1-\alpha}{2\alpha}} \rho \sigma^{\frac{1-\alpha}{2\alpha}} \right)^{\alpha} }}{\Tr{\rho}},
    \end{equation}
    for \((\alpha < 1 \wedge \rho \not\perp \sigma ) \vee \supp(\rho) \subseteq \supp{\sigma}\), and \(\infty\) otherwise. 
    
    The \emph{Petz quantum \Renyi divergence} is defined as
    \begin{equation}
        \renyiPetzDiv_\alpha\rel{\rho}{\sigma}\defvar \frac{1}{\alpha-1} \log \frac{\Tr{\rho^\alpha \sigma^{1-\alpha}}}{\Tr{\rho}},
    \end{equation}
    for \((\alpha < 1 \wedge \rho \not\perp \sigma ) \vee \supp(\rho) \subseteq \supp{\sigma}\),
    and \(\infty\) otherwise. In any statement that applies to both divergences, we will write $\renyiDiv_\alpha$ to denote either divergence.
\end{definition}

\begin{definition}\label{def:Conditional Renyi entropy}
For $\alpha \geq 0$ and $\rho_{AB} \in \dop{=}(AB)$ the \emph{quantum conditional \Renyi entropies} are defined as:
\begin{align}
    \renyiPetzDown_\alpha \con{A}{B}_\rho &\defvar -\renyiPetzDiv_\alpha \rel{\rho_{AB}}{I_A \otimes \rho_B} \\
    \renyiPetzUp_\alpha \con{A}{B}_\rho &\defvar \sup_{\sigma_B\in\dop{=}(B)} -\renyiPetzDiv_\alpha \rel{\rho_{AB}}{I_A \otimes \sigma_B} \\
    \renyiSandDown_\alpha \con{A}{B}_\rho &\defvar -\renyiSandDiv_\alpha \rel{\rho_{AB}}{I_A \otimes \rho_B} \\
    \renyiSandUp_\alpha \con{A}{B}_\rho &\defvar \sup_{\sigma_B\in\dop{=}(B)} -\renyiSandDiv_\alpha \rel{\rho_{AB}}{I_A \otimes \sigma_B}
\end{align}
Again, in any statement that applies to all conditional \Renyi entropies, we use $\renyiEnt_\alpha$.
\end{definition}

\subsection{\texorpdfstring{$f$}{f}-weighted \texorpdfstring{\Renyi}{Renyi} entropies}
\begin{definition}[$f$-weighted \Renyi entropies (partially restated from \cite{Arqand_2025, vanHimbeeck_2024})]\label{def:f-weighted entropies}
	Let $\rho \in \dop{=}(\CP Q Q')$ be a state where $\CP$ is classical with alphabet $\alphCP$. A \emph{tradeoff function\footnote{Ref.~\cite{Arqand_2025} instead referred to this as a ``quantum estimation score-system'' (QES).} on $\CP$} is a function $f:\alphCP \to \R$; equivalently, we may denote it as a real-valued tuple $\mbf{f} \in \R^{|\alphCP|}$ where each term in the tuple specifies the value $f(\cP)$. Given a tradeoff function $f$ and a value $\alpha\in(0,1)\cup (1,\infty)$, we define two versions of an \emph{$f$-weighted \Renyi entropy}  of order $\alpha$ for $\rho$, as follows:
	\begin{equation}\label{eq:upfweighted entropy}
		\begin{split}
			&\frenyiSandUp_\alpha(Q|\CP Q')_{\rho} \defvar \frac{\alpha}{1-\alpha} \log \left( \sum_{\cP} \rho(\cP) \, 2^{\frac{1-\alpha}{\alpha} \left(-f(\cP) + \renyiSandUp_{\alpha}(Q|Q')_{\rho_{|\cP}} \right) } \right) ,
		\end{split}
	\end{equation}
	\begin{equation}\label{eq:downfweighted entropy}
		\begin{split}
			&\frenyiSandDown_\alpha(Q|\CP Q')_{\rho} \defvar \frac{1}{1-\alpha} \log \left( \sum_{\cP} \rho(\cP) \, 2^{(1-\alpha) \left(-f(\cP) + \renyiSandDown_{\alpha}(Q|Q')_{\rho_{|\cP}} \right) } \right),
		\end{split}
	\end{equation}
	where the sums run over all $\cP$ values such that $\rho(\cP)>0$. We extend both definitions to $\alpha=\infty$ by taking the $\alpha\to\infty$ limit.
\end{definition}

\begin{definition}[Mixed $f$-weighted \Renyi entropies]\label{def:mixed f-weighted entropies}
	Let $\rho \in \dop{=}(\CP Q Q')$ be a state where $\CP$ is classical with alphabet $\alphCP$. A \emph{tradeoff function on $\CP$} is a function $f:\alphCP \to \R$, equivalently denoted by a real-valued tuple $\mbf{f} \in \R^{|\alphCP|}$ where each term in the tuple specifies the value $f(\cP)$. Given a tradeoff function $f$ and a value $\alpha\in(0,1)\cup (1,\infty)$, we define two versions of a \emph{mixed $f$-weighted \Renyi entropy} of order $\alpha$ for $\rho$, as follows:
	\begin{equation}\label{eq:mixedupfweighted entropy}
		\begin{split}
			&\frenyiMixedUp_\alpha(Q|\CP Q')_{\rho} \defvar \frac{\alpha}{1-\alpha} \log \left( \sum_{\cP} \rho(\cP) \, 2^{\frac{1-\alpha}{\alpha} \left(-f(\cP) + \renyiSandDown_{\alpha}(Q|Q')_{\rho_{|\cP}} \right) } \right) ,
		\end{split}
	\end{equation}
	\begin{equation}\label{eq:mixeddownfweighted entropy}
		\begin{split}
			&\frenyiMixedDown_\alpha(Q|\CP Q')_{\rho} \defvar \frac{1}{1-\alpha} \log \left( \sum_{\cP} \rho(\cP) \, 2^{(1-\alpha) \left(-f(\cP) + \renyiSandUp_{\alpha}(Q|Q')_{\rho_{|\cP}} \right) } \right),
		\end{split}
	\end{equation}
	where the sums run over all $\cP$ values such that $\rho(\cP)>0$. We extend both definitions to $\alpha=\infty$ by taking the $\alpha\to\infty$ limit.
\end{definition}
The choice of $\rotcurvearrowdown$ ($\rotcurvearrowup$) is meant to indicate that $\frenyiMixedUp$ ($\frenyiMixedDown$) is created from $\frenyiSandUp$ ($\frenyiSandDown$) upon replacing the conditional \Renyi entropy in the original definition \cref{def:f-weighted entropies} with its lower bound $\renyiSandDown$ (upper bound $\renyiSandUp$). 

\begin{definition}[Normalized tradeoff functions]\label{def:Full f-weighting}
	Let \(f\) be a tradeoff function on a register \(\CP\) as in \cref{def:f-weighted entropies}. Given any set of states $\mathcal{S} \subseteq \dop{=}(\CP Q Q')$, we define the \emph{$\frenyiSandUp_\alpha$-normalization constant} and \emph{$\frenyiSandDown_\alpha$-normalization constant} for that set to be, respectively,
	\begin{gather} 
		\kapup \defvar \inf_{\rho\in\mathcal{S}} \frenyiSandUp_\alpha(Q|\CP Q')_{\rho} ,\\
		\kapdown \defvar \inf_{\rho\in\mathcal{S}} \frenyiSandDown_\alpha(Q|\CP Q')_{\rho} .
	\end{gather}
    Similarly, we define the \emph{$\frenyiMixedUp_{\alpha}$-normalization constant} and  by
    \begin{gather}
        \kapmix \defvar \inf_{\rho \in \mathcal{S}} \frenyiMixedUp_\alpha(Q|\CP Q')_{\rho}, \\
        \kapmixdown \defvar \inf_{\rho \in \mathcal{S}} \frenyiMixedDown_\alpha(Q|\CP Q')_{\rho}.
    \end{gather}
	Given either of the above values, we then define corresponding \emph{$\frenyiSandUp_\alpha$-normalized}, \emph{$\frenyiSandDown_\alpha$-normalized}, \emph{$\frenyiMixedUp_\alpha$-normalized} or \emph{$\frenyiMixedDown_\alpha$-normalized} tradeoff functions \(\hat{f}\) respectively, via
	\begin{gather}
		\hat{f}(\cP) \defvar f(\cP) + \kapup, \quad
		\hat{f}(\cP) \defvar f(\cP) + \kapdown, \\
        \hat{f}(\cP) \defvar f(\cP) + \kapmix, \quad
		\hat{f}(\cP) \defvar f(\cP) + \kapmixdown.
	\end{gather}
\end{definition}

\begin{remark}
    By \cite[Lemma 10.2.5]{Kamin_2026}, it holds
    \begin{equation}
        \kapup \geq \kapmix \geq \kapdown, \quad
        \kapup \geq \kapmixdown \geq \kapdown,
    \end{equation}
    and, the ordering between $\kapmix$ and $\kapmixdown$ depends on the specific state. However, the mixed normalization constant $\kapmix$ can be immediately calculated from the methods in \cite{Kamin_2025a}. Furthermore, by \cite[ Corollary 10.2.4]{Kamin_2026}, the $f$-weightings obtained from \cite[Eq.~(118)]{Kamin_2025a}, are the optimal choices for $\kapmix$.
\end{remark}

%--------------------------------------------------------
%         Derivation of the POVM ring expressions
%--------------------------------------------------------

\section{Derivation of the POVM ring expressions}\label{apdx:POVMexpressions}
First, we consider the ideal detector scenario. Since we chose to use the displaced photon-number basis, we need to find expressions for the displaced, coarse-grained POVMs
\begin{align}
    M_{\beta_i, k}^{\textrm{ideal}} &:= \ket{i}\!\!\bra{i}_A \otimes \hat{D}(\beta_i) \int_{B_k} \frac{1}{\pi} \ket{\zeta}\!\!\bra{\zeta}~d^2\zeta \hat{D}^{\dagger}(\beta_i),
\end{align}
where $B_k$ is a ring in phase-space with inner radius $r_k$ and outer radius $r_{k+1}$.
We represent them in the displaced Fock basis $\ket{n, \beta_i} := \hat{D}(\beta_i)\ket{n}, ~ n \in \mathbbm{N}$,
\begin{equation}
    \begin{aligned}
    M_{\beta_i, k}^{\textrm{ideal}}  &:= \ket{i}\!\!\bra{i}_A \otimes \sum_{m,n=0}^{\infty} \hat{D}(\beta_i) \int_{B_k} \frac{1}{\pi} \ket{m}\!\!\bra{m} \ket{\zeta}\!\!\bra{\zeta} \ket{n}\!\!\bra{n}_B~d^2\zeta \hat{D}^{\dagger}(\beta_i)\\
    &= \ket{i}\!\!\bra{i}_A \otimes \sum_{m,n=0}^{\infty} \ket{m,\beta_i}\!\!\bra{n, \beta_i}_B \int_{B_k} \frac{1}{\pi} e^{-\abs{\zeta}^2} \frac{\zeta^{m+n+1}}{\sqrt{m!} \sqrt{n!}}~d^2\zeta \\
    &= \ket{i}\!\!\bra{i}_A \otimes \sum_{m,n=0}^{\infty} \ket{m,\beta_i}\!\!\bra{n, \beta_i}_B \int_{r_{k}}^{r_{k+1}} e^{-r^2} \frac{r^{m+n+1}}{\pi \sqrt{m!} \sqrt{n!}}~dr \int_{0}^{2\pi} e^{i (m-n)\theta}~d\theta\\
    &= \ket{i}\!\!\bra{i}_A \otimes \sum_{m,n=0}^{\infty} \ket{m,\beta_i}\!\!\bra{n, \beta_i}_B \int_{r_{k}}^{r_{k+1}} e^{-r^2} \frac{r^{m+n+1}}{\pi\sqrt{m!} \sqrt{n!}}~dr 2\pi \delta_{m,n}\\
    &= \ket{i}\!\!\bra{i}_A \otimes \sum_{n=0}^{\infty} \frac{2}{n!} \ket{n,\beta_i}\!\!\bra{n, \beta_i}_B \int_{r_{k}}^{r_{k+1}}  e^{-r^2} r^{2n+1}~dr \\
    &= \ket{i}\!\!\bra{i}_A \otimes \sum_{n=0}^{\infty} \ket{n,\beta_i}\!\!\bra{n, \beta_i}_B  \frac{\Gamma(n+1, r_{k}^2) - \Gamma(n+1, r_{k+1}^2)}{n!},
\end{aligned}
\end{equation}
where, for the third equality, we switched to polar coordinates, and for the last equality, we used the definition of the upper incomplete Gamma function. For numerical reasons, we require a photon-number cutoff, which we rigorously account for via a dimension-reduction argument. Therefore, we replace the upper bound of the sum by $N_c$, the photon cutoff number.

Second, we turn to trusted non-ideal detectors, where the displaced coarse-grained POVMs read,
\begin{equation}
    \begin{aligned}
    M_{\beta_i, k}^{\textrm{non-ideal}} &:= \ket{i}\!\!\bra{i}_A \otimes \hat{D}(\beta_i) \int_{B_k} G_{\zeta}~ d^2\zeta \hat{D}^{\dagger}(\beta_i),\\
    &= \ket{i}\!\!\bra{i}_A \otimes \sum_{m,n=0}^{\infty} \hat{D}(\beta_i) \ket{m}\!\!\bra{m}_B \int_{B_k} G_{\zeta}~ d^2\zeta \ket{n}\!\!\bra{n}\hat{D}^{\dagger}(\beta_i),\\
    &= \ket{i}\!\!\bra{i}_A \otimes \sum_{m,n=0}^{\infty} \ket{m, \beta_i}\!\!\bra{n, \beta_i}_B  \int_{B_k} \bra{m} G_{\zeta} \ket{n}~ d^2\zeta.
\end{aligned}
\end{equation}
Using eqs. (6.14) and (6.15) from Ref. \cite{Mollow_1967}, we obtain for $n\geq m$
\begin{align}
    \bra{m}G_{\zeta}\ket{n} = \frac{C_{m,n}}{\pi} e^{-a |\zeta|^2} (\zeta^*)^{n-m} L_m^{(n-m)}\left(- \frac{|\zeta|^2}{b}\right),
\end{align}
where $a:= \frac{1}{\eta_d(1+\bar{n}_d)}$, $b:= \eta_d \bar{n}_d(1+\bar{n}_d)$, $C_{m,n} := \frac{1}{\eta_d^{\frac{n-m+2}{2}}} \frac{\bar{n}_d^m}{(1+\bar{n}_d)^{n+1}} \sqrt{\frac{m!}{n!}}$, $\bar{n}_d := \frac{1-\eta_d+\nu_{\mathrm{el}}}{\eta_d}$, and $L_k^{(\alpha)}(x):= \sum_{j=0}^{k} (-1)^j \begin{pmatrix}
    k+\alpha \\ k-j
\end{pmatrix} \frac{x^j}{j!}$ is the generalized Laguerre polynomial of degree $k$ with parameter $\alpha$. This leads to
\begin{equation}
    \begin{aligned}
    M_{\beta_i, k}^{\textrm{non-ideal}} &= \ket{i}\!\!\bra{i}_A \otimes \sum_{m,n=0}^{\infty} \ket{m, \beta_i}\!\!\bra{n, \beta_i}_B \frac{C_{m,n}}{\pi} \sum_{j=0}^{n}\begin{pmatrix}
    m \\ n-j
\end{pmatrix} \frac{(-1)^j}{b^j j!}      \int_{B_k} e^{-a|\zeta|^2} (\zeta^*)^{m-n} (-1)^j|\zeta|^{2j} d^2\zeta \\
&= \ket{i}\!\!\bra{i}_A \otimes \sum_{m,n=0}^{\infty} \ket{m, \beta_i}\!\!\bra{n, \beta_i}_B \frac{C_{m,n}}{\pi} \sum_{j=0}^{n} \begin{pmatrix}
    m \\ n-j
\end{pmatrix} \frac{1}{b^j j!}      \int_{r_{k}}^{r_{k+1}} e^{- a r^2} r^{m-n+1} r^{2j} dr \int_{0}^{2\pi} e^{-i (m-n)\theta} ~d\theta\\
&= \ket{i}\!\!\bra{i}_A \otimes \sum_{m,n=0}^{\infty} \ket{m, \beta_i}\!\!\bra{n, \beta_i}_B \frac{C_{m,n}}{\pi} \sum_{j=0}^{n} \begin{pmatrix}
    m \\ n-j
\end{pmatrix}   \frac{1}{b^j j!}    \int_{r_{k}}^{r_{k+1}} e^{- a r^2} r^{m-n+1} r^{2j} dr 2\pi \delta_{m,n}\\
&= \ket{i}\!\!\bra{i}_A \otimes \sum_{m=0}^{\infty} \ket{m, \beta_i}\!\!\bra{m, \beta_i}_B 2C_{m,m} \sum_{j=0}^{m} \begin{pmatrix}
    m \\ m-j
\end{pmatrix} \frac{1}{b^j j!}      \int_{r_{k}}^{r_{k+1}} e^{- a r^2} r^{2j+1} dr\\
&= \ket{i}\!\!\bra{i}_A \otimes \sum_{m=0}^{\infty} \ket{m, \beta_i}\!\!\bra{m, \beta_i}_B 2C_{m,m} \sum_{j=0}^{m} \begin{pmatrix}
    m \\ m-j
\end{pmatrix}  \frac{1}{b^j j!}     \int_{r_{k}}^{r_{k+1}} e^{- a r^2} r^{2j+1} dr \\
&= \ket{i}\!\!\bra{i}_A \otimes \sum_{m=0}^{\infty} \ket{m, \beta_i}\!\!\bra{m, \beta_i}_B C_{m,m} \sum_{j=0}^{m} \begin{pmatrix}
    m \\ m-j
\end{pmatrix} \frac{\left(\Gamma(j+1,a r_{k}^2) - \Gamma(j+1, a r_{k+1}^2)\right)}{a^{j+1}b^jj!}   
\end{aligned}
\end{equation}
As in the ideal case, for numerical reasons we replace the upper bound of the sum by the photon cutoff $N_c$. 

\subsection{Derivation of the simulated expectations}
We simulate the quantum channel as phase-invariant Gaussian channel with transmittance $\eta$ and excess noise $\xi$. 
Then, we obtain
\begin{equation}
    \begin{aligned}
    m_{\beta_i, k}^{\textrm{ideal}} &= \Tr{\hat{D}(\beta_i)\rho_{\mathrm{th}}(\bar{n}) \hat{D}^{\dagger}(\beta_i) M_{\beta_i, k}^{\textrm{ideal}}}\\
    &= \Tr{\hat{D}(\beta_i) \sum_{\ell=0}^{\infty} \frac{\bar{n}^\ell}{(1+\bar{n})^{\ell+1}} \ket{\ell}\!\!\bra{\ell} \hat{D}^{\dagger}(\beta_i) \sum_{n=0}^{\infty} \ket{n,\beta_i}\!\!\bra{n, \beta_i}  \frac{\Gamma(n+1, r_{k}^2) - \Gamma(n+1, r_{k+1}^2)}{n!}} \\
    &=  \Tr{\sum_{\ell, n=0}^{\infty} \ket{\ell, \beta_i}\!\!\bra{\ell, \beta_i}  \ket{n,\beta_i}\!\!\bra{n, \beta_i}  \frac{\Gamma(n+1, r_{k}^2) - \Gamma(n+1, r_{k+1}^2)}{n!} \frac{\bar{n}^\ell}{(1+\bar{n})^{\ell+1}}} \\
     &= \sum_{n=0}^{\infty}  \frac{\Gamma(n+1, r_{k}^2) - \Gamma(n+1, r_{k+1}^2)}{n!} \frac{\bar{n}^n}{(1+\bar{n})^{n+1}},
\end{aligned} 
\end{equation} 
where, again, we replace the upper bound of the last sum by our finite cutoff $N_c$.

Similarly, we obtain for the non-ideal detector scenario,
\begin{equation}
    \begin{aligned}
    m_{\beta_i, k}^{\textrm{non-ideal}} &= \Tr{\hat{D}(\beta_i)\rho_{\mathrm{th}}(\bar{n}) \hat{D}^{\dagger}(\beta_i) M_{\beta_i, k}^{\textrm{non-ideal}}}\\
    &= \Tr{ \sum_{\ell, m=0}^{\infty} \frac{\bar{n}^\ell}{(1+\bar{n})^{\ell+1}} \ket{\ell, \beta_i}\!\!\bra{\ell, \beta_i} \ket{m, \beta_i}\!\!\bra{m, \beta_i} C_{m,m} \sum_{j=0}^{m} \begin{pmatrix}
    m \\ m-j
\end{pmatrix} \frac{ \left(\Gamma(j+1,a r_{k}^2) - \Gamma(j+1, a r_{k+1}^2)\right)}{a^{j+1}b^jj!}}  \\
&= \Tr{ \sum_{ m=0}^{\infty} \ket{m, \beta_i}\!\!\bra{m, \beta_i}  \frac{\bar{n}^m C_{m,m}}{(1+\bar{n})^{m+1}}   \sum_{j=0}^{m} \begin{pmatrix}
    m \\ m-j
\end{pmatrix}  \frac{ \left(\Gamma(j+1,a r_{k}^2) - \Gamma(j+1, a r_{k+1}^2)\right)}{a^{j+1}b^jj!}}\\
&=\sum_{ m=0}^{\infty}  \frac{\bar{n}^m C_{m,m}}{(1+\bar{n})^{m+1}}  \sum_{j=0}^{m} \begin{pmatrix}
    m \\ m-j
\end{pmatrix} \frac{ \left(\Gamma(j+1,a r_{k}^2) - \Gamma(j+1, a r_{k+1}^2)\right)}{a^{j+1}b^jj!}
\end{aligned}  
\end{equation}

%--------------------------------------------------------
%         Proof of the Dimension Reduction Statement
%--------------------------------------------------------

\twocolumngrid
\section{Proof of the Dimension Reduction Statement}\label{App:Dim_Red_Proof}
%\subsection{Variable-length Dimension Reduction}
Here, we give the full proof for the variable-length dimension reduction statement (Theorem~\ref{thrm:Dim_red_kappa_special}) and the preceding lemma (\cref{lem:tagged_QKD_entropy_bound}).

\begin{proof}[Proof of \cref{lem:tagged_QKD_entropy_bound}]
    We begin with the case $0 < W_{|\gen} < 1$ and explain the remaining cases briefly at the end of the proof as they are straightforward. For the case $0 < W_{| \gen } < 1$, the proof splits into 3 pieces. First, we show without loss of generality we can work with a purification of $\wt{\rho}_{A\wt{B} \vert \gen}$ of a specific form. Second, we use this specific form to bound the entropy of $\wt{\omega}$ by the re-normalized, projected purification $\hat{\rho}_{ABE \vert \gen}$ acted on by the true (i.e.~not tagged) QKD map. Finally, we bound the difference between the true state and this re-normalized state under the QKD map and use uniform continuity of the sandwiched \Renyi entropy to establish the claimed bound.

Let us start by showing we can consider a purification of $\wt{\rho}_{A\wt{B}\vert \gen}$ with special structure without loss of generality. For notational simplicity, we define $Q = I_{A} \otimes \Pi \otimes I_{E}$.  One purification of \(\widetilde{\rho}_{A\widetilde{B}|\gen}\) is the state 
    \begin{equation}
    \begin{aligned}
        \ket{\psi}_{A\widetilde{B}RF} = &\sqrt{1-W_{|\gen}} \ket{\varphi}_{A\widetilde{B}R} \ket{0}_F \\
        &+ \sqrt{W_{|\gen}} \ket{\chi}_{A\widetilde{B}R} \ket{1}_F,
         \end{aligned}
    \end{equation}
    where $\ket{\varphi}$ is a purification of the re-normalized post-measurement state
    \begin{equation}
        \frac{Q\hat{\rho}_{ABE|\gen}Q}{1-W_{|\gen}} \ ,
    \end{equation}
    with the $E$ system embedded in $R$, and $\ket{\chi}$ is a purification of ,
    \begin{equation}
        \frac{1}{W_{|\gen}}\left(\pTr{B}{\hat{\rho}_{AB|\gen}}-\pTr{B_{N_c}}{\rho_{AB_{N_c}|\gen}}\right)\otimes\pure{\star} \ . 
    \end{equation}
    As all purifications of the state are equivalent up to an isometry on the purifying space, by choosing $RF \cong E\wt{E}$, there exists a unitary such that $\dyad{\psi} = (I_{A\wt{B}} \otimes U_{E\wt{E} \to FR})\wt{\rho}_{A\wt{B}E\wt{E}}(I_{A\wt{B}} \otimes U_{E\wt{E} \to FR})^{\dagger}$. As the R\'{e}nyi entropy is invariant under isometries and $\EATchannQKDTag_j$ commutes with this isometry,
    \begin{align}
        \renyiSandUp_{\alpha}(S|IE\widetilde E)_{\widetilde{\omega}_{|\gen}} = \renyiSandUp_{\alpha}(S|IRF)_{U\widetilde{\omega}_{|\gen}U^{\dagger}} \ . 
    \end{align}
    Thus, we may assume $\widetilde{\omega}_{|\gen}$ is in fact produced by the purification $\dyad{\psi}_{A\wt{B}RF}$, which has this convenient flag structure.

    Now, by definition of the tagged QKD channel $\EATchannQKDTag_j$, $\EATchannQKDTag(\dyad{\psi})$ will be a quantum-classical state where the $F$ register acts as the classical register. To see this, define
    \begin{align}
        \widetilde{\omega}^Q &\defvar \left(\EATchannQKDTag_j \otimes \id_{R} \right)\left[ \pure{\varphi}_{A\widetilde{B}R} \right] \\
        \widetilde{\omega}^{\star} &\defvar \left(\EATchannQKDTag_j \otimes \id_{R} \right)\left[ \pure{\chi}_{A\widetilde{B}R} \right]\\
       \notag &= \pure{\perp}_S \otimes \nu_{IR} \ ,
    \end{align}
    where the final equality uses that the tagged QKD channel always maps the $\star$-component to the discard output in register $S$ and we leave the state $\nu_{IR}$ undetermined. Using these definitions,
        \begin{equation}
        \begin{aligned}[t]
        \widetilde{\omega}_{SIRF} &= \left(\EATchannQKDTag_j \otimes \id_{RF} \right) \left[\pure{\psi}_{A\widetilde{B}RF} \right] \\
        &=\left(1-W_{|\gen}\right) \widetilde{\omega}^Q \otimes \pure{0}_F \\
        &+ W_{|\gen} \widetilde{\omega}^{\star} \otimes \pure{1}_F.
        \end{aligned}
    \end{equation}
    Now, conditioning on a classical register such as $F$ lets us write the conditional \Renyi entropy as
    \begin{equation}
        \begin{aligned}
        &\renyiSandUp_{\alpha}(S|IE\widetilde E)_{\widetilde{\omega}_{|\gen}} = \renyiSandUp_{\alpha}(S|IFR)_{\widetilde{\omega}_{|\gen}}\\
        = &\frac{\alpha}{1-\alpha}\log\Bigg(
        (1-W_{|\gen}) 2^{\frac{1-\alpha}{\alpha}
        \renyiSandUp_{\alpha}(S|IR)_{\widetilde{\omega}^Q}} + W_{|\gen} \Bigg) \\
        \leq &\renyiSandUp_{\alpha}(S|IR)_{\widetilde{\omega}^Q}
    \end{aligned}
    \end{equation}
    where we used that $\renyiSandUp_{\alpha}(S|IR)_{\widetilde{\omega}^{\star}} = 0$ and the inequality follows because the conditional sandwiched \Renyi entropy is nonnegative for classical $S$. 
    Finally, since \(F\) is fixed to \(\ket{0}\) on the projected state, adding it to the conditioning system does not change the entropy. After identifying \(RF\) isometrically with \(E\widetilde E\), we use the same symbol \(\widetilde{\omega}^{Q}\) for the resulting state \(\widetilde{\omega}^{Q}_{SIR}\otimes\pure{0}_F\) on \(SIE\widetilde E\). We therefore obtain
\begin{equation}
        \begin{aligned}
        \renyiSandUp_{\alpha}(S|IR)_{\widetilde{\omega}^{Q}}
        &=
        \renyiSandUp_{\alpha}(S|IRF)_{\widetilde{\omega}^{Q}\otimes\pure{0}_F}\\
        &=
        \renyiSandUp_{\alpha}(S|IE\widetilde E)_{\widetilde{\omega}^{Q}}\\
        &\leq
        \renyiSandUp_{\alpha}(S|IE)_{\widetilde{\omega}^{Q}_{SIE}}\\
        &=
        \renyiSandUp_{\alpha}(S|IE)_{\left(\EATchannQKD_j\otimes\id_E\right)\left[\frac{Q\hat{\rho}_{ABE|\gen}Q}{1-W_{|\gen}}\right]},
    \end{aligned}
\end{equation}
    where the inequality follows from strong subadditivity of the conditional \Renyi entropy \cref{prop:Inf-strong-subadd}, while the final equality follows from the choice of \(\ket{\varphi}\) and the agreement of the original and tagged QKD channels on inputs supported on \(B_{N_c}\). This yields the following inequality for the output state of the tagged state and channel
    \begin{equation}\label{eq:proof_entr_tagged_ntr_bnd}
        \renyiSandUp_{\alpha}(S|IE\widetilde E)_{\widetilde{\omega}_{|\gen}} \leq \renyiSandUp_{\alpha}(S|IE)_{\left(\EATchannQKD_j\otimes\id_E\right)\left[\frac{Q\hat{\rho}_{ABE|\gen}Q}{1-W_{|\gen}}\right]}.
    \end{equation}
    
    Next, we aim to relate this conditional \Renyi entropy of the renormalized projected state to the conditional \Renyi entropy of the original state $\hat{\rho}_{ABE|\gen}$. We will do so by using the continuity bound for the sandwiched conditional \Renyi entropy from \cite[Corollary~5.1]{Bluhm_2026}, which requires a bound on the trace distance. 
    
    For this purpose, we apply the normalized gentle measurement lemma \cref{lem:Norm_gentle_measurement_lemma} to \(\hat{\rho}_{ABE|\gen}\) with the projector \(Q = I_A\otimes\Pi\otimes I_E\), which gives
    \begin{equation}\label{eq:gentle-projected-generation-state}
        \frac{1}{2}\norm{
        \hat{\rho}_{ABE|\gen} - \frac{Q \hat{\rho}_{ABE|\gen} Q}{1-W_{|\gen}} }_1 \leq \sqrt{W_{|\gen}}.
    \end{equation}
    Hence, by the data-processing inequality for the trace distance, we find
    \begin{equation}\label{eq:gentle-projected-generation-output}
        \frac{1}{2}\norm{\omega_{|\gen} - \left(\EATchannQKD_j\otimes\id_E\right) \left[ 
        \frac{Q \hat{\rho}_{ABE|\gen} Q}{1-W_{|\gen}} \right]
        }_1 \leq \sqrt{W_{|\gen}}.
    \end{equation}
    Then, due to the continuity bound for the sandwiched conditional \Renyi entropy from \cite[Corollary~5.1]{Bluhm_2026}, with \(d_S^2\) replaced by \(d_S\) since \(S\) is classical, we find
    \begin{align}
        &\renyiSandUp_{\alpha}(S|IE)_{\omega_{|\gen}} \\
        &\geq
        \renyiSandUp_{\alpha}(S|IE)_{\left(\EATchannQKD_j\otimes\id_E\right)\left[
        \frac{Q\hat{\rho}_{ABE|\gen}Q}{1-W_{|\gen}} \right]}
        - \Delta_{\alpha}\!\left(\sqrt{W_{|\gen}}\right) 
    \end{align}
    Now, we can insert \cref{eq:proof_entr_tagged_ntr_bnd} to find
    \begin{align}    
        \renyiSandUp_{\alpha}(S|IE)_{\omega_{|\gen}} \geq \renyiSandUp_{\alpha}(S|IE\widetilde E)_{\widetilde{\omega}_{|\gen}}
        - \Delta_{\alpha}\!\left(\sqrt{W_{|\gen}}\right),
    \end{align}
    which concludes the claim for $0 < W_{|\gen} <1$.

    If \(W_{|\gen}=1\), then \(\widetilde{\rho}_{A\widetilde{B}|\gen}\) is supported entirely on the \(\star\)-subspace and \(\EATchannQKDTag_j\) maps it to the deterministic discard output. Thus,
    \begin{equation}
        \renyiSandUp_{\alpha}(S|IE\widetilde E)_{\widetilde{\omega}_{|\gen}}=0.
    \end{equation}
    Since \(S\) is classical,
    \begin{equation}
        \renyiSandUp_{\alpha}(S|IE)_{\omega_{|\gen}}\geq0,
    \end{equation}
    and using $\Delta_{\alpha}(1) \geq 0$, the claim also follows for \(W_{|\gen}=1\).

    If \(W_{|\gen}=0\), then \(\hat{\rho}_{AB|\gen}\) is supported on \(AB_{N_c}\), and \(\widetilde{\rho}_{A\widetilde{B}|\gen}\) agrees with it under the canonical embedding into the tagged space. Since the original and tagged QKD channels agree on this subspace, their outputs are isometrically equivalent. Their conditional \Renyi entropies are therefore equal, and the claim follows from \(\Delta_\alpha(0)=0\).
\end{proof}

\begin{proof}[Proof of \cref{thrm:Dim_red_kappa_special}]
    The proof consists of four main parts: rewriting the marginal constraint, bounding the test-round constraints using the statistic-generating map \(\probst_j\), bounding the conditional entropy, and incorporating the dimension-reduction constraints.

    We first introduce weights that keep track of the probability mass outside the cutoff subspace in the different protocol branches, i.e. test and generation rounds. These weights will play a crucial role in any of the bounds we derive in this proof. For any state \(\hat{\rho}_{AB}\in\dop{=}(AB)\) satisfying the marginal constraint \(\pTr{B}{\hat{\rho}_{AB}}=\sigma_A^{(j)}\), define
    \begin{align}
        W_{\hat{\rho}} &\defvar \Tr{\left(I_A\otimes\Pibar\right)\hat{\rho}_{AB}},\\
        W_{\hat{\rho},\test} &\defvar \Tr{\left(I_A\otimes\Pibar\right)\hat{\rho}_{AB|\test}}\\
        \notag&=\frac{1}{\gamma}\Tr{\left(\Pi^{\test}\otimes\Pibar\right)\hat{\rho}_{AB}},\\
        W_{\hat{\rho},\gen} &\defvar \Tr{\left(I_A\otimes\Pibar\right)\hat{\rho}_{AB|\gen}}\\
        \notag&=\frac{1}{1-\gamma}\Tr{\left(\Pi^{\gen}\otimes\Pibar\right)\hat{\rho}_{AB}},\\
        W_{\hat{\rho},\test}^{\annA,x} &\defvar \Tr{\Pibar\hat{\rho}_{B|(\test\wedge\annA,x)}}.
    \end{align}
    Here, \(\hat{\rho}_{AB|\test}\) and \(\hat{\rho}_{AB|\gen}\) denote the normalized states conditioned on a test round and a generation round, respectively. The second equality in each of the definitions of \(W_{\hat{\rho},\test}\) and \(W_{\hat{\rho},\gen}\) follows directly from the definition of the corresponding conditional state and the cyclicity of the trace. For outcomes satisfying \(p(\annA,x|\test)=0\), we set \(W_{\hat{\rho},\test}^{\annA,x}=0\).
    
    Because Alice chooses a test or generation round with probabilities \(\gamma\) and \(1-\gamma\), respectively, and \(\Pi^{\test}+\Pi^{\gen}=I_A\), the definitions above imply
\begin{equation}
        \begin{aligned}
        W_{\hat{\rho}}
        &=\Tr{\left(\left(\Pi^{\test}+\Pi^{\gen}\right)\otimes\Pibar\right)\hat{\rho}_{AB}}\\
        &=\Tr{\left(\Pi^{\test}\otimes\Pibar\right)\hat{\rho}_{AB}}+\Tr{\left(\Pi^{\gen}\otimes\Pibar\right)\hat{\rho}_{AB}}\\
        &=\gamma W_{\hat{\rho},\test}+(1-\gamma)W_{\hat{\rho},\gen}.
    \end{aligned}
\end{equation}
    Furthermore, the POVM elements conditioned on test rounds satisfy
    \begin{equation}
        \sum_{\annA,x}M_{\annA,x}^{A|\test} = \Lambda_{\supp\left(\Pi^{\test}\right)},
    \end{equation}
    where $\Lambda_{\supp\left(\Pi^{\test}\right)}$ is the projector onto the support of $\Pi^{\test}$, see \cite[App.~H]{Kamin_2025a}. Therefore,
\begin{equation}
        \begin{aligned}
        W_{\hat{\rho},\test}
        &= \Tr{\left(I_A \otimes\Pibar\right)\hat{\rho}_{AB|\test}}\\
        &= \Tr{\left(\Lambda_{\supp\left(\Pi^{\test}\right)} \otimes\Pibar\right)\hat{\rho}_{AB|\test}}\\
        &= \sum_{\annA,x}\Tr{\left(M_{\annA,x}^{A|\test}\otimes\Pibar\right)\hat{\rho}_{AB|\test}}\\
        &= \sum_{\annA,x}p(\annA,x|\test)\Tr{\Pibar\hat{\rho}_{B|(\test\wedge\annA,x)}}\\
        &= \sum_{\annA,x}p(\annA,x|\test)W_{\hat{\rho},\test}^{\annA,x}.
    \end{aligned}
\end{equation}
    Here, the fourth equality follows directly from the definition of the normalized conditional state \(\hat{\rho}_{B|(\test\wedge\annA,x)}\), according to which
    \begin{equation}
        \begin{split}
            &\Tr{\left(M_{\annA,x}^{A|\test}\otimes\Pibar\right)\hat{\rho}_{AB|\test}} \\ 
            =& p(\annA,x|\test)\Tr{\Pibar\hat{\rho}_{B|(\test\wedge\annA,x)}},
        \end{split}
    \end{equation}
    while the final equality follows from the definition of \(W_{\hat{\rho},\test}^{\annA,x}\).
    Next, we project Bob's system \(B\) onto the cutoff subspace \(B_{N_c}\) and define the resulting subnormalized state \(\rho\in\dop{\leq}(AB_{N_c})\) as
    \begin{equation}
        \rho_{AB_{N_c}} \defvar \left(I_A\otimes\Pi\right)\hat{\rho}_{AB}\left(I_A\otimes\Pi\right).
    \end{equation}    
    By the definition of \(W_{\hat{\rho}}\), the subnormalized projected state $\rho_{AB_{N_c}}$ satisfies
    \begin{equation}
        \Tr{\rho_{AB_{N_c}}}=\Tr{\left(I_A\otimes\Pi\right)\hat{\rho}_{AB}}=1-W_{\hat{\rho}}.
    \end{equation}
    Moreover, using \(\Pi\Pibar=0\) and the fact that the off-diagonal terms vanish under the partial trace over \(B\), we obtain
    \begin{align}
        \pTr{B}{\left(I_A\otimes\Pibar\right)\hat{\rho}_{AB}\left(I_A\otimes\Pibar\right)} = \sigma_A^{(j)}-\pTr{B_{N_c}}{\rho_{AB_{N_c}}}.
    \end{align}
    Since the partial trace of a positive semidefinite operator is itself positive semidefinite, it follows that
    \begin{equation}
        \sigma_A^{(j)}-\pTr{B_{N_c}}{\rho}\geq0,
    \end{equation}
    or, equivalently,
    \begin{equation}
        \pTr{B_{N_c}}{\rho}\leq\sigma_A^{(j)}.
    \end{equation}    

    We can therefore define the normalized tagged state $\widetilde{\rho} \in \dop{=}(A \widetilde{B})$ by
    \begin{equation}
        \widetilde{\rho}_{A\widetilde{B}} \defvar \rho_{AB_{N_c}} + \left( \sigma_A^{(j)} - \pTr{B_{N_c}}{\rho_{AB_{N_c}}} \right) \otimes \pure{\star},
    \end{equation}
    where we again defined $\widetilde{B} = B_{N_c} \oplus \Span\{ \ket{\star} \}$. This state is positive semidefinite and normalized, since
    \begin{equation}
        \Tr{\widetilde{\rho}_{A\widetilde{B}}}
        =\Tr{\rho_{AB_{N_c}}}+\Tr{\left(\sigma_A^{(j)}-\pTr{B_{N_c}}{\rho_{AB_{N_c}}}\right)}
        =1.
    \end{equation}
    Moreover, its marginal on Alice's system is
\begin{equation}
        \begin{aligned}
        \pTr{\widetilde{B}}{\widetilde{\rho}_{A\widetilde{B}}} &= \pTr{B_{N_c}}{\rho_{AB_{N_c}}} +  \sigma_A^{(j)} - \pTr{B_{N_c}}{\rho_{AB_{N_c}}} \\
        &= \sigma_A^{(j)}.
    \end{aligned}
\end{equation}
    Since the test and generation events are determined entirely by Alice, conditioning on either event commutes with the projection-and-tagging operation on Bob's system. In particular,
    \begin{equation}
        \rho_{AB_{N_c}|\test} = \left(I_A \otimes \Pi \right) \hat{\rho}_{AB|\test} \left(I_A \otimes \Pi \right)
    \end{equation}
    and
    \begin{equation}
        \widetilde{\rho}_{A\widetilde{B}|\test}  = \begin{aligned}[t]
        &\rho_{AB_{N_c}|\test} + \bigl( \pTr{B}{ \hat{\rho}_{AB|\test}} \\
        &- \pTr{B_{N_c}}{\rho_{AB_{N_c}|\test}} \bigr) \otimes \pure{\star},
        \end{aligned}
    \end{equation}
    where
    \begin{equation}
        \pTr{B}{\hat{\rho}_{AB|\test}} = \frac{\sqrt{\Pi^{\test}}\sigma_A^{(j)}\sqrt{\Pi^{\test}}}{\gamma}.
    \end{equation}
    The analogous identities hold for generation rounds, with \(\Pi^{\test}\) and \(\gamma\) replaced by \(\Pi^{\gen}\) and \(1-\gamma\), respectively.

    By the definitions of the conditional weights, the projected conditional states satisfy
    \begin{align}
        \Tr{\rho_{AB_{N_c}|\test}}&=\Tr{\left(I_A\otimes\Pi\right)\hat{\rho}_{AB|\test}}=1-W_{\hat{\rho},\test},\\
        \Tr{\rho_{AB_{N_c}|\gen}}&=\Tr{\left(I_A\otimes\Pi\right)\hat{\rho}_{AB|\gen}}=1-W_{\hat{\rho},\gen}.
    \end{align}
    Likewise, for every \((\annA,x)\) such that \(p(\annA,x|\test)>0\), defining
    \begin{equation}
        \rho_{AB_{N_c}|(\test\wedge\annA,x)}\defvar\left(I_A\otimes\Pi\right)\hat{\rho}_{AB|(\test\wedge\annA,x)}\left(I_A\otimes\Pi\right),
    \end{equation}
    we obtain
    \begin{equation}
        \Tr{\rho_{AB_{N_c}|(\test\wedge\annA,x)}}=1-W_{\hat{\rho},\test}^{\annA,x}.
    \end{equation}

    Next, we derive upper and lower bounds on the test-round statistics following \cite[Theorem~4]{Upadhyaya_2021}. Since \([\Pi,M_{\annB,y}^B]=0\), Bob's POVM elements decompose as
    \begin{equation}\label{eq:M_B_decompose_Pi}
        M_{\annB,y}^B=\Pi M_{\annB,y}^B\Pi+\Pibar M_{\annB,y}^B\Pibar.
    \end{equation}
    \begin{widetext}
    Thus, for the lower bound on \(\bsym{\nu}_{|\test}\), we find
   \begin{equation}
        \begin{aligned}
        \bsym{\nu}_{|\test}
        &=\probst_j[\hat{\rho}_{AB|\test}]\\
        &=\sum_{\cP\in\Ct}\sum_{\substack{(\annA,x,\annB,y)\\\in\phi^{-1}(\cP)}}\Tr{\left(M_{\annA,x}^{A|\test}\otimes M_{\annB,y}^B\right)\hat{\rho}_{AB|\test}}\hat{e}_{\cP}\\
        &\geq\sum_{\cP\in\Ct}\sum_{\substack{(\annA,x,\annB,y)\\\in\phi^{-1}(\cP)}}\Tr{\left(M_{\annA,x}^{A|\test}\otimes M_{\annB,y}^B\right)\left(I_A\otimes\Pi\right)\hat{\rho}_{AB|\test}\left(I_A\otimes\Pi\right)}\hat{e}_{\cP}\\
        &=\begin{aligned}[t]
            \sum_{\cP\in\Ct}\sum_{\substack{(\annA,x,\annB,y)\\\in\phi^{-1}(\cP)}}\Tr{\left(M_{\annA,x}^{A|\test} \otimes \widetilde{M}_{\annB,y}^B \right)
            \Bigg(\rho_{AB_{N_c}|\test}
            +\left(\pTr{B}{\hat{\rho}_{AB|\test}}-\pTr{B_{N_c}}{\rho_{AB_{N_c}|\test}}\right)\otimes\pure{\star}\Bigg)}
            \hat{e}_{\cP}
        \end{aligned}\\
        &=\sum_{\cP\in\Ct}\sum_{\substack{(\annA,x,\annB,y)\\\in\phi^{-1}(\cP)}}\Tr{\left(M_{\annA,x}^{A|\test}\otimes\widetilde{M}_{\annB,y}^B\right)\widetilde{\rho}_{A\widetilde{B}|\test}}\hat{e}_{\cP}\\
        &=\probstTag_j[\rho_{AB_{N_c}|\test}].
    \end{aligned}
   \end{equation}

    Here, the inequality in step three follows from \cref{eq:M_B_decompose_Pi}, since the contribution from the complementary subspace is nonnegative due to
    \begin{equation}
        M_{\annA,x}^{A|\test} \otimes \Pibar M_{\annB,y}^B \Pibar \geq 0.
    \end{equation}
    Step four follows because \(\widetilde{M}_{\annB,y}^B\) agrees with \(\Pi M_{\annB,y}^B\Pi\) on \(B_{N_c}\) and vanishes on \(\Span\{\ket{\star}\}\). Steps five and six use the definitions of the normalized tagged state and the tagged statistic-generating map, respectively.

    Similarly, for the upper bound, we find
  \begin{equation}
        \begin{aligned}
        \bsym{\nu}_{|\test}
        &\leq\sum_{\cP\in\Ct}\sum_{\substack{(\annA,x,\annB,y)\\\in\phi^{-1}(\cP)}}\Tr{\left(M_{\annA,x}^{A|\test}\otimes\widetilde{M}_{\annB,y}^{B}\right)\widetilde{\rho}_{A\widetilde{B}|\test}}\hat{e}_{\cP}
        + W_{\hat{\rho},\test} \sum_{\cP\in\Ct}\sum_{\substack{(\annA,x,\annB,y)\\\in\phi^{-1}(\cP)}} \left\|M_{\annA,x}^{A|\test} \otimes \Pibar M_{\annB,y}^B \Pibar \right\|_{\infty} \hat{e}_{\cP}\\
        &=\sum_{\cP\in\Ct} \sum_{\substack{(\annA,x,\annB,y)\\\in\phi^{-1}(\cP)}} \Tr{\left(M_{\annA,x}^{A|\test}\otimes\widetilde{M}_{\annB,y}^{B}\right) \widetilde{\rho}_{|\test}} \hat{e}_{\cP}
        +W_{\hat{\rho},\test} \bsym{\delta}\\
        &=\probstTag_j[\rho_{AB_{N_c}|\test}]+ W_{\hat{\rho},\test} \bsym{\delta}.
    \end{aligned}
  \end{equation}
    \end{widetext}
    Here, the inequality follows from Hölder's inequality and the definition of $W_{\hat{\rho},\test}$.

    As the third part of the proof, we bound the entropy for the state conditioned on a key generation round. For this purpose, define the following states
    \begin{equation}
    \begin{gathered}
        \widetilde{\omega}_{SI\CP E \widetilde{E}|\gen} \defvar \EATchannQKDTag_j \left[\widetilde{\rho}_{A\widetilde{B} E \widetilde{E}|\gen} \right],\\
        \omega_{SI\CP E|\gen} \defvar \EATchannQKD_j \left[\hat{\rho}_{ABE|\gen}\right],
    \end{gathered}
    \end{equation}
    where $\widetilde{\rho}_{A\widetilde{B} E \widetilde{E}|\gen}$ and $\hat{\rho}_{ABE|\gen}$ are purifications of $\widetilde{\rho}_{A\widetilde{B}|\gen}$ and $\hat{\rho}_{AB|\gen}$, respectively.
    Then, by \cref{lem:tagged_QKD_entropy_bound}, it holds
    \begin{equation}
        \renyiSandUp_{\alpha}(S|IE)_{\omega_{|\gen}} \geq \renyiSandUp_{\alpha}(S|IE\widetilde E)_{\widetilde{\omega}_{|\gen}} - \Delta_{\alpha}\!\left(\sqrt{W_{\hat{\rho},\gen}}\right).
    \end{equation}    

    As the fourth step, we incorporate the dimension-reduction constraints. Applying the operator inequality
    \begin{equation}
        \Pibar\leq rV_1
    \end{equation}
    to the normalized conditional state \(\hat{\rho}_{B|(\test\wedge\annA,x)}\) gives
\begin{equation}
        \begin{aligned}
        W_{\hat{\rho},\test}^{\annA,x}
        &=\Tr{\Pibar\hat{\rho}_{B|(\test\wedge\annA,x)}}\\
        &\leq r\Tr{V_1\hat{\rho}_{B|(\test\wedge\annA,x)}}.
    \end{aligned}
\end{equation}
    Multiplying both sides by \(p(\annA,x|\test)>0\), we obtain
    \begin{equation}
        p(\annA,x|\test)W_{\hat{\rho},\test}^{\annA,x}
        \leq r\,p(\annA,x|\test)\Tr{V_1\hat{\rho}_{B|(\test\wedge\annA,x)}}.
    \end{equation}
    By the definition of the normalized conditional state,
\begin{equation}
        \begin{aligned}
        &p(\annA,x|\test)\Tr{V_1\hat{\rho}_{B|(\test\wedge\annA,x)}}\\
        &=\Tr{\left(M_{\annA,x}^{A|\test}\otimes V_1\right)\hat{\rho}_{AB|\test}}\\
        &=\nu_{\annA,x,1}.
    \end{aligned}
\end{equation}
    Therefore, we find
    \begin{equation}
        p(\annA,x|\test)W_{\hat{\rho},\test}^{\annA,x}
        \leq r \nu_{\annA,x,1}.
    \end{equation}
    If $p(\annA,x|\test)=0$, this constraint holds trivially because then also $\nu_{\annA,x,1}=0$ and the constraint amounts to $0\leq 0$.
    
    The weights associated with the projected, subnormalized state $\rho$ agree with those defined from the original state $\hat{\rho}$. Indeed, the trace identities derived above imply
    \begin{align}
        W_{\rho} &\defvar 1-\Tr{\rho} = W_{\hat{\rho}},\\
        W_{\rho,\test} &\defvar 1-\Tr{\rho_{|\test}} =W_{\hat{\rho},\test},\\
        W_{\rho,\gen} &\defvar 1-\Tr{\rho_{|\gen}} =W_{\hat{\rho},\gen},\\
        W_{\rho,\test}^{\annA,x} &\defvar 1-\Tr{\rho_{|(\test\wedge\annA,x)}} = W_{\hat{\rho},\test}^{\annA,x}.
    \end{align}
    We may therefore drop the state labels and denote these weights by \(W\), \(W_{\test}\), \(W_{\gen}\), and \(W_{\test}^{\annA,x}\), respectively.

    The last step of the proof concerns the feasible sets. So far the weights depend on the particular state. We will regard them as optimization variables themselves due to the following argument.
    
    Every feasible infinite-dimensional state \(\hat{\rho}_{AB}\) induces a projected state \(\rho_{AB_{N_c}}\), a vector \(\bsym{\nu}_{|\test}\), and weights \(W\), \(W_{\test}\), \(W_{\gen}\), and \(W_{\test}^{\annA,x}\) satisfying all the constraints in the theorem statement. Treating \(W_{\test}^{\annA,x}\) and \(W_{\gen}\) as optimization variables subject only to the stated constraints can only enlarge the feasible set and therefore cannot increase the infimum. Combining this relaxation with the bounds derived above proves the claimed lower bound on \(\kapup_j\).
\end{proof}

%--------------------------------------------------------
%         Auxiliary Infinite Dimensional Results
%--------------------------------------------------------

\section{Auxiliary Infinite Dimensional Results}\label{app:inf-dim-stuff}
Here we prove various results in operator theory and entropies on separable Hilbert spaces that we need for the other proofs. In Section \ref{app:inf-dim-prelim}, we provide what we believe to be sufficient background to follow this section. In Section \ref{app:inf-dim-variational-form}, we prove the variational form of Schatten norms for separable Hilbert spaces as we could not find a reference. In Section \ref{app:inf-dim-norm-contracts}, we extend some known Schatten norm contractivity results to separable Hilbert spaces. In Section \ref{app:inf-dim-LHL-against-inf-side-info}, we extend Dupuis' leftover hashing lemma for R\'{e}nyi entropies \cite{Dupuis_2022} to applying to infinite-dimensional side information. In Section \ref{app:inf-dim-duality}, we establish duality between Petz and Sandwiched \Renyi entropies with infinite-dimensional side-information so long as the unconditioned system is finite-dimensional. In Section \ref{app:inf-dim-DPI}, we explain how the DPI extends to separable Hilbert spaces. Finally, in Section \ref{app:inf-dim-classical-side-info}, we establish the chain rule for classical side-information when other systems are separable Hilbert spaces. To the best of our knowledge, this proof method is novel in that it is direct rather than using duality as in \cite{Tomamichel_2016}.

\subsection{Preliminaries}\label{app:inf-dim-prelim}
\paragraph{Operator Theory} We let $A$ and $B$ denote separable Hilbert spaces. We let $\cB(A,B)$  (resp.~$\cK(A,B$)~) denote the bounded (resp.~compact) operators from $A$ to $B$. For notational simplicity, we let $\cB(A) \coloneq \cB(A,A)$ and similarly for $\cK(A)$. We recall the Schatten $p$-norms for $p \in (1,\infty)$:
    \begin{align}\label{eq:Schatten-norm}
     \Vert T \Vert_{p} \coloneq \left[\Tr{\vert T \vert^{p}} \right]^{1/p} \ ,
    \end{align}
where $\vert T \vert = \sqrt{T^{\ast}T}$ and $\Tr{\cdot}$ is the trace. An operator $T :A \to B$ is Schatten $p$-class if and only if $\Vert T \Vert_{p} < +\infty$ and denote the set of Schatten $p$-class operators from $A$ to $B$ by $\cS_{p}(A,B)$. We note that it is known that the set of Schatten $1$-class operators are the trace class operators, $\cT(A,B)$. These definitions are critical for two reasons. First, these norms are used to define the sandwiched R\'{e}nyi divergences. Second, as finite-dimensional matrices are Schatten $p$-class for all $p$, but not all operators are, our basic concern is somehow this leads to the break down of some step in Dupuis' proof.

Following \cite{Mosonyi_2023}, we denote bounded self-adjoint (i.e. $X = X^{\ast}$) operators on $A$ by $\cB(A)_{sa}$ and remind the reader that being `Hermitian' is a strictly weaker condition in infinite dimensions than being self-adjoint. We denote the set of bounded non-zero positive semidefinite operators on $A$ by $\cB(A)_{\gneq 0}$ and $\cB(A)_{\geq 0} = \cB(A)_{\gneq 0} \cup \{0_{A}\}$. We denote the set of density matrices on $A$ by $\cS(A)$. We denote the set of normal operators on $A$ by $\cN(A)$. 

With regards to linear maps, we will be interested in ones that take bounded operators to bounded operators, i.e. linear maps $\msf{N}:\cB(A) \to \cB(B)$. If a linear map satisfies $\msf{N}(X^{\ast}) = \msf{N}(X)^{\ast}$ we call it self-adjoint as it will take self-adjoint operators to self-adjoint operators. For a linear map $\msf{N}$, the $\cS_{p} \to \cS_{q}$ norm is 
\begin{align}\label{eq:norm-between-schatten-spaces}
    \Vert N \Vert_{p \to q} \coloneq \sup_{X \neq 0} \frac{\Vert N(X) \Vert_{q}}{\Vert X \Vert_{p}} \ . 
\end{align}

We recall that compact operators are basically operators where all the standard finite dimensional stuff still works (see e.g.~\cite{Conway_2007}). In particular, they have singular value decompositions, i.e. for $T \in \cK(A,B)$, 
\begin{align}
    T = \sum_{n=1}^{\infty} \sqrt{\lambda_{n}}\ket{\phi_{n}}\bra{\psi_{n}} \ , 
\end{align}
where $(\lambda_{n})$ are the eigenvalues of $T^{\ast}T$, $(\ket{\psi_{n}})_{n}$ is an orthonormal basis and $(\ket{\phi_{n}} \coloneq W\ket{\psi_{n}})_{n}$ where $W$ is a partial isometry. This is particularly useful as then defining $\sigma_{n}(T) \coloneq \sqrt{\lambda_{n}}$, we have
\begin{align}
    \Vert T \Vert_{p} = \left[\Tr{\vert T \vert^{p}} \right]^{1/p} = \left[\sum_{n=1}^{\infty} \sigma_{n}(T)^{p}\right]^{1/p} = \Vert \sigma(T) \Vert_{p} \ .
\end{align}
Moreover, for compact, self-adjoint operators the functional calculus works the same way as in finite dimensions (except note that it may change whether it is of a certain Schatten $p$-class or even bounded). 

% {\color{blue}*** I wound up not needing topology, but I left it in case it helps with Lars's stuff ****}
% \paragraph{Topology} We next state some topological definitions. We start with the standard ones.
% \begin{definition}
%     Let $T\in \cB(A)$ and $(T_{n})_{n} \subset \cB(A)$.
%     \begin{enumerate}
%         \item $(T_{n})_{n}$ converges strongly to $T$ if $\lim_{n \to \infty} \Vert Tx - T_{n}x\Vert_{A} = 0 $ for all $x \in A$.
%         \item $(T_{n})_{n}$ converges weakly to $T$ if $\lim_{n \to \infty} \langle y, T_{n}x \rangle = \langle y , Tx \rangle_{A}$ for all $x,y \in \cH$. 
%     \end{enumerate}
% \end{definition}
% We then define a map generalization of strong convergence.
% \begin{definition}
%     \cite{Shirokov_2007} Let $\Phi \in \cL(\cT(A),\cT(B))$ and $(\Phi_{n})_{n} \subset \cL(\cT(A),\cT(B))$. Then we say it converges strongly if 
%     \begin{align}
%         \lim_{n \to \infty} \Vert (\Phi - \Phi_{n})(X) \Vert_{1} = 0 \quad \forall X \in \cT(A) \ . 
%     \end{align}
% \end{definition}
% {\color{blue} **** }

\paragraph{Entropies} We now define the relevant sandwiched R\'{e}nyi quantities for separable Hilbert spaces.
\begin{definition}\label{def:inf-dim-renyis}
    Let $P,Q \in \cB(A)_{\gneq 0}$. Let $\alpha > 1$.
    We define the Sandwiched and Petz mean functionals by: %This uses Petz is alpha,z where z = 1 and Sandwiched is when z = alpha
    \begin{align}
        \wt{Q}_{\alpha}(P \Vert Q) &\coloneq \begin{cases}\label{eq:inf-dim-sandwiched-Q-function}
            \Vert Q^{\frac{1-\alpha}{2\alpha}}P^{\frac{1}{2}} \Vert_{2\alpha}^{2\alpha} & \ran(P^{\frac{\alpha}{2\alpha}}) \subseteq \ran(Q^{\frac{\alpha-1}{2\alpha}}) \\ 
            +\infty & \text{otherwise} 
        \end{cases} \\
        \ol{Q}_{\alpha}(P \Vert Q) &\coloneq \begin{cases}
            \Vert Q^{\frac{1-\alpha}{2}} P^{\frac{\alpha}{2}} \Vert_{2}^{2} & \ran(P^{\frac{\alpha}{2}}) \subseteq \ran(Q^{\frac{\alpha-1}{2}}) \\ 
            +\infty & \text{otherwise} \ . 
        \end{cases}
    \end{align}
    
    We then define the sandwiched \Renyi and Petz divergences respectively as
    \begin{align}
        &\wt{D}_{\alpha}(P \Vert Q) \coloneq \frac{1}{\alpha-1}\log \wt{Q}_{\alpha}(P \Vert Q) - \frac{1}{\alpha-1}\log(\Tr{P}), \\
        &\wt{D}_{\alpha}(P \Vert Q) \coloneq \frac{1}{\alpha-1}\log \ol{Q}_{\alpha}(P \Vert Q)  - \frac{1}{\alpha-1}\log(\Tr{P}) \ .
    \end{align}
\end{definition}
\begin{remark} For simplicity, the above are stated as a definition. However, that these definitions are well-defined come from the following points from \cite{Mosonyi_2023}: Lemma III.1, Definition III.2, Definition III.9, the unnumbered equation beneath Eqn.~(3.14), and that the $(\alpha,z)$-divergence is the sandwiched \Renyi divergence when $\alpha = z$.
\end{remark}

We may then define the relevant \Renyi entropies.
\begin{definition}
    Let $\rho \in \cS_{=}(A \otimes B),\sigma \in \cS(B)_{=}$. The sandwiched and Petz \Renyi entropy of $\rho$ conditioned on $\sigma$ are
    \begin{align}
        \mbb{H}_{\alpha}(A \vert B)_{\rho \vert \sigma} \coloneq -\mbb{D}_{\alpha}(\rho_{AB} \Vert I_{A} \otimes \sigma_{B}) \ , 
    \end{align}
    where $\mbb{D} \in \{\wt{D},\ol{D}\}$. Then we have
    \begin{align}
        \mbb{H}^{\uparrow}_{\alpha}(A \vert B)_{\rho} &\coloneq \sup_{\sigma_{B} \in S_{=}(B)} \mbb{H}_{\alpha}(A \vert B)_{\rho \vert \sigma} \\
        \mbb{H}_{\alpha}^{\downarrow}(A \vert B)_{\rho} &\coloneq \mbb{H}_{\alpha}(A \vert B)_{\rho \vert \rho_{B}} \ . 
    \end{align}
\end{definition}

\begin{remark}\label{rem:non-trivial-cases}
    We remark the non-trivial cases of the mean functionals arise for both up and down arrow entropies. To see this, it suffices to see it for the down arrow case. By the same argument as in finite dimensions (c.f.~\cite[Lemma B.4.1]{Renner_2006}), for $0 \leq P_{AB} \in \cB(A \otimes B)$, $\ran(P_{AB}) \subset \ran(P_{A}) \otimes \ran(P_{B})$. Then we have for $\alpha >1, z > 1$
    \begin{align*}
        \ran((I_{A} \otimes \rho_{B})^{\frac{\alpha-1}{2z}}) &= A \otimes \ran(\rho_{B}^{\frac{\alpha-1}{2z}}) \\
        &\supseteq A \otimes \ran(\rho_{B}^{\frac{\alpha}{2z}}) \\
        &\supseteq \ran(\rho_{A}^{\frac{\alpha}{2z}}) \otimes \ran(\rho_{B}^{\frac{\alpha}{2z}}) \supset \ran(\rho_{AB}^{\frac{\alpha}{2z}}) \ , 
    \end{align*}
    where we used $\ran(P^{\alpha}) \supseteq \ran(P^{\beta})$ for $P \geq 0$ and $\beta > \alpha$. This shows that we have the conditions for the norm definition.
\end{remark}

Finally, as the $f$-weighted \Renyi entropies are simply defined by evaluating the sandwiched \Renyi entropies conditioned on a classical register, they immediately extend.
\begin{definition}\label{def:inf-dim}
	Let $Q,Q'$ be separable Hilbert spaces and $\alphCP$ be a finite alphabet. Consider classical-quantum state $\rho_{\CP Q Q'} = \sum_{c \in \CP} p_{C}(c)\dyad{c}_{C} \otimes \rho_{\vert \cP}$ where $\{\rho_{\vert \cP}\}_{\cP \in \alphCP} \subset S_{=}(QQ')$. Let $f:\alphCP \to \R$ and $\alpha\in(0,1)\cup (1,\infty)$, we define the up-arrow and down-arrow \emph{$f$-weighted \Renyi entropy} of order $\alpha$ for $\rho$, as follows:
	\begin{equation}\label{eq:upfweighted entropy inf}
    \begin{aligned}
			&\hspace{-3mm}\frenyiSandUp_\alpha(Q|\CP Q')_{\rho} \\
            & \hspace{-3mm} \defvar \frac{\alpha}{1-\alpha} \log \left( \sum_{\cP} \rho(\cP) \, 2^{\frac{1-\alpha}{\alpha} \left(-f(\cP) + \renyiSandUp_{\alpha}(Q|Q')_{\rho_{|\cP}} \right) } \right) ,
	\end{aligned}
    \end{equation}
	\begin{equation}\label{eq:downfweighted entropy inf}
		\begin{aligned}
			&\hspace{-3mm} \frenyiSandDown_\alpha(Q|\CP Q')_{\rho} \\ 
            &\hspace{-5mm} \defvar \frac{1}{1-\alpha} \log \left( \sum_{\cP} \rho(\cP) \, 2^{(1-\alpha) \left(-f(\cP) + \renyiSandDown_{\alpha}(Q|Q')_{\rho_{|\cP}} \right) } \right),
		\end{aligned}
	\end{equation}
	where the sums run over all $\cP$ values such that $\rho(\cP)>0$. We extend both definitions to $\alpha=\infty$ by taking the $\alpha\to\infty$ limit.
\end{definition}

\paragraph{\texorpdfstring{$\lambda$}{lambda}-Expected Contractive Maps}
Following \cite{Colomer_2024}, we use the following definition.
\begin{definition}\label{eq:contractive-maps}
    Let $\omega \in \cS(B)$. Consider the replacer channel $\cM^{\omega}_{A\to B}(\cdot) \coloneq \Tr{\cdot }\omega_{B}$. We say a family of channels $(\cR^{h}_{A \to B})_{h}$ equipped with distribution $q$ is $\lambda$-expected-contractive towards $\omega$ if for all $\rho \in \cB(A \otimes E)$,
    \begin{align}
        \mbb{E}_{h \sim q} \Vert (\cR^{h} - \cM^{\omega})(\rho_{AE}) \Vert_{2} \leq \lambda \Vert \rho_{AE} \Vert_{2} \ .
    \end{align}
\end{definition}
\noindent We use this definition as it is more general. 

\subsection{Variational Form of Schatten Norms}\label{app:inf-dim-variational-form}
Here we prove the variational form of Schatten $p$-norms in infinite dimensions, which we could not find explicitly in the literature, but clearly should be true so long as the operator is appropriately Schatten class.

The following proposition explains when the trace works the same way in infinite dimensions.
\begin{proposition}\label{prop:trace-properties} ~
    \begin{enumerate}
        \item  \cite[Chap. 3, Lemma 8.1]{Gohberg_1978} $\Tr{X}$ is independent of the choice of orthonormal basis. It follows $X \in \cS_{1}(A)$ if and only if $\Tr{\vert X \vert} < \infty$. 
        \item \cite[Chap. 3, Theorem 8.2]{Gohberg_1978} Let $X \in \cS_{\infty}(A,B)$, $Y \in \cB(B,A)$ such that $XY \in \cS_{1}(A)$ and $YX \in \cS_{1}(B)$, then $\Tr{XY} = \Tr{YX}$.
        \item \cite[Chap. 3, Theorem 8.4]{Gohberg_1978} If $X \in \cS_{1}(A)$, $\Tr{X} = \sum_{i=1}^{\nu(X)} \lambda_{i}(X)$ where $\nu(X)$ is the sum of the algebraic multiplicities of all the non-zero eigenvalues $X$.
        \item \cite[Chap. 3, Theorem 8.5]{Gohberg_1978} If $X \in \cS_{1}(A)$, $\vert \Tr{X} \vert \leq \Vert X \Vert_{1}$.
    \end{enumerate}
\end{proposition}
The following states H\"{o}lder's inequality holds in infinite dimensions.
\begin{proposition}\label{prop:Holder's-inequality}
    \cite[Special Case of Eq.~7.5]{Gohberg_1978} Let $0 < p_{i} \leq \infty$, $X_{i} \in \cS_{p_{i}}(A_{i},A_{i+1})$, and $p^{-1} = \sum_{i=1}^{n} p_{i}^{-1} \leq 1$. For $X \coloneq X_{1}X_{2}\cdots X_{n}$,
        \begin{align}
            \Vert X \Vert_{p} \leq \Pi_{i=1}^{n} \Vert X_{i} \Vert_{p_{i}} \ . 
        \end{align}
        In particular, for $1 \leq p,q \leq \infty$ such that $\frac{1}{p} + \frac{1}{q} = 1$ and $X \in \cS_{p}(A,B)$, $Y \in \cS_{q}(B,A)$, then
        \begin{align}
            \Vert XY \Vert_{1} \leq \Vert X \Vert_{p} \Vert Y \Vert_{q} \quad \Vert YX \Vert_{1} \leq \Vert X \Vert_{p} \Vert Y \Vert_{q} \ , 
        \end{align}
        which shows $XY, YX \in \cS_{1}$.
\end{proposition}
Now we prove the variational form.
\begin{proposition}\label{prop:variational-form-in-Hilbert}
    Let $1 < p,q < \infty$ such that $\frac{1}{p} + \frac{1}{q} = 1$. If $X \in \cS_{p}(A,B)$, then 
    \begin{align}\label{eq:variational-form-of-Schatten-norm}
        \Vert X \Vert_{p} = \max\{ \vert \langle Y , X \rangle \vert : Y \in \cB(A,B) \, , \, \Vert Y \Vert_{q} \leq 1 \} \ .
    \end{align}
    Moreover, if $X \in \cB(A)_{sa} \cap \cS_{p}(A)$, then
\begin{equation}
        \begin{aligned}
        \Vert X \Vert_{p} &= \max\{ \vert \langle Y , X \rangle \vert : Y \in \cB(\cH)_{sa} \, , \, \Vert Y \Vert_{q} \leq 1 \} \ . 
    \end{aligned}
\end{equation}
    In both cases one may require the $q$-norm of $Y$ to be an equality.
\end{proposition}
\begin{proof}
    In general,
    \begin{align*}
        \langle Y, X \rangle &= \Tr{Y^{\ast}X} \leq \vert \Tr{Y^{\ast}X} \vert \leq \Vert Y^{\ast}X \Vert_{1}\\
        &\leq \Vert X \Vert_{p} \Vert Y^{\ast} \Vert_{q} = \Vert X \Vert_{p} \Vert Y \Vert_{q} \ , 
    \end{align*}
    where we used previously mentioned properties. Thus, the RHS of \eqref{eq:variational-form-of-Schatten-norm} (with a supremum) can only be upper bounded by $\Vert X \Vert_{p}$. It thus suffices to show we can achieve $\Vert X \Vert_{p}$. If $X = 0$, both sides of \eqref{eq:variational-form-of-Schatten-norm} are zero trivially, so we focus on the case $X \neq 0$. Our construction is the construction Nelson provides for proving H\"{o}lder's inequality for von Neumann algebras \cite[Section 3]{Nelson_1974}. By the polar decomposition, $X = U\vert X \vert$ where $U: \text{Im}(X^{\ast}) \to \text{Im}(X)$. Define $Y = \frac{U\vert X \vert^{p-1}}{\Vert X \Vert^{p/q}_{p}}$. Then
    \begin{align}
        \Vert Y \Vert_{q} = \frac{\Vert \vert X \vert^{p-1} \Vert_{q}}{\Vert X \Vert^{p/q}_{p}} = \frac{\Vert \vert X \vert^{p/q} \Vert_{q}}{\Vert X \Vert^{p/q}_{p}} = 1
    \end{align}
    where the first equality uses $\Vert U\vert X \vert^{p-1} \Vert_{q} = \left\Vert \sqrt{\vert X \vert^{2(p-1)}} \right \Vert_{q} = \Vert \vert X \vert^{p-1} \Vert_{q}$, the second uses $p-1=\frac{p}{q}$, and the last is a direct calculation. Moreover,
    {\allowdisplaybreaks
    \begin{align}
        \langle Y, X \rangle &= \Tr{Y^{\ast}X} \\
        \notag&= \left( \Vert X \Vert^{p}_{p} \right)^{-1/q} \Tr{\vert X \vert^{p-1}U^{\ast}U \vert X \vert} \\
        \notag&= \left( \Vert X \Vert^{p}_{p} \right)^{-1/q} \Tr{\vert X \vert^{p}} \\
        \notag&= \left( \Vert X \Vert^{p}_{p} \right)^{-1/q} \sum_{i=1}^{\infty} \sqrt{\lambda_{i}(X)^{p}} \\
        \notag&= \left( \Vert X \Vert^{p}_{p} \right)^{-1/q} \sum_{i=1}^{\infty} \sigma_{i}(X)^{p} \\
        \notag&=  \left( \Vert X \Vert^{p}_{p} \right)^{-1/q} \Vert X \Vert^{p}_{p} \\
        \notag&= \left( \Vert X \Vert^{p}_{p} \right)^{1-1/q} \\
        \notag&= \Vert X \Vert_{p}
    \end{align}
    }
    where the third equality uses $U^{\ast}U$ is the projector onto $\text{Im}(X)$ and the final equality uses $\frac{1}{p} = 1 - 1/q$. Thus, the RHS of \eqref{eq:variational-form-of-Schatten-norm} achieves the LHS of \eqref{eq:variational-form-of-Schatten-norm}. 

    Finally, when $X$ is self-adjoint, by the spectral theorem $X = \sum_{i=1}^{\infty} \lambda_{i} \dyad{\nu_{i}}$. It follows that the unitary for the polar decomposition of $X$ is 
    \begin{align}
        U = \sum_{i = 1}^{\infty} \text{sgn}(\lambda_{i})\dyad{\nu_{i}} = \text{sgn}(X)
    \end{align} 
    where $\text{sgn}(x) \coloneq \begin{cases}
        1 & x > 0 \\ 0 & x = 0 \\ -1 & x < 0 
    \end{cases}$. As $U$ and $X$ commute, it is straightforward to confirm that in this case our definition of $Y$ is self-adjoint. Thus, we may restrict to self-adjoint operators in this case. This completes the proof.
\end{proof}

Similar to the above, we state the variational form for the Schatten $1$-norm, which requires the supremum. The proof is the same as the above but considering the sequence of rank-one operators $(Y_{\ve} = \dyad{\nu_{\ve}})_{\ve > 0}$ where $\Vert X\ket{\nu_{\epsilon}} \Vert_{\infty} \geq X - \epsilon$ for every epsilon and then letting $\ve$ go to zero.
\begin{proposition}\label{prop:duality-of-Schatten-1-norm}
    If $X \in \cS_{1}(A,B)$, then 
    \begin{align}
        \Vert X \Vert_{1}  = \sup \{ \vert \langle Y , X \rangle \vert : Y \in \cB(A,B) \, , \, \Vert Y \Vert_{\infty} \leq 1 \} \ .
    \end{align}
\end{proposition}

We now prove the Reverse H\"{o}lder's Inequality. The barrier to directly lifting the proof for finite-dimensions, \cite[Lemma 3.3]{Tomamichel_2016}, is that we do not a priori know $\Tr{MN^{-1}}$ is trace class even if this quantity is finite without further assumptions, so we cannot use the results in Proposition \ref{prop:trace-properties} without choosing a heuristic. Instead, we provide a tedious but direct proof.
\begin{proposition}(Extended Reverse H\"{o}lder's Inequality)\label{prop:reverse-Holders}
    Let $A$ be a separable Hilbert space. Let $p \in (0,1)$, $\frac{1}{p}+\frac{1}{q} = 1$. Let $M \in \cS_{p}(A) \cap \cB_{\geq 0}(A)$ and $0 \neq N \in \cS_{-q}(A) \cap \cB_{\geq 0}(A)$. Then,
    \begin{align}\label{eq:exteneded-reverse-holder}
        \lim_{\ve \downarrow 0} \Tr{M^{1/2}(N+ \ve I)^{-1}M^{1/2}} \geq \Vert M \Vert_{p} \Vert N^{-1} \Vert_{-q}^{-1} \ ,
    \end{align}
    where the limit is in $[0,+\infty]$ and is only finite if $\supp(N) \subseteq \supp(M)$.
\end{proposition}
\begin{proof}
    First, we show the limit exists. Note for $\ve_{1} \geq \ve_{2}$, $(N + \ve_{2}I)^{-1} \geq (N + \ve_{1}I)^{-1}$, so $\Tr{M(N+\ve I)^{-1}}$ monotonically increases as $\ve$ decreases. Similarly, for all $\ve > 0$, $(N+\ve I)^{-1} \leq \ve^{-1}I$, so 
    $$\Tr{M(N+\ve I)^{-1}} \leq \Tr{M} \ve^{-1} < +\infty \ , $$
    where the strict inequality is because $0 \leq M$ and trace class. Thus, the limit exists by the monotone convergence theorem.

    Next, assume $\supp(M) \not \subseteq \supp(N)$. Let $\Pi$ be the projector onto the kernel of $N$. Then $\Tr{M \Pi} > 0$ by assumption on the non-inclusion of the supports. As $(N + \ve I)^{-1} = (N + \ve (I - \Pi))^{-1} + \ve^{-1} \Pi$, we have the lower bound 
    $$\Tr{M(N+\ve I)^{-1}} \geq \ve^{-1}\Tr{M\Pi} \ . $$
    As this lower bound goes to infinity as $\ve$ goes to zero, the quantity is finite only if $\supp(M) \subseteq \supp(N)$. 

    We now focus on the case $\supp(M) \subseteq \supp(N)$. If the limit is infinite, \eqref{eq:exteneded-reverse-holder} is trivially true, so we consider the case where the trace is finite. The proof will be to find a simple form for the trace quantity and then use scalar Jensen's inequality and standard H\"{o}lder's inequality to relate it to $\Vert M \Vert_{p}^{p}$. Then re-ordering things will complete the proof. 
    
    First, using the spectral decomposition theorem, $N = \sum_{i} n_{i} \dyad{\nu_{i}}$, 
    \begin{align}
        & \lim_{\ve \downarrow 0} \Tr{M^{1/2}(N+\ve I)^{-1}M^{1/2}} \nonumber \\
        &= \lim_{\ve \downarrow 0} \sum_{i} (n_{i}+\ve)^{-1} \bra{i}M\ket{\nu_{i}} \nonumber \\
        &= \sum_{i} n_{i}^{-1} \bra{\nu_{i}}M\ket{\nu_{i}} \ , \label{eq:simplified-rev-holder-limit}
    \end{align}
    where the second equality uses the monotone convergence of non-negative sums, which applies as $M,N \geq 0$. 
    
    Second, denote the spectral decomposition $M = \sum_{j} \lambda_{j} \dyad{\omega_{j}}$. For each $i$, define conditional probability $(r_{j \vert i} \coloneq \vert \bra{\omega_{j}}\ket{\nu_{i}} \vert^{2})_{j}$, which is a probability distribution by the support inclusion. Then,
    \begin{align}
        \bra{\nu_{i}}M^{p}\ket{\nu_{i}} &= \sum_{j} r_{j \vert i} \lambda_{j}^{p} \nonumber \\ &\leq \left( \sum_{j} r_{j \vert i} \lambda_{j} \right)^{p} = \left(\bra{\nu_{i}}M\ket{\nu_{i}}\right)^{p} \ , 
    \end{align}
    where we used that the power function is concave for $0 < p < 1$ and Jensen's inequality. It follows
    \begin{align}
        \Vert M \Vert^{p}_{p} &= \Tr{M^{p}} \nonumber \\
        &= \sum_{i} \bra{\nu_{i}}M^{p}\ket{\nu_{i}} \nonumber \\
        &\leq \sum_{i} \left(\bra{\nu_{i}}M\ket{\nu_{i}} \right)^{p}  \nonumber \\
        &= \sum_{i} \left(\frac{\bra{\nu_{i}}M\ket{\nu_{i}}}{n_{i}}\right)^{p} n_{i}^{p} \eqqcolon \sum_{i} x_{i}y_{i} \ . \label{eq:Mpp-bound}
    \end{align}
    Now applying standard H\"{o}lder's inequality for $p' \coloneq 1/p$,
    \begin{align}
        \sum_{i} x_{i}y_{i} 
        &\leq \left( \sum_{i} x_{i}^{p'} \right)^{1/p'} \left( \sum_{i} y_{i}^{\frac{p'}{p'-1}} \right)^{\frac{p'-1}{p'}} \nonumber \\
        &= \left(\sum_{i} \frac{\bra{\nu_{i}}M\ket{\nu_{i}}}{n_{i}} \right)^{p} \left(\sum_{i} n_{i}^{\frac{p}{1-p}} \right)^{1-p} \ , 
    \end{align}
    where the equality is simplifying terms. Finally, as $\frac{p}{1-p} = -q$ ,
    \begin{align}
        \left(\sum_{i} n_{i}^{\frac{p}{1-p}} \right)^{1-p} &= {\Tr{N^{-q}}}^{1-p} \nonumber \\
        &= \Vert N^{-1} \Vert_{-q}^{-q(1-p)} = \Vert N^{-1} \Vert_{-q}^{p} \ . 
    \end{align}
    Combining all of the previous bounds, we have 
    \begin{align*}
        \Vert M \Vert_{p}^{p} \leq  \left(\lim_{\ve \downarrow 0} \Tr{M^{1/2}(N+\ve I)^{-1}M^{1/2}}\right)^{p} \Vert N^{-1} \Vert^{p}_{-q} \ .
    \end{align*}
    Taking the $p^{th}$ root and re-ordering completes the proof.
\end{proof}

\subsection{Positive, Trace Non-Increasing Maps are Schatten \texorpdfstring{$1$}{1}-Norm Contractions}\label{app:inf-dim-norm-contracts}
Dupuis shows in finite dimensions that $\Vert \msf{N}(X) \Vert_{1} \leq \Vert X \Vert_{1}$ if $X$ is normal and $\msf{N}$ is a completely positive, trace non-increasing map. This is unecessarily restrictive in the dimension, the class of input operators, and the positivity of the map. First, we recall the Russo-Dye theorem.
\begin{proposition}[Russo-Dye Theorem]
    (See e.g. \cite[Theorem 1.3.3]{Stormer_1974}). Let $\msf{N}:\cB(A) \to \cB(B)$ be a positive, self-adjoint linear map. Then $\Vert \msf{N} \Vert_{\infty \to \infty} = \Vert \msf{N}(I) \Vert_{\infty}$.
\end{proposition}
This implies a positive sub-unital map can only decrease the infinity norm of a bounded operator as we now show.
\begin{proposition}\label{prop:PSU-inf-contraction}
    Let $X \in \cB(A)$ and $\msf{N}$ be a positive, sub-unital map. Then $\Vert \msf{N}(X) \Vert_{\infty} \leq \Vert X \Vert_{\infty}$.
\end{proposition}
\begin{proof}
    First, $\Vert \msf{N} \Vert_{\infty \to \infty} = \Vert \msf{N}(I) \Vert_{\infty} \leq \Vert I \Vert_{\infty} = 1$ where the first equality is the Russo-Dye theorem and the inequality is the assumption the map is positive and sub-unital, so $0 \leq \msf{N}(I) \leq I$. Now let $X \in \cB(A)$. Then $\Vert \msf{N}(X) \Vert_{\infty} \leq \Vert \msf{N} \Vert_{\infty \to \infty}\Vert X \Vert_{\infty} \leq \Vert X \Vert_{\infty}$.
\end{proof}
We now convert this into a monotonicity property for the Schatten $1$-norm. This has been done in finite dimensions in \cite{Watrous_2018}.
\begin{proposition}\label{prop:PTNI-1-norm-contr}
    Let $\cN: \cB(A) \to \cB(B)$ be a positive, trace non-increasing map. Let $X \in \cB(A)$ and $\cN(X) \in \cS_{1}(B)$. Then $\Vert \cN(X) \Vert_{1} \leq \Vert X \Vert_{1}$.
\end{proposition}
\begin{proof}
    Starting and ending using the variational formula (Proposition \ref{prop:duality-of-Schatten-1-norm}),
\begin{equation}
        \begin{aligned}
        \Vert \msf{N}(X) \Vert_{1} &= \sup\{ \vert \langle Y, \msf{N}(X) \rangle \vert : Y \in \cB(B) \, , \, \Vert Y \Vert_{\infty} \leq 1 \} \\ 
        &= \sup\{ \vert \langle \msf{N}^{\ast}(Y), X \rangle \vert :  Y \in \cB(B) \, , \, \Vert Y \Vert_{\infty} \leq 1 \}  \\
        &\leq \sup\{ \vert \langle Y', X \rangle \vert :  Y' \in \cB(A) \, , \, \Vert Y' \Vert_{\infty} \leq 1 \} \\
        &= \Vert X \Vert_{1}
    \end{aligned}
\end{equation}
    where the second equality uses the definition of the adjoint map and the inequality uses the adjoint of a positive, trace non-increasing map is a positive, unital map and so by Proposition \ref{prop:PSU-inf-contraction} we have only increased the feasible set. This completes the proof.
\end{proof}

\subsection{\texorpdfstring{\Renyi}{Renyi} Leftover Hashing against Infinite-Dimensional Side-Information}\label{app:inf-dim-LHL-against-inf-side-info}
We now extend Dupuis's result. The following will make use of the following complex interpolation result.
\begin{proposition}[Hadamard's Three-Line Theorem] \label{prop:Hadamard-3-line} Let $S = \{z \in \mbb{C} \vert 0 \leq \operatorname{Re}[z] \leq 1\}$. Let $f:S \to \mbb{C}$ be a bounded function such that it is holomorphic on the interior of $S$ and continuous on the border. Let $1 \leq p_{0} \leq p_{1}$ and $0 < \theta < 1$ and define $p_{0} \leq p_{\theta} \leq p_{1}$ via 
\begin{align}
    \frac{1}{p_{\theta}} = \frac{1-\theta}{p_{0}} + \frac{\theta}{p_{1}} \ . 
\end{align}
Furthermore, for $k \in (0,1)$, let $M_{k} \coloneq \sup_{t \in \mbb{R}} \vert f(k+it) \vert$. Then for any $0 \leq \theta \leq 1$,
\begin{align}
    \vert f(\theta) \vert \leq M_{0}^{1-\theta}M_{1}^{\theta} \ . 
\end{align}
\end{proposition}

We first extend \cite[Lemma 8]{Dupuis_2022} to infinite dimensions. The changes are as follows. First, we require the maps $\msf{N}^{h}$ to be self-adjoint, which is required in \cite{Dupuis_2022}, but omitted. Second, we mildly modify the function $f$ to which Proposition \ref{prop:Hadamard-3-line} is applied in \cite[Lemma 8]{Dupuis_2022}. This is because one cannot apply the identity $\Vert Y^{a + it} \Vert_{\alpha} = \Vert Y^{a} \Vert_{\alpha}$ in general if $Y$ has negative eigenvalues, which is used both in the finite-dimensional proof and our proof, so it is a fix to both papers. Lastly, we impose assumptions on the norms of the maps $(\msf{N}^{h})_{h \in \msf{H}}$, which is required to guarantee the function $f$ is bounded. This last point is unique to infinite-dimensions as boundedness is for free in finite-dimension.
\begin{proposition}(Extension of \cite[Lemma 7]{Dupuis_2022}) \label{prop:extended-technical-lemma} Let $h$ take values in finite set $\msf{H}$ with non-zero probability for each value. For each $h \in \msf{H}$, let $\msf{N}^{h}:\cB(A \otimes E) \to \cB(B \otimes E)$ be a self-adjoint map such that
\begin{equation}\label{eq:map-norm-boundedness}
\max_{h \in \msf{H}} \sup_{p \in [1,2]} \Vert \msf{N}^{h} \Vert_{p \to p} \eqqcolon C < +\infty \ . 
\end{equation}
Let $\alpha \in (1,2]$. For any $\rho_{AE} \in \cS(A \otimes B)$ and $\sigma_{E} \in \cS(\cH_{E})$ such that $\Vert \sigma_{E}^{\frac{1-\alpha}{2\alpha}}\rho_{AE}\sigma_{E}^{\frac{1-\alpha}{2\alpha}} \Vert_{\alpha} = 1$. Then,
\begin{equation}
\begin{aligned}
    &\mbb{E}_{h} \left \Vert \msf{N}^{h}\left( \sigma_{E}^{\frac{1-\alpha}{2\alpha}}\rho_{AE}\sigma_{E}^{\frac{1-\alpha}{2\alpha}} \right) \right \Vert_{\alpha} \\
    &\leq \left ( \sup_{\eta \in \cN(A \otimes E): \Vert \eta \Vert_{1} = 1} \mbb{E}_{h} \Vert \msf{N}^{h}(\eta) \Vert_{1}  \ \right)^{\frac{2}{\alpha}-1} \\ 
    &\hspace{5mm} \cdot \left ( \sup_{\eta \in \cN(A \otimes E): \Vert \eta \Vert_{2} = 1} \mbb{E}_{h} \Vert \msf{N}^{h}(\eta) \Vert_{2}  \ \right)^{2\left(\frac{\alpha -1}{\alpha} \right)} \ . 
\end{aligned}
\end{equation}
\end{proposition}
\begin{proof}
        For notational simplicity, let $X \coloneq \sigma_{E}^{\frac{1-\alpha}{2\alpha}}\rho_{AE}\sigma_{E}^{\frac{1-\alpha}{2\alpha}}$ and note that it is always self-adjoint. By assumption, for each $h$, $\msf{N}^{h}(X)$ is self-adjoint. Thus, by Proposition \ref{prop:variational-form-in-Hilbert} and its proof, the self adjoint operator 
        \begin{align}
            Y_{h} \coloneq \frac{U_{h} \vert \msf{N}^{h}(X) \vert^{\alpha -1}}{\Vert \msf{N}^{h}(X) \Vert_{\alpha}^{\alpha-1}}
        \end{align}
        where the unitary $U_{h} \coloneq \text{sgn}\left( \msf{N}^{h}(X) \right)$ is such that $\Vert Y_{h} \Vert_{\frac{\alpha}{\alpha-1}} = 1$ and 
        \begin{align}
            &\left \Vert \msf{N}^{h}\left(X \right) \right \Vert_{\alpha}  = \Tr{Y_{h}\msf{N}^{h}\left( X \right)} \ .
        \end{align}
        We now define $f:S \to \mbb{C}$ as
        \begin{equation}\label{eq:interpolation-function}
        \begin{aligned}
            f(z) 
            &\coloneq \mbb{E}_{h} \Tr{U_{h} \vert Y_{h}\vert^{z\left(\frac{\alpha}{2(\alpha-1)}\right)}\msf{N}^{h}\left(  X^{(1-z)\alpha+z(\alpha/2)} \right)} \ . 
        \end{aligned}
        \end{equation}
        As in the finite dimensional case, we wish to use this as the function in Proposition \ref{prop:Hadamard-3-line}. This function satisfies the conditions because our assumption on $\msf{N}^{h}$ implies it is bounded as we now show. To this end, define $q \coloneq \frac{\alpha}{\alpha-1}$ and note $(1-z)\alpha + z(\alpha/2) = \alpha(1-z/2)$. So if we consider the argument of the expectation in \eqref{eq:interpolation-function} re-expressed in these terms,
        {\allowdisplaybreaks
        \begin{align}
        & \left \vert \Tr{U_{h} \vert Y_{h} \vert^{qz/2}\msf{N}^{h}\left( X^{\alpha(1-z/2)} \right)} \right \vert \label{eq:bounded-trace} \\ 
        &\leq \left \Vert \vert Y_{h} \vert ^{qz/2}\msf{N}^{h}\left( X^{\alpha(1-z/2)} \right) \right \Vert_{1} \\
        &\leq \left \Vert \vert Y_{h} \vert ^{qz/2} \right\Vert_{p} \left \Vert \msf{N}^{h}\left( X^{\alpha(1-z/2)} \right) \right \Vert_{\frac{p}{p-1}} \\
        &\leq \left \Vert \vert Y_{h} \vert ^{qz/2} \right\Vert_{p} \left \Vert X^{\alpha(1-z/2)} \right \Vert_{\frac{p}{p-1}} \Vert \msf{N}^{h} \Vert_{\frac{p}{p-1} \to \frac{p}{p-1}}
        \end{align} 
        where the first inequality is Proposition \ref{prop:trace-properties}, the second is Proposition \ref{prop:Holder's-inequality}, and the last is by definition of the $\cS_{p} \to \cS_{p}$ norm. 
        }
        Moreover, using that for positive semidefinite operator $Q$, $\Vert Q^{a+ib} \Vert_{p} = \Vert X^{a} \Vert_{p}$ for all Schatten $p$-norm, 
        \begin{align}
            &\left \vert \Tr{U_{h} \vert Y_{h} \vert^{qz/2}\msf{N}^{h}\left( X^{\alpha(1-z/2)} \right)} \right \vert \label{eq:bounded-trace} \\ 
            &\leq \left \Vert \vert Y_{h} \vert ^{qx/2} \right\Vert_{p} \left \Vert X^{\alpha(1-x/2)} \right \Vert_{\frac{p}{p-1}} \Vert \msf{N}^{h} \Vert_{\frac{p}{p-1} \to \frac{p}{p-1}} \ ,
        \end{align}
        where $\text{Re}(z) \eqqcolon x \in [0,1]$. 
        
        Now we define $k_{1} \coloneq qx/2$, $k_{2} \coloneq \alpha(1-x/2)$, and choosing $p \coloneq 2/x$. Arithmetic will verify $k_{1}p = \frac{\alpha}{\alpha-1}$ and $k_{2}\frac{p}{p-1} = \alpha$. Thus, using for Hermitian operator $T$ $\Vert T^{k} \Vert_{\alpha} = \Vert T \Vert_{k\alpha}^{k}$,  
        {\allowdisplaybreaks
        \begin{align}
            &\Vert \vert Y_{h} \vert^{qx/2} \Vert_{p} = \Vert Y_{h} \Vert^{k_{1}}_{k_{1}p} = \Vert Y_{h} \Vert^{k_{1}}_{\frac{\alpha}{\alpha-1}} = 1 \\ 
            & \Vert X^{\alpha(1-x/2)} \Vert_{\frac{p}{p-1}} = \Vert X \Vert^{k_{2}}_{k^{2}\frac{p}{p-1}} = \Vert X \Vert^{k_{2}}_{\alpha} = 1 \ , 
        \end{align}
        }
        where the final equalities are $\Vert Y_{h} \Vert_{\frac{\alpha}{\alpha-1}} = 1$ by construction and $\Vert \sigma_{E}^{\frac{1-\alpha}{2\alpha}}\rho_{AE}\sigma_{E}^{\frac{1-\alpha}{2\alpha}} \Vert_{\alpha} = 1$ by assumption. In summary, for $z \in S$,
        \begin{align}
            &\left \vert \Tr{U_{h} \vert Y_{h} \vert^{qz/2}\msf{N}^{h}\left( X^{\alpha(1-z/2)} \right)} \right \vert \\
            &\leq \Vert \msf{N}^{h} \Vert_{\frac{p}{p-1} \to \frac{p}{p-1}}
            = \Vert \msf{N}^{h} \Vert_{\frac{1}{1-x/2} \to \frac{1}{1-x/2}} \ .
        \end{align}
        Noting that $\frac{1}{1-x/2} \in [1,2]$, our assumption in \eqref{eq:map-norm-boundedness} guarantees boundedness of the function on the strip.  The rest of the proof is a straightforward modification of the proof of \cite[Lemma 8]{Dupuis_2022}, but we include it for clarity.
        
        Now, in the notation of Proposition \ref{prop:Hadamard-3-line}, we bound $M_{0}$ and $M_{1}$ for our function. To bound $M_{0}$, observe
        \begin{align}
            \vert f(it) \vert &\leq \mbb{E}_{h} \left\vert \Tr{U_{h} \vert Y_{h} \vert^{q\frac{it}{2}}\msf{N}^{h}\left( X^{\alpha(1-\frac{it}{2})} \right)} \right\vert \\
            &\leq \mbb{E}_{h} \left \Vert U_{h} \vert Y_{h} \vert^{q\frac{it}{2}} \right \Vert_{\infty} \left \Vert \msf{N}^{h}\left( X^{\alpha(1-\frac{it}{2})} \right) \right\Vert_{1} \\
            &= \mbb{E}_{h} \left \Vert \msf{N}^{h}\left( X^{\alpha(1-\frac{it}{2})} \right) \right\Vert_{1} \\ 
            &\leq \sup_{\eta \in \cN(A \otimes E): \Vert \eta \Vert_{1} = 1} \mbb{E}_{h} \Vert \msf{N}^{h}(\eta) \Vert_{1} \ , 
        \end{align}
        where the first inequality is triangle inequality, the second is H\"{o}lder's inequality (Proposition \ref{prop:Holder's-inequality}), the equality is that $U_{h} \vert Y_{h} \vert^{q\frac{it}{2}}$ is a unitary, and the final inequality is because $X^{\alpha(1-\frac{it}{2})}$ is a normal operator satisfying $\Vert X^{\alpha(1-\frac{it}{2})} \Vert_{1} = 1$ as direct calculations will verify, so it is feasible for the final supremization.

        Similarly, to bound $M_{1}$, observe 
        \begin{align}
            \Vert U_{h} \vert Y_{h} \vert^{q\frac{1+it}{2}} \Vert_{2}^{2} = \Vert \vert Y_{h} \vert^{q/2} \Vert_{2}^{2} = \Vert Y_{h} \Vert_{q} = 1 \ , 
        \end{align}
        where we used unitary invariance of the norm and the same powers of Hermitian operators of Schatten norms property again. Also observe
        \begin{align}
            \Vert X^{\alpha(1-\frac{1+it}{2}} \Vert_{2}^{2} = \Vert X^{\alpha/2(1+it)} \Vert_{2}^{2} = \Tr{X^{\alpha}} = \Vert X \Vert^{\alpha}_{\alpha} = 1 \ ,
        \end{align}
        where we used $X$ is positive semidefinite and our assumption on its $\alpha$ Schatten norm. Thus,
        \begin{align}
            \vert f(1+it) \vert & \leq \mbb{E}_{h} \left\vert \Tr{U_{h} \vert Y_{h} \vert^{q\frac{1+it}{2}}\msf{N}^{h}\left( X^{\alpha(1-\frac{1+it}{2})} \right)} \right\vert \notag \\
            &\leq \mbb{E}_{h} \Vert U_{h} \vert Y_{h} \vert^{q\frac{1+it}{2}} \Vert_{2}  \Vert \msf{N}^{h}\left(X^{\alpha(1-\frac{1+it}{2})}\right) \Vert_{2} \notag \\
            &\leq \sup_{\eta \in \cN(A \otimes E): \Vert \eta \Vert_{2} = 1} \mbb{E}_{h} \Vert \msf{N}^{h}(\eta) \Vert_{2} \ . 
        \end{align}
        Using Proposition \ref{prop:Hadamard-3-line} with these bounds completes the proof.
\end{proof}

\begin{proposition}(Extension of \cite[Theorem 9]{Dupuis_2022})\label{prop:inf-LHL}
    Let $A,B$, and $E$ be separable Hilbert spaces where $B$ is finite. Let $(\cR^{h}_{A \to B}:\cT(A) \to \cT(B))_{h \in \msf{H}}$ equipped with distribution $q_{h}$ be $\lambda$-expected contractive towards $\omega_{B} \in \cS(B)$ (Definition \ref{eq:contractive-maps}) and such that $(\cN^{h} \coloneq (\cR^{h} - \cM^{\omega}) \otimes \id_{E} )_{h \in \msf{H}}$ satisfy the condition \eqref{eq:map-norm-boundedness}. Then for quantum states $\rho_{AE} \in \cS(A \otimes E)$ and $\sigma_{E}$,
    \begin{equation}
    \begin{aligned}
        & \mbb{E}_{h} \left \Vert (\cR^{h} - \cM^{\omega})(\rho_{AE}) \right \Vert_{1} \\
        & \leq 2^{\frac{2}{\alpha}-1} \cdot 2^{\frac{\alpha-1}{\alpha}(\log\vert B \vert - \renyiSandUp_{\alpha}(A \vert E)_{\rho \vert \sigma}+2\log(\lambda))} \ . 
    \end{aligned}
    \end{equation}
    In particular, if $\rho_{AE}$ is CQ and $\{\cR^{h}_{A \to B} : h \in \msf{H}\}$ is a $\lambda$-randomizing family of hash functions, this bound holds.
\end{proposition}
\begin{remark}
    We don't specify support conditions because when they are not met the bound holds trivially.
\end{remark}
\begin{proof}
    The proof is identical to the one given in \cite{Dupuis_2022} except we appeal to our Proposition \ref{prop:PTNI-1-norm-contr} instead of \cite[Lemma 11]{Dupuis_2022}and our Proposition \ref{prop:extended-technical-lemma} instead of \cite[Lemma 8]{Dupuis_2022}, which is why we require that $(\cR^{h} - \cM^{\omega}) \otimes \id_{E}$ have the appropriate bounded $p \to p$-norms. 

    Finally, the reason the bound always holds for a $\lambda$-randomizing family of hash functions is because in this setting $\cR^{h} - \cM^{\omega}$ is a generalized classical-to-quantum map in the sense of the following proposition because $\omega = I_{B}/\vert B \vert$. Thus, by the following proposition $\Vert \cR^{h} - \cM^{\omega}) \otimes \id_{E} \Vert_{p \to p} < +\infty$ for all $1 \leq p \leq +\infty$, which covers the interval needed. 
\end{proof}

\begin{proposition}\label{prop:gen-clas-to-q-map-bounded-p-to-p-norm}
    Let $A$ and $B$ be finite-dimensional Hilbert spaces and $E$ be a separable Hilbert space. Consider a `generalized' classical-to-quantum map, $\Phi_{A \to B}[Y] = \sum_{i} \bra{i}Y\ket{i}Z_{i}$ where $(Z_{i})_{i}$ are Hermitian operators. Then $\Vert \Phi_{A \to B} \otimes \id_{E} \Vert_{p \to p} < +\infty$ for every $p \in [1,\infty]$.
\end{proposition}
\begin{proof}
    Note $\Phi_{A \to B} = \Phi_{A \to B} \circ \mathcal{P}_{A}$ .where $\mathcal{P}_{A}$ is the pinching channel in the computational basis. Then we have
    \begin{equation}
    \begin{aligned}
        &\Vert \Phi \otimes \id_{E} \Vert_{p \to p} \\
        &= \sup_{\sum_{i} \dyad{i} \otimes X_{i} : \sum_{i} \Vert X_{i} \Vert_{p} = 1} \left \Vert \sum_{i} \Phi(\dyad{i}) \otimes X_{i} \right \Vert_{p} \\
        &\leq \sup_{\sum_{i} \dyad{i} \otimes X_{i} : \sum_{i} \Vert X_{i} \Vert_{p} = 1}  \sum_{i} \left \Vert  \Phi(\dyad{i}) \right \Vert_{p} \left \Vert X_{i} \right \Vert_{p} \\
        &\leq \sum_{i} \Vert \Phi(\dyad{i}) \Vert_{p} \ \\
        &< +\infty
    \end{aligned}
    \end{equation}
    where the equality uses the pinching inequality for unitarily invariant norms, c.f.~\cite{Tomamichel_2016}, and \eqref{eq:norm-between-schatten-spaces}, the first inequality is triangle inequality and multiplicativity over tensor products of the Schatten norms, the second inequality is using the argument of the supremum, and strict inequality is because a finite-dimensional map always has bounded $p \to p$ norm.
\end{proof}

\subsection{Duality between Petz and Sandwiched \Renyi Entropies}\label{app:inf-dim-duality}

Here we establish an extension of the duality between Petz and Sandwiched \Renyi entropy when the system $A$ that is the first argument of the conditional entropy remains finite.
\begin{proposition}\label{prop:Petz-Sandwiched-duality}
    For any pure state $\ket{\rho}_{ABC}$ where $\vert A \vert < +\infty$, 
    \begin{align}
        \overline{H}^{\uparrow}_{\alpha}(A \vert B)_{\rho} + \wt{H}^{\downarrow}_{\beta}(A \vert C)_{\rho} = 0 
    \end{align}
    when $\alpha, \beta \in [0,\infty] $ satisfying $\alpha \beta = 1$.
\end{proposition}
To prove this duality, we will lift the finite-dimensional proof (\cite[Proposition 5.8]{Tomamichel_2016},\cite{tomamichel_2014a}). To this end, we need an extension of the non-variational form of the optimized Petz \Renyi divergence. 

\begin{proposition}\label{prop:non-var-form-of-petz-up}
    Let $\alpha \in (0,1) \cup (1,\infty)$ and $\rho_{AB} \in S_{=}(AB)$ where $\vert A \vert < +\infty$. Then, 
    \begin{align}
        \overline{H}^{\uparrow}_{\alpha}(A \vert B)_{\rho} = \frac{\alpha}{1-\alpha}\log \Tr{\left( \operatorname{Tr}_{A}\left[\rho_{AB}^{\alpha}\right] \right)^{\frac{1}{\alpha}}} \ .
    \end{align}
\end{proposition}
\begin{proof}
    The proof of the finite-dimensional result (\cite[Lemma 5.3]{Tomamichel_2016}) relies on two technical steps. The first step is to introduce a candidate optimizer for $\overline{H}^{\uparrow}_{\alpha}(A \vert B)_{\rho} = \sup_{\sigma_{B} \in S_{=}(B)} - \overline{D}_{\alpha}(\rho_{AB} \Vert I_{A} \otimes \sigma_{B})$, 
    which is 
    \begin{align}
        \sigma_{B}^{\star} \coloneq \frac{\left( \operatorname{Tr}_{A}\left[\rho_{AB}^{\alpha}\right] \right)^{\frac{1}{\alpha}}}{\Tr{\left( \operatorname{Tr}_{A}\left[\rho_{AB}^{\alpha}\right] \right)^{\frac{1}{\alpha}}}} \ . 
    \end{align}
    This step is fine so long as $\sigma^{\star}_{B}$ is well-defined. The second step is to prove this candidate is optimal using the H\"{o}lder's or reverse H\"{o}lder's inequality depending on if $\alpha < 1$ or $\alpha > 1$ respectively. As we have already shown H\"{o}lder's and reverse H\"{o}lder's inequality hold in our setting (Props.~\ref{prop:Holder's-inequality}~and~\ref{prop:reverse-Holders}), we are fine as long as we show $\left(\operatorname{Tr}_{A}[\rho_{AB}^{\alpha}]\right)^{\frac{1}{\alpha}}$ is trace class so that $\sigma^{\star}$ is well-defined. 
    
    To that end, first define $q \coloneq 1/\alpha$. It follows
        \begin{align}
            \Vert \rho^{\alpha}_{AB} \Vert_{q}^{q} = \Tr{\left(\rho^{\alpha}_{AB}\right)^{q}} = \Tr{\rho} = 1 \ ,
        \end{align}
    so $\rho^{\alpha}_{AB} \in \cS_{q}(AB)$. The idea is now to decompose the state into a finite sum of terms on the $B$ space and then bound the norm of those terms individually. Recall we may express the partial trace as (see e.g.~\cite{Wilde_2016}),
        \begin{align*}
            \operatorname{Tr}_{A}\left[\rho^{\alpha}_{AB}\right] = \sum_{i \in \{1,...,\vert A \vert\}} \bra{i} \otimes I_{B} \rho^{\alpha}_{AB} \ket{i} \otimes I_{B} \eqqcolon \sum_{i} X_{i} \ , 
        \end{align*}
        where $X_{i} \coloneq \bra{i} \otimes I_{B} \rho^{\alpha}_{AB} \ket{i} \otimes I_{B}$. Defining isometries $V_{i}:B \to A \otimes B$ by $V_{i}\ket{\psi} = \ket{i}_{A}\ket{\psi}_{B}$, we have $X_{i} = V_{i}^{\ast}\rho_{AB}^{\alpha}V_{i}$.

        Now, as $\rho^{\alpha}$ is compact, isometries are bounded, and Schatten norms can be expressed in terms of singular values, it follows from singular value bounds, see e.g.~\cite[Theorem 1.6]{Simon2005-wg}, that
        \begin{align*}
            \Vert X_{i} \Vert_{q} = \Vert V^{\ast}_{i}\rho^{\alpha}V_{i} \Vert_{q} \leq \Vert V^{\ast}_{i} \Vert_{\infty} \Vert T \Vert_{q} \Vert V_{i} \Vert_{\infty} = \Vert X \Vert_{q} \ , 
        \end{align*}
        where the final equality uses the infinity norm of an isometry is one. Using that the power function is monotone increasing, we conclude $\Vert X_{i} \Vert_{q}^{q} \leq \Vert X \Vert_{q}^{q} = 1$ for all $i$. Using this, when $\alpha > 1$ so that $q < 1$, we apply \cite[Theorem 2.8]{mccarthy_1967} to conclude $\Vert \operatorname{Tr}_{A}\left[\rho^{\alpha}_{AB}\right] \Vert_{q}^{q} \leq \sum_{i} \Vert X_{i} \Vert_{q}^{q} \leq \vert A \vert$,
        where the final inequality uses the bound we just established. When $\alpha < 1$, by the triangle inequality, monotonicity of the power function, and our bound on $\Vert X \Vert_{q}^{q}$, we have $\Vert \operatorname{Tr}_{A}[\rho_{AB}^{\alpha}] \Vert_{q}^{q}\leq \vert A \vert^{q}$. Therefore, for $\alpha > 0$, $\operatorname{Tr}_{A}\left[\rho^{\alpha}_{AB}\right] \in \cS_{q}(B)$. Finally,
        \begin{align}
            \Tr{\left( \operatorname{Tr}_{A}\left[\rho_{AB}^{\alpha}\right] \right)^{\frac{1}{\alpha}}} &= \Tr{\left( \operatorname{Tr}_{A}\left[\rho_{AB}^{\alpha}\right] \right)^{q}} \\
            \notag&= \Vert \operatorname{Tr}_{A}\left[\rho_{AB}^{\alpha}\right] \Vert_{q}^{q} < +\infty \ .
        \end{align}
        Thus, $\left( \operatorname{Tr}_{A}\left[\rho_{AB}^{\alpha}\right] \right)^{\frac{1}{\alpha}}$ is trace class as its trace is finite and it is positive semidefinite, so $\sigma^{\star}_{B}$ is well-defined.        
\end{proof}

\begin{remark}
    It is easy to see the proof breaks if we no longer had a finite sum on the $A$ system. More concerningly, if $A$ is not finite and $\alpha > 1$, then $\sigma^{\star}_{B}$ need not be finite as one can choose $\rho_{AB}$ pure with Schmidt coefficients defined by a probability distribution $(p_{n})_{n \in \mbb{N}}$ such that $\sum_{n} p_{n} = 1$ but $\sum_{n} p_{n}^{1/\alpha} = +\infty$. Similar constructions break the proof for $\alpha < 1$.
\end{remark}

\begin{proof}[Proof of Prop.~\ref{prop:Petz-Sandwiched-duality}]
    The proof is the same as \cite[Proposition 5.8]{Tomamichel_2016} as one can appeal to Proposition \ref{prop:non-var-form-of-petz-up} and the Schmidt decomposition argument still works in infinite dimensions.
\end{proof}

\subsection{On Data Processing and Strong Sub-Additivity}\label{app:inf-dim-DPI}

We use the following fact, often referred to as strong subadditivity for (optimized) Sandwiched \Renyi entropy.
\begin{proposition}\label{prop:Inf-strong-subadd}
    For separable Hilbert spaces $A$, $B$, and $C$ where $\vert A \vert < +\infty$ and $\alpha \in [1/2,1) \cup (1,\infty]$,
    \begin{align}
        \wt{H}^{\uparrow}_{\alpha}(A \vert BC) \leq \wt{H}^{\uparrow}_{\alpha}(A \vert B) \ .  
    \end{align}
\end{proposition}
This fact follows directly from the data processing inequality of the sandwiched \Renyi divergence as then for a density matrix $\sigma_{B}$ and a corresponding extension $\sigma_{BC}$,
\begin{align}
    -\wt{D}_{\alpha}(\rho_{AB} \Vert I_{A} \otimes \sigma_{B}) &= -\wt{D}_{\alpha}(\operatorname{Tr}_{C}[\rho_{ABC}] \Vert I_{A} \otimes \operatorname{Tr}_{C}[\sigma_{BC}]) \\
    \notag&\geq -\wt{D}_{\alpha}(\rho_{ABC} \Vert I_{A} \otimes \sigma_{BC}) \ ,
\end{align}
so one can use the definition of supremizing over the state to conclude the proof. The fact thus may be proven using the data processing inequality of sandwiched \Renyi divergence for von Neumann algebras \cite[Theorem 14]{Berta_2018}. Alternatively, for $\alpha \in (1,\infty]$, \cite[Proposition 3.40]{Mosonyi_2023} establishes the limiting behaviour of finite approximations of sandwiched \Renyi divergences when the arguments are trace-class, positive semidefinite operators. Thus, by letting $\vert A \vert < +\infty$ so that $I_{A} \otimes \sigma_{BC}$ is always trace class, one can extend the data processing for finite dimensions (see e.g.~\cite{Tomamichel_2016}) to infinite-dimensions using the same approximation argument as extending DPI of the quantum relative entropy from finite-dimensions to infinite-dimensions \cite[Step 2 of Proof of Theorem 1]{muller2017monotonicity}. 

\subsection{Chain Rule for Classical Side-Information}\label{app:inf-dim-classical-side-info}
To make things simpler to verify for a reader, we prove the chain rule for classical side-information for optimized Sandwiched \Renyi divergences with infinite-dimensional side-information without appealing to duality of the sandwiched entropy as is done in \cite{Tomamichel_2016}. Instead, we do a trace inequality argument.

\begin{proposition}\label{prop:cl-chain-rule}
    Let $A$, $B$, and $X$ be separable Hilbert spaces where $\vert X \vert$ is finite. Then, for quantum-classical state $\rho_{ABX} = \sum_{x} p_{X}(x) \rho_{AB}^{x} \otimes \dyad{x}_{X}$,
    \begin{align}
        \wt{H}^{\uparrow}_{\alpha}(A \vert BX)_{\rho} \geq \wt{H}^{\uparrow}_{\alpha}(A \vert B)_{\rho} -\log \vert X \vert \ .
    \end{align}
\end{proposition}
\begin{proof}
    We define $\pi_{X} \coloneq \frac{1}{X}I_{X}$, which is well defined as $X$ is a finite-dimensional classical register. The basic proof idea is to get a bound by replacing a feasible point of the optimization for $\wt{H}^{\uparrow}_{\alpha}(A \vert BX)_{\rho}$ with a new feasible point $\tau_{B} \otimes \pi_{X}$ that only picks up a $\log|X|$ correction term. By optimizing then completes the proof.

    By re-writing the norm representation of the $\wt{Q}$ function in \eqref{eq:inf-dim-sandwiched-Q-function} as a trace and merging the negative in the definition of optimized we have
    \begin{align}
        &\exp((1-\alpha)\wt{H}^{\uparrow}_{\alpha}(A \vert BX)) \\
        \notag&= \operatorname{opt}_{\sigma_{XB} \in S_{=}(XB)} \wt{Q}(\rho_{AB} \Vert I_{A} \otimes \sigma_{BX}) \\ 
        \notag&= \operatorname{opt}_{\sigma_{XB} \in S_{=}(XB)}\Tr{\left(\sigma_{BX}^{\frac{1-\alpha}{2\alpha}}\rho_{ABX}\sigma_{BB}^{\frac{1-\alpha}{2\alpha}} \right)^{\alpha}} \ , 
    \end{align}
    where $\operatorname{opt}$ is $\sup$ for $\alpha < 1$ and $\inf$ for $\alpha > 1$, because we have multiplied the entropy by $(1-\alpha)$ whose negativity depends on the value of $\alpha$.
    
    \begin{align}
        &\wt{Q}(\rho_{ABX} \Vert I_{A} \otimes \tau_{B} \otimes \pi_{X}) \\
        \notag&=\Tr{\left( \left(\pi_{X} \otimes \tau_{B}\right)^{\frac{1-\alpha}{2\alpha}}\rho_{ABX} \left(\pi_{X} \otimes \tau_{B}\right)^{\frac{1-\alpha}{2\alpha}} \right)^{\alpha}} \\ 
        \notag&= \sum_{x} \left(\frac{1}{\vert X \vert}\right)^{1-\alpha}p_{X}(x)^{\alpha} \\
        \notag&\hspace{2cm} \cdot \Tr{\left( \left(\tau_{B}\right)^{\frac{1-\alpha}{2\alpha}}\rho_{AB}^{x} \left(\tau_{B}\right)^{\frac{1-\alpha}{2\alpha}} \right)^{\alpha}} \\
        \notag&= \vert X \vert^{\alpha-1} \sum_{x} \Tr{\left(p_{X}(x) \left(\tau_{B}\right)^{\frac{1-\alpha}{2\alpha}}\rho_{AB}^{x} \left(\tau_{B}\right)^{\frac{1-\alpha}{2\alpha}} \right)^{\alpha}} \ . \label{eq:cl-chain-rule-step-1}
    \end{align}
    For notational simplicity, define the positive-semidefinite operators $Y_{x} \coloneq p_{X}(x) \left(\tau_{B}\right)^{\frac{1-\alpha}{2\alpha}}\rho_{AB}^{x} \left(\tau_{B}\right)^{\frac{1-\alpha}{2\alpha}}$. Note that by linearity $\Tr{\left(\sum_{x} Y_{x}\right)^{\alpha}} = \Tr{\left(\left(\tau_{B}\right)^{\frac{1-\alpha}{2\alpha}}\rho_{AB}\left(\tau_{B}\right)^{\frac{1-\alpha}{2\alpha}}\right)^{\alpha}} = \wt{Q}(\rho_{AB} \Vert I_{A} \otimes \tau_{B})$. Observe that when $\alpha < 1$, $Y_{x}$ is always bounded and for $\alpha > 1$, $Y_{x}$ is always bounded for $\tau_{B}$ such that $\wt{Q}(\rho_{AB} \Vert I_{A} \otimes \tau_{B}) < +\infty$ as follows from \cite[Lemma 3.1]{Mosonyi_2023}.
    
    For $Y_{x}$ that are bounded, positive semidefinite operators and the function $f_{\alpha}(x)\coloneq x^{\alpha}$ satisfies $f_{\alpha}(0) = 0$ for all $\alpha \in (0,\infty)$ and is concave for $0 < \alpha < 1$ and convex for $1 < \alpha < \infty$, we have by \cite[Theorem 2]{Rotfeld_1969}, 
    \begin{align}
        \sum_{x} \Tr{Y_{x}^{\alpha}} \geq \Tr{\left(\sum_{x}Y_{x}\right)^{\alpha}} \quad \alpha \in (0,1) 
    \end{align}
    and the inequality swapped for $\alpha > 1$. Thus, for $\alpha \in (0,1)$
    \begin{equation}
    \begin{aligned}
        &\wt{Q}(\rho_{ABX} \Vert I_{A} \otimes \tau_{B} \otimes \pi_{X}) \\ 
        & \hspace{1cm} \geq \vert X \vert^{\alpha -1}  \wt{Q}(\rho_{AB} \Vert I_{A} \otimes \tau_{B}) \ , 
    \end{aligned}
    \end{equation}
    and the inequality switched for $\alpha > 1$. 
    
    For $\alpha < 1$, by supremizing over the choice of $\tau_{B}$, we obtain
    \begin{align}
        & \exp((1-\alpha)H^{\uparrow}_{\alpha}(A \vert BX)) \\
        \notag& \hspace{5mm} \geq \vert X \vert^{\alpha-1} \exp((1-\alpha)H^{\uparrow}_{\alpha}(A \vert B)_{\rho}) \ . 
    \end{align}
    Taking the logarithm and dividing by $1-\alpha > 0$, completes the proof. For $\alpha > 1$, by infimizing over $\tau_{B}$, we obtain
    \begin{align}
        & \exp((1-\alpha)H^{\uparrow}_{\alpha}(A \vert BX)) \\
        \notag& \hspace{5mm} \leq \vert X \vert^{\alpha-1} \exp((1-\alpha)H^{\uparrow}_{\alpha}(A \vert B)_{\rho}) \ . 
    \end{align}
    We remark that while $Y_{x}$ were only bounded if $\wt{Q}_{\alpha}(\rho_{AB} \Vert I_{A} \otimes \tau_{B}) < +\infty$, if the infimum were not finite, the bound holds trivially true, so the bound holds true for all states $\tau_{B}$. Taking the logarithm and dividing by $1-\alpha < 0$, which we note flips the inequality back appropriately, completes the proof.
\end{proof}

%--------------------------------------------------------
%         Full Proof of the iMEAT
%--------------------------------------------------------

\section{Full Proof of the Infinite-Dimensional Marginal-Constrained Entropy Accumulation Theorem}\label{App:Inf_MEAT_Proof}
    The marginal-constrained entropy accumulation theorem from Ref.~\cite{Arqand_2025} ultimately proves a lower bound on the $n$-round \Renyi entropy of the form
    \begin{equation}
        \renyiSandUp_\alpha(S_1^n | \CP_1^n E_n)_{\rho_{|\Omega}} \geq  n  h^\uparrow_{\alpha}
    - \frac{\alpha}{\alpha-1} \log\frac{1}{p_\Omega},
    \end{equation}
    where $h^\uparrow_{\alpha}$ is a constant we will define later. Most importantly, this was only shown for \emph{finite-dimensional} spaces. Therefore, so far, the MEAT is not applicable to any CV QKD protocols or other applications which require infinite dimensions.
    
    Hence, the goal of this section is to formally establish that this bound still holds for infinite-dimensional separable Hilbert spaces. 
    
    Similarly, we will extend the following bound for $f$-weighted entropies to infinite-dimensional separable Hilbert spaces, which is required to prove security for variable-length protocols
    \begin{equation}
        	H_\alpha^{\uparrow,f_\mathrm{full}}(S_1^n|\CP_1^nE_n)_\rho\geq\sum_j\min_{\cP_1^{j-1}}\kappa_{\cP_1^{j-1}}.
    \end{equation}
    Again, here $\kappa_{\cP_1^{j-1}}$ are constants we will define in more detail later, as this is simply an informal summary.

    As already mentioned in the main text, the proof of the $f$-weighted marginal-constrained entropy accumulation theorem relies on the uniform continuity of the $f$-weighted \Renyi entropy. Therefore, we begin by proving the uniform continuity of both the conditional sandwiched \Renyi entropy and the $f$-weighted \Renyi entropy in \cref{app:subsec_continuity} for infnite-dimensional side-information.
    
    Next, we introduce the projected MEAT channels in \cref{app:subsec_conv_to_finite}, which will allow us to apply the finite-dimensional $f$-weighted entropy accumulation theorem and certain parts of the finite-dimensional MEAT. 
    
    In the subsequent section, we establish several preliminary lemmata to prove the $f$-weighted marginal-constrained entropy accumulation theorem in \cref{app:Proof_of_f_weighted_MEAT}, utilizing the established uniform continuity and taking the limit of the truncations.
    
    As noted previously, this approach does not work for the standard marginal-constrained entropy accumulation theorem because the objective in $h^\uparrow_{\alpha}$ is not uniformly continuous nor are the sets involved in $h_{\alpha}^{\uparrow}$ compact. Hence, in \cref{app:subsec_prelim_inf_MEAT}, we introduce a perturbation function $V$ closely related to the objective function in $h_{\alpha}^{\uparrow}$. Most importantly, we prove that it is lower semicontinuous and that, by the Fenchel--Moreau theorem \cite[Theorem~13.37]{Bauschke_2017}, $V=V^{**}$, where $V^{**}$ is the convex biconjugate. This establishes the strong duality of this perturbation function. Crucially, when the perturbation is set to $0$, we recover $h_{\alpha}^{\uparrow}$ and the corresponding dual expression in terms of the $f$-weighted \Renyi entropy. Finally, in \cref{app:subsec_proof_inf_MEAT}, we combine these results to prove the infinite-dimensional MEAT.

\subsection{Continuity Bounds}\label{app:subsec_continuity}

\begin{theorem}[Continuity Bound for Sandwiched \Renyi Entropy with infinite dimensional Side-information]\label{thrm:Continuity_Renyi}
	Let $\epsilon \in [0,1]$ and let $\rho,\sigma \in \dop{=}(QE)$, satisfying
	\begin{equation}
		\frac{1}{2} \norm{\rho-\sigma}_1 \leq \epsilon.
	\end{equation}
	Furthermore, let $\dim(Q)<\infty$, but $E$ can be possibly infinite dimensional.
	Then, the sandwiched \Renyi entropy is uniformly continuous and we find for $\alpha \in (1,\infty)$:
	\begin{equation}
		\abs{ \renyiSandUp_{\alpha}(Q|E)_{\rho} - \renyiSandUp_{\alpha}(Q|E)_{\sigma} } \leq \Delta_{\alpha}(\epsilon),
	\end{equation}
	where $\Delta_{\alpha}(\epsilon)$ is defined as
	\begin{widetext}
	\begin{equation}
		\Delta_{\alpha}(\epsilon)\defvar\min\left\{
		\begin{aligned} &\log(1+\epsilon)+\frac{1}{\alpha-1}\log\!\left(1+\epsilon d^{\alpha-1}-\frac{\epsilon^\alpha}{(1+\epsilon)^{\alpha-1}}\right),\\ &\frac{\alpha}{\alpha-1}\log\!\left(1+\epsilon d^{\frac{\alpha-1}{\alpha}}\right),\\ &\log(1+\epsilon)+\frac{\alpha}{\alpha-1}\log\!\left(1+\epsilon d^{\frac{\alpha-1}{\alpha}}-\frac{\epsilon^{2-\frac{1}{\alpha}}}{(1+\epsilon)^{\frac{\alpha-1}{\alpha}}}\right)
		\end{aligned}
		\right\},
	\end{equation}
	where \(d=\dim(Q)\) if $Q$ is classical and $d=\dim(Q)^2$ otherwise.
	\end{widetext}
\end{theorem}
\begin{proof}
	The statement already holds for finite dimensional $E$ by \cite[Corollary~5.1]{Bluhm_2026} and therefore, we only need to show the extension to infinite dimensions. If $Q$ is classical, the same proof as in \cite{Bluhm_2026} gives the stated replacement $\dim(Q)^2 \mapsto \dim(Q)$, since $\renyiSandUp_\alpha(Q|E)_\omega\geq0$ for every state $\omega$ which is classical on $Q$.
	
	Let $P_m$ be finite-rank projections on $E$ satisfying $P_m\leq P_{m+1}$ and $P_m\xrightarrow{m\to\infty}I_E$ strongly, let $E_m\defvar\supp(P_m)\oplus\Span\{\ket{\perp}\}$, and define the channel
	\begin{equation}
		\mathcal{T}_m(X)\defvar P_mXP_m+\Tr{\left(I_E-P_m\right)X}\pure{\perp}.
	\end{equation}
	Define
	\begin{equation}
		\rho^{(m)}\defvar(\id_Q\otimes\mathcal{T}_m)(\rho),\qquad
		\sigma^{(m)}\defvar(\id_Q\otimes\mathcal{T}_m)(\sigma).
	\end{equation}
	By the data-processing inequality for the trace distance,
	\begin{equation}
		\frac{1}{2}\norm{\rho^{(m)}-\sigma^{(m)}}_1
		\leq\frac{1}{2}\norm{\rho-\sigma}_1
		\leq\epsilon.
	\end{equation}
	Since $E_m$ is finite-dimensional, \cite[Corollary~5.1]{Bluhm_2026} gives
	\begin{equation}\label{eq:finite_continuity_cutoffs}
		\abs{\renyiSandUp_\alpha(Q|E_m)_{\rho^{(m)}}-\renyiSandUp_\alpha(Q|E_m)_{\sigma^{(m)}}}
		\leq\Delta_\alpha(\epsilon).
	\end{equation}
	
	It remains only to show that, for each $\omega\in\{\rho,\sigma\}$,
	\begin{equation}\label{eq:entropy_cutoff_convergence}
		\lim_{m\to\infty}\renyiSandUp_\alpha(Q|E_m)_{\omega^{(m)}}
		=\renyiSandUp_\alpha(Q|E)_\omega.
	\end{equation}
	Therefore, let
	\begin{equation}
		R_m\defvar I_Q\otimes P_m,\; r_m\defvar\Tr{R_m\rho}, \; s_m\defvar\Tr{R_m\sigma},
	\end{equation}
	and for all sufficiently large $m$, define
	\begin{equation}
		\widehat{\rho}^{(m)}\defvar\frac{R_m\rho R_m}{r_m},\qquad
		\widehat{\sigma}^{(m)}\defvar\frac{R_m\sigma R_m}{s_m}.
	\end{equation}
	Here we regard these states as states on $QE_m$ with zero support on the discard subspace. Applying the finite-dimensional approximation of the sandwiched \Renyi divergence and the directed minimax result \cite[Proposition~III.39 and Lemma~II.5]{Mosonyi_2023} to the definition of the conditional \Renyi entropy gives
	\begin{align}
		\renyiSandUp_\alpha(Q|E)_\rho
		&=\lim_{m\to\infty}\left(\renyiSandUp_\alpha(Q|E_m)_{\widehat{\rho}^{(m)}}-\frac{\alpha}{\alpha-1}\log r_m\right), \\
		\renyiSandUp_\alpha(Q|E)_\sigma
		&=\lim_{m\to\infty}\left(\renyiSandUp_\alpha(Q|E_m)_{\widehat{\sigma}^{(m)}}-\frac{\alpha}{\alpha-1}\log s_m\right).
	\end{align}
	Since $P_m\to I_E$ strongly, we have $r_m \to 1$ and $s_m \to 1$ and thus,
	\begin{align}\label{eq:normalized_projected_entropy_limits}
		\lim_{m\to\infty}\renyiSandUp_\alpha(Q|E_m)_{\widehat{\rho}^{(m)}}
		&=\renyiSandUp_\alpha(Q|E)_\rho, \\
		\lim_{m\to\infty}\renyiSandUp_\alpha(Q|E_m)_{\widehat{\sigma}^{(m)}}
		&=\renyiSandUp_\alpha(Q|E)_\sigma.
	\end{align}
	By definition, $R_m\rho R_m=r_m\widehat{\rho}^{(m)}$, while the discard term in $\rho^{(m)}$ has trace $1-r_m$ and support orthogonal to $\widehat{\rho}^{(m)}$. Similarly, $R_m\sigma R_m=s_m\widehat{\sigma}^{(m)}$, while the discard term in $\sigma^{(m)}$ has trace $1-s_m$ and support orthogonal to $\widehat{\sigma}^{(m)}$. Hence,
	\begin{equation}
		\frac{1}{2}\norm{\rho^{(m)}-\widehat{\rho}^{(m)}}_1=1-r_m,\qquad
		\frac{1}{2}\norm{\sigma^{(m)}-\widehat{\sigma}^{(m)}}_1=1-s_m.
	\end{equation}
	Applying the finite-dimensional continuity bound gives
	\begin{align}
		\abs{\renyiSandUp_\alpha(Q|E_m)_{\rho^{(m)}}-\renyiSandUp_\alpha(Q|E_m)_{\widehat{\rho}^{(m)}}}
		&\leq\Delta_\alpha(1-r_m) \xrightarrow{m\to\infty} 0, \\
		\abs{\renyiSandUp_\alpha(Q|E_m)_{\sigma^{(m)}}-\renyiSandUp_\alpha(Q|E_m)_{\widehat{\sigma}^{(m)}}}
		&\leq\Delta_\alpha(1-s_m) \xrightarrow{m\to\infty} 0.
	\end{align}
	Together with \cref{eq:normalized_projected_entropy_limits}, this yields
	\begin{align}
		\lim_{m\to\infty}\renyiSandUp_\alpha(Q|E_m)_{\rho^{(m)}}
		&=\renyiSandUp_\alpha(Q|E)_\rho, \label{eq:rho-m-entropy-converge} \\
		\lim_{m\to\infty}\renyiSandUp_\alpha(Q|E_m)_{\sigma^{(m)}}
		&=\renyiSandUp_\alpha(Q|E)_\sigma. \label{eq:sigma-m-entropy-converge}
	\end{align}
    By adding and subtracting the entropies of $\rho^{(m)}$ and $\sigma^{(m)}$ and taking the triangle inequality,
    {\allowdisplaybreaks
    \begin{align}
        &\abs{\renyiSandUp_\alpha(Q|E)_\rho-\renyiSandUp_\alpha(Q|E)_\sigma} \\
       \notag &\leq \abs{\renyiSandUp_\alpha(Q|E)_\rho - \renyiSandUp_\alpha(Q|E_m)_{\rho^{(m)}}} \\ 
        \notag& \hspace{3mm} + \abs{\renyiSandUp_\alpha(Q|E_m)_{\sigma^{(m)}} - \renyiSandUp_\alpha(Q|E)_\sigma} \\ 
        \notag& \hspace{7mm} + \abs{\renyiSandUp_\alpha(Q|E_m)_{\sigma^{(m)}} - \renyiSandUp_\alpha(Q|E)_{\rho^{(m)}}} \ .
    \end{align}
    }
    Taking $m \to \infty$ and using \eqref{eq:finite_continuity_cutoffs}, \eqref{eq:rho-m-entropy-converge}, and \eqref{eq:sigma-m-entropy-converge}, we find
    \begin{equation}
		\abs{\renyiSandUp_\alpha(Q|E)_\rho-\renyiSandUp_\alpha(Q|E)_\sigma}
		\leq\Delta_\alpha(\epsilon),
	\end{equation}
	which proves the theorem.
	% Taking $m\to\infty$ in \cref{eq:finite_continuity_cutoffs} we therefore find
	% \begin{equation}
	% 	\abs{\renyiSandUp_\alpha(Q|E)_\rho-\renyiSandUp_\alpha(Q|E)_\sigma}
	% 	\leq\Delta_\alpha(\epsilon),
	% \end{equation}
	% which proves the theorem.
\end{proof}

\begin{theorem}[Continuity Bound for $f$-weighted \Renyi Entropies]\label{thrm:Continuity_f_weighted}
	Let $\epsilon \in [0,1]$ and let $\rho,\sigma \in \dop{=}(\CP QQ')$ be classical on $\CP$, satisfying
	\begin{equation}
		\frac{1}{2} \norm{\rho-\sigma}_1 \leq \epsilon.
	\end{equation}
	Furthermore, let $f: \alphCP_1^n \to \R$ be a tradeoff function and let $\dim\left(\CP\right)<\infty$ and $\dim\left(Q\right) <\infty$. Crucially, the register $Q'$ can be infinite-dimensional. Then, the $f$-weighted \Renyi entropy is uniformly continuous and we find for $\alpha \in (1,\infty)$:
	\begin{equation}
		\abs{ \frenyiSandUp_{\alpha}(Q|\CP Q')_{\rho} - \frenyiSandUp_{\alpha}(Q|\CP Q')_{\sigma} } \leq \Delta_{\alpha,f}(\epsilon),
	\end{equation}
	where $\Delta_{\alpha,f}(\epsilon)$ is defined as
	\begin{widetext}
		\begin{equation}\label{eq:correction_fweighted_cont}
			\Delta_{\alpha,f}(\epsilon) \defvar \min \begin{cases}
				\begin{aligned}
					&\log(1 + \epsilon) + \frac{1}{\alpha - 1}\log(1 + \epsilon \bar{d}^{\alpha - 1} - \frac{\epsilon^\alpha}{(1 + \epsilon)^{\alpha - 1}})
				\end{aligned},\\
				\frac{\alpha}{\alpha - 1}\log(1 + \epsilon \bar{d}^{\frac{\alpha - 1}{\alpha}}) ,\\
				\begin{aligned}
					&\log(1 + \epsilon) + \frac{\alpha}{\alpha - 1}\log(1 + \epsilon \bar{d}^{ \frac{\alpha - 1}{\alpha}} - \frac{\epsilon^{2 - \frac{1}{\alpha}}}{(1 + \epsilon)^{\frac{\alpha - 1}{\alpha}}}).
				\end{aligned}
			\end{cases}
		\end{equation}
		
	\end{widetext}
	The effective dimension $\bar{d}$ depends on whether $Q$ is classical or quantum and is given by
	\begin{equation}
		\bar{d} = \begin{cases}
			\left(\dim(Q)\left(\ceil{2^{\max{f}-\min{f}}} +1\right) \right)^2 , &\text{$Q$ quantum},\\
			\dim(Q)\left(\ceil{2^{\max{f}-\min{f}}} +1 \right), &\text{$Q$ classical}.
		\end{cases}
	\end{equation}
\end{theorem}
\begin{proof}
	Let $\delta >0$ be arbitrary and set $M\defvar \max_{\cP}{f}(\cP) +\delta$. Then, the constant $M$ satisfies $M-f(\cP) >0$ for all $\cP$ in $\alphCP$.

    First, we invoke the construction from \cite[Lemma~4.2]{Arqand_2025}, which provides a read-and-prepare channel $\mathcal{D} \in \CPTP(\CP,\CP D)$ and register $D$ with $\dim(D)< \infty$, such that    
	\begin{align}\label{eq:proof_cont_bnd_f_weighted_plus_M}
		\frenyiSandUp_{\alpha}(Q|\CP Q')_\rho &= \renyiSandUp_{\alpha}(DQ|\CP Q')_{\rho} -M , \\
		\frenyiSandUp_{\alpha}(Q|\CP Q')_\sigma &= \renyiSandUp_{\alpha}(DQ|\CP Q')_{\sigma} -M.
	\end{align}
    While \cite[Lemma~4.2]{Arqand_2025} was formulated for finite-dimensional spaces, the identity remains valid even if the separable register $Q'$ is infinite-dimensional. The original proof carries over unchanged because it relies only on properties that still hold in infinite dimensions: the additivity of the sandwiched \Renyi divergence under tensor products with the finite-dimensional register $D$, and a decomposition over the finite classical register $\CP$.
	
	The proof of \cite[Lemma 4.2]{Arqand_2025} constructs the conditional states $\rho_{D|\cP}$ as a mixture between an arbitrary but fixed pure state $\ket{\psi}$ and the maximally mixed state, that is
	\begin{equation}
		\rho_{D|\cP} = \lambda_{\cP} \pure{\psi} + \left( 1- \lambda_{\cP} \right) \frac{I_D}{\dim(D)},
	\end{equation}
	where $\lambda_{\cP}$ is a parameter chosen such that $\renyiSandUp_{\alpha}(D)_{\rho_{D|\cP}} = M- f(\cP)$ for all $\cP$. Thus, as required to apply \cite[Lemma 4.2]{Arqand_2025}, the dimension of $D$ can be chosen as 
	\begin{equation}
		\dim(D) = \ceil{\max_{\cP} 2^{M-f(\cP)}} = \ceil{2^{M-\min{f}}}.
	\end{equation}

    Having established the properties of the extension to $D$, we now derive the new continuity bound. Using the relations in \cref{eq:proof_cont_bnd_f_weighted_plus_M} we find,
\begin{equation}
    	\begin{aligned}
		&\abs{ \frenyiSandUp_{\alpha}(Q|\CP Q')_{\rho} - \frenyiSandUp_{\alpha}(Q|\CP Q')_{\sigma} } \\
		&= \abs{M - \renyiSandUp_{\alpha}(DQ|\CP Q')_{\sigma} - M + \renyiSandUp_{\alpha}(DQ|\CP Q')_{\rho} } \\
		&= \abs{\renyiSandUp_{\alpha}(DQ|\CP Q')_{\sigma} - \renyiSandUp_{\alpha}(DQ|\CP Q')_{\rho}},
	\end{aligned}
\end{equation}
	since the constants $M$ cancel each other out. Furthermore, the extensions created by the read-and-prepare channel $\mathcal{D}$ satisfy due to the data processing inequality
\begin{equation}
    \frac{1}{2} \norm{ \rho_{DQ\CP Q'} - \sigma_{DQ \CP Q'} }_1 \leq \frac{1}{2} \norm{ \rho_{Q\CP Q'} - \sigma_{Q \CP Q'} }_1 \leq \epsilon
\end{equation}
	
	Therefore, one can apply the continuity bound for sandwiched conditional \Renyi entropies from \cref{thrm:Continuity_Renyi}. Since $D$ can be chosen classical, whenever $Q$ is classical the composite register $DQ$ is classical as well, and the corresponding tightening for classical registers applies, i.e. omitting the square from the dimension of the primary system.
	
	Moreover, the continuity bound is monotonically increasing in the effective dimension of the $DQ$ register, and for all sufficiently small $\delta>0$,
\begin{equation}
    	\begin{aligned}
		\dim(D) &= \ceil{2^{M-\min f}} = \ceil{2^{\max f + \delta -\min f}} \\
		&\leq \ceil{2^{\max f-\min f}}+1,
	\end{aligned}
\end{equation}
	and the theorem statement follows.
\end{proof}

\subsection{Conversion to Finite Dimensional Spaces}\label{app:subsec_conv_to_finite}
In this section, we introduce sequences of projected versions of the infinite-dimensional MEAT channels. One sequence still acts on the infinite-dimensional spaces and one sequence acts on finite-dimensional spaces. We show the latter sequence is equivalent to the former sequence up to an appropriate embedding while preserving the marginal constraint up to a known truncation (\cref{thrm:Finite_MEAT_Chains}). We show the former sequence converges to the infinite-dimensional MEAT channels (\cref{lem:Convergence_of_Projected_Channels}).

In a subsequent section, we apply the finite-dimensional version of the $f$-weighted marginal-constrained entropy accumulation theorem \cite[Theorem~4.1a]{Arqand_2025} to a chain of these projected channels and afterwards prove their convergence to the infinite-dimensional chain.

\subsubsection{Finite Dimensional Spaces}
Since finite-dimensional registers require no truncation, throughout this subsection we assume that all registers $A_j$ and $E_j$ are infinite-dimensional or are otherwise left unchanged.

\begin{definition}[Projections and discarded states]
On the systems $A_j$ for $j=0,\dots, n-1$ and $E_j$ for $j=0,\dots n$ let us define the projections
\begin{itemize}
    \item[1)] $\{P_{j}^{(d)}\}$ on $A_j$ with $\dim(\supp(P_{j}^{(d)})) = d$,
    \item[2)] $\{Q_{j}^{(d)}\}$ on $E_j$ with $\dim(\supp(Q_{j}^{(d)})) = d$,
\end{itemize}
satisfying
\begin{equation}
    P_{j}^{(d)} \leq P_{j}^{(d+1)}, \quad \text{and} \quad Q_{j}^{(d)} \leq Q_{j}^{(d+1)},
\end{equation}
for all $d\in \N$. Furthermore, we define the following pure discard states
\begin{itemize}
    \item[1)] $\pure{\perp_{j}}_{A_j} \in S_{=}(A_j)$ with $P_{j}^{(d)}\ket{\perp_j} = 0$.
    \item[2)] $\pure{\varnothing_{j}}_{E_j} \in S_{=}(E_j)$ with $Q_{j}^{(d)}\ket{\varnothing_{j}} = 0$.
\end{itemize}
Making use of these discard states we define the following extended projections
\begin{itemize}
    \item[1)] $\widebar{P}_{j}^{(d)} \defvar P_{j}^{(d)} + \pure{\perp_{j}}$,
    \item[2)] $\widebar{Q}_{j}^{(d)} \defvar Q_{j}^{(d)} + \pure{\varnothing_{j}}$,
\end{itemize}
on $A_j$ and $E_j$, respectively. These extended projections are additionally required to converge to the identity strongly, that is
\begin{align}
    \lim_{d \to \infty} \norm{\widebar{P}_{j}^{(d)} \ket{\psi} - \ket{\psi}}_2 = 0, \\ 
    \lim_{d \to \infty} \norm{\widebar{Q}_{j}^{(d)} \ket{\phi} - \ket{\phi}}_2 = 0,
\end{align}
for all $\ket{\psi}\in A_j$ and $\ket{\phi}\in E_j$, where $\norm{\cdot}_2$ denotes the Hilbert-space vector norm.
\end{definition}

\begin{definition}[Finite Dimensional Spaces]
For each infinite dimensional space $A_j$ and $E_j$ we define corresponding finite dimensional spaces as
\begin{itemize}
    \item[1)] $A_j^{(d)}$ of dimension $d+1$,
    \item[2)] $E_j^{(d)}$ of dimension $d+1$.
\end{itemize}
Additionally, in each space we identify a special unit vector $\ket{f}$ that will be mapped to the discard state under the embedding isometries we introduce in Definition~\ref{def:Isometries_Fin_Inf}.

Moreover, let $E$ denote a space into which $E_{j}$ can be embedded for all $j$. Similarly, let $\widebar{E}$ be a space into which $E_{j}^{(d)}$ can be embedded for all $j,d$. We further choose the latter space such that the family of embedded subspaces $\{E_j^{(d)}\}_{j,d}$ has the same embedding structure as the family
\begin{equation}
    \left\{ \supp\left(\widebar{Q}_j^{(d)}\right) \right\}_{j,d}
\end{equation}
inside $E$. In particular, all inclusion, intersection, and orthogonality relations between the spaces $E_j^{(d)}$ agree with those between the corresponding subspaces $\supp(\widebar{Q}_j^{(d)})$.
\end{definition}

\begin{definition}[Isometries between Finite and Infinite Spaces]\label{def:Isometries_Fin_Inf}
For all $j=0,\dots,n-1$ and dimensions $d\in \N$ we define the isometries $V_{j}^{(d)} : A_{j}^{(d)} \rightarrow A_j$ mapping between the finite and infinite $A$ spaces, which satisfy
\begin{align}
    (V_{j}^{(d)})^{\dagger} V_{j}^{(d)} &= \id_{A^{(d)}_j}, \\
    V_{j}^{(d)} (V_{j}^{(d)})^{\dagger} &= \widebar{P}_{j}^{(d)}, \\
    V_{j}^{(d)}\ket{f}_{A_{j}^{(d)}} &= \ket{\perp_{j}}_{A_{j}}.
\end{align}
Furthermore, for each $j=0,\dots,n$ and dimension $d \in \N$, we define the isometries $W_{j}^{(d)} : E_{j}^{(d)} \rightarrow E_j$ mapping between the finite and infinite $E$ registers, which in turn satisfy
\begin{align}
    (W_{j}^{(d)})^{\dagger} W_{j}^{(d)} &= \id_{E_j^{(d)}}, \\
    W_{j}^{(d)} (W_{j}^{(d)})^{\dagger} &= \widebar{Q}_{j}^{(d)}, \\
    W_{j}^{(d)}\ket{f}_{E_{j}^{(d)}} &= \ket{\varnothing_{j}}_{E_{j}}.
\end{align}
Whenever an isometry acts on a state containing additional registers, it is understood to be tensored with the identity operator on all remaining registers.
\end{definition}

\begin{lemma}[Trace Norm Convergence of the Trace-Preserving Truncation Map]\label{lem:truncation_convergence}
    Let $\sigma \in \dop{=}(Q)$ and let $\{P_d\}$ be a sequence of orthogonal projectors such that $P_d \to I$ in the strong operator topology as $d \to \infty$. Let $\tau \in \dop{=}(Q)$ be an arbitrary fixed quantum state, and define the trace-preserving map
    \begin{equation}
        \Pi_d(\rho) \defvar P_d \rho P_d + \Tr{\left(I - P_d \right) \rho} \tau.
    \end{equation}
    Then, the image of $\sigma$ under this map converges to $\sigma$ in trace norm, i.e.
    \begin{equation}
        \lim_{d \to \infty} \| \Pi_d(\sigma) - \sigma \|_1 = 0.
    \end{equation}
    Moreover, for any auxiliary register $R$ and any $\nu_{RQ}\in\dop{=}(RQ)$, it holds
    \begin{equation}
    	\lim_{d\to\infty}\norm{\left(\id_R\otimes\Pi_d\right)(\nu_{RQ})-\nu_{RQ}}_1=0.
    \end{equation}
\end{lemma}
\begin{proof}
    By the triangle inequality, and noting that $\|\tau\|_1 = 1$ and $\Tr{P_d \sigma} \leq 1$, we bound the total error as
    \begin{equation}\label{eq:triangle_strong_tracenorm}
            \| \Pi_d(\sigma) - \sigma \|_1 \leq \| P_d \sigma P_d - \sigma \|_1 + 1 - \Tr{P_d \sigma}.
    \end{equation}
    
    We can bound the first term by adding and subtracting $P_d \sigma$ inside the norm and applying the triangle inequality, which gives
\begin{equation}
        \begin{aligned}
        &\| P_d \sigma P_d - \sigma \|_1  \\
        &= \| P_d \sigma P_d - P_d \sigma + P_d \sigma - \sigma \|_1 \\
        &\leq \| P_d (\sigma P_d - \sigma) \|_1 + \| (P_d - I) \sigma \|_1.
    \end{aligned}
\end{equation}
    For the first term, we apply Hölder's inequality, $\|AB\|_1 \leq \|A\|_\infty \|B\|_1$. Because $P_d$ is a projector, its operator norm is bounded by $\|P_d\|_\infty \leq 1$, yielding
    \begin{equation}
        \| P_d (\sigma P_d - \sigma) \|_1 \leq \| \sigma P_d - \sigma \|_1.
    \end{equation}
    Because the trace norm is invariant under the adjoint operation ($\| X^\dagger \|_1 = \| X \|_1$) and both $P_d$ and $\sigma$ are Hermitian, we have $\| \sigma P_d - \sigma \|_1 = \| P_d \sigma - \sigma \|_1 = \| (I - P_d) \sigma \|_1$. The second term is trivially $\| (P_d - I) \sigma \|_1 = \| (I - P_d) \sigma \|_1$. Combining these results yields
    \begin{equation}
        \| P_d \sigma P_d - \sigma \|_1 \leq 2 \| (I - P_d) \sigma \|_1.
    \end{equation}
    
    Next, we apply Hölder's inequality another time, but with $p=q=2$ and find
\begin{equation}\label{eq:HS_bound}
        \begin{aligned}
            2 \| (I - P_d) \sigma \|_1 &= 2\| (I - P_d) \sqrt{\sigma} \sqrt{\sigma} \|_1 \\
            &\leq2 \| (I - P_d) \sqrt{\sigma} \|_2 \, \| \sqrt{\sigma} \|_2 \\
            &= 2\sqrt{ \Tr{\sigma (I - P_d)} } \sqrt{ \Tr{\sigma} } \\
            &= 2\sqrt{ 1 - \Tr{P_d \sigma} }.
    \end{aligned}
\end{equation}
    
    Now, it remains to evaluate the limit of $\Tr{P_d \sigma}$. Since $P_d \xrightarrow{d \to \infty} I$ strongly, it holds
    \begin{equation}
        \lim_{d \to \infty} \Tr{P_d \sigma} =1.
    \end{equation}
    
    Substituting this limit into \cref{eq:HS_bound} shows that $2 \| (I - P_d) \sigma \|_1$ vanishes for $d\to \infty$. Consequently, both terms in \cref{eq:triangle_strong_tracenorm} go to zero.
    
    The same argument applies when $\Pi_d$ acts locally on $Q$ while an arbitrary auxiliary register $R$ is left unchanged. In particular, for every $\nu_{RQ}\in\dop{=}(RQ)$,
\begin{equation}
        \begin{aligned}
    	&\norm{\left(\id_R\otimes\Pi_d\right)(\nu_{RQ})-\nu_{RQ}}_1 \\
        \leq &\begin{aligned}[t]
            &2\sqrt{\Tr{\left(I_R\otimes(I_Q-P_d)\right)\nu_{RQ}}} \\
            &+ \Tr{\left(I_R\otimes(I_Q-P_d)\right)\nu_{RQ}}
        \end{aligned} \\
    	\leq &2\sqrt{\Tr{\left(I_Q-P_d\right)\nu_Q}} +\Tr{\left(I_Q-P_d\right)\nu_Q} \\
    	&\xrightarrow{d\to\infty} 0.
    \end{aligned}
\end{equation}
\end{proof}

\subsubsection{Projected MEAT Channels}
\begin{definition}[Projection Channels]
For each $j=1,\dots,n$, we define the \emph{input} projection channel $\ProjIn{j}{d} : S_{=}(A_{j-1}  E_{j-1}) \longrightarrow S_{=}(A_{j-1} E_{j-1})$ as
\begin{equation}
    \begin{aligned}
    \ProjIn{j}{d}(\rho) \defvar &\left(\widebar{P}_{j-1}^{(d)} \otimes \widebar{Q}_{j-1}^{(d)}\right) \rho \left( \widebar{P}_{j-1}^{(d)} \otimes \widebar{Q}_{j-1}^{(d)} \right) \\
    &+ \Tr{\left(I_{A_{j-1} E_{j-1}} - \widebar{P}_{j-1}^{(d)} \otimes \widebar{Q}_{j-1}^{(d)} \right) \rho} \\
    &\cdot \ketbra{\perp_{j-1} \varnothing_{j-1}}{\perp_{j-1} \varnothing_{j-1}}.
    \end{aligned}
\end{equation}
Furthermore, for each $j=1,\dots,n$, we define the \emph{output} projection channel $\ProjOut{j}{d} : S_{=}(E_j) \rightarrow S_{=}(E_j)$ as
\begin{equation}
    \ProjOut{j}{d}(\rho) := \widebar{Q}_{j}^{(d)} \rho \widebar{Q}_{j}^{(d)} + \Tr{\left( I_{E_{j}} - \widebar{Q}_{j}^{(d)}\right)\rho} \pure{\varnothing_{j}}.
\end{equation}

Finally, we define the \emph{initial} projection channel $\Pi_0^{(d)} : S_{=}(A_0^{n-1} E_0) \rightarrow S_{=}(A_0^{n-1} E_0)$ as a tensor product of projection channels on $A_0,\dots,A_{n-1}$ and $E_0$. Therefore, for every $j=0,\dots,n-1$, define the projection channel on $A_j$ as $\Pi_{A_j}^{(d)}: S_{=}(A_j) \rightarrow S_{=}(A_j)$ by
\begin{equation}\label{eq:proj-chan-on-Aj}
    \Pi_{A_j}^{(d)}(\rho) \defvar\widebar P_j^{(d)}\rho\widebar P_j^{(d)} 
    + \Tr{\left(I_{A_j}-\widebar P_j^{(d)}\right)\rho} \pure{\perp_j}.
\end{equation}
Similarly, define the projection channel on $E_0$ as $\Pi_{E_0}^{(d)}:S_{=}(E_0)\rightarrow S_{=}(E_0)$ by
\begin{equation}
    \Pi_{E_0}^{(d)}(\rho) \defvar \widebar Q_0^{(d)}\rho\widebar Q_0^{(d)}
    + \Tr{\left(I_{E_0}-\widebar Q_0^{(d)}\right)\rho}\pure{\varnothing_0}.
\end{equation}
Then, we define the initial projection channel $\Pi_0^{(d)} $ by
\begin{equation}\label{eq:proj-chan-on-Aj-Ej}
    \Pi_0^{(d)} \defvar \left(\bigotimes_{j=0}^{n-1}\Pi_{A_j}^{(d)}\right) \otimes\Pi_{E_0}^{(d)}.
\end{equation}
Whenever one of the channels defined above is applied to a state containing additional registers, it is understood to be tensored with the identity channel on all remaining registers.
\end{definition}

\begin{definition}[Projected MEAT Channels]
For each $j=1,\dots,n$ we define the finite-dimensional and infinite dimensional projected MEAT channels. The \emph{infinite dimensional projected} MEAT channel $\EATchannProjInf_j^{(k,l)} : S_{=}(A_{j-1}E_{j-1}) \longrightarrow S_{=}(S_j \CP_j E_j)$ is defined as
\begin{equation}
    \EATchannProjInf_j^{(k,l)}(\rho) \defvar \ProjOut{j}{k} \circ \EATchann_j \circ \ProjIn{j}{l}(\rho).
\end{equation}

Additionally, we define the \emph{finite-dimensional projected} MEAT channel $\EATchannProj_j^{(k,l)} : S_{=}(A_{j-1}^{(l)}E_{j-1}^{(l)}) \longrightarrow S_{=}(S_j\CP_jE_j^{(k)})$ by
\begin{equation}\label{eq:finite_proj_channel}
    \begin{aligned}
        \EATchannProj_j^{(k,l)}(\tau) \defvar {}& \left(I_{S_j\CP_j}\otimes(W_j^{(k)})^\dagger\right)\rho'\left(I_{S_j\CP_j}\otimes W_j^{(k)}\right) \\
        &+ \pTr{E_j}{\left(I_{S_j\CP_j}\otimes\left(I_{E_j}-\widebar Q_j^{(k)}\right)\right)\rho'}\otimes\pure{f}_{E_j^{(k)}},
    \end{aligned}
\end{equation}
where
\begin{equation}
    \rho' \defvar \EATchann_j\circ\ProjIn{j}{l}\left[\left(V_{j-1}^{(l)}\otimes W_{j-1}^{(l)}\right)\tau\left(V_{j-1}^{(l)}\otimes W_{j-1}^{(l)}\right)^\dagger\right].
\end{equation}
\end{definition}

\begin{remark}
    By the properties of the isometries in \cref{def:Isometries_Fin_Inf} the map in \cref{eq:finite_proj_channel} is trace preserving. Furthermore, it is completely positive by inspection.
\end{remark}

\subsubsection{Equivalences between Finite- and Infinite-Dimensional Channels}

\begin{lemma}[Single Round Equivalence]\label{Lem:Single_Round_Equiv}
Let $\tau \in S_{=}(A_{j-1}^{(l)}E_{j-1}^{(l)})$ be arbitrary. The output of the projected finite-dimensional MEAT channel satisfies:
\begin{equation}
\begin{aligned}
    &W_{j}^{(k)} \EATchannProj_j^{(k,l)}(\tau) (W_{j}^{(k)})^{\dagger} \\
    &= \ProjOut{j}{k} \circ \EATchann_j \left[ (V_{j-1}^{(l)} \otimes W_{j-1}^{(l)}) \tau (V_{j-1}^{(l)} \otimes W_{j-1}^{(l)})^{\dagger} \right]
    \end{aligned}
\end{equation}
\end{lemma}
\begin{proof}
Let
\begin{equation}
    \omega \defvar (V_{j-1}^{(l)}\otimes W_{j-1}^{(l)})\tau(V_{j-1}^{(l)}\otimes W_{j-1}^{(l)})^\dagger.
\end{equation}
Since $\omega$ is supported on $\supp(\widebar P_{j-1}^{(l)})\otimes\supp(\widebar Q_{j-1}^{(l)})$, we have $\ProjIn{j}{l}(\omega)=\omega$. Define
\begin{equation}
    \rho' \defvar \EATchann_j(\omega), \quad \eta_{S_j\CP_j} \defvar \pTr{E_j}{\left(I_{E_j}-\widebar Q_j^{(k)}\right)\rho'}.
\end{equation}
Using $\widebar Q_j^{(k)}=W_j^{(k)}(W_j^{(k)})^\dagger$ and $W_j^{(k)}\pure f(W_j^{(k)})^\dagger=\pure{\varnothing_j}$, we obtain
\begin{equation}
 \begin{aligned}
    &\ProjOut{j}{k}(\rho') \\
    &= \widebar Q_j^{(k)}\rho'\widebar Q_j^{(k)}+\eta_{S_j\CP_j}\otimes\pure{\varnothing_j} \\
    &= W_j^{(k)}\left((W_j^{(k)})^\dagger\rho'W_j^{(k)}+\eta_{S_j\CP_j}\otimes\pure f\right)(W_j^{(k)})^\dagger \\
    &= W_j^{(k)}\EATchannProj_j^{(k,l)}(\tau)(W_j^{(k)})^\dagger.
\end{aligned}   
\end{equation}

\end{proof}

\begin{lemma}[Embedding Channels]\label{Lem:Embedding_Channels}
For all $j=0,\dots,n-1$, define
\begin{equation}
    U_j \defvar (V_j^{(l)})^\dagger V_j^{(d)}\otimes(W_j^{(l)})^\dagger W_j^{(k)}
\end{equation}
and the map $\Xi_j:S_{=}(A_j^{(d)}E_j^{(k)})\rightarrow S_{=}(A_j^{(l)}E_j^{(l)})$ by
\begin{equation}
    \Xi_j(\tau)\defvar U_j\tau U_j^\dagger.
\end{equation}
Then, for all $l\geq d$ and $l\geq k$, the operator $U_j$ is an isometry and $\Xi_j$ is an isometric CPTP channel. Moreover, they satisfy
\begin{equation}
\begin{aligned}
    &\ProjIn{j+1}{l}\left[(V_j^{(d)} \otimes W_j^{(k)}) \tau (V_j^{(d)}\otimes W_j^{(k)})^\dagger\right] \\
    &= (V_j^{(l)} \otimes W_j^{(l)}) \Xi_j(\tau) (V_j^{(l)} \otimes W_j^{(l)})^\dagger.
\end{aligned}
\end{equation}
\end{lemma}
\begin{proof}
\begin{widetext}
Since $l\geq d$ and $l\geq k$, it holds
\begin{equation}
    \supp(\widebar P_j^{(d)}) \subseteq \supp(\widebar P_j^{(l)}), \quad \supp(\widebar Q_j^{(k)}) \subseteq \supp(\widebar Q_j^{(l)}).
\end{equation}
Therefore, $\widebar P_j^{(l)}V_j^{(d)}=V_j^{(d)}$ and $\widebar Q_j^{(l)}W_j^{(k)}=W_j^{(k)}$. Thus,
\begin{align}
    U_{j}^{\dagger} U_{j} &= \left((V_{j}^{(d)})^{\dagger} V_{j}^{(l)} \otimes (W_{j}^{(k)})^{\dagger} W_{j}^{(l)}\right)\left((V_{j}^{(l)})^{\dagger} V_{j}^{(d)} \otimes (W_{j}^{(l)})^{\dagger} W_{j}^{(k)}\right) \\
    \notag&= (V_{j}^{(d)})^{\dagger} V_{j}^{(l)} (V_{j}^{(l)})^{\dagger} V_{j}^{(d)} \otimes (W_{j}^{(k)})^{\dagger} W_{j}^{(l)} (W_{j}^{(l)})^{\dagger} W_{j}^{(k)} \\
    \notag&= (V_{j}^{(d)})^{\dagger} \widebar{P}_{j}^{(l)} V_{j}^{(d)} \otimes (W_{j}^{(k)})^{\dagger} \widebar{Q}_{j}^{(l)} W_{j}^{(k)} \\
    \notag&= (V_{j}^{(d)})^{\dagger} V_{j}^{(d)} \otimes (W_{j}^{(k)})^{\dagger} W_{j}^{(k)} \\
    \notag&= I
\end{align}
Hence, $U_j$ is an isometry and $\Xi_j$ is a CPTP map for each $j=0,\dots,n-1$.

Next, we show the claimed equivalence between the input projection channel and the embedding channel.
\begin{align}
    &\ProjIn{j+1}{l} \left[ (V_{j}^{(d)} \otimes W_{j}^{(k)}) \tau (V_{j}^{(d)} \otimes W_{j}^{(k)})^{\dagger} \right] \\
\notag&= (\widebar{P}_{j}^{(l)} \otimes \widebar{Q}_{j}^{(l)}) (V_{j}^{(d)} \otimes W_{j}^{(k)}) \tau (V_{j}^{(d)} \otimes W_{j}^{(k)})^{\dagger} (\widebar{P}_{j}^{(l)} \otimes \widebar{Q}_{j}^{(l)}) \\
    \notag&\quad + \Tr{(I - \widebar{P}_{j}^{(l)} \otimes \widebar{Q}_{j}^{(l)}) (V_{j}^{(d)} \otimes W_{j}^{(k)}) \tau (V_{j}^{(d)} \otimes W_{j}^{(k)})^{\dagger}} \ketbra{\perp_{j}, \varnothing_{j}}{\perp_{j}, \varnothing_{j}} \\
   \notag &= \widebar{P}_{j}^{(l)} V_{j}^{(d)} \otimes \widebar{Q}_{j}^{(l)} W_{j}^{(k)} \tau (V_{j}^{(d)} \otimes W_{j}^{(k)})^{\dagger} (\widebar{P}_{j}^{(l)} \otimes \widebar{Q}_{j}^{(l)}) \\
   \notag &\quad + \Tr{(V_{j}^{(d)} \otimes W_{j}^{(k)}) \tau (V_{j}^{(d)} \otimes W_{j}^{(k)})^{\dagger} - (V_{j}^{(d)} \otimes W_{j}^{(k)}) \tau (V_{j}^{(d)} \otimes W_{j}^{(k)})^{\dagger}} \ketbra{\perp_{j}, \varnothing_{j}}{\perp_{j}, \varnothing_{j}} \\
    \notag&= (V_{j}^{(l)} \otimes W_{j}^{(l)}) \left( (V_{j}^{(l)})^{\dagger} V_{j}^{(d)} \otimes (W_{j}^{(l)})^{\dagger} W_{j}^{(k)} \right) \tau \left( (V_{j}^{(l)})^{\dagger} V_{j}^{(d)} \otimes (W_{j}^{(l)})^{\dagger} W_{j}^{(k)} \right)^{\dagger} (V_{j}^{(l)} \otimes W_{j}^{(l)})^{\dagger} + 0 \\
    \notag&= (V_{j}^{(l)} \otimes W_{j}^{(l)}) U_{j} \tau U_{j}^{\dagger} (V_{j}^{(l)} \otimes W_{j}^{(l)})^{\dagger} \\
    \notag&= (V_{j}^{(l)} \otimes W_{j}^{(l)}) \Xi_{j}(\tau) (V_{j}^{(l)} \otimes W_{j}^{(l)})^{\dagger},
\end{align}
where for the second equality we used $\widebar{P}_{j}^{(l)} V_{j}^{(d)} = V_{j}^{(d)}$ and $\widebar{Q}_{j}^{(l)} W_{j}^{(k)} = W_{j}^{(k)}$, the third is expanding terms, and the fourth and fifth are our definitions of $U_{j}$ and $\Xi_{j}$.
\end{widetext}
\end{proof}

\begin{theorem}[Finite MEAT Chains]\label{thrm:Finite_MEAT_Chains}
Let $\omega_{A_0^{n-1}E_0} \in S_{=}(A_0^{n-1}E_0)$ be an arbitrary state satisfying the marginal $\omega_{A_0^{n-1}} = \bigotimes_{j=0}^{n-1} \sigma_{A_j}^{(j)}$. Furthermore, let $\tau \in S_{=}((A^{(d)})_{0}^{n-1}E_0^{(d)})$ be the corresponding finite-dimensional state defined by
\begin{widetext}
\begin{equation}
    \tau_{(A^{(d)})_{0}^{n-1}E_0^{(d)}} \defvar \left(\left(\bigotimes_{j=0}^{n-1}(V_j^{(d)})^\dagger\right)\otimes(W_0^{(d)})^\dagger\right)\Pi_0^{(d)}[\omega]\left(\left(\bigotimes_{j=0}^{n-1}V_j^{(d)}\right)\otimes W_0^{(d)}\right),
\end{equation}
which also satisfies a marginal constraint $\tau_{(A^{(d)})_{0}^{n-1}} = \bigotimes_{j=0}^{n-1} \xi_{A_j^{(d)}}^{(j)}$ where
\begin{equation}
    \xi_{A_j^{(d)}}^{(j)}\defvar (V_j^{(d)})^\dagger\sigma_{A_j}^{(j)}V_j^{(d)}+\Tr{\left(I_{A_j}-\widebar P_j^{(d)}\right)\sigma_{A_j}^{(j)}}\pure{f}_{A_j^{(d)}}.
\end{equation}
\end{widetext}
Set $k_0\defvar d$, and for each $j=1,\dots,n$, let $\Xi_{j-1}$ denote the embedding channel defined with parameters $(d,k_{j-1},l_j)$. Additionally, assume that $l_j\geq d$ and $l_j\geq k_{j-1}$ for all $j=1,\dots,n$.

Then, the chain of $n$ concatenated infinite dimensional projected MEAT channels (acting on the infinite-dimensional spaces) is equivalent to the chain of finite dimensional projected MEAT channels (acting purely on the finite-dimensional spaces) up to an isometry, i.e.,
\begin{widetext}
\begin{equation}
    \EATchannProjInf_n^{(k_n, l_n)} \circ \dots \circ \EATchannProjInf_1^{(k_1, l_1)} \circ \Pi_0^{(d)} [\omega_{A_0^{n-1}E_0}] = W_{n}^{(k_n)} \left( (\EATchannProj_n^{(k_n, l_n)} \circ \Xi_{n-1}) \circ \dots \circ (\EATchannProj_1^{(k_1, l_1)} \circ \Xi_{0}) [\tau_{(A^{(d)})_{0}^{n-1}E_0^{(d)}}] \right) \left(W_{n}^{(k_n)}\right)^{\dagger}.
\end{equation}
\end{widetext}
\end{theorem}

\begin{proof}
First, define the projected state $\rho := \Pi_0^{(d)}(\omega_{A_0^{n-1}E_0})$. Since the image of $\Pi_0^{(d)}$ is supported on $\left(\bigotimes_{j=0}^{n-1}\supp(\widebar P_j^{(d)})\right)\otimes\supp(\widebar Q_0^{(d)})$, the definition of $\tau$ implies
\begin{equation}\label{eq:proof_meat_chain_rho_tau}
    \rho = \left( \left(\bigotimes_{j=0}^{n-1} V_{j}^{(d)}\right) \otimes W_{0}^{(d)} \right) \tau \left( \left(\bigotimes_{j=0}^{n-1} V_{j}^{(d)}\right) \otimes W_{0}^{(d)} \right)^{\dagger} \ . 
\end{equation}
Here, we used $V_j^{(d)}(V_j^{(d)})^\dagger=\widebar P_j^{(d)}$ and $W_0^{(d)}(W_0^{(d)})^\dagger=\widebar Q_0^{(d)}$.  Additionally, $\tau$ satisfies the claimed marginal since
{\allowdisplaybreaks
    \begin{align}
    &\pTr{E_0^{(d)}}{\tau_{(A^{(d)})_0^{n-1}E_0^{(d)}}} \\
    \notag&= \left(\bigotimes_{j=0}^{n-1}(V_j^{(d)})^\dagger\right)\rho_{A_0^{n-1}}\left(\bigotimes_{j=0}^{n-1}V_j^{(d)}\right) \\
    \notag&= \bigotimes_{j=0}^{n-1}\left((V_j^{(d)})^\dagger\Pi_{A_j}^{(d)}(\sigma_{A_j}^{(j)})V_j^{(d)}\right) \\
    \notag&= \bigotimes_{j=0}^{n-1}\xi_{A_j^{(d)}}^{(j)}.
\end{align}
}
For the second equality, we used the definition of $\rho=\Pi_0^{(d)}(\omega)$, the trace-preserving property of $\Pi_{E_0}^{(d)}$, and $\omega_{A_0^{n-1}}=\bigotimes_{j=0}^{n-1}\sigma_{A_j}^{(j)}$, which imply $\rho_{A_0^{n-1}}=\bigotimes_{j=0}^{n-1}\Pi_{A_j}^{(d)}(\sigma_{A_j}^{(j)})$.

Now, we explicitly show the conversion for the first channel and the theorem statement follows by simply repeating the same steps.

\begin{widetext}
We find:
\begin{equation}
    \begin{aligned}
    &\EATchannProjInf_1^{(k_1, l_1)} \circ \Pi_0^{(d)}(\omega_{A_0^{n-1}E_0}) \\
    &= \ProjOut{1}{k_1} \circ \EATchann_1 \circ \ProjIn{1}{l_1} \circ \Pi_0^{(d)}(\omega) \\
    &= \ProjOut{1}{k_1} \circ \EATchann_1 \circ \ProjIn{1}{l_1}(\rho) \\
    &= \ProjOut{1}{k_1} \circ \EATchann_1 \circ \ProjIn{1}{l_1} \left[ \left( \left(\bigotimes_{j=0}^{n-1} V_{j}^{(d)}\right) \otimes W_{0}^{(d)} \right) \tau \left( \left(\bigotimes_{j=0}^{n-1} V_{j}^{(d)}\right) \otimes W_{0}^{(d)} \right)^{\dagger} \right] \\
    &= \left(\bigotimes_{j=1}^{n-1} V_{j}^{(d)}\right) \left\{ \ProjOut{1}{k_1} \circ \EATchann_1 \left[ (V_{0}^{(l_1)} \otimes W_{0}^{(l_1)}) \Xi_{0}(\tau) (V_{0}^{(l_1)} \otimes W_{0}^{(l_1)})^{\dagger} \right] \right\} \left(\bigotimes_{j=1}^{n-1} V_{j}^{(d)}\right)^{\dagger} \\
    &=\left(\bigotimes_{j=1}^{n-1} V_{j}^{(d)}\right) \left\{ W_{1}^{(k_1)} \left(\EATchannProj_1^{(k_1,l_1)}\circ\Xi_0\right)[\tau] (W_{1}^{(k_1)})^{\dagger} \right\} \left(\bigotimes_{j=1}^{n-1} V_{j}^{(d)}\right)^{\dagger} \\
    &= \left( \left(\bigotimes_{j=1}^{n-1} V_{j}^{(d)}\right) \otimes W_{1}^{(k_1)} \right) \left(\EATchannProj_1^{(k_1, l_1)} \circ \Xi_{0}\right)[\tau] \left( \left(\bigotimes_{j=1}^{n-1} (V_{j}^{(d)})^{\dagger}\right) \otimes (W_{1}^{(k_1)})^{\dagger} \right).
\end{aligned}
\end{equation}
\end{widetext}
The steps result from the following arguments. In the first and second step, we simply write out the definition of $\EATchannProjInf_1^{(k_1,l_1)}$ and insert the definition of $\rho$. In step three, we rewrite $\rho$ in terms of $\tau$ as stated in \cref{eq:proof_meat_chain_rho_tau}. Steps four and five follow by \cref{Lem:Embedding_Channels,Lem:Single_Round_Equiv}, respectively.

Iterating the same argument over the remaining $n-1$ rounds yields the theorem statement.
\end{proof}

From now on, for notational simplicity, for each fixed $d$, we set $k_0\defvar d$ and use the shorthand
\begin{equation}\label{eq:shorthand_iterated_lim}
	\lim_{\substack{\mbf{k}\to\infty\\\mbf{l}\to\infty}}
	\defvar
	\lim_{k_1\to\infty}\lim_{\substack{l_1\to\infty\\l_1\geq\max\{d,k_0\}}}\dots\lim_{k_n\to\infty}\lim_{\substack{l_n\to\infty\\l_n\geq\max\{d,k_{n-1}\}}},
\end{equation}
where the analogous restriction $l_j\geq\max\{d,k_{j-1}\}$ is imposed for every intermediate $j$. This restriction does not change the iterated limit because, for every fixed $d$ and $k_{j-1}$, it only removes finitely many initial values of $l_j$. Moreover, all limits of states below are understood with respect to the trace norm.

\begin{lemma}[Convergence of Projected Channels]\label{lem:Convergence_of_Projected_Channels}
Let $\alpha \in (1,\infty)$ and let $f:\alphCP \to \R$ be an arbitrary tradeoff function. For every $\omega_{A_0^{n-1}E_0}\in S_{=}(A_0^{n-1}E_0)$ and every event $\Omega$ on $\CP_1^n$ whose probability under the original output satisfies
\begin{equation}
    p_\Omega\defvar\Pr_{\left(\EATchann_n\circ\dots\circ\EATchann_1\right)[\omega]}[\Omega]\in(0,1],
\end{equation}
the following statements hold:
\begin{enumerate}
    \item The outputs of the infinite-dimensional projected channels converge to the output of the original channels, i.e.
    \begin{equation}\label{eq:Convergence_of_Projected_Channels}
        \begin{aligned}
        &\lim_{d \to \infty} \lim_{\substack{\mbf{k} \rightarrow \infty \\ \mbf{l}\rightarrow \infty}} \left( \EATchannProjInf_n^{(k_n, l_n)} \circ \dots \circ \EATchannProjInf_1^{(k_1, l_1)} \left(\Pi_0^{(d)}\left(\omega_{A_0^{n-1} E_0}\right)\right) \right)\\
        &=\EATchann_n \circ \dots \circ \EATchann_1 \left(\omega_{A_0^{n-1} E_0}\right).
        \end{aligned}
    \end{equation}
    \item The same limit as in \cref{eq:Convergence_of_Projected_Channels} holds conditioned on event $\Omega$ on $\CP_1^n$.
    \item The $f$-weighted sandwiched conditional \Renyi entropy of the outputs of the infinite-dimensional projected channels converges to the original entropy, i.e.
    \begin{equation}\label{eq:Convergence_of_Projected_Channels_f_entropies}
    \begin{aligned}
        &\lim_{d \to \infty} \lim_{\substack{\mbf{k} \rightarrow \infty \\ \mbf{l}\rightarrow \infty}} \frenyiSandUp_{\alpha}(S_1^n | \CP_1^n E_n)_{\EATchannProjInf_n^{(k_n, l_n)} \circ \dots \circ \EATchannProjInf_1^{(k_1, l_1)} \left(\Pi_0^{(d)} \left( \omega_{A_0^{n-1} E_0} \right) \right)} \\
        &=\frenyiSandUp_{\alpha}(S_1^n | \CP_1^n E_n)_{\EATchann_n \circ \dots \circ \EATchann_1 \left(\omega_{A_0^{n-1} E_0}\right)}. 
    \end{aligned}
    \end{equation}
    \item The sandwiched conditional \Renyi entropy conditioned on $\Omega$ converges to the corresponding entropy of the original output, i.e.\begin{equation}\label{eq:Convergence_of_Projected_Channels_entropies}
    \begin{aligned}
        &\lim_{d \to \infty} \lim_{\substack{\mbf{k} \rightarrow \infty \\ \mbf{l}\rightarrow \infty}} \renyiSandUp_{\alpha}(S_1^n | \CP_1^n E_n)_{\EATchannProjInf_n^{(k_n, l_n)} \circ \dots \circ \EATchannProjInf_1^{(k_1, l_1)} \left(\Pi_0^{(d)} \left( \omega_{A_0^{n-1} E_0} \right) \right)\big|_\Omega} \\
        &=\renyiSandUp_{\alpha}(S_1^n | \CP_1^n E_n)_{\EATchann_n \circ \dots \circ \EATchann_1 \left(\omega_{A_0^{n-1} E_0}\right) \big|_\Omega}. 
    \end{aligned}
    \end{equation}
\end{enumerate}
\end{lemma}

\begin{proof}
By \cref{thrm:Continuity_f_weighted,thrm:Continuity_Renyi}, the $f$-weighted \Renyi entropy and the conditional \Renyi entropy are uniformly continuous. Thus, we only need to show \cref{eq:Convergence_of_Projected_Channels} and its version conditioned on event $\Omega$.

We start with the version not conditioned on $\Omega$. For any fixed state $\rho\in\dop{=}(S_j\CP_jE_j)$, define
\begin{equation}
    \begin{aligned}
	&\epsilon_{j,k}^{\mathrm{out}}\defvar\Tr{\left(I_{E_j}-\widebar Q_j^{(k)}\right)\rho_{E_j}}, \\
	&\nu_{S_j\CP_j}^{(k)}\defvar\pTr{E_j}{\left(I_{E_j}-\widebar Q_j^{(k)}\right)\rho}.
\end{aligned}
\end{equation}
Using the definition of $\ProjOut{j}{k}$, the triangle inequality, and the estimate established in the proof of \cref{lem:truncation_convergence}, applied locally to $E_j$ with $S_j\CP_j$ as the auxiliary register, we obtain
\begin{equation}
    \begin{aligned}
	\norm{\rho-\ProjOut{j}{k}(\rho)}_1
	&\leq \begin{aligned}[t]
		&\norm{\rho-\widebar Q_j^{(k)}\rho\widebar Q_j^{(k)}}_1 \\
		&+ \norm{\nu_{S_j\CP_j}^{(k)}\otimes\pure{\varnothing_j}}_1 
	\end{aligned}\\
	&\leq2\sqrt{\epsilon_{j,k}^{\mathrm{out}}}+\epsilon_{j,k}^{\mathrm{out}}.
\end{aligned}
\end{equation}

Since $\widebar Q_j^{(k)}\to I_{E_j}$ strongly, $\epsilon_{j,k}^{\mathrm{out}}\to0$, and hence
\begin{equation}\label{eq:proof_lim_out}
	\lim_{k\to\infty}\norm{\rho-\ProjOut{j}{k}(\rho)}_1=0.
\end{equation}

Similarly, for any fixed state $\rho\in\dop{=}(A_{j-1}E_{j-1})$, define
\begin{equation}
	\epsilon_{j,l}^{\mathrm{in}}\defvar\Tr{\left(I_{A_{j-1}E_{j-1}}-\widebar P_{j-1}^{(l)}\otimes\widebar Q_{j-1}^{(l)}\right)\rho}.
\end{equation}
Then
\begin{equation}\label{eq:proof_lim_in}
	\norm{\rho-\ProjIn{j}{l}(\rho)}_1\leq2\sqrt{\epsilon_{j,l}^{\mathrm{in}}}+\epsilon_{j,l}^{\mathrm{in}}\xrightarrow{l\to\infty}0.
\end{equation}

For the initial projection, define

\begin{align}
	&\epsilon_{A_j,d}\defvar\Tr{\left(I_{A_j}-\widebar P_j^{(d)}\right)\omega_{A_j}}, \\ &\epsilon_{E_0,d}\defvar\Tr{\left(I_{E_0}-\widebar Q_0^{(d)}\right)\omega_{E_0}}.
\end{align}

By the tensor-product definition of $\Pi_0^{(d)}$, see \cref{eq:proj-chan-on-Aj-Ej}, we have the telescoping identity
\begin{align}
	\Pi_0^{(d)}(\omega) - \omega
	&=\left(\bigotimes_{i=0}^{n-1}\Pi_{A_i}^{(d)}\right) \left[ \Pi_{E_0}^{(d)}(\omega) - \omega \right] \\
	\notag&\quad+\sum_{j=0}^{n-1}\left(\bigotimes_{i=j+1}^{n-1}\Pi_{A_i}^{(d)}\right)\left[\Pi_{A_j}^{(d)}(\omega)-\omega\right],
\end{align}
where the empty tensor product is understood as the identity channel. Therefore, applying the triangle inequality and the data-processing inequality for the trace distance under CPTP maps gives 
\begin{equation}
	\norm{\Pi_0^{(d)}(\omega)-\omega}_1\leq\norm{\Pi_{E_0}^{(d)}(\omega)-\omega}_1+\sum_{j=0}^{n-1}\norm{\Pi_{A_j}^{(d)}(\omega)-\omega}_1.
\end{equation}
Applying \cref{lem:truncation_convergence} once more, yields
\begin{equation}\label{eq:proof_lim_init}
    \begin{aligned}
	\norm{\omega-\Pi_0^{(d)}(\omega)}_1
	&\leq \begin{aligned}[t]
		&\sum_{j=0}^{n-1}\left(2\sqrt{\epsilon_{A_j,d}}+\epsilon_{A_j,d}\right) \\
	 	&+ 2\sqrt{\epsilon_{E_0,d}} + \epsilon_{E_0,d}
	 \end{aligned} \\
	&\xrightarrow{d\to\infty}0.
\end{aligned}
\end{equation}

Then, for a single round $j$, the triangle inequality and the data processing inequality give
\begin{equation}
    \begin{aligned}
	&\norm{\EATchannProjInf_j^{(k,l)}(\rho)-\EATchann_j(\rho)}_1 \\
	&\leq\norm{\ProjOut{j}{k}\!\left(\EATchann_j(\ProjIn{j}{l}(\rho))\right)-\ProjOut{j}{k}\!\left(\EATchann_j(\rho)\right)}_1 \\
	&\quad+\norm{\ProjOut{j}{k}\!\left(\EATchann_j(\rho)\right)-\EATchann_j(\rho)}_1 \\
	&\leq\norm{\ProjIn{j}{l}(\rho)-\rho}_1+\norm{\ProjOut{j}{k}\!\left(\EATchann_j(\rho)\right)-\EATchann_j(\rho)}_1.
\end{aligned}
\end{equation}
Combining this bound with \eqref{eq:proof_lim_out} and \eqref{eq:proof_lim_in}, one obtains
\begin{equation}\label{eq:single_round_bound}
	\lim_{k\to\infty}\lim_{l\to\infty}\norm{\EATchannProjInf_j^{(k,l)}(\rho)-\EATchann_j(\rho)}_1=0.
\end{equation}

Next, we bound the total trace distance for the concatenated channels applied to the initial state $\omega \defvar \omega_{A_0^{n-1} E_0}$. Using a telescoping sum we find
\begin{widetext}
\begin{equation}
    \begin{aligned}
    & \norm{ \EATchannProjInf_n^{(k_n, l_n)} \circ \dots \circ \EATchannProjInf_1^{(k_1, l_1)} (\Pi_0^{(d)}(\omega)) - \EATchann_n \circ \dots \circ \EATchann_1 (\omega) }_1 \\
    &= \begin{aligned}[t]
    \Big\| &\EATchannProjInf_n^{(k_n, l_n)} \circ \dots \circ \EATchannProjInf_1^{(k_1, l_1)} (\Pi_0^{(d)}(\omega)) - \EATchann_n \circ \EATchannProjInf_{n-1}^{(k_{n-1}, l_{n-1})} \circ \dots \circ \EATchannProjInf_1^{(k_1, l_1)} (\Pi_0^{(d)}(\omega)) \\
    &+ \EATchann_n \circ \EATchannProjInf_{n-1}^{(k_{n-1}, l_{n-1})} \circ \dots \circ \EATchannProjInf_1^{(k_1, l_1)} (\Pi_0^{(d)}(\omega)) - \dots \\
    &+ \EATchann_n \circ \dots \circ \EATchann_2 \circ \EATchannProjInf_1^{(k_1, l_1)} (\Pi_0^{(d)}(\omega)) - \EATchann_n \circ \dots \circ \EATchann_1 (\Pi_0^{(d)}(\omega)) \\
    &+ \EATchann_n \circ \dots \circ \EATchann_1 (\Pi_0^{(d)}(\omega)) - \EATchann_n \circ \dots \circ \EATchann_1 (\omega) \Big\|_1
    \end{aligned} \\
    &\le \norm{ \EATchannProjInf_n^{(k_n, l_n)} \circ \dots \circ \EATchannProjInf_1^{(k_1, l_1)} (\Pi_0^{(d)}(\omega)) - \EATchann_n \circ \EATchannProjInf_{n-1}^{(k_{n-1}, l_{n-1})} \circ \dots \circ \EATchannProjInf_1^{(k_1, l_1)} (\Pi_0^{(d)}(\omega)) }_1 \\
    &\quad + \dots + \norm{ \EATchann_n \circ \dots \circ \EATchann_1 (\Pi_0^{(d)}(\omega)) - \EATchann_n \circ \dots \circ \EATchann_1 (\omega) }_1
\end{aligned}
\end{equation}

Because the subsequent original channels $\EATchann_j$ are CPTP, we apply the data processing inequality to remove them from each term and find
\begin{equation}
    \begin{aligned}
    &\norm{ \EATchannProjInf_n^{(k_n, l_n)} \circ \dots \circ \EATchannProjInf_1^{(k_1, l_1)} (\Pi_0^{(d)}(\omega)) - \EATchann_n \circ \dots \circ \EATchann_1 (\omega) }_1 \\
    &\le \norm{ \EATchannProjInf_n^{(k_n,l_n)}(\rho_{n-1}) - \EATchann_n(\rho_{n-1}) }_1 + \dots + \norm{ \EATchannProjInf_1^{(k_1,l_1)}(\rho_0) - \EATchann_1(\rho_0) }_1 + \norm{ \Pi_0^{(d)}(\omega) - \omega }_1,
\end{aligned}
\end{equation}
where we defined the intermediate states $\rho_0 \defvar \Pi_0^{(d)}(\omega)$ and $\rho_{j-1} \defvar \EATchannProjInf_{j-1}^{(k_{j-1}, l_{j-1})} \circ \dots \circ \EATchannProjInf_1^{(k_1, l_1)} ( \rho_0 )$.

Taking the sequential limits $\lim_{d \to \infty} \lim_{k_1 \to \infty} \dots \lim_{l_n \to \infty}$, we evaluate the sum from right to left. Due to Eqs.~\eqref{eq:proof_lim_out}, \eqref{eq:proof_lim_in}, and \eqref{eq:proof_lim_init}, the trace distance for each individual projection step vanishes and thus, the full trace distance goes to zero
\begin{equation}
    \lim_{d \to \infty} \lim_{\substack{\mbf{k} \rightarrow \infty \\ \mbf{l}\rightarrow \infty}}  \norm{
    \EATchannProjInf_n^{(k_n, l_n)} \circ \dots \circ \EATchannProjInf_1^{(k_1, l_1)} (\Pi_0^{(d)}(\omega))
    - \EATchann_n \circ \dots \circ \EATchann_1 (\omega) }_1
    = 0,
\end{equation}
which proves \cref{eq:Convergence_of_Projected_Channels}. 

Next, because the trace distance between the unconditional states vanishes and the output state conditioned on $\Omega$ has a fixed, strictly positive probability $p_\Omega \in (0,1]$, the trace distance between the states conditioned on $\Omega$ must equally vanish:
\begin{equation}
    \lim_{d \to \infty} \lim_{\substack{\mbf{k} \rightarrow \infty \\ \mbf{l}\rightarrow \infty}} \norm{ \left( \EATchannProjInf_n^{(k_n, l_n)} \circ \dots \circ \EATchannProjInf_1^{(k_1, l_1)} (\Pi_0^{(d)}(\omega)) \right)_{|\Omega} - \EATchann_n \circ \dots \circ \EATchann_1 (\omega)_{|\Omega} }_1 = 0,
\end{equation}
which proves the version of \cref{eq:Convergence_of_Projected_Channels} conditioned on event $\Omega$.
\end{widetext}
The statements regarding the entropies in \cref{eq:Convergence_of_Projected_Channels_f_entropies,eq:Convergence_of_Projected_Channels_entropies} then follow directly from \cref{thrm:Continuity_f_weighted,thrm:Continuity_Renyi}, respectively.
\end{proof}

\subsection{\texorpdfstring{$f$}{f}-weighted entropy accumulation in infinite dimensions}
In this section, we prove that $f$-weighted entropy may accumulate with infinite-dimensional side-information. We begin by stating the theorem.
\fWeightedMEATInf*

To prove this theorem, we will relate the previously defined finite-dimensional projected channels to the infinite-dimensional channels via \cref{lem:Convergence_of_Projected_Channels}. As a first step,  we show in the next subsection that the corresponding finite-dimensional normalization constants are, in the limit, lower bounded by the infinite-dimensional quantities $\kappa_{\cP_1^{j-1}}$.

\subsubsection{Preliminary Lemmata and Theorems}
In preparation for the proof let us define the following sets
\begin{align}
	&\begin{aligned}
    &\widebar{\Sigma}_j(k_{j-1},k_j,l_j,d) \\
    &\defvar \begin{aligned}[t]
        \Big\{ &\left(\left(\EATchannProj_j^{(k_j,l_j)}\circ\Xi_{j-1}\right)\otimes\id_R\right)[\tau] \mid \\
    	&\tau \in S_{=}(A_{j-1}^{(d)} E_{j-1}^{(k_{j-1})} R), \tau_{A_{j-1}^{(d)}} = \xi_{A_{j-1}^{(d)}}^{(j-1)} \Big\},
    	\end{aligned}
    \end{aligned} \label{eq:def_Sigma_bar} \\
    &\begin{aligned}
    &\widetilde{\Sigma}_j(k_j,l_j,d)\\
    &\defvar \begin{aligned}[t]
    	\Big\{ &\left(\EATchannProjInf_j^{(k_j,l_j)}\otimes\id_{\widetilde E}\right)[\omega]
    	\Big| \\
    	&\omega\in S_{=}(A_{j-1}E_{j-1}\widetilde E), \omega_{A_{j-1}}=\widebar{\sigma}_{A_{j-1}}^{(j-1)},\\
    	&\omega=
    	\left(\widebar P_{j-1}^{(l_j)}\otimes\widebar Q_{j-1}^{(l_j)}
    	\otimes I_{\widetilde E}\right)
    	\omega
    	\left(\widebar P_{j-1}^{(l_j)}\otimes\widebar Q_{j-1}^{(l_j)}
    	\otimes I_{\widetilde E}\right)
    	\Big\}.
    \end{aligned}
\end{aligned} \label{eq:def_Sigma_tilde}
\end{align}

In both cases $R$ ($\widetilde{E}$) is a register of large enough dimension to serve as a purifying register for any of the $A_{j-1}^{(d)}E_{j-1}^{(k_{j-1})}$ ($A_{j-1}E_{j-1}$) registers. Additionally, each $\xi_{A_{j-1}^{(d)}}^{(j-1)}$ satisfies
\begin{equation}
    V_{j-1}^{(d)} \xi_{A_{j-1}^{(d)}}^{(j-1)} V_{j-1}^{(d)\dagger} = \widebar{\sigma}_{A_{j-1}}^{(j-1)},
\end{equation}
where $\widebar{\sigma}_{A_{j-1}}^{(j-1)}$ are defined as
\begin{equation}
	\widebar{\sigma}_{A_{j-1}}^{(j-1)} = \pTr{E_0A_0^{j-2}A_j^{n-1}}{\Pi_0^{(d)}\left(\omega_{A_0^{n-1}E_0}\right)}.
\end{equation}
Finally, we note, that the sets $\widebar{\Sigma}_j(k_{j-1},k_j,l_j,d)$ are defined entirely on finite dimensional spaces whereas the sets $\widetilde{\Sigma}_j\left(k_j,l_j,d\right)$ are defined in the infinite dimensional spaces.

Using these sets, we define the following optimization problems
\begin{align}
    \widebar{\kappa}_{\cP_1^{j-1}} &\defvar \inf_{\nu\in\widebar{\Sigma}_j(k_{j-1},k_j,l_j,d)} H^{\uparrow, f_{|\cP_1^{j-1}}}_{\alpha}(S_j| \CP_j E_{j}^{(k_{j})} R)_{\nu}, \label{eq:def_kappa_bar} \\
    \widetilde{\kappa}_{\cP_1^{j-1}} &\defvar \inf_{\nu\in\wt{\Sigma}_j\left(k_j,l_j,d\right)} H^{\uparrow, f_{|\cP_1^{j-1}}}_{\alpha}(S_j| \CP_j E_j \widetilde{E})_{\nu}.
\end{align}
For notational simplicity, we leave the dependence of $\widebar{\kappa}_{\cP_1^{j-1}}$ on $(k_{j-1},k_j,l_j,d)$ and that of $\widetilde{\kappa}_{\cP_1^{j-1}}$ on $(k_j,l_j,d)$ implicit.

\begin{lemma}[Infinite-Space Lower Bound]\label{lem:kappa_inf-space_lower_bnd}
	Let $d,k_{j-1},k_j,l_j$ satisfy $l_j\geq d$ and $l_j\geq k_{j-1}$ and $k_0 =d$. Then for all $ j =1,\dots,n$, it holds
	\begin{equation}
		\widebar{\kappa}_{\cP_1^{j-1}} \geq \widetilde{\kappa}_{\cP_1^{j-1}}.
	\end{equation}
\end{lemma}
\begin{proof}
    We simply show that any feasible point of $\widebar{\Sigma}_{j}$ can be converted to a state in $\widetilde{\Sigma}_{j}$ via an isometry on part of the registered being conditioned upon. This is enough to show minimizing the conditional entropy over $\wt{\Sigma}_{j}$ can only result in a lower value than minimizing over $\widebar{\Sigma}_{j}$.

	Let $\nu \in \widebar{\Sigma}_j $ be arbitrary. Then, there exists $\tau \in S_{=}(A_{j-1}^{(d)} E_{j-1}^{(k_{j-1})} R)$ with $\tau_{A_{j-1}^{(d)}} = \xi_{A_{j-1}^{(d)}}^{(j-1)}$ such that  
	\begin{equation}
		\nu = \left(\left(\EATchannProj_j^{(k_j,l_j)}\circ\Xi_{j-1}\right)\otimes\id_R\right)[\tau].
	\end{equation}
	Since $R$ is a purifying register for the finite-dimensional registers $\{A_{j-1}^{(d)}E_{j-1}^{(k_{j-1})}\}_{j=1,\dots, n}$ and $\widetilde{E}$ is a purifying register for the infinite-dimensional registers $\{A_{j-1}E_{j-1}\}_{j=1,\dots, n}$, it holds $\dim(R) \leq \dim(\widetilde{E})$. Therefore, there exists an isometry $V_R:R\rightarrow\widetilde E$. 
	
	Using this isometry $V_R$, define the state $\omega \in \dop{=}(A_{j-1}E_{j-1}\widetilde{E})$ as
	\begin{equation}
		\omega \defvar \left(V_{j-1}^{(d)}\otimes W_{j-1}^{(k_{j-1})} \otimes V_R\right) \tau  \left(V_{j-1}^{(d)}\otimes W_{j-1}^{(k_{j-1})} \otimes V_R \right)^\dagger.
	\end{equation}
	Its marginal on $A_{j-1}$ satisfies
    \begin{equation}
	\begin{aligned}
		\omega_{A_{j-1}}
		&=V_{j-1}^{(d)}\tau_{A_{j-1}^{(d)}}V_{j-1}^{(d)\dagger} \\
		&=V_{j-1}^{(d)}\xi_{A_{j-1}^{(d)}}^{(j-1)}V_{j-1}^{(d)\dagger} \\
		&=\widebar{\sigma}_{A_{j-1}}^{(j-1)}.
	\end{aligned}       
    \end{equation}
	Moreover, since $l_j\geq d$ and $l_j\geq k_{j-1}$, we have
	\begin{equation}
    \begin{aligned}
		&\widebar P_{j-1}^{(l_j)}V_{j-1}^{(d)}=V_{j-1}^{(d)}, \\
		&\widebar Q_{j-1}^{(l_j)}W_{j-1}^{(k_{j-1})}=W_{j-1}^{(k_{j-1})},
	\end{aligned}
    \end{equation}
	and thus,
	\begin{equation}
		\omega = \left(\widebar P_{j-1}^{(l_j)}\otimes\widebar Q_{j-1}^{(l_j)}\otimes I_{\widetilde E}\right)
		\omega \left(\widebar P_{j-1}^{(l_j)}\otimes\widebar Q_{j-1}^{(l_j)}\otimes I_{\widetilde E}\right).
	\end{equation}
	
	Therefore, the state $\widetilde{\nu}$ defined as
	\begin{equation}
		\widetilde{\nu} \defvar \left(\EATchannProjInf_j^{(k_j,l_j)}\otimes\id_{\widetilde E}\right)[\omega],
	\end{equation}
	is an element of $\widetilde{\Sigma}_j$. Starting from the definition of $\widetilde{\nu}$, a simple calculation applying \cref{Lem:Embedding_Channels,Lem:Single_Round_Equiv} shows that the two output states $\nu$ and $\widetilde{\nu}$ satisfy
	\begin{equation}
		\widetilde{\nu} = \left(W_j^{(k_j)}\otimes V_R\right) \nu \left(W_j^{(k_j)}\otimes V_R\right)^\dagger.
	\end{equation}
	Therefore, by the isometric invariance of the $f$-weighted \Renyi entropy \cite[Lemma~4.4]{Arqand_2025},
	\begin{equation}
		H^{\uparrow,f_{|\cP_1^{j-1}}}_{\alpha}(S_j|\CP_jE_j^{(k_j)}R)_{\nu} = H^{\uparrow,f_{|\cP_1^{j-1}}}_{\alpha}(S_j|\CP_jE_j\widetilde E)_{\widetilde{\nu}}.
	\end{equation}
	Hence, for every feasible point in $\widebar{\Sigma}_j$ we can construct a feasible point in $\widetilde{\Sigma}_j$ with the same objective value. Taking the infimum yields
	\begin{equation}
		\widebar{\kappa}_{\cP_1^{j-1}} \geq \widetilde{\kappa}_{\cP_1^{j-1}},
	\end{equation}
	which completes the proof.
\end{proof}

As the previous lemma shows we cannot use $\widebar{\Sigma}_{j}$, we need a different approach. To this end, we define another optimization problem as
\begin{equation}
    \hat{\kappa}_{\cP_1^{j-1}} \defvar \inf_{\nu\in\hat{\Sigma}_j\left(d\right)} H^{\uparrow, f_{|\cP_1^{j-1}}}_{\alpha}(S_j| \CP_j E_j \wt{E})_{\nu},
\end{equation}
where the sets $\hat{\Sigma}_j\left(d\right)$ are defined as
\begin{equation}
    \begin{aligned}[t]
    \hat{\Sigma}_j\left(d\right) \defvar \Big\{ &\left( \EATchann_j \otimes \id_{\wt{E}}\right)\left[\omega\right] \mid \\ & \omega \in S_{=}(A_{j-1} E_{j-1} \wt{E}), \omega_{A_{j-1}} = \widebar{\sigma}_{A_{j-1}}^{(j-1)} \Big\}.
    \end{aligned}
\end{equation}
The marginal states $\widebar{\sigma}_{A_{j-1}}^{(j-1)}$ are again defined as
\begin{equation}
	\widebar{\sigma}_{A_{j-1}}^{(j-1)} = \pTr{E_0A_0^{j-2}A_j^{n-1}}{\Pi_0^{(d)}\left(\omega_{A_0^{n-1}E_0}\right)}.
\end{equation}

\begin{lemma}[Intermediate Limit Bound]\label{lem:kappa_intermediate_limit_bnd}
	For every $d,k_j,l_j$ satisfying $l_j\geq d$, it holds
	\begin{equation}
	    \widetilde{\kappa}_{\cP_1^{j-1}} \geq \hat{\kappa}_{\cP_1^{j-1}}.
	\end{equation}
\end{lemma}
\begin{proof}
    First, note that for all $k_j,l_j$ satisfying $l_j\geq d$,
    \begin{align}
    &\begin{aligned}
    	&\widetilde{\Sigma}_j(k_j,l_j,d)\\
    	&= \begin{aligned}[t]
    		\Big\{ &\left(\EATchannProjInf_j^{(k_j,l_j)}\otimes\id_{\widetilde E}\right)[\omega]
    		\Big| \\
    		&\omega\in S_{=}(A_{j-1}E_{j-1}\widetilde E), \omega_{A_{j-1}}=\widebar{\sigma}_{A_{j-1}}^{(j-1)},\\
    		&\omega=
    		\left(\widebar P_{j-1}^{(l_j)}\otimes\widebar Q_{j-1}^{(l_j)}
    		\otimes I_{\widetilde E}\right)
    		\omega
    		\left(\widebar P_{j-1}^{(l_j)}\otimes\widebar Q_{j-1}^{(l_j)}
    		\otimes I_{\widetilde E}\right)
    		\Big\}
    	\end{aligned}
    \end{aligned} \\
    &\subseteq \begin{aligned}[t]
        \bigg\{ &\nu \in \mathrm{S}_=(S_j \CP_j E_j \wt{E}) \;\bigg|\; \exists \omega \in \mathrm{S}_=(A_{j-1}E_{j-1}\wt{E}), \\
        &\left(\left(\ProjOut{j}{k_j} \circ \EATchann_j \right) \otimes \id_{\wt{E}} \right)\left[\omega\right] = \nu, \omega_{A_{j-1}} = \widebar{\sigma}_{A_{j-1}}^{(j-1)} \bigg\}
    \end{aligned}
    \end{align}
    where the inclusion follows because for every input state $\omega$ occurring in the definition of
    $\widetilde{\Sigma}_j(k_j,l_j,d)$, the support condition implies
    \begin{equation}
    	\left(\ProjIn{j}{l_j}\otimes\id_{\widetilde E}\right)[\omega]=\omega.
    \end{equation}
    Hence, we find
    \begin{equation}
        \widetilde{\Sigma}_j(k_j,l_j,d) \subseteq \ProjOut{j}{k_j}\left( \hat{\Sigma}_j\left(d\right) \right),
    \end{equation}
    and therefore
    \begin{align}
        \widetilde{\kappa}_{\cP_1^{j-1}} &= \inf_{\nu\in\widetilde{\Sigma}_j\left(k_j,l_j,d\right)} H^{\uparrow, f_{|\cP_1^{j-1}}}_{\alpha}(S_j| \CP_j E_{j} \wt{E})_{\nu} \\
        &\geq \inf_{\nu\in \ProjOut{j}{k_j}\left( \hat{\Sigma}_j\left(d\right) \right)} H^{\uparrow, f_{|\cP_1^{j-1}}}_{\alpha}(S_j| \CP_j E_j \wt{E})_{\nu} \\
        &= \inf_{\nu\in \hat{\Sigma}_j\left(d\right)} H^{\uparrow, f_{|\cP_1^{j-1}}}_{\alpha}(S_j| \CP_j E_j \wt{E})_{\ProjOut{j}{k_j}\left( \nu \right)}.
    \end{align}
    Then, by the data-processing inequality of the $f$-weighted \Renyi entropy \cite[Lemma~4.4]{Arqand_2025}, which by \cref{app:inf-dim-DPI} also holds for infinite dimensional side-information as $\vert S_{j} \vert < +\infty$, we find
    \begin{equation}
    \begin{aligned}
    	&\inf_{\nu\in \hat{\Sigma}_j\left(d\right)} H^{\uparrow, f_{|\cP_1^{j-1}}}_{\alpha}(S_j| \CP_j E_j \wt{E})_{\ProjOut{j}{k_j}\left( \nu \right)} \\
    	&\geq \inf_{\nu\in \hat{\Sigma}_j\left(d\right)} H^{\uparrow, f_{|\cP_1^{j-1}}}_{\alpha}(S_j| \CP_j E_j \wt{E})_{\nu} =  \hat{\kappa}_{\cP_1^{j-1}}.
    \end{aligned}
    \end{equation}
    Thus, we have in summary
    \begin{equation}
    	\widetilde{\kappa}_{\cP_1^{j-1}} \geq \hat{\kappa}_{\cP_1^{j-1}},
    \end{equation}
    which concludes the proof.
\end{proof}

Let us define one more optimization problem as
\begin{equation}
    \kappa_{\cP_1^{j-1}} \defvar \inf_{\nu\in\Sigma_j} H^{\uparrow, f_{|\cP_1^{j-1}}}_{\alpha}(S_j| \CP_j E_j \wt{E})_{\nu},
\end{equation}
where the sets $\Sigma_j$ are defined as
\begin{equation}
    \begin{aligned}[t]
    \Sigma_j \defvar \Big\{ &\left( \EATchann_j \otimes \id_{\wt{E}}\right)\left[\omega\right] \mid \\ & \omega \in S_{=}(A_{j-1} E_{j-1} \wt{E}), \omega_{A_{j-1}} = \sigma_{A_{j-1}}^{(j-1)} \Big\}.
    \end{aligned}
\end{equation}

\begin{lemma}[Marginal Limit Bound]\label{lem:kappa_marginal_limit_bnd}
    The limit as $d \rightarrow \infty$ of the optimization problem $\hat{\kappa}_{\cP_1^{j-1}}$ is equal to the optimization problem $\kappa_{\cP_1^{j-1}}$, i.e. it holds:
	\begin{equation}
		\lim_{d\to\infty}\hat{\kappa}_{\cP_1^{j-1}}=\kappa_{\cP_1^{j-1}}.
	\end{equation}
\end{lemma}
\begin{proof}
	Since $\Pi_0^{(d)}(\omega) \to \omega$ in trace norm, taking the marginal on $A_{j-1}$ gives
	\begin{equation}
		\epsilon_d \defvar \frac{1}{2} \norm{\widebar{\sigma}_{A_{j-1}}^{(j-1)} - \sigma_{A_{j-1}}^{(j-1)}}_1 \xrightarrow{d\to\infty} 0.
	\end{equation}
    For any arbitrary $\delta > 0$, let $\hat{\omega}^{(d)} \in S_{=}(A_{j-1} E_{j-1} \wt{E})$ with marginal $\hat{\omega}^{(d)}_{A_{j-1}} = \widebar{\sigma}_{A_{j-1}}^{(j-1)}$ be a state whose image $\hat{\nu}^{(d)} = (\EATchann_j \otimes \id_{\wt{E}})[\hat{\omega}^{(d)}]$ achieves the infimum of $\hat{\kappa}_{\cP_1^{j-1}}$ up to $\delta$, that is
    \begin{equation}
        H^{\uparrow, f_{|\cP_1^{j-1}}}_{\alpha}(S_j| \CP_j E_j \wt{E})_{\hat{\nu}^{(d)}} \leq \hat{\kappa}_{\cP_1^{j-1}} + \delta.
    \end{equation}

    By Uhlmann's theorem, see \cite{Hou_2012} for infinite dimensions, and the Fuchs--van de Graaf inequalities, there exists a joint state $\omega^{(d)} \in S_{=}(A_{j-1} E_{j-1} \wt{E})$ that satisfies the exact marginal constraint $\omega^{(d)}_{A_{j-1}} = \sigma_{A_{j-1}}^{(j-1)}$ such that its distance to $\hat{\omega}^{(d)}$ is bounded by the function $g(\epsilon_d) =\sqrt{2\epsilon_d - \epsilon_d^2}$ which vanishes as $\epsilon_d \to 0$. That is,
    \begin{equation}
        \frac{1}{2} \norm{\omega^{(d)}_{A_{j-1} E_{j-1} \wt{E}} - \hat{\omega}^{(d)}_{A_{j-1} E_{j-1} \wt{E}} }_1 \leq g(\epsilon_d).
    \end{equation}
    %Maybe add more steps later

    Applying the channel, we define $\nu^{(d)} \defvar (\EATchann_j \otimes \id_{\wt{E}})[\omega^{(d)}]$. By the data-processing inequality of the trace distance, the images, $\hat{\nu}^{(d)}$ and $\nu^{(d)}$ are also close
    \begin{equation}
        \frac{1}{2} \norm{ \nu^{(d)} - \hat{\nu}^{(d)} }_1 \leq \frac{1}{2} \norm{\omega^{(d)} - \hat{\omega}^{(d)} }_1 \leq g(\epsilon_d).
    \end{equation}

    Crucially, because $\omega^{(d)}_{A_{j-1}} = \sigma_{A_{j-1}}^{(j-1)}$, the state $\nu^{(d)}$ is a valid element of the set $\Sigma_j$. By definition, it therefore provides an upper bound for the target optimization problem $\kappa_{\cP_1^{j-1}}$, that is
    \begin{equation}\label{eq:kappa_bound_nud}
        \kappa_{\cP_1^{j-1}} \leq H^{\uparrow, f_{|\cP_1^{j-1}}}_{\alpha}(S_j| \CP_j E_j \wt{E})_{\nu^{(d)}}.
    \end{equation}

    By \cref{thrm:Continuity_f_weighted}, the $f$-weighted conditional \Renyi entropy is uniformly continuous. Hence, the difference of the $f$-weighted \Renyi entropies of $\hat{\nu}^{(d)}$ and $\nu^{(d)}$ is bounded as
    \begin{equation}
    \begin{split}
        \Big| &H^{\uparrow, f_{|\cP_1^{j-1}}}_{\alpha}(S_j| \CP_j E_j \wt{E})_{\nu^{(d)}} \\
        &- H^{\uparrow, f_{|\cP_1^{j-1}}}_{\alpha}(S_j| \CP_j E_j \wt{E})_{\hat{\nu}^{(d)}} \Big| \leq \Delta_{\alpha,f_{|\cP_1^{j-1}}}(g(\epsilon_d)),
        \end{split}
    \end{equation}
    where $\Delta_{\alpha,f_{|\cP_1^{j-1}}}$ is defined as \cref{eq:correction_fweighted_cont} of \cref{thrm:Continuity_f_weighted} and $\lim_{d \to \infty} \Delta_{\alpha,f_{|\cP_1^{j-1}}}(g(\epsilon_d)) = 0$ since $g(\epsilon_d) \to 0$ for $d\to \infty$.

    Using this continuity bound alongside \cref{eq:kappa_bound_nud}, we find
    \begin{equation}
    \begin{aligned}
        \kappa_{\cP_1^{j-1}} &\leq H^{\uparrow, f_{|\cP_1^{j-1}}}_{\alpha}(S_j| \CP_j E_j \wt{E})_{\nu^{(d)}} \\
        &\leq H^{\uparrow, f_{|\cP_1^{j-1}}}_{\alpha}(S_j| \CP_j E_j \wt{E})_{\hat{\nu}^{(d)}} + \Delta_{\alpha,f_{|\cP_1^{j-1}}}(g(\epsilon_d)) \\
        &\leq \hat{\kappa}_{\cP_1^{j-1}} + \delta + \Delta_{\alpha,f_{|\cP_1^{j-1}}}(g(\epsilon_d)).
    \end{aligned}       
    \end{equation}

    Since $\delta > 0$ was chosen arbitrarily, it follows that
    \begin{equation}
    	\kappa_{\cP_1^{j-1}} \leq \hat{\kappa}_{\cP_1^{j-1}} + \Delta_{\alpha,f_{|\cP_1^{j-1}}}(g(\epsilon_d)).
    \end{equation}
    Repeating the same argument with the roles of the two marginal constraints swapped, now starting from a state whose image achieves $\kappa_{\cP_1^{j-1}}$ up to $\delta$, yields the reverse inequality
    \begin{equation}
    	\kappa_{\cP_1^{j-1}} \geq \hat{\kappa}_{\cP_1^{j-1}} - \Delta_{\alpha,f_{|\cP_1^{j-1}}}(g(\epsilon_d)),
    \end{equation}
    and thus,
    \begin{equation}
    	\abs{ \kappa_{\cP_1^{j-1}} - \hat{\kappa}_{\cP_1^{j-1}} } \leq \Delta_{\alpha,f_{|\cP_1^{j-1}}}(g(\epsilon_d)).
    \end{equation}
    
    Taking the limit $d \to \infty$ on both sides causes the correction term $\Delta_{\alpha,f_{|\cP_1^{j-1}}}(g(\epsilon_d))$ to vanish, yielding the final statement:
    \begin{equation}
    	\lim_{d\to\infty}\hat{\kappa}_{\cP_1^{j-1}}=\kappa_{\cP_1^{j-1}}.
    \end{equation}
\end{proof}

\subsubsection{Proof of Theorem\texorpdfstring{~\ref{thrm:fweighted_meat_simp}}{ 9}}\label{app:Proof_of_f_weighted_MEAT}
\begin{proof}[Proof of Theorem~\ref{thrm:fweighted_meat_simp}]
After establishing the previous lemmata, the proof follows by combining \cref{eq:Convergence_of_Projected_Channels_f_entropies,thrm:Finite_MEAT_Chains}, applying the finite-dimensional $f$-weighted entropy accumulation theorem, and then using \cref{lem:kappa_inf-space_lower_bnd,lem:kappa_intermediate_limit_bnd,lem:kappa_marginal_limit_bnd}.

Throughout this proof, the cutoff limits are understood in the sense of \cref{eq:shorthand_iterated_lim}; in particular, we set $k_0=d$ and restrict the cutoff parameters to $l_j\geq d$ and $l_j\geq k_{j-1}$ for every $j=1,\dots,n$. Therefore, the assumptions of \cref{thrm:Finite_MEAT_Chains,lem:kappa_inf-space_lower_bnd,lem:kappa_intermediate_limit_bnd} are satisfied throughout the proof.

\begin{widetext}
By \cref{eq:Convergence_of_Projected_Channels_f_entropies,thrm:Finite_MEAT_Chains}, the conditional \Renyi entropy of the final state $\rho$ created by the infinite-dimensional channels acting on the infinite-dimensional state $\omega_{A_0^{n-1}E_0}$ is equal to
\begin{equation}
\begin{aligned}
    H_\alpha^{\uparrow,f_\mathrm{full}}(S_1^n|\CP_1^nE_n)_\rho 
    &= \lim_{d \to \infty} \lim_{\substack{\mbf{k} \rightarrow \infty \\ \mbf{l}\rightarrow \infty}} H_\alpha^{\uparrow,f_\mathrm{full}}(S_1^n | \CP_1^n E_n)_{W_n^{(k_n)} \left( (\EATchannProj_n^{(k_n, l_n)} \circ \Xi_{n-1}) \circ \dots \circ (\EATchannProj_1^{(k_1, l_1)} \circ \Xi_{0}) [\tau_{(A^{(d)})_{0}^{n-1}E_0^{(d)}}] \right) \left(W_n^{(k_n)}\right)^{\dagger}} \\
    &=\lim_{d \to \infty} \lim_{\substack{\mbf{k} \rightarrow \infty \\ \mbf{l}\rightarrow \infty}} H_\alpha^{\uparrow,f_\mathrm{full}}(S_1^n | \CP_1^n E_n^{(k_n)})_{\left( (\EATchannProj_n^{(k_n, l_n)} \circ \Xi_{n-1}) \circ \dots \circ (\EATchannProj_1^{(k_1, l_1)} \circ \Xi_{0}) \left[\tau_{\left(A^{(d)}\right)_{0}^{n-1}E_0^{(d)}}\right] \right) },
\end{aligned}
\end{equation}
where for the second equality we used the isometric invariance of the $f$-weighted \Renyi entropy \cite[Lemma~4.4]{Arqand_2025}. The state $\tau_{\left(A^{(d)}\right)_{0}^{n-1}E_0^{(d)}}$ now satisfies the following $d$-dependent marginal,
\begin{equation}
    \tau_{\left(A^{(d)}\right)_{0}^{n-1}} = \bigotimes_{j=1}^n \xi_{A_{j-1}^{(d)}}^{(j-1)},
\end{equation}
where each $\xi_{A_{j-1}^{(d)}}^{(j-1)}$ satisfies $V_{j-1}^{(d)} \xi_{A_{j-1}^{(d)}}^{(j-1)} V_{j-1}^{(d)\dagger} = \widebar{\sigma}_{A_{j-1}}^{(j-1)},$ and
$\widebar{\sigma}_{A_{j-1}}^{(j-1)} = \pTr{E_0A_0^{j-2} A_j^{n-1}}{\Pi_0^{(d)}\left(\omega_{A_0^{n-1}E_0}\right)}.$

Crucially, $H_\alpha^{\uparrow,f_\mathrm{full}}(S_1^n | \CP_1^n E_n^{(k_n)})$ involves only finite-dimensional states and registers for each $d, \mbf{k}$ and  $\mbf{l}$. Therefore, we can apply the finite-dimensional marginal $f$-weighted entropy accumulation theorem \cite[Theorem~4.1a]{Arqand_2025}, which yields

\begin{align}
     H_\alpha^{\uparrow,f_\mathrm{full}}(S_1^n | \CP_1^n E_n^{(k_n)})_{\left( (\EATchannProj_n^{(k_n, l_n)} \circ \Xi_{n-1}) \circ \dots \circ (\EATchannProj_1^{(k_1, l_1)} \circ \Xi_{0}) \left[\tau_{\left(A^{(d)}\right)_{0}^{n-1}E_0^{(d)}}\right] \right) }
     \geq \sum_{j=1}^n \min_{\cP_1^{j-1}} \widebar{\kappa}_{\cP_1^{j-1}},
\end{align}
where $\widebar{\kappa}_{\cP_1^{j-1}}$ is defined as in \cref{eq:def_kappa_bar}
\begin{equation}
    \widebar{\kappa}_{\cP_1^{j-1}} = \inf_{\nu\in\widebar{\Sigma}_j(k_{j-1},k_j,l_j,d)} H^{\uparrow, f_{|\cP_1^{j-1}}}_{\alpha}(S_j| \CP_j E_{j}^{(k_{j})} R)_{\nu},
\end{equation}
and $\widebar{\Sigma}_j(k_{j-1},k_j,l_j,d)$ is defined as in \cref{eq:def_Sigma_bar}
\begin{align}
\widebar{\Sigma}_j(k_{j-1},k_j,l_j,d) =
    \Big\{ \left( \left( \EATchannProj_j^{(k_j, l_j)} \circ \Xi_{j-1} \right) \otimes \id_R \right) [\tau] \mid \tau \in S_{=}(A_{j-1}^{(d)} E_{j-1}^{(k_{j-1})} R), \tau_{A_{j-1}^{(d)}} = \xi_{A_{j-1}^{(d)}}^{(j-1)} \Big\}.
\end{align}
For every fixed admissible choice of $d,\mbf{k},\mbf{l}$ and every $\cP_1^{j-1}$, \cref{lem:kappa_inf-space_lower_bnd,lem:kappa_intermediate_limit_bnd} give $\widebar{\kappa}_{\cP_1^{j-1}}
	\geq\widetilde{\kappa}_{\cP_1^{j-1}}
	\geq\hat{\kappa}_{\cP_1^{j-1}}.$
    
Thus, the finite-dimensional entropy-accumulation bound implies, for every such fixed choice of cutoff parameters,
\begin{equation}
	H_\alpha^{\uparrow,f_\mathrm{full}}(S_1^n | \CP_1^n E_n^{(k_n)})_{\left( (\EATchannProj_n^{(k_n, l_n)} \circ \Xi_{n-1}) \circ \dots \circ (\EATchannProj_1^{(k_1, l_1)} \circ \Xi_{0}) \left[\tau_{\left(A^{(d)}\right)_{0}^{n-1}E_0^{(d)}}\right] \right) } 
	\geq \sum_{j=1}^n\min_{\cP_1^{j-1}}\hat{\kappa}_{\cP_1^{j-1}}.
\end{equation}
\end{widetext}
Taking the cutoff limits $\mbf{k},\mbf{l} \to \infty$ yields
\begin{equation}
	H_\alpha^{\uparrow,f_\mathrm{full}}(S_1^n|\CP_1^nE_n)_\rho
	\geq\lim_{d\to\infty}\sum_{j=1}^n
	\min_{\cP_1^{j-1}}\hat{\kappa}_{\cP_1^{j-1}}.
\end{equation}
By \cref{lem:kappa_marginal_limit_bnd}, for each $\cP_{1}^{j-1}$ and $j \in \{1,...,n\}$, the limit over $d$ exists. Since both the sum and each minimum range over finite sets, we can interchange them with the limit. Thus, we find
\begin{equation}
\begin{aligned}
	H_\alpha^{\uparrow,f_\mathrm{full}}(S_1^n|\CP_1^nE_n)_\rho
	&\geq \sum_{j=1}^n \min_{\cP_1^{j-1}} \lim_{d\to\infty}\hat{\kappa}_{\cP_1^{j-1}} \\
	&= \sum_{j=1}^n \min_{\cP_1^{j-1}}\kappa_{\cP_1^{j-1}},
\end{aligned}    
\end{equation}
which proves the first statement of the theorem in \cref{eq:qes eat_simp}.

Finally, by the normalization property of the $f$-weighted \Renyi entropy \cite[Lemma~4.9]{Arqand_2025}, for every $\cP_1^{j-1}$ we have
\begin{equation}
	\inf_{\nu\in\Sigma_j}\widetilde{H}_{\alpha}^{\uparrow,\hat{f}_{|\cP_1^{j-1}}}(S_j|\CP_jE_j\widetilde E)_\nu
	=\kappa_{\cP_1^{j-1}}-\kappa_{\cP_1^{j-1}}=0.
\end{equation}
Applying \cref{eq:qes eat_simp}, which we just proved to the tradeoff functions $\hat{f}_{|\cP_1^{j-1}}$ therefore yields
\begin{equation}
	\widetilde{H}^{\uparrow,\hat{f}_\mathrm{full}}(S_1^n|\CP_1^nE_n)_\rho\geq0,
\end{equation}
which proves the statement regarding the ``normalized'' tradeoff function in \cref{eq:chainQESnorm_simp}.
\end{proof}

\subsection{Infinite-Dimensional MEAT Proof}
We now state the main technical theorem, which we prove in this subsection.
\MEATInf*

At a technical level, to establish this theorem, our aim is to prove the strong duality of a function with its Lagrange dual. The function and its Lagrange dual follows from \cite{Arqand_2025}. What is new is the proof of strong duality without relying upon the compactness of a set $\Theta$ that will be defined later. To prove this strong duality, we will use the Fenchel-Moreau theorem, which relies upon the lower semi-continuity of a relevant perturbation function. Thus, the section is structured by proving some preliminary lemmata, including the lower semi-continuity of this aforementioned perturbation function, then proving the strong duality in \cref{thrm:convduality}, and finally proving the infinite-dimensional MEAT.

\subsubsection{Preliminary Lemmata}\label{app:subsec_prelim_inf_MEAT}
First, for reference, we define the Legendre--Fenchel convex conjugate as in \cite{Boyd_2004}
\begin{definition}[Convex conjugate]\label{def:conjugate}
	For a function $F:D\to\mathbb{R}\cup\{-\infty,+\infty\}$ where $D$ is a convex subset of $\mathbb{R}^k$, its \emph{convex conjugate} is the function $F^*:\mathbb{R}^k\to\mathbb{R}\cup\{-\infty,+\infty\}$ given by
	\begin{align}
		F^*(\mbf{y}) \defvar \sup_{\mbf{x}\in D} \left(\mbf{x}\cdot\mbf{y} - F(\mbf{x})\right).
	\end{align}
\end{definition}

Now, in preparation for the lower-semicontinuity result we prove the following limiting property of the relative entropy. 

\begin{lemma}[One-sided continuity of classical relative entropy]
\label{lem:one-sided-relative-entropy-continuity}
Let $\CP$ be a finite classical register with alphabet $\alphCP$,
$\lambda\in\mathbb{P}_{\CP}$ and let $\{\lambda_n\}_{n\in\N}, \{p_n\}_{n\in \N} \subseteq \mathbb{P}_{\CP}$ be sequences that satisfy
\begin{equation}
    \lambda_n \xrightarrow{n\to \infty} \lambda,
\end{equation}
and
\begin{equation}
    \limsup_{n\to\infty}D(\lambda_n\Vert p_n)<\infty.
\end{equation}
Then
\begin{equation}
    \limsup_{n\to\infty} \left[ 
        D(\lambda\Vert p_n) - D(\lambda_n\Vert p_n)
    \right] \leq 0.
\end{equation}
\end{lemma}

\begin{proof}
The assumption $\limsup_{n\to\infty}D(\lambda_n\Vert p_n)<\infty$ implies that $D(\lambda_n\Vert p_n)<\infty$ for all sufficiently large $n$. Moreover, since $\lambda_n(\cP)\to\lambda(\cP)>0$ for every $\cP\in\supp(\lambda)$, this implies that $p_n(\cP)>0$ on $\supp(\lambda)$ for all sufficiently large $n$, and hence $D(\lambda\Vert p_n)<\infty$. Since discarding finitely many initial terms does not affect the limit superior, we restrict to such $n$.
Hence, we rewrite the relative entropies as
\begin{equation}
\begin{aligned}
    &D(\lambda\Vert p_n) - D(\lambda_n\Vert p_n) \\
    = &H(\lambda_n) - H(\lambda) + \sum_{\cP} \left(\lambda_n(\cP) -\lambda(\cP) \right) \log(p_n(\cP)).
\end{aligned}
\end{equation}
The Shannon entropy is continuous and therefore the first part of the expression $H(\lambda_n) - H(\lambda)$ satisfies
\begin{equation}
    \limsup_n \left( H(\lambda_n) - H(\lambda) \right) = 0.
\end{equation}
For the second part, consider an arbitrary $\cP \in \supp(\lambda)$. We show that, for every $\cP\in \supp(\lambda)$, the sequence $\{p_n(\cP)\}_n$ is bounded away from zero. Since $\lambda_n(\cP)\to\lambda(\cP)>0$, there exist $\delta_{\cP}>0$ and $N_{\cP} \in \N$ such that
\begin{equation}
    \lambda_n(\cP)\geq\delta_{\cP}
\end{equation}
for every $n\geq N_{\cP}$. Next, consider the stochastic map $T_{\cP}:\CP\to\{0,1\}$ that records whether the outcome is $\cP$, so that
\begin{equation}
    T_{\cP}(r) = \left(r(\cP),1-r(\cP)\right)
\end{equation}
for every $r\in\mathbb{P}_{\CP}$. Then, by the data-processing inequality for classical relative entropy,
\begin{equation}
\begin{aligned}
    D(\lambda_n\Vert p_n) &\geq D(T_{\cP}(\lambda_n) \Vert T_{\cP}(p_n)) \\
    &= d\bigl(\lambda_n(\cP)\Vert p_n(\cP)\bigr),
\end{aligned}
\end{equation}
where
\begin{equation}
    d(x\Vert y) \defvar x\log\frac{x}{y} + (1-x)\log\frac{1-x}{1-y}
\end{equation}
denotes binary relative entropy.

If $p_n(\cP)$ were not bounded away from zero, there would be a
subsequence along which $p_n(\cP)\to0$. Along this subsequence,
$\lambda_n(\cP)\geq\delta_{\cP}$, and hence
\begin{equation}
    d\bigl(\lambda_n(\cP)\Vert p_n(\cP)\bigr)
    \longrightarrow+\infty.
\end{equation}
This would contradict
\begin{equation}
    \limsup_{n\to\infty}D(\lambda_n\Vert p_n)<\infty.
\end{equation}
Consequently, for every $\cP\in \supp(\lambda)$, there exist $\eta_{\cP}>0$ and
$N_{\cP}\in\mathbb{N}$ such that
\begin{equation}
    p_n(\cP)\geq\eta_{\cP}
\end{equation}
for all $n\geq N_{\cP}$. Therefore, since $\lambda_n \to \lambda$ it follows
\begin{equation}
    \limsup_{n\to \infty} \left(\lambda_n(\cP) -\lambda(\cP) \right) \log(p_n(\cP)) = 0,
\end{equation}
and the sum over $\cP \in \supp(\lambda)$ goes to zero.

Finally, consider $\cP \notin \supp(\lambda)$. Here, $\lambda(\cP) = 0$ and the second term becomes
\begin{equation}
    \sum_{\cP \notin \supp(\lambda)} \lambda_n(\cP) \log p_n(\cP).
\end{equation}
Because $\lambda_n(\cP) \geq 0$ and $p_n(\cP) \leq 1$ and thus $\log(p_n(\cP)) \leq 0$, every term in this sum is non-positive. Hence, it follows for all $\cP \notin \supp(\lambda)$,
\begin{equation}
    \limsup_{n\to \infty} \left(\lambda_n(\cP) -\lambda(\cP) \right) \log(p_n(\cP)) \leq 0.
\end{equation}
Combining both cases of the second term with the first term we therefore reach the lemma statement.
\end{proof}

For $\alpha \in (1,2)$ and $\rho \in S_{=}(Q\CP Q')$, let us define the following pair of functions on $\R^{|\alphCP|}$
\begin{align}\label{eq:GandGstar_f}
    \begin{aligned}
        &G_{\alpha,\rho}({\mbf{f}}) \coloneqq - H^{\uparrow ,f}_\alpha(Q|\CP Q')_{\rho}, \\
        &G_{\alpha,\rho}^*(\bsym{\lambda}) \coloneqq 
        \begin{cases}
            \begin{aligned}
                &\frac{\alpha}{\alpha-1}D\left(\bsym{\lambda} \middle\Vert \bsym{\rho}_{\CP}\right) \\ &+\sum\limits_{\cP\in\supp(\bsym{\rho}_{\CP})}\lambda(\cP)\renyiSandUp_{\alpha}(Q|Q')_{\rho_{|\cP}}
            \end{aligned} &
            \bsym{\lambda} \in \mathbb{P}_{\CP} \; ,\\
            +\infty  & \text{else,}
        \end{cases}
    \end{aligned}
\end{align}
which by \cite[Lemma~4.11]{Arqand_2025} are convex conjugates. The calculation in the proof of \cite[Lemma~4.11]{Arqand_2025} remains valid when $Q'$ is infinite-dimensional and separable, since the decomposition with respect to the finite classical register $\CP$ contains only finitely many terms and the register $Q$ remains finite-dimensional.

Furthermore, we define a perturbation function based on $G_{\alpha,\rho}$, and then prove it is lower semicontinuous. This will be critical for our aforementioned application of the Fenchel-Moreau theorem.

\begin{definition}[Perturbation function for $G_{\alpha,\rho}$]\label{def:perturbation_func}
    Let $G_{\alpha,\rho}$ be defined as in \cref{eq:GandGstar_f} and let $\mathcal{R} \subseteq \dop{=}(QQ'\CP)$ with finite-dimensional spaces $Q$ and $\CP$ be given as
    \begin{equation}
        \mathcal{R} \subseteq \left\{\rho \in \dop{=}(QQ'\CP) \mid \rho \text{ is classical on } \CP \right\}.
    \end{equation}
    The space $Q'$ may be infinite-dimensional. Furthermore, define
    \begin{equation}
        \mathcal{Z} \defvar \left\{\mbf{z} \in \R^{|\alphCP|} \mid \sum_{\cP} z(\cP) =0 \right\},
    \end{equation}
    and let $S_{\Omega} \subseteq \mathbb{P}_{\CP}$ be compact.
    
    Then, we define the perturbation function $V: \R^{|\alphCP|} \to (-\infty,+\infty]$ associated with the family $\{G_{\alpha,\rho}\}_{\rho\in \mathcal{R}}$ as
    \begin{equation}\label{eq:perturbation_func_V}
        V(\mbf{z}) \defvar \begin{cases}
            \inf_{\substack{\mbf{q} \in S_{\Omega}, \rho \in \mathcal{R} \\ \lambda \in \mathbb{P}_{\CP}, \\ \mbf{q} - \bsym{\lambda} = \mbf{z}}} G_{\alpha,\rho}^*(\bsym{\lambda}) & \mbf{z} \in \mathcal{Z}, \\
            +\infty & \text{else}.
        \end{cases}
    \end{equation}
\end{definition}

\begin{lemma}[Lower semicontinuity of Perturbation Function]\label{lem:LSC_perturbation_func}
    The perturbation function $V$ as defined in \cref{def:perturbation_func} is lower semicontinuous, that is for any sequence $\{\mbf{z}_n\}_{n\in \N} \subseteq \R^{|\alphCP|}$ with $\mbf{z}_n \xrightarrow{n \to \infty} \mbf{z}$ it holds
    \begin{equation}
        V(\mbf{z}) \leq \liminf_{n \to \infty} V(\mbf{z}_n).
    \end{equation}
\end{lemma}
\begin{proof}
Let $\{\mbf{z}_n\}_{n\in\N} \subseteq\R^{|\alphCP|}$ be a sequence satisfying
\begin{equation}
    \mbf{z}_n \xrightarrow{n\to\infty} \mbf{z},
\end{equation}
and set
\begin{equation}
    v \defvar \liminf_{n\to\infty}V(\mbf{z}_n).
\end{equation}
If $v=+\infty$, there is nothing to prove, and therefore, we assume $v < \infty $. Moreover, it holds $v>-\infty$, since
\begin{equation}
    G_{\alpha,\rho}^*(\bsym{\lambda}) \geq -\log\dim(Q) > - \infty,
\end{equation}
for all $\rho \in \mathcal{R}$ since $Q$ is finite-dimensional.

Next, choose a subsequence $\{\mbf{z}_{n_k}\}_k$ of $\{\mbf{z}_n\}_n$ achieving the limit inferior, that is
\begin{equation}
    \lim_{k \to \infty} V(\mbf{z}_{n_k})= v.
\end{equation}
Since $v\in\R$, the subsequence can be chosen such that $V(\mbf{z}_{n_k})<+\infty$ for all $k\in\N$. By the definition of $V$, this implies that $\mbf{z}_{n_k}\in\mathcal{Z}$ for all $k\in\N$. Furthermore, since $\mbf{z}_n$ converges to $\mbf{z}$, so does the subsequence ${\mbf{z}_{n_k}}_k$ and the closedness of $\mathcal{Z}$ therefore
implies that $\mbf{z}\in\mathcal{Z}$.

Now, for every $k\in\N$, choose approximate optimizers $\mbf{q}_k\in S_{\Omega}$, $\rho^k\in\mathcal{R}$, and $\bsym{\lambda}_k\in\mathbb{P}_{\CP}$ of 
\begin{equation}
    \inf_{\substack{q \in S_{\Omega}, \rho \in \mathcal{R} \\ \lambda \in \mathbb{P}_{\CP}, \\ \mbf{q} - \bsym{\lambda} = \mbf{z}_{n_k} }} G_{\alpha,\rho}^*(\bsym{\lambda}),
\end{equation}
such that
\begin{equation}
    \mbf{q}_k-\bsym{\lambda}_k = \mbf{z}_{n_k}
\end{equation}
and
\begin{equation}
    G_{\alpha,\rho^k}^*(\bsym{\lambda}_k) \leq V(\mbf{z}_{n_k})+\frac{1}{k},
\end{equation}
which can be chosen this way by the definition of an infimum. 

Since $S_{\Omega}$ is compact, there exists a convergent subsequence
$\{\mbf{q}_{k_j}\}_{j\in\N}$ with limit $\mbf{q}\in S_{\Omega}$, that is
\begin{equation}
    \mbf{q}_{k_j} \xrightarrow{j\to\infty} \mbf{q}\in S_{\Omega}.
\end{equation}
Since $\mbf{z}_n$ converges to $\mbf{z}$, also $\mbf{z}_{n_{k_j}}$ converges to $\mbf{z}$. Consequently,
\begin{equation}
\bsym{\lambda}_{k_j} = \mbf{q}_{k_j}-\mbf{z}_{n_{k_j}} \xrightarrow{j\to\infty} \mbf{q}-\mbf{z}.
\end{equation}
If we define $\bsym{\lambda} \defvar \mbf{q} -\mbf{z}$, we thus find
\begin{equation}
    \bsym{\lambda}_{k_j} \xrightarrow{j \to \infty} \bsym{\lambda},
\end{equation}
and since $\mathbb{P}_{\CP}$ is closed, it follows that $\bsym{\lambda}\in\mathbb{P}_{\CP}$. Furthermore, by the definition of $\bsym{\lambda}$, it trivially holds
\begin{equation}
    \mbf{q}-\bsym{\lambda}=\mbf{z}.
\end{equation}

For ease of notation, we relabel the extracted sequences
\begin{equation}
\begin{split}
    &\{\mbf{z}_{n_{k_j}}\}_{j\in\N},\quad \{\mbf{q}_{k_j}\}_{j\in\N},\\
    &\{\rho^{k_j}\}_{j\in\N} \quad \text{and} \quad \{\bsym{\lambda}_{k_j}\}_{j\in\N}
\end{split}
\end{equation}
as
\begin{equation}
\begin{split}
    &\{\mbf{z}_m\}_{m\in\N}, \quad \{\mbf{q}_m\}_{m\in\N}, \\
    &\{\rho^m\}_{m\in\N}, \quad \text{and} \quad \{\bsym{\lambda}_m\}_{m\in\N},
\end{split}
\end{equation}
respectively. Furthermore, we define
\begin{equation}
    \delta_m \defvar \frac{1}{k_m}.
\end{equation}
Then,
\begin{equation}
    \begin{split}
    V(\mbf{z}_m) &\xrightarrow{m\to\infty} v, \quad \mbf{z}_m \xrightarrow{m\to\infty} \mbf{z}, \\
    \mbf{q}_m &\xrightarrow{m\to\infty} \mbf{q}, \quad \bsym{\lambda}_m \xrightarrow{m\to\infty} \bsym{\lambda}.
    \end{split}
\end{equation}
Moreover, it holds
\begin{equation}
    \mbf{q}_m-\bsym{\lambda}_m=\mbf{z}_m
\end{equation}
for every $m\in\N$, and
\begin{equation}
    \mbf{q}-\bsym{\lambda}=\mbf{z}.
\end{equation}
Then, the approximate optimality condition becomes
\begin{equation}
    G_{\alpha,\rho^m}^*(\bsym{\lambda}_m) \leq V(\mbf{z}_m)+\delta_m,
\end{equation}
where
\begin{equation}
    \delta_m \xrightarrow{m\to\infty} 0.
\end{equation}

After this relabelling, we proceed with the proof. For every $m \in \N$, define
\begin{equation}
    \mbf{p}_m \defvar \bsym{\rho}^{m}_{\CP},
\end{equation}
and for each $\cP\in\supp(\mbf{p}_m)$, define
\begin{equation}
    h_m(\cP) \defvar \renyiSandUp_{\alpha} (Q|Q')_{\rho^m_{|\cP}},
\end{equation}
and set $h_m(\cP)=0$ for $\cP\notin\supp(\mbf{p}_m)$. This does not change the value of $G_{\alpha,\rho}^*(\bsym{\lambda_{m}})$ when the relative entropy term is finite.

The approximate optimality condition implies that $ G_{\alpha,\rho^m}^*(\bsym{\lambda}_m)$ is bounded from above.
Additionally, since $Q$ is finite-dimensional, for every $\cP\in\alphCP$ it holds
\begin{equation}
	-\log\dim(Q)\leq h_m(\cP)\leq\log\dim(Q).
\end{equation}
Therefore, we conclude
\begin{equation}
    \sum_{\cP\in\alphCP} \lambda_m(\cP)h_m(\cP) \geq -\log\dim(Q),
\end{equation}
from which it follows that
\begin{equation}
    \limsup_{m\to\infty} D(\bsym{\lambda}_m\Vert\mbf{p}_m) < \infty,
\end{equation}
as otherwise $ G_{\alpha,\rho^m}^*(\bsym{\lambda}_m)$ would not be bounded above by definition of the functional.

Therefore, all conditions of \cref{lem:one-sided-relative-entropy-continuity} are satisfied, and we find 
\begin{equation}\label{eq:limsup_rel_ent_diff}
    \limsup_{m\to\infty} \big[ D(\bsym{\lambda}\Vert\mbf{p}_m) - D(\bsym{\lambda}_m \Vert \mbf{p}_m) \big] \leq 0.
\end{equation}
In particular, $D(\bsym{\lambda}\Vert\mbf{p}_m)$ is finite for all sufficiently large $m$.

Moreover,
\begin{equation}
 \left| \sum_{\cP\in\alphCP} \bigl( \lambda(\cP)-\lambda_m(\cP) \bigr)h_m(\cP) \right| 
    \leq \log\dim(Q) \left\| \bsym{\lambda}-\bsym{\lambda}_m \right\|_1,
\end{equation}
which converges to $0$, because $\lambda_m \to \lambda$.

Combining this with the preceding bound in \cref{eq:limsup_rel_ent_diff}, it holds
\begin{equation}
\begin{split}
    \limsup_{m\to\infty} \left[ G_{\alpha,\rho^m}^*(\bsym{\lambda}) - G_{\alpha,\rho^m}^*(\bsym{\lambda}_m) \right] \leq 0.
\end{split}
\end{equation}

Now, let $\varepsilon>0$ be arbitrary. By the definition of the limit superior, there exists an $M_1 \in \N$ such that, for every $m \geq M_1$, it holds
\begin{equation}
\begin{split}
    \left[ G_{\alpha,\rho^m}^*(\bsym{\lambda}) - G_{\alpha,\rho^m}^*(\bsym{\lambda}_m) \right] \leq \frac{\varepsilon}{3},
\end{split}
\end{equation}
and thus, it holds equivalently
\begin{equation}
    G_{\alpha,\rho^m}^*(\bsym{\lambda}) \leq G_{\alpha,\rho^m}^*(\bsym{\lambda}_m) + \frac{\varepsilon}{3},
\end{equation}
for all $m \geq M_1$.

Furthermore, by the convergence $\delta_m \xrightarrow{m\to\infty} 0$ and the convergence $V(\mbf{z}_m)\xrightarrow{m\to\infty}v$, there exists an $M_2\in\N$ such that, for every $m\geq M_2$,
\begin{equation}
\begin{split}
\delta_m\leq\frac{\varepsilon}{3} \quad \text{and} \quad V(\mbf{z}_m)\leq v+\frac{\varepsilon}{3}.
\end{split}
\end{equation}
We note that, by construction, the triple $(\mbf{q},\rho^m,\bsym{\lambda})$ is feasible for $V(\mbf{z})$ and therefore, for every $m\geq \max(M_1,M_2)$,
\begin{equation}
\begin{split}
V(\mbf{z}) &\leq G_{\alpha,\rho^m}^*(\bsym{\lambda}) \leq G_{\alpha,\rho^m}^*(\bsym{\lambda}_m) +\frac{\varepsilon}{3}\\
&\leq V(\mbf{z}_m)+\delta_m +\frac{\varepsilon}{3} \leq v+\varepsilon.
\end{split}
\end{equation}

Since $\varepsilon$ was arbitrary we let $\varepsilon\downarrow0$, which gives
\begin{equation}
    V(\mbf{z}) \leq v = \liminf_{n\to\infty}V(\mbf{z}_n).
\end{equation}
Hence, $V$ is lower semicontinuous.
\end{proof}

In preparation of the final theorem in this section and the next lemma we require the following definition of the purifying function.
\begin{definition}[Purifying function (adapted from {\cite[Defn.~2.10]{Arqand_2025}})]\label{def:purify}
	For registers $Q,Q'$ with $\dim(Q) \leq \dim(Q')$, a \emph{purifying function for $Q$ onto $Q'$} is a function $\pf: \dop{\leq}(Q) \to \dop{\leq}(QQ')$ such that for any state $\rho_Q$, the state $\pf(\rho_Q)$ is a purification of $\rho_Q$ onto the register $Q'$, i.e.~a (possibly subnormalized) rank-$1$ operator such that $\pTr{Q'}{\pf(\rho_Q)} = \rho_Q$.
\end{definition}

\begin{theorem}\label{thrm:convduality}
Let $\EATchann\in\CPTP(Q,SE\CP)$ where the output is always classical on $\CP$, let $\widetilde{E}$ be a register such that $\dim(Q) \leq \dim(\widetilde{E})$, and let $\pf$ be a purifying function for $Q$ onto $\widetilde{E}$ (Definition~\ref{def:purify}). 
Furthermore, let $\dim(S) < \infty$ and $\dim(\CP)<\infty$, however, $Q$ and $E$ may be chosen infinite-dimensional.

Let $\mathbb{P}_{\CP}$ denote the set of probability distributions on $\CP$, and consider any $S_\Omega \subseteq \mathbb{P}_{\CP}$ and any $\Theta \subseteq \dop{=}(Q)$. For any $\omega\in\dop{=}(Q)$, define a corresponding state $\nu^\omega \defvar \EATchann\left[ \pf\left(\omega\right)\right]$. Then, for $\alpha \in (1,\infty)$
the optimization 
\begin{align}\label{eq:dualf}
\sup_{\mbf{f} \in \R^{|\alphCP|}} \inf_{\mbf{q} \in S_\Omega}
\inf_{\omega\in\Theta}
\left(\frenyiSandUp_{\alpha}(S | \CP E\widetilde{E})_{\nu^\omega} + \mbf{f}\cdot\mbf{q} \right),
\end{align}
is the Lagrange dual problem of the following constrained optimization:
\begin{align}\label{eq:primalomega}
\begin{aligned}[t]
\inf_{\mbf{q} \in S_\Omega} \inf_{\bsym{\lambda} \in \mathbb{P}_{\CP}} &\inf_{\omega\in\Theta}
\begin{aligned}[t]
    \Biggl( &\frac{\alpha}{{\alpha}-1}D\left(\bsym{\lambda} \middle\Vert \bsym{\nu}^\omega_{\CP}\right) \\
    &+ \sum_{\cP\in\supp(\bsym{\nu}^\omega_{\CP})}\lambda(\cP)\renyiSandUp_{\alpha}(S|E\widetilde{E})_{\nu^\omega_{|\cP}} \Biggr)
\end{aligned} \\
&\text{s.t.}\quad \mbf{q}-\bsym{\lambda} = \mbf{0}
\end{aligned},
\end{align}
in which the objective function is jointly convex in $\bsym{\lambda}$ and $\omega$ (and $\mbf{q}$, trivially).

Furthermore, if $S_\Omega$ is nonempty, compact and convex, $\Theta$ is convex, and the primal admits a finite feasible point, i.e. there exist $\mbf{q}\in S_{\Omega}$ and $\omega \in \Theta$ such that
\begin{equation}
    D\left(\mbf{q} \middle\Vert \bsym{\nu}^\omega_{\CP}\right) < +\infty,
\end{equation}
then the primal and dual values are equal.
\end{theorem}

\begin{remark}
    As in \cite[Lemma~4.12]{Arqand_2025}, the preceding theorem establishes equality of the primal and dual values but does not guarantee attainment of the dual value. Indeed, the application of the Fenchel-Moreau theorem does not imply that the supremum in the dual problem is attained. Unlike the Clark-Duffin argument used
    in \cite[Lemma~4.12]{Arqand_2025}, our argument does not require compactness of $\Theta$. Such compactness cannot in general be inferred from boundedness in infinite dimensions, since the set of density operators is not compact in the trace-norm topology.
\end{remark}

\begin{proof}
    The fact that both optimization problems are Lagrange dual to each other follows from \cite[Lemma~4.12]{Arqand_2025}. Hence, the main difficulty of the proof of this theorem is the strong duality in the case of a convex set $S_{\Omega}$ and feasible points.

    To prove the strong duality, we consider the perturbation function $V$ associated with the family $\{G_{\alpha,\nu^{\omega}}\}_{\omega \in \Theta}$, corresponding to the choice
    \begin{equation}
        \mathcal{R} = \{\nu^{\omega} \mid \omega \in \Theta, \nu^{\omega} = \EATchann\left[\pf\left(\omega\right)\right] \}.
    \end{equation}
    We will apply the Fenchel-Moreau theorem \cite[Theorem~13.37]{Bauschke_2017}, which states that under certain conditions a function $f$ is equal to its biconjugate, $f^{**}$. As we will see below, the primal problem is equal to $V(\mbf{0})$, while the dual problem is equal to $V^{**}(\mbf{0})$.

    Hence, we calculate $V^{**}$ for states $\nu^{\omega}$. Since $\mbf{q}\in S_{\Omega}\subseteq\mathbb{P}_{\CP}$ and $\bsym{\lambda}\in\mathbb{P}_{\CP}$, it holds
    \begin{equation}
        \sum_{\cP\in\alphCP}\bigl(q(\cP)-\lambda(\cP)\bigr)=0.
    \end{equation}
    Consequently, the constraint $\mbf{q}-\bsym{\lambda}=\mbf{z}$ has no feasible points whenever $\mbf{z}\notin\mathcal{Z}$. Using the convention that the infimum over the empty set is $+\infty$, the infimum in the following calculation agrees with the definition of $V$ on all of $\R^{|\alphCP|}$. By the definition of the Fenchel conjugate and inserting the definition of the perturbation function from \cref{def:perturbation_func}, we find
    \begin{equation}
    \begin{aligned}
        V^*(\mbf{u}) 
        &= \sup_{\mbf{z} \in \R^{|\alphCP|}} \left\{ \langle\mbf{u}, \mbf{z}\rangle - V(\mbf{z}) \right\} \\
        &= \sup_{\mbf{z} \in \R^{|\alphCP|}} \left\{ \langle\mbf{u}, \mbf{z}\rangle - \inf_{\substack{\mbf{q} \in S_{\Omega}, \omega \in \Theta \\ \bsym{\lambda} \in \mathbb{P}_{\CP}, \\ \mbf{q} - \bsym{\lambda} = \mbf{z}}} G_{\alpha,\nu^{\omega}}^*(\bsym{\lambda}) \right\} \\
        &= \sup_{\substack{\mbf{z} \in \R^{|\alphCP|}, \mbf{q} \in S_{\Omega}, \\ \omega \in \Theta, \bsym{\lambda} \in \mathbb{P}_{\CP}, \\ \mbf{q} - \bsym{\lambda} = \mbf{z}}} \left\{ \langle\mbf{u}, \mbf{z}\rangle - G_{\alpha,\nu^{\omega}}^*(\bsym{\lambda}) \right\}.
    \end{aligned}    
    \end{equation}
    Next, by substituting $\mbf{z} = \mbf{q} - \bsym{\lambda}$ into the objective function, we can evaluate the supremum over $\mbf{z}$ directly and find
    \begin{equation}
    \begin{aligned}
        V^*(\mbf{u})
        &= \sup_{\substack{\mbf{q} \in S_{\Omega}, \\ \omega \in \Theta, \bsym{\lambda} \in \mathbb{P}_{\CP}}}
        \left\{ \langle\mbf{u}, \mbf{q} - \bsym{\lambda}\rangle - G_{\alpha,\nu^{\omega}}^*(\bsym{\lambda}) \right\} \\
        &= \sup_{\mbf{q} \in S_{\Omega}} \langle\mbf{u}, \mbf{q}\rangle 
        + \sup_{\substack{\omega \in \Theta, \\ \bsym{\lambda} \in \mathbb{P}_{\CP}}} 
        \left\{ \langle\mbf{u}, -\bsym{\lambda}\rangle - G_{\alpha,\nu^{\omega}}^*(\bsym{\lambda}) 
        \right\} \\
        &= \sup_{\mbf{q} \in S_{\Omega} } \langle\mbf{u}, \mbf{q}\rangle + 
        \sup_{\omega \in \Theta} \left( 
            \sup_{\bsym{\lambda} \in \mathbb{P}_{\CP}} 
            \left\{ 
                \langle-\mbf{u}, \bsym{\lambda}\rangle - G_{\alpha,\nu^{\omega}}^*(\bsym{\lambda}) 
            \right\} 
        \right)
    \end{aligned}      
    \end{equation}
    Now, we notice that the inner supremum over $\bsym{\lambda}$ is equal to $G_{\alpha,\nu^\omega}^{**}(-\mbf{u})$, since $G_{\alpha,\nu^\omega}^*(\bsym{\lambda})=+\infty$ for every $\bsym{\lambda}\notin\mathbb{P}_{\CP}$. By \cite[Proof~of~Lemma~4.11]{Arqand_2025}, $G_{\alpha,\nu^\omega}$ can be written as a log-sum-exponential and is therefore finite-valued, convex, and continuous on $\R^{|\alphCP|}$. In particular, it is proper and lower semicontinuous. The Fenchel-Moreau theorem therefore gives
    \begin{equation}\label{eq:proof_G=G**}
        G_{\alpha,\nu^\omega}=G_{\alpha,\nu^\omega}^{**}
    \end{equation}
    for every $\omega\in\Theta$. Thus, we can simplify the previous expression for $V^*$ and find
    \begin{equation}
        V^*(\mbf{u}) 
        = \sup_{\mbf{q} \in S_{\Omega} } \langle\mbf{u}, \mbf{q}\rangle + \sup_{\omega \in \Theta} G_{\alpha,\nu^{\omega}}(-\mbf{u}).
    \end{equation}
    Next, we calculate $V^{**}$. Again by the definition of the Fenchel conjugate we find
    \begin{equation}
    \begin{aligned}
        V^{**}(\mbf{y}) 
        &= \sup_{\mbf{u} \in \R^{|\alphCP|}}  \left\{  \langle\mbf{y}, \mbf{u}\rangle - V^*(\mbf{u}) \right\}  \\
        &= \sup_{\mbf{u} \in \R^{|\alphCP|}} 
        \left\{ 
            \langle\mbf{y}, \mbf{u}\rangle 
            - 
            \sup_{\substack{\mbf{q} \in S_{\Omega}, \\ \omega \in \Theta } }
            \left( 
                \langle\mbf{u}, \mbf{q}\rangle + G_{\alpha,\nu^{\omega}}(-\mbf{u}) 
            \right) 
        \right\} \\
        &= \sup_{\mbf{u}\in \R^{|\alphCP|}} \inf_{\substack{\mbf{q} \in S_{\Omega}, \\ \omega \in \Theta }}
        \left\{ 
            \langle\mbf{y}, \mbf{u}\rangle - \langle\mbf{u}, \mbf{q}\rangle - G_{\alpha,\nu^{\omega}}(-\mbf{u}) 
        \right\} \\
        &= \sup_{\mbf{u}\in \R^{|\alphCP|}} \inf_{\substack{\mbf{q} \in S_{\Omega},\\ \omega \in \Theta }} 
        \left\{ 
            \langle\mbf{u}, \mbf{y}-\mbf{q}\rangle - G_{\alpha,\nu^{\omega}}(-\mbf{u}) 
        \right\},
    \end{aligned}
\end{equation}
    where in step three, we applied $-\sup(x) = \inf(-x)$ and in step four, we simply redistributed the inner products. Substituting $-\mbf{f} \coloneqq \mbf{u}$, the final expression for $V^{**}$ becomes
    \begin{align}\label{eq:proof_V**}
        V^{**}(\mbf{y}) &= 
        \sup_{\mbf{f}\in \R^{|\alphCP|}} \inf_{\substack{\mbf{q} \in S_{\Omega},\\ \omega \in \Theta }} 
        \left\{ 
            \langle\mbf{f}, \mbf{q}-\mbf{y}\rangle - G_{\alpha,\nu^{\omega}}(\mbf{f}) 
        \right\}.
    \end{align}

    With the formulation of $V^{**}$ at hand, we can insert $G_{\alpha,\nu^{\omega}}$, which by \cref{eq:GandGstar_f} is defined as
    \begin{equation}
        G_{\alpha,\nu^{\omega}}(\mbf{f}) = -\frenyiSandUp_{\alpha}(S|\CP E\widetilde{E})_{\nu^{\omega}},
    \end{equation}
    and find the equivalent formulation of $V^{**}$
    \begin{equation}
        V^{**}(\mbf{y}) = 
        \sup_{\mbf{f}\in \R^{|\alphCP|}} \inf_{\substack{\mbf{q} \in S_{\Omega},\\ \omega \in \Theta }}
        \left\{ 
            \langle\mbf{f}, \mbf{q}-\mbf{y}\rangle + \frenyiSandUp_{\alpha}(S|\CP E\widetilde{E})_{\nu^{\omega}}
        \right\}.
    \end{equation}
    
    Then, we observe the following
    \begin{equation}
        V^{**}(\mbf{0}) = 
        \sup_{\mbf{f}\in \R^{|\alphCP|}} \inf_{\substack{\mbf{q} \in S_{\Omega},\\ \omega \in \Theta }} 
        \left\{ 
            \frenyiSandUp_{\alpha}(S|\CP E\widetilde{E})_{\nu^{\omega}} + \langle\mbf{f}, \mbf{q}\rangle
        \right\},
    \end{equation}
    which is the dual problem as defined in \cref{eq:dualf}. On the other hand, after inserting the definition of $G_{\alpha,\nu^{\omega}}^*$ as in \cref{eq:GandGstar_f}, $V(\mbf{0})$ is given by
    \begin{align}
        \begin{aligned}[t]
        V(\mbf{0}) = &\begin{aligned}[t]&\inf_{\substack{\mbf{q} \in S_\Omega, \omega\in\Theta, \\ \bsym{\lambda} \in \mathbb{P}_{\CP}}}
            \Biggl( \frac{\alpha}{{\alpha}-1}D\left(\bsym{\lambda} \middle\Vert \bsym{\nu}^\omega_{\CP}\right) \\
            &\qquad + \sum_{\cP\in\supp(\bsym{\nu}^\omega_{\CP})}\lambda(\cP)\renyiSandUp_{\alpha}(S|E\widetilde{E})_{\nu^\omega_{|\cP}} \Biggr).
        \end{aligned} \\
        &\text{s.t.}\quad \mbf{q}-\bsym{\lambda} = \mbf{0}
        \end{aligned}
    \end{align}
    Finally, to show strong duality we will apply the Fenchel-Moreau theorem to the perturbation function $V$. In order to do so, we need to show that $V$ is proper, lower semicontinuous and convex.

    By \cref{lem:LSC_perturbation_func}, $V$ is lower semicontinuous. Furthermore, by \cite[Lemma~5.4]{Arqand_2025b}, the map $(\bsym{\lambda},\omega)\mapsto G_{\alpha,\nu^\omega}^*(\bsym{\lambda})$ is jointly convex. The proof of this joint convexity statement carries over to infinite-dimensional separable Hilbert spaces, since purifications of the same state are related by an isometry and the data-processing inequality for the $f$-weighted entropy also holds for possibly infinite-dimensional conditioning registers. Additionally, since $S_{\Omega}$, $\Theta$, and $\mathbb{P}_{\CP}$ are convex sets, and the constraint $\mbf{q}-\bsym{\lambda}=\mbf{z}$ is affine, $V$ is obtained by minimizing a jointly convex function with respect to a subset of variables over a convex set. Therefore, $V$ is convex by \cite[Section~3.2.5]{Boyd_2004}.

    Lastly, since the register $S$ is finite-dimensional, for every $\omega\in\Theta$ and $\cP\in\supp(\bsym{\nu}^\omega_{\CP})$, the conditional sandwiched \Renyi entropy is finite and satisfies
    \begin{equation}
        \renyiSandUp_{\alpha}(S|E\widetilde{E})_{\nu^\omega_{|\cP}}\geq-\log\dim(S)>-\infty.
    \end{equation}
    Together with the non-negativity of the relative entropy, this implies
    \begin{align*}
        &G_{\alpha,\nu^\omega}^*(\bsym{\lambda})\geq-\log\dim(S) \quad \text{and} \\
        &V(\mbf{z})\geq -\log\dim(S),
    \end{align*}
    and hence, $V$ never takes the value $-\infty$. By assumption, there exist $\mbf{q}\in S_{\Omega}$ and $\omega\in\Theta$ such that $D(\mbf{q}\Vert\bsym{\nu}^\omega_{\CP})<+\infty$. Choosing $\bsym{\lambda}=\mbf{q}$ gives a feasible point for $V(\mbf{0})$, and the finiteness of the conditional sandwiched \Renyi entropies implies
    \begin{equation}
    V(\mbf{0})\leq G_{\alpha,\nu^\omega}^*(\mbf{q})<+\infty.
    \end{equation}
    Therefore, $V$ is proper and all conditions of the Fenchel-Moreau theorem \cite[Theorem~13.37]{Bauschke_2017} are satisfied. It follows
    \begin{equation}
        V = V^{**}.
    \end{equation}
    In particular, we find $V(\mbf{0}) = V^{**}(\mbf{0})$, which implies the theorem statement.
\end{proof}

\begin{lemma}[Existence of a finite feasible point]\label{lem:finite_feasible_point}
Consider the setting of \cref{thrm:MEAT_inf}. Let $A$ be a common register into which all $A_{j-1}$ can be isometrically embedded, and let $J$ be a classical register with alphabet $\{1,\dots,n\}$. Define the set
\begin{equation}
    \Theta_{AJE} \defvar \left\{ \omega\in\dop{=}(AEJ) \middle| \ \begin{array}{l}
    \omega \text{ is classical on } J, \\
    \omega_{AE|j} \in \dop{=}(A_{j-1}E_{j-1}), \\
    \omega_{A|j} = \sigma_{A_{j-1}}^{(j-1)}
    \text{ for every } j
    \end{array}
    \right\}.
\end{equation}

Furthermore, let $ \EATchann\in\CPTP(AEJ,S\CP E)$ be the $J$-controlled channel whose restriction to the outcome $J=j$ implements $\EATchann_j$ for every $j\in\{1,\dots,n\}$. Here, each $\EATchann_j$ is understood as an arbitrary CPTP extension to the common input register $AE$ that agrees with the original channel on the embedded subspace $A_{j-1}E_{j-1}$. 

Finally, enlarging $\widetilde{E}$ if necessary, let $\pf$ be a purifying function for $AEJ$ onto $\widetilde{E}$ and, for every $\omega\in\Theta_{AJE}$, define
\begin{equation}
    \nu^\omega \defvar \EATchann\left[\pf(\omega)\right].
\end{equation}
Then, there exist $\mbf{q}\in S_\Omega$ and $\omega^\star\in\Theta_{AJE}$ such that
\begin{equation}
    \supp(\mbf{q})\subseteq\supp\left(\bsym{\nu}_{\CP}^{\omega^\star}\right).
\end{equation}
Consequently,
\begin{equation}
    D\left(\mbf{q}\middle\Vert\bsym{\nu}_{\CP}^{\omega^\star}\right)<+\infty.
\end{equation}
\end{lemma}

\begin{proof}
Define $\mbf{q}\in\mathbb{P}_{\CP}$ by
\begin{equation}\label{eq:proof_def_feasible_q}
    \begin{split}
    q(\cP) &\defvar \sum_{\cP_1^n} \Pr_{\rho_{|\Omega}}[\cP_1^n] \freq_{\cP_1^n}(\cP) \\
    &= \frac{1}{n}\sum_{j=1}^n \Pr_{\rho_{|\Omega}}[\CP_j=\cP].
    \end{split}
\end{equation}
Since every frequency distribution appearing in this convex combination belongs to $S_\Omega$ and $S_\Omega$ is convex, it follows that $\mbf{q}\in S_\Omega$.

For every $j\in\{1,\dots,n\}$, define the input state of the channel $\EATchann_j$ by
\begin{equation}
    \eta^{(j-1)}_{A_{j-1}E_{j-1}} \defvar \pTr{\mathrm{all}\setminus A_{j-1}E_{j-1}}{\left(\EATchann_{j-1}\circ\cdots\circ\EATchann_1\right)\left[\omega_{A_0^{n-1}E_0}\right]},
\end{equation}
where $\mathrm{all}\setminus A_{j-1}E_{j-1}$ denotes the registers $S_1^{j-1}\CP_1^{j-1}A_j^{n-1}$, and the composition and any empty register products are interpreted trivially. Furthermore, since the preceding channels act trivially on $A_{j-1}$ and are trace preserving, it holds
\begin{equation}
\eta^{(j-1)}_{A_{j-1}} = \omega_{A_{j-1}} = \sigma_{A_{j-1}}^{(j-1)}.
\end{equation}

After applying the isometric embeddings into the common registers $A$ and $E$, define
\begin{equation}
    \omega^\star_{AEJ}
    \defvar
    \frac{1}{n}\sum_{j=1}^n
    \eta^{(j-1)}_{A_{j-1}E_{j-1}}
    \otimes
    \ket{j}\!\bra{j}_J.
\end{equation}
The state $\omega^\star$ is classical on $J$, for every $j\in\{1,\dots,n\}$ it satisfies
\begin{equation}
\begin{gathered}
    \omega^\star_{AE|j} = \eta^{(j-1)}_{A_{j-1}E_{j-1}} \in \dop{=}(A_{j-1}E_{j-1}), \\
    \omega^\star_{A|j} = \eta^{(j-1)}_{A_{j-1}} = \sigma_{A_{j-1}}^{(j-1)},
\end{gathered}
\end{equation}
and therefore
\begin{equation}
    \omega^\star\in\Theta_{AJE}. 
\end{equation}

Since $\pf(\omega^\star)$ has marginal $\omega^\star$ on $AEJ$, applying the controlled channel and tracing out all registers except $\CP$ gives
\begin{equation}
    \begin{split}
        \nu_{\CP}^{\omega^\star}
        &= \pTr{S E\widetilde{E} J}{\EATchann\left[\pf(\omega^\star)\right]} \\
        &= \pTr{S E J}{\EATchann\left[\omega^\star\right]} \\
        &= \frac{1}{n} \sum_{j=1}^n \pTr{S_jE_j}{\EATchann_j\left[\eta^{(j-1)}_{A_{j-1}E_{j-1}}\right]} \\
        &= \frac{1}{n} \sum_{j=1}^n \rho_{\CP_j}.
    \end{split}
\end{equation}

Hence, for every $\cP\in\alphCP$, it holds
\begin{equation}
    \bsym{\nu}_{\CP}^{\omega^\star}(\cP) = \frac{1}{n} \sum_{j=1}^n \Pr_{\rho}[\CP_j=\cP].
\end{equation}

We have thus now built our $\nu^{\omega^{\star}}$ and $\mbf{q}$. It remains to show $\supp(\mbf{q})\subseteq\supp\left(\bsym{\nu}_{\CP}^{\omega^\star}\right)$. To this end, we consider each $\cP \in \alphCP$. If $q(\cP) = 0$, then it is not relevant for showing $\supp(\mbf{q})\subseteq\supp\left(\bsym{\nu}_{\CP}^{\omega^\star}\right)$. 

Now, fix an arbitrary $\cP\in\alphCP$. We distinguish two cases. First, suppose that $q(\cP)=0$, which imposes no condition on $\bsym{\nu}_{\CP}^{\omega^\star}(\cP)$. It remains to consider the case $q(\cP)>0$. In this case, by the definition of $\mbf{q}$ in \cref{eq:proof_def_feasible_q}, there exists $j_0\in\{1,\dots,n\}$ such that
\begin{equation}
    \Pr_{\rho_{|\Omega}}[\CP_{j_0}=\cP] > 0.
\end{equation}
Then, using $p_\Omega>0$, we obtain
\begin{equation}
    \begin{split}
        \bsym{\nu}_{\CP}^{\omega^\star}(\cP)
        &\geq \frac{1}{n} \Pr_{\rho}[\CP_{j_0}=\cP] \\
        &\geq \frac{1}{n} \Pr_{\rho}[\CP_{j_0}=\cP \wedge \Omega] \\
        &= \frac{p_\Omega}{n} \Pr_{\rho_{|\Omega}}[\CP_{j_0}=\cP] > 0.
    \end{split}
\end{equation}
Since $\cP\in\alphCP$ was arbitrary, every $\cP$ satisfying $q(\cP)>0$ also satisfies $\bsym{\nu}_{\CP}^{\omega^\star}(\cP)>0$ and therefore,
\begin{equation}
    \supp(\mbf{q}) \subseteq \supp\left(\bsym{\nu}_{\CP}^{\omega^\star}\right).
\end{equation}
The alphabet $\alphCP$ is finite, and thus the support inclusion implies
\begin{equation}
    D\left(\mbf{q} \middle\Vert \bsym{\nu}_{\CP}^{\omega^\star}\right) < +\infty,
\end{equation}
which concludes the proof.
\end{proof}

\subsubsection{Proof of Theorem\texorpdfstring{~\ref{thrm:MEAT_inf}}{ 8}}\label{app:subsec_proof_inf_MEAT}
\begin{proof}[Proof of \cref{thrm:MEAT_inf}]
    The proof proceeds in two high-level steps. In the first step we will make use of the finite-dimensional projected states and channels to relate the conditional \Renyi entropy to the $f$-weighted \Renyi entropy just as in the finite-dimensional proof in \cite{Arqand_2025}. In the second step, we reformulate the problem and apply \cref{thrm:convduality} which establishes strong Lagrange duality even in infinite dimensions where compactness of the underlying sets is not guaranteed.
    
    As in the definition of the shorthand for the cutoff limits in \cref{eq:shorthand_iterated_lim}, we set $k_0=d$ and let all cutoff parameters occurring below satisfy $l_j\geq d$ and $l_j\geq k_{j-1}$ for every $j=1,\dots,n$.  This restriction is again valid because, for every fixed $d$ and $k_{j-1}$, it only removes finitely many initial values of $l_j$ and does not change the limit. Therefore, the assumptions of \cref{thrm:Finite_MEAT_Chains,lem:kappa_inf-space_lower_bnd,lem:kappa_intermediate_limit_bnd} are satisfied throughout the proof.
    
    Without loss of generality, we will assume 
    \begin{equation}\label{eq:event-support-condition-assumption}
		\Omega\subseteq\supp\left(\bsym{\rho}_{\CP_1^n}\right).
	\end{equation}
    To see why we can make this assumption, consider $\Omega \not \subseteq\supp\left(\bsym{\rho}_{\CP_1^n}\right)$. In this case, we can define $\Omega' \coloneq \Omega \cap \supp\left(\bsym{\rho}_{\CP_1^n}\right)$, which is an event only defined on $\CP^{n}_{1}$ as needed. Because $\Omega'$ only removes sequences $\cP^{n}_{1}$ from $\Omega$ such that that $\Pr_{\rho}[\cP^{n}_{1}] = 0$, we may conclude $p_{\Omega'} = p_{\Omega} > 0$ and $\rho_{\vert \Omega'} = \rho_{\vert \Omega}$. Furthermore, $\Omega' \subseteq\supp\left(\bsym{\rho}_{\CP_1^n}\right)$ and we may define $S_{\Omega'} = S_{\Omega}$. Thus, we may assume $\Omega\subseteq\supp\left(\bsym{\rho}_{\CP_1^n}\right)$.
    
\begin{widetext}
	Starting with the proof, by \cref{eq:Convergence_of_Projected_Channels_entropies,thrm:Finite_MEAT_Chains}, the conditional \Renyi entropy of the final state $\rho$ created by the infinite-dimensional channels acting on the infinite-dimensional state $\omega_{A_0^{n-1}E_0}$ is equal to
\begin{equation}\label{eq:first_step_entr_bound}
\begin{aligned}
    \renyiSandUp_\alpha(S_1^n | \CP_1^n E_n)_{\rho_{|\Omega}} 
    &=\lim_{d \to \infty} \lim_{\substack{\mbf{k} \rightarrow \infty \\ \mbf{l}\rightarrow \infty}} \renyiSandUp_\alpha(S_1^n | \CP_1^n E_n)_{\EATchannProjInf_n^{(k_n, l_n)} \circ \cdots \circ \EATchannProjInf_1^{(k_1, l_1)} \left( \Pi_0^{(d)}(\omega_{A_0^{n-1}E_0}) \right)\big|_\Omega} \\ 
    &=\lim_{d \to \infty} \lim_{\substack{\mbf{k} \rightarrow \infty \\ \mbf{l}\rightarrow \infty}} \renyiSandUp_\alpha(S_1^n | \CP_1^n E_n^{(k_n)})_{\EATchannProj_n^{(k_n, l_n)} \circ \Xi_{n-1} \circ \cdots \circ \EATchannProj_1^{(k_1, l_1)} \circ \Xi_0 (\tau)_{|\Omega}},
\end{aligned}
\end{equation}
where for the second equality we used the isometric invariance of the conditional \Renyi entropy.
\end{widetext}
The state $\tau_{\left(A^{(d)}\right)_{0}^{n-1}E_0^{(d)}}$ now satisfies the following $d$-dependent marginal,
\begin{equation}
    \tau_{\left(A^{(d)}\right)_{0}^{n-1}} = \bigotimes_{j=1}^n \xi_{A_{j-1}^{(d)}}^{(j-1)},
\end{equation}
where each $\xi_{A_{j-1}^{(d)}}^{(j-1)}$ satisfies
\begin{equation}
    V_{j-1}^{(d)} \xi_{A_{j-1}^{(d)}}^{(j-1)} V_{j-1}^{(d)\dagger} = \widebar{\sigma}_{A_{j-1}}^{(j-1)},
\end{equation}
and
\begin{equation}
    \widebar{\sigma}_{A_{j-1}}^{(j-1)} = \pTr{E_0A_{0}^{j-2}A_j^{n-1} }{\Pi_0^{(d)}\left(\omega_{A_0^{n-1}E_0}\right)}.
\end{equation}

Crucially, $\renyiSandUp_\alpha(S_1^n | \CP_1^n E_n^{(k_n)})$ involves only finite-dimensional states and registers for each $d, \mbf{k}$ and  $\mbf{l}$ and we aim to apply all steps until Eq.~(125) of the proof of the original finite-dimensional marginal-constrained entropy accumulation theorem \cite[Proof of Theorem 4.2b]{Arqand_2025} to this conditional \Renyi entropy. 

To this end, we define
\begin{align}
    \hat{\rho} &\defvar \begin{aligned}[t]
    \Big(&\left(\EATchannProj_n^{(k_n,l_n)}\circ\Xi_{n-1}\right)\circ \\ 
    &\cdots \circ\left(\EATchannProj_1^{(k_1,l_1)}\circ\Xi_0\right)\Big)[\tau],
    \end{aligned} \\
	\bar{p}_{\Omega} &\defvar \Pr_{\hat{\rho}}[\Omega], 
\end{align}
as the output state of the projected channel sequence and the probability of the event $\Omega$ under this state, respectively. Due to \cref{lem:Convergence_of_Projected_Channels}, we have
\begin{equation}
	\lim_{d\to\infty}\lim_{\substack{\mbf{k}\to\infty\\\mbf{l}\to\infty}}\bar{p}_{\Omega}= p_\Omega.
\end{equation}
Since $p_\Omega>0$, we may, successively in the order of the iterated limits, restrict each cutoff parameter to sufficiently large values for which $\bar{p}_\Omega>0$. This does not change any of the limits in the equation above nor in \cref{eq:first_step_entr_bound}.

However, to apply the mentioned proof steps until Eq.~(125) of \cite[Proof of Theorem 4.2b]{Arqand_2025} we additionally require that for all $\cP_1^n$ with $\Pr_{\hat{\rho}_{|\Omega}}[\cP_1^n] >0$ it follows
\begin{equation}
    \freq_{\cP_1^n}\in S_\Omega.
\end{equation}
So far, we did not establish this property. Since for all $\cP_1^n$ with $\Pr_{\hat{\rho}_{|\Omega}}[\cP_1^n] >0$ it holds $\cP_1^n \in \Omega$ and by \cref{eq:event-support-condition-assumption}, we can conclude
\begin{equation}
    \Pr_{\rho_{|\Omega}}[\cP_1^n] >0,
\end{equation}
where $\rho$ is the infinite dimensional state. Now, we can apply the assumption that for all $\cP_1^n$ with $\Pr_{\rho_{|\Omega}}[\cP_1^n] >0$, we have $\freq_{\cP_1^n}\in S_\Omega$. Hence, we established that for all $\cP_1^n$ with $\Pr_{\hat{\rho}_{|\Omega}}[\cP_1^n] >0$ it follows
\begin{equation}
    \freq_{\cP_1^n}\in S_\Omega.
\end{equation}

Therefore, we indeed satisfy the requirements to apply all steps until Eq.~(125) of the proof of the original finite-dimensional marginal-constrained entropy accumulation theorem \cite[Proof of Theorem 4.2b]{Arqand_2025} to the conditional \Renyi entropy $\renyiSandUp_\alpha(S_1^n | \CP_1^n E_n^{(k_n)})$. We find
\begin{widetext}
\begin{equation}
\begin{aligned}
    &\renyiSandUp_\alpha(S_1^n | \CP_1^n E_n^{(k_n)})_{\EATchannProj_n^{(k_n, l_n)} \circ \Xi_{n-1} \circ \cdots \circ \EATchannProj_1^{(k_1, l_1)} \circ \Xi_0 (\tau)_{|\Omega}} \\
    &\geq \left(\sum_{j=1}^n  \inf_{\nu\in\widebar{\Sigma}_j(k_{j-1},k_j,l_j,d)} H^{\uparrow,f}_{\alpha} (S_j|\CP_jE_j^{(k_j)}R)_\nu \right) + n \inf_{\mbf{q} \in S_\Omega} \mbf{f}\cdot\mbf{q} - \frac{\alpha}{\alpha-1} \log\frac{1}{\bar{p}_\Omega},
\end{aligned}
\end{equation}
where $\widebar{\Sigma}_j(k_{j-1},k_j,l_j,d)$ is defined as in \cref{eq:def_Sigma_bar}
\begin{align}
\widebar{\Sigma}_j(k_{j-1},k_j,l_j,d) \defvar
    \Big\{ \left( \left(\EATchannProj_j^{(k_j, l_j)} \circ \Xi_{j-1} \right) \otimes \id_R \right) [\tau] \mid \tau \in S_{=}(A_{j-1}^{(d)} E_{j-1}^{(k_{j-1})} R), \tau_{A_{j-1}^{(d)}} = \xi_{A_{j-1}^{(d)}}^{(j-1)} \Big\}.
\end{align}

For every fixed admissible choice of $d,\mbf{k},\mbf{l}$ satisfying $l_j\geq d$ and $l_j\geq k_{j-1}$ for every $j=1,\dots,n$, and with $f_{|\cP_1^{j-1}}=f$, \cref{lem:kappa_inf-space_lower_bnd,lem:kappa_intermediate_limit_bnd} give

\begin{equation}
    \inf_{\nu\in\widebar{\Sigma}_j(k_{j-1},k_j,l_j,d)}
	\frenyiSandUp_{\alpha}(S_j|\CP_jE_j^{(k_j)}R)_\nu \geq \inf_{\nu\in\hat{\Sigma}_j(d)} \frenyiSandUp_{\alpha}(S_j|\CP_jE_j\widetilde E)_\nu.
\end{equation}

Thus, the finite-dimensional bound implies
\begin{equation}
\begin{aligned}
	&\renyiSandUp_\alpha(S_1^n|\CP_1^nE_n^{(k_n)})_{\EATchannProj_n^{(k_n,l_n)}\circ\Xi_{n-1}\circ\cdots\circ\EATchannProj_1^{(k_1,l_1)}\circ\Xi_0(\tau)_{|\Omega}} \geq
	\sum_{j=1}^n\inf_{\nu\in\hat{\Sigma}_j(d)}
	H^{\uparrow,f}_{\alpha}(S_j|\CP_jE_j\widetilde E)_\nu
	+n\inf_{\mbf{q}\in S_\Omega}\mbf{f}\cdot\mbf{q}
	-\frac{\alpha}{\alpha-1}\log\frac{1}{\bar p_\Omega}.
\end{aligned}
\end{equation}
Then, taking the cutoff limits $d,\mbf{k},\mbf{l} \to \infty$, using the convergence of $\bar{p}_\Omega$, and applying \cref{lem:kappa_marginal_limit_bnd}, we obtain

\begin{align}\label{eq:entropy-bound-that-depends-on-rhoOmega-and-pOmega}
	\renyiSandUp_\alpha(S_1^n|\CP_1^nE_n)_{\rho_{|\Omega}}
	&\geq \sum_{j=1}^n\inf_{\nu\in\Sigma_j} H^{\uparrow,f}_{\alpha}(S_j|\CP_jE_j\widetilde E)_\nu + n\inf_{\mbf{q}\in S_\Omega}\mbf{f}\cdot\mbf{q}
	-\frac{\alpha}{\alpha-1}\log\frac{1}{p_\Omega} \ , 
\end{align}
where note that all that changed on the right hand side was replacing the set $\hat{\Sigma}_j(d)$ with $\Sigma_j$ and $\bar{p}_\Omega$ with $p_\Omega$. 

Finally, we can apply the last step in the proof \cite[Theorem 4.2b]{Arqand_2025} to reach the second line of Eq.~(125) which relates the optimization over $\Sigma_j$ to the marginal-constrained convex range $\Sigma_{S\CP E\widetilde{E}}$. This gives us our first intermediate result
\begin{align}
\renyiSandUp_\alpha(S_1^n | \CP_1^n E_n)_{\EATchann_n \circ \cdots \circ \EATchann_1 (\omega_{A_0^{n-1}E_0})_{|\Omega}}
&\geq n \inf_{\mbf{q} \in S_\Omega}
\inf_{\nu\in\Sigma_{S\CP E\widetilde{E}}}
\left(H^{\uparrow, f}_{\alpha}(S | \CP E\widetilde{E})_{\nu} + \mbf{f}\cdot\mbf{q} \right)
- \frac{\alpha}{\alpha-1} \log\frac{1}{p_\Omega}.
\end{align}

Now, let $\Theta_{AJE}$, $\EATchann$, $\pf$, and the family $\{\nu^\omega\}_{\omega\in\Theta_{AJE}}$ be defined as in \cref{lem:finite_feasible_point}. By the same construction as in \cite[Proof of Theorem~4.2b]{Arqand_2025}, the marginal-constrained convex range can be reparametrized in terms of $\EATchann$ and $\Theta_{AJE}$. Hence,
\begin{equation}
    \inf_{\mbf{q} \in S_\Omega}
\inf_{\nu\in\Sigma_{S\CP E\widetilde{E}}}
\left(H^{\uparrow, f}_{\alpha}(S | \CP E\widetilde{E})_{\nu} + \mbf{f}\cdot\mbf{q} \right) = \inf_{\mbf{q} \in S_\Omega}
\inf_{\substack{\omega \in \dop{=}(AEJ\widetilde{E}) \\ \text{s.t. } \omega_{AEJ} \in \Theta_{AJE}}}
\left(H^{\uparrow, f}_{\alpha}(S | \CP E\widetilde{E})_{\EATchann\left[\omega\right]} + \mbf{f}\cdot\mbf{q} \right).
\end{equation}

Next, by the same purification argument as in \cite[Proof of Theorem~4.2b]{Arqand_2025}, using \cite[Lemma~4.10]{Arqand_2025}, restricting the optimization to purifications does not change the value of the infimum. The same argument applies to separable Hilbert spaces because every extension can be obtained from a purification by applying a channel to the purifying register. By the data-processing inequality for the $f$-weighted \Renyi entropy, this channel cannot decrease the entropy, while every purification is itself a feasible extension. Thus,
\begin{equation}
    \inf_{\mbf{q} \in S_\Omega}
\inf_{\nu\in\Sigma_{S\CP E\widetilde{E}}}
\left(H^{\uparrow, f}_{\alpha}(S | \CP E\widetilde{E})_{\nu} + \mbf{f}\cdot\mbf{q} \right) = \inf_{\mbf{q} \in S_\Omega}
\inf_{\substack{\omega_{AEJ} \in \Theta_{AJE}}}
\left(H^{\uparrow, f}_{\alpha}(S | \CP E\widetilde{E})_{\nu^{\omega}} + \mbf{f}\cdot\mbf{q} \right),
\end{equation}
where $\nu^\omega \defvar \EATchann\left[\pf(\omega)\right]$ is the state defined in \cref{lem:finite_feasible_point}.
\end{widetext}

This statement holds for any tradeoff function $\mbf{f}$ and thus, we can take the supremum over it, resulting in exactly the statement of \cref{thrm:convduality} with $\Theta=\Theta_{AJE}$, if all conditions are satisfied.

Therefore, we now check that the conditions of \cref{thrm:convduality} are indeed satisfied. Since $\alphCP$ is finite, the set of probability distributions $\mathbb{P}_{\CP}$ is compact. Hence, $S_\Omega$ is compact because it is closed in $\mathbb{P}_{\CP}$. Moreover, $S_\Omega$ is nonempty because $p_\Omega>0$ and every frequency distribution occurring with nonzero probability in $\rho_{|\Omega}$ belongs to $S_\Omega$. The set $\Theta_{AJE}$ is convex by construction, and \cref{lem:finite_feasible_point} shows that the primal problem admits a finite feasible point.

Hence, all conditions of \cref{thrm:convduality} are satisfied, and the theorem statement follows after imposing the constraint $ \mbf{q} - \bsym{\lambda} = \mbf{0}$ and reverting the reparametrization from $\left(\EATchann,\Theta_{AJE}\right)$ back to $\Sigma_{S\CP E\widetilde{E}}$.
\end{proof}

%--------------------------------------------------------
%               Remarks on Navarro et. al. 
%--------------------------------------------------------

\section{Remarks on the security argument of Ref.\texorpdfstring{~\cite{Navarro_2026a}}{ 33}}\label{apdx:Remarks_Paper}
In this Appendix, we briefly comment on existing issues in Ref. \cite{Navarro_2026a} as well as on how these issues relate and are resolved in this work.

\subsection{Setting \texorpdfstring{$\mbf{q} = \widetilde{\mathcal{M}}^K(\omega_{AB})$}{q=M(w)}}
\textbf{Issue:} In both finite and infinite dimensions, the infimum in
\begin{equation}
    \inf_{\mbf{q}} \inf_{\omega} \frac{\alpha}{\alpha-1} D\left( \mbf{q} || \omega_c \right) + q(\perp) \renyiSandUp_{\alpha}(Z|E)_{\omega |\perp}
\end{equation}
is not generally achieved for $\mbf{q} = \widetilde{\mathcal{M}}^K(\omega_{AB})_C$. Instead, setting $\mbf{q} = \widetilde{\mathcal{M}}^K(\omega_{AB})$ removes feasible $(\mbf{q},\omega_{AB})$ pairs that might achieve a smaller objective, thereby increasing the infimum (i.e., yielding an upper bound on the key rate). To make this step sound, one must either keep the KL term, or compensate conservatively, which might be loose. For the QPSK protocol, one can even construct explicit counterexamples. 

\textbf{Relation to our work:} 
In our work, we do not need this step because we estimate the weight $W$ differently. Notably, we combine weight estimation with the key-rate calculation, which completely avoids this issue for fixed-length protocols. Moreover, to the best of our knowledge, this is the only viable solution allowing for an extension to variable-length key rates.

% \subsection{\texorpdfstring{$\widetilde{\mathcal{M}}^K$}{M} is not linear}
% \textbf{Issue:} The channel defined in Eq. (35) of \cite{Navarro_2026a} is not linear, not a channel, and thus it does not satisfy the data processing inequality. Consequently, Eq. (70) bounding the trace distance of $\widetilde{\mathcal{M}}^K$ applied to the full and the projected state, is invalid. This bound, however, is essential for the dimension reduction argument to hold.

% \textbf{Relation to our work:} 
% Our work avoids this by using a normalised state with a flag space. We extend the channel $\mathcal{M}^{\textrm{QKD}}$ to $\tilde{\mathcal{M}}^{\textrm{QKD}}$ such that it discards the flag space. Then, we proceed with our argument by using strong subadditivity, concluding that the \Renyi entropy of $\tilde{\mathcal{M}}^{\textrm{QKD}}$ applied to the tagged state lower bounds the original expression.

\subsection{Infinite-dimensions treated incorrectly}
We split this issue into three separate points:

\begin{enumerate}
    \item[a)]  \textbf{Eve's system cannot be assumed to be finite-dimensional.}\\
\textbf{Issue:}
 Bob's heterodyne POVM inherently lives in an infinite-dimensional Hilbert space. Assuming Eve has all power allowed by quantum mechanics, we need to consider the situation where Eve holds a purification of the state $\omega_{AB}$ that leaves Alice's lab, which is infinite-dimensional. Consequently, the purifying register must have infinite dimensions as well.

\textbf{Relation to our work: }
We never assume Eve to be finite-dimensional.

\item[b)] \textbf{ The proof of the $h_{\alpha}^{\uparrow}$ bound on the \Renyi entropy $\renyiSandUp_{\alpha}$ is incorrect.}\\
\textbf{Issue:} The goal of the specific version of MEAT applied in Ref.~\cite{Navarro_2026a} is to show
\begin{equation}
    \renyiSandUp_{\alpha}(S_1^n | \hat{C}_1^n E_n)_{\mathcal{M}^{\otimes n}(\rho) |\Omega} \geq h_{\alpha}^{\uparrow} - \frac{\alpha}{\alpha-1} \log\left(\frac{1}{p_{\Omega}}\right),
\end{equation}
where $\Omega$ is an event on $\hat{C}_1^n$, $\rho_{A_1^nB_1^nE}$ satisfies the marginal $\rho_{A_0^{n-1}} = \bigotimes_{j=1}^{n} \sigma_A$ and $h_{\alpha}^{\uparrow}$ is an optimization problem involving only single-round quantities. Following the argument in Appendix B of Ref.~\cite{Navarro_2026a} to obtain an infinite dimensional generalization, we find
\begin{equation}
\begin{aligned}
    &\renyiSandUp_{\alpha}(S_1^n | \hat{C}_1^n E_n)_{\mathcal{M}^{\otimes n} | \Omega} \\
    & \geq \lim_{d \rightarrow \infty} \left( \renyiSandUp_{\alpha}(S_1^n | \hat{C}_1^n E_n)_{\mathcal{M}(\rho)_d^{\otimes n} | \Omega }  - \Delta_{\alpha}(\epsilon_d) \right)\\
    &= \lim_{d \rightarrow \infty} \renyiSandUp_{\alpha}(S_1^n | \hat{C}_1^n E_n)_{\mathcal{M}(\rho)_d^{\otimes n} | \Omega }
\end{aligned}
\end{equation}
since $\epsilon_d \rightarrow 0$.

Now, applying the MEAT naively as described in Appendix. B of Ref.~\cite{Navarro_2026} to the \Renyi entropy would give
\begin{align}
\begin{aligned}
    &\lim_{d \rightarrow \infty} \renyiSandUp_{\alpha}(S_1^n | \hat{C}_1^n E_n)_{\mathcal{M}_d^{\otimes n}|\Omega } \\ \geq &\lim_{d\rightarrow \infty} \left( h_{\alpha}^{\uparrow} - \log\left(\frac{1}{\tilde{p}_{\Omega}(d)}\right). \right),
    \end{aligned}
\end{align}
where $h_{\alpha}^{\uparrow}(d)$ is now a dimension dependent version of $h_{\alpha}^{\uparrow}$ from the original version of the MEAT given by
\begin{align}
    h_{\alpha}^{\uparrow}(d) = \begin{aligned}[t]
    \inf_{\mbf{q} \in S_{\Omega}} \inf_{\omega \in \Sigma_d} &\frac{\alpha}{\alpha-1} D\left( \mbf{q} \mid \omega_{\CP} \right) \\ 
    &+ \sum_{\hat{c}} q(\hat{c}) \renyiSandUp_{\alpha}(S|E\tilde{E})_{\omega_{|\hat{c}}},
    \end{aligned}
\end{align}
where $\Sigma_d = \{\mathcal{M}_d(\rho) ~|~ \rho \in S_{=}(AE\tilde{E}), \rho_A = \sigma_A\}$.

Crucially, the function in the objective is not uniformly continuous. Thus, the limit and the infimum cannot be swapped without further justification, that is
\begin{equation}
    h_{\alpha}^{\uparrow}(d=\infty) \neq \begin{aligned}[t]
    \inf_{\mbf{q} \in S_{\Omega}} \inf_{\omega \in \Sigma_{\infty}} &\frac{\alpha}{\alpha-1} D\left( \mbf{q} \mid\mid \omega_{\CP} \right) \\ 
    &+ \sum_{\hat{c}} q(\hat{c}) \renyiSandUp_{\alpha}(S|E\tilde{E})_{\omega_{|\hat{c}}},
    \end{aligned}
\end{equation}
Yet, the quantity with the infimum and limit swapped containing the infinite dimensional set $\Sigma_{\infty}$ is the basis of the full security proof in the paper.

\textbf{Relation to our work: }
We circumvent this issue by explicitly proving an infinite-dimensional marginal-constrained entropy accumulation theorem (iMEAT). Crucially, our proof never tries to make this argument, because, to the best of our knowledge, it cannot be shown in general.

\item[c)] Appeal to lift in Ref [51, Appendix A] (reference numbers as in \cite{Navarro_2026a})\\
\textbf{Issue:}
The authors argue that one could alternatively use the argument of Ref. [51, Appendix A] (Ref.~\cite{Tupkary2026} of our work), to lift the security from finite dimensions to infinite dimensions. However, this effectively forces on to use the key lengths obtained via $\liminf_{d\to \infty} h_{\alpha}^{\uparrow}(d)$ instead of $h_{\alpha}^{\uparrow}(d=\infty)$ unless Eve's dimension can be assumed to be finite. Hence, one recovers a similar issue as outlined in b).

Here, we note that we made a slightly informal conversion from $\epsilon$-parameters to key lengths, as Ref. [51, Appendix A] argument is stated in terms of security parameters.

\textbf{Relation to our work:}
Again, we circumvent this issue by proving an infinite dimensional generalisation of the MEAT.

\end{enumerate}
In summary, two main issues remain. First, the dimension reduction as presented does not provide a correct bound on the \Renyi entropy. Second, the sketched generalisation to an infinite-dimensional MEAT is unjustified. We resolve both in our work and, at the same time, provide greater scope.

\end{document}